\documentclass[11pt,reqno]{amsart}
\title{Long-range order in discrete spin systems}
\date{\today}

\author{Ron Peled}
\address{Ron Peled\hfill\break
	Department of Mathematics\\
	University of Maryland\\
	College Park, MD, USA.}
\email{peledron@umd.edu}
\urladdr{http://sites.google.com/umd.edu/peledron/}

\author{Yinon Spinka}
\address{Yinon Spinka\hfill\break
	School of Mathematical Sciences\\
	Tel Aviv University\\
	Tel Aviv, 69978, Israel.}
\email{yinonspi@tauex.tau.ac.il}
\urladdr{http://www.math.tau.ac.il/~yinonspi/}

\usepackage{amsmath, amsthm, amsopn, amssymb, graphicx, enumerate, enumitem, geometry, afterpage, caption, subcaption, color, textcomp, stmaryrd, bbm, tikz, etoolbox, thmtools}
\usepackage[english]{babel}
\usetikzlibrary{arrows,automata,shapes}
\usepackage[colorlinks=true,urlcolor=blue,pdfborder={0 0 0}]{hyperref}
\hypersetup{linkcolor=[rgb]{0,0,0.6}}
\hypersetup{citecolor=[rgb]{0,0.6,0}}

\usepackage[nameinlink]{cleveref}
  \crefname{theorem}{Theorem}{Theorems}
  \crefname{thm}{Theorem}{Theorems}
  \crefname{mainthm}{Theorem}{Theorems}
  \crefname{lemma}{Lemma}{Lemmas}
  \crefname{lem}{Lemma}{Lemmas}
  \crefname{remark}{Remark}{Remarks}
  \crefname{prop}{Proposition}{Propositions}
  \crefname{defn}{Definition}{Definitions}
  \crefname{corollary}{Corollary}{Corollaries}
  \crefname{cor}{Corollary}{Corollaries}
  \crefname{section}{Section}{Sections}
  \crefname{figure}{Figure}{Figures}
  \crefname{cond}{Condition}{Conditions}

\newtheorem{thm}{Theorem}[section]
\newtheorem{lemma}[thm]{Lemma}
\newtheorem{prop}[thm]{Proposition}

\newtheorem{cor}[thm]{Corollary}

\newtheorem{cond}[thm]{Condition}
\theoremstyle{definition}
\newtheorem*{remark}{Remark}

\makeatletter
\def\subsubsection{\@startsection{subsubsection}{3}%
  \z@{.5\linespacing\@plus.7\linespacing}{-.5em}%
  {\normalfont\bfseries}}
\makeatother

\newcommand{\fq}{\mathfrak{q}}
\newcommand{\cA}{\mathcal{A}}

\newcommand{\cP}{\mathcal{P}}
\newcommand{\cC}{\mathcal{C}}
\newcommand{\cH}{\mathcal{H}}

\newcommand{\cB}{\mathcal{B}}
\newcommand{\cR}{\mathcal{R}}
\newcommand{\cQ}{\mathcal{Q}}

\newcommand{\cF}{\mathcal{F}}
\newcommand{\cS}{\mathcal{S}}
\newcommand{\cG}{\mathcal{G}}

\newcommand{\cX}{\mathcal{X}}
\newcommand{\cU}{\mathcal{U}}
\newcommand{\cV}{\mathcal{V}}
\newcommand{\cI}{\mathcal{I}}
\newcommand{\cW}{\mathcal{W}}
\newcommand{\N}{\mathbb{N}}
\newcommand{\R}{\mathbb{R}}
\newcommand{\Z}{\mathbb{Z}}
\renewcommand{\SS}{\mathbb{S}}
\newcommand{\HH}{\mathbb{H}}
\newcommand{\T}{\mathbb{T}}
\newcommand{\E}{\mathbb{E}}

\renewcommand{\Pr}{\mathbb{P}}
\newcommand{\cE}{\mathcal{E}}
\newcommand{\1}{\mathbf{1}}

\makeatletter
\DeclareRobustCommand{\cev}[1]{%
	\mathpalette\do@cev{#1}%
}
\newcommand{\do@cev}[2]{%
	\fix@cev{#1}{+}%
	\reflectbox{$\m@th#1\vec{\reflectbox{$\fix@cev{#1}{-}\m@th#1#2\fix@cev{#1}{+}$}}$}%
	\fix@cev{#1}{-}%
}
\newcommand{\fix@cev}[2]{%
	\ifx#1\displaystyle
	\mkern#23mu
	\else
	\ifx#1\textstyle
	\mkern#23mu
	\else
	\ifx#1\scriptstyle
	\mkern#22mu
	\else
	\mkern#22mu
	\fi
	\fi
	\fi
}
\makeatother

\newcommand{\dpartial}{\vec{\partial}}
\newcommand{\dpartialrev}{\cev{\partial}}
\newcommand{\intB}{\partial_{\bullet}}
\newcommand{\extB}{\partial_{\circ}}
\newcommand{\intextB}{\partial_{\ins \out}}
\newcommand{\ins}{\bullet}
\newcommand{\out}{\circ}
\newcommand{\up}{\,\uparrow\,}
\newcommand{\down}{\,\downarrow\,}

\newcommand{\Even}{\mathrm{Even}}
\newcommand{\Odd}{\mathrm{Odd}}
\DeclareMathOperator\dist{dist}
\DeclareMathOperator\diam{diam}

\newcommand{\distTV}{\mathrm{d_{TV}}}
\DeclareMathOperator\Int{int}

\newcommand{\phase}{\mathcal{P}}
\newcommand{\phasedom}{\mathcal{P}}
\newcommand{\phasemax}{\mathcal{P}_{\mathsf{max}}}
\newcommand{\Hom}{\mathsf{Hom}}
\newcommand{\bad}{\mathsf{none}}
\newcommand{\unbal}{\mathsf{unbal}}
\newcommand{\overlap}{\mathsf{overlap}}
\newcommand{\Ent}{\mathsf{Ent}}
\DeclareMathOperator\supp{supp}

\newcommand{\rest}{\mathsf{rest}}
\newcommand{\highlyrest}{\mathsf{high}}
\newcommand{\nondom}{\mathsf{nondom}}
\newcommand{\unique}{\mathsf{uniq}}
\newcommand{\iso}{\mathsf{iso}}
\newcommand{\even}{\mathsf{even}}
\newcommand{\odd}{\mathsf{odd}}
\newcommand{\bdry}{\mathsf{bdry}}
\newcommand{\inner}{\mathsf{int}}
\newcommand{\hole}{\mathsf{defect}}
\newcommand{\breakups}{\cX}

\makeatletter
\newcommand*\rel@kern[1]{\kern#1\dimexpr\macc@kerna}
\newcommand*\widebar[1]{%
  \begingroup
  \def\mathaccent##1##2{%
    \rel@kern{0.8}%
    \overline{\rel@kern{-0.8}\macc@nucleus\rel@kern{0.2}}%
    \rel@kern{-0.2}%
  }%
  \macc@depth\@ne
  \let\math@bgroup\@empty \let\math@egroup\macc@set@skewchar
  \mathsurround\z@ \frozen@everymath{\mathgroup\macc@group\relax}%
  \macc@set@skewchar\relax
  \let\mathaccentV\macc@nested@a
  \macc@nested@a\relax111{#1}%
  \endgroup
}
\makeatother

\begin{document}
\begin{abstract}
We establish long-range order for discrete nearest-neighbor spin systems on $\Z^d$ satisfying a certain symmetry assumption, when the dimension $d$ is higher than an explicitly described threshold. The results characterize all periodic, maximal-pressure Gibbs states of the system. The results further apply in low dimensions provided that the lattice $\Z^d$ is replaced by $\Z^{d_1}\times\T^{d_2}$ with $d_1\ge 2$ and $d=d_1+d_2$ sufficiently high, where $\T$ is a cycle of even length. Applications to specific systems are discussed in detail and models for which new results are provided include the antiferromagnetic Potts model, Lipschitz height functions, and the hard-core, Widom--Rowlinson and beach models and their multi-type extensions. We also establish a formula conjectured by Jenssen and Keevash for the topological pressure in the high-dimensional limit.
\end{abstract}

\maketitle

\section{Introduction}
\label{sec:intro}
The spin systems studied in statistical physics may exhibit a variety of behaviors. A prototypical example is the Ising model of a ferromagnet, which is disordered at high temperatures, displays critical behavior at a precise critical temperature and transitions to an ordered, spontaneously magnetized state at lower temperatures. Similar behavior has been observed in many models of statistical mechanics. A main goal of statistical physics is to determine, for a given set of parameters of the model, the possible states of the system.

Dobrushin~\cite{Dobrushin1968TheDe} found an explicit condition guaranteeing that the system has a unique Gibbs state. The condition is fulfilled when the spin at each vertex is ``not too sensitive'' to the values of spins at other vertices. Thus Dobrushin's uniqueness condition can be seen as a condition guaranteeing the absence of long-range order in the system. One naturally seeks also a complementary condition, guaranteeing long-range order in the system. Significant progress in this direction was made by the pioneering work of Pirogov--Sinai~\cite{pirogov1975phase,pirogov1976phase} (see~\cite[Chapter 7]{friedli2017statistical} for a pedagogic introduction). There, conditions are given for the \emph{ground states} of the system, its various orderings at zero temperature, to describe the low-temperature regime. While the work of Pirogov--Sinai and its various continuations is quite extensive, it does not apply to systems for which even the ground states are not understood. This is typically the case for systems with \emph{residual entropy} -- systems in which the number of ground state configurations grows exponentially with the volume, such as, for example, the antiferromagnetic Potts model with at least $3$ states. For the theory of Pirogov and Sinai to apply, one further needs to verify a so-called Peierls condition, whose verification requires case-by-case considerations and may require significant effort in certain cases.

In this work, we study discrete spin systems on the $\Z^d$ lattice having nearest-neighbor isotropic interactions and satisfying a certain symmetry assumption, with our main result being an explicit quantitative condition guaranteeing the existence of long-range order in the spin system. When the condition is fulfilled, we further classify all possibilities for the emergent ordering, in the form of all periodic, maximal-pressure Gibbs states of the spin system. These states are necessarily of the following form: In a sampled configuration $f$ there will be two subsets $A$ and $B$ of the spin space, satisfying that the interaction weight between every spin in $A$ and every spin in $B$ is maximal, such that $f$ mostly takes values in $A$ on the even sublattice of $\Z^d$ and mostly takes values in $B$ on the odd sublattice of $\Z^d$. Our quantitative condition quantifies notions such as ``low temperature'' and ``large entropic gap between different possibilities for the long-range order'' and provides new results for many spin systems, including especially those with residual entropy. For the discrete spin systems under study, the condition is always satisfied in sufficiently high dimensions, and is satisfied also in low dimensions provided the $\Z^d$ lattice is replaced by an enhanced version of it, the lattice $\Z^{d_1}\times\T^{d_2}$ with $d_1\ge 2$ and $d=d_1+d_2$ sufficiently high, where $\T$ is a cycle of even length. This work continues our previous study of proper colorings of $\Z^d$ addressed in the companion paper~\cite{peledspinka2018colorings}.

\subsection{The model}\label{sec:the model}
We study the possible types of emergent long-range order in \emph{discrete spin systems} on $\Z^d$. The spin systems considered are described by a finite \emph{spin space} $\SS$, a collection $(\lambda_i)_{i \in \SS}$ of positive numbers called the \emph{single-site activities}, and a collection $(\lambda_{i,j})_{i,j \in \SS}$ of non-negative numbers called the \emph{pair interactions}. The pair interactions are symmetric, i.e., $\lambda_{i,j}=\lambda_{j,i}$ for all $i,j \in \SS$, and at least one is positive. The probability of a \emph{configuration} $f \colon \Lambda \to \SS$ is proportional to
\begin{equation}\label{eq:config-weight}
\omega_f :=  \prod_{v \in \Lambda} \lambda_{f(v)} \prod_{\{u,v\} \in E(\Lambda)} \lambda_{f(u),f(v)} ,
\end{equation}
where $\Lambda$ is a finite subset of $\Z^d$ and $E(\Lambda)$ is the set of edges of $\Z^d$ whose two endpoints belong to~$\Lambda$. Hard constraints arise when some of the $\lambda_{i,j}$ are zero. In fact, in many models of interest the $\lambda_{i,j}$ take values in $\{0,1\}$; Such models are termed (weighted) \emph{homomorphism models} as the configurations with positive probability are the graph homomorphisms from $\Lambda$ to the graph with vertex set $\SS$ and edge set $\{\{i,j\} : \lambda_{i,j}=1\}$. Classical models obtained as special cases of the above framework include the Ising, Potts, proper coloring, hard-core, Widom--Rowlinson, clock and beach models; see \cref{sec:first applications} and \cref{sec:applications} for more details.

The emergent long-range order will involve spins interacting with the maximal pair interaction weight. Thus we define a \emph{pattern} as a pair $(A,B)$ of subsets of $\SS$ such that
\begin{equation}
  \lambda_{a,b}=\max_{i,j \in \SS} \lambda_{i,j}\qquad\text{for all $a \in A$ and $b \in B$}.\label{eq:pattern def}
\end{equation}
The single-site activities then play a role in singling out \emph{dominant} patterns, defined as patterns maximizing $(\sum_{a\in A} \lambda_a)(\sum_{b\in B} \lambda_b)$ among all patterns. Two patterns $(A,B)$ and $(A',B')$ are called \emph{equivalent} if there is a bijection $\varphi \colon \SS \to \SS$ such that
\begin{equation}\label{eq:equivalent-patterns-def}
\{\varphi(A),\varphi(B)\}=\{A',B'\},\qquad \lambda_{\varphi(i)}=\lambda_i,\qquad \lambda_{\varphi(i),\varphi(j)}=\lambda_{i,j}\qquad\text{for all }i,j \in \SS.
\end{equation}
We emphasize that if $(A,B)$ is a pattern with $A\neq B$, then $(A,B)$ and $(B,A)$ are two equivalent, albeit distinct, patterns.

Our results apply to spin systems in which all dominant patterns are equivalent. We describe these results in the next section (with quantitative refinements presented in \cref{sec:quantitative-results}) and then provide examples and first applications in \cref{sec:first applications}.

\subsection{Qualitative results}\label{sec:non-quantitative results}
\label{sec:intro-results}
We begin with required notation.
A vertex of $\Z^d$ is called {\em odd (even)} if the sum of its coordinates is odd (even). A \emph{domain} is a non-empty finite $\Lambda \subset \Z^d$ such that both $\Lambda$ and $\Z^d\setminus\Lambda$ are connected, and its \emph{internal vertex-boundary}, denoted $\intB \Lambda$, is the set of vertices in $\Lambda$ that are adjacent to a vertex outside $\Lambda$.
Given a configuration $f \colon \Lambda \to \SS$ and a vertex $v \in \Lambda$, we say that
\[ \text{$v$ is \emph{in the $(A,B)$-pattern} if either $v$ is even and $f(v) \in A$, or $v$ is odd and $f(v) \in B$}.\]
We also say that a set of vertices is in the $(A,B)$-pattern if all its elements are such.
Let $\Omega_{\Lambda,(A,B)}$ be the set of configurations $f \in \SS^\Lambda$ satisfying that $\intB \Lambda$ is in the $(A,B)$-pattern.
Let $\Pr_{\Lambda,(A,B)}$ be the probability measure on $\Omega_{\Lambda,(A,B)}$ in which the probability of a configuration $f$ is proportional to $\omega_f$ defined in~\eqref{eq:config-weight}.

\begin{thm}\label{thm:long-range-order-NQ}
Fix a spin system in which all dominant patterns are equivalent. There exists a function $\epsilon(d)$ such that $\epsilon(d) \to 0$ as $d \to \infty$ and such that, for any dominant pattern $(A,B)$, any domain $\Lambda$ and any vertex $v \in \Lambda$,
\[ \Pr_{\Lambda,(A,B)}\big(v\text{ is not in the $(A,B)$-pattern}\big)\le \epsilon(d) .\]
\end{thm}

The theorem establishes the existence of long-range order in high dimensions, as the effect of the imposed boundary conditions on the distribution of $f(v)$ does not vanish in the limit as the domain $\Lambda$ increases to the whole of $\Z^d$. The statement in the theorem quantifies the probability of a single-site deviation from the boundary pattern. Extensions to larger spatial deviations are provided in \cref{sec:large-violations}.

It is natural to wonder whether other restrictions on the boundary values besides the one used in \cref{thm:long-range-order-NQ} would lead to other behaviors of the configuration in the bulk of the domain. This idea is captured by the notion of a Gibbs state: a probability measure on configurations $f \in \SS^{\Z^d}$ for which the conditional distribution on any finite set, given the configuration outside the set, is given by the weights in~\eqref{eq:config-weight} (see \cref{sec:gibbs} for a precise definition).
A fundamental problem in statistical physics is to understand the set of Gibbs states corresponding to a given model. In many models, including those considered here, it is evident that there is at least one Gibbs state and the next question arising is to ascertain whether there is more than one. Dobrushin gave a fundamental criterion for uniqueness of Gibbs states~\cite{Dobrushin1968TheDe}, which, when applied to the models here, may be translated to a condition on $\SS$, $(\lambda_i)$ and $(\lambda_{i,j})$.

In the opposite direction, results showing multiplicity of Gibbs states are in general more difficult to obtain. For some models, this question may be trivial to answer due to the existence of ``frozen Gibbs states'' -- measures supported on a single admissible configuration $f$ (i.e., having the property that $\lambda_{f(u),f(v)}>0$ for any edge $\{u,v\}$ of $\Z^d$) satisfying that~$f$ cannot be modified on any finite set while preserving its admissibility -- for example, such configurations exist in the proper $q$-coloring model when $q\le d+1$~\cite{alon2019mixing}. To avoid this (and similar) degenerate situation, one often restricts consideration to Gibbs states of \emph{maximal pressure} -- Gibbs states which are \emph{periodic}, i.e., invariant under translations by a full-rank sublattice of $\Z^d$, and whose pressure equals the topological pressure of the spin system (see \cref{sec:equilibrium-states}) -- and the challenge is then to determine whether there is more than one such measure. A concrete question is to determine whether multiple Gibbs states of maximal pressure exist for a given spin system. In fact, \cref{thm:long-range-order-NQ} immediately implies the existence of multiple Gibbs states when the dimension $d$ is sufficiently large, one for each dominant pattern $(A,B)$, and it is not overly difficult
to establish that these have maximal pressure. This fact, along with additional properties, constitutes our second main result. In the result, $(\Z^d)^{\otimes 2}$ denotes the graph on $\Z^d$ in which two vertices are adjacent if their distance in $\Z^d$ is $1$ or $2$.

\begin{thm}\label{thm:existence_Gibbs_states-NQ}
Fix a spin system in which all dominant patterns are equivalent. There exists $d_0$ such that in any dimension $d \ge d_0$, there exists a distinct extremal (periodic) maximal-pressure Gibbs state $\mu_{(A,B)}$ for each dominant pattern $(A,B)$. Moreover, for any sequence of domains $\Lambda_n$ increasing to~$\Z^d$, the measures $\Pr_{\Lambda_n,(A,B)}$ converge weakly to $\mu_{(A,B)}$ as $n \to \infty$. In particular, $\mu_{(A,B)}$ is invariant to automorphisms of $\Z^d$ preserving the two sublattices. Moreover $\mu_{(A,B)}$ is supported on the set of configurations with an infinite connected component of vertices in the $(A,B)$-pattern, whose complement has only finite $(\Z^d)^{\otimes 2}$-connected components.
\end{thm}

\cref{thm:existence_Gibbs_states-NQ} shows that there are at least as many extremal Gibbs states of maximal pressure as there are dominant patterns. Our third result shows that these exhaust all possibilities.

\begin{thm}\label{thm:characterization_of_Gibbs_states-NQ}
Fix a spin system in which all dominant patterns are equivalent. There exists $d_0$ such that in any dimension $d \ge d_0$, every (periodic) maximal-pressure Gibbs state is a mixture of the measures $\{\mu_{(A,B)}\}$.
\end{thm}

The above results are qualitative in the sense that the dependence of the threshold dimension $d_0$ in \cref{thm:existence_Gibbs_states-NQ} and \cref{thm:characterization_of_Gibbs_states-NQ} on the parameters of the spin system, and also the function $\epsilon(d)$ in \cref{thm:long-range-order-NQ}, are not made explicit. Quantitative refinements providing estimates for these quantities are presented in \cref{sec:quantitative-results} below.

A version of the above qualitative results applies in any dimension $d\ge 2$ provided the underlying graph is suitably modified. Precisely, the results remain true as stated when $\Z^d$ is replaced by a graph of the form $\Z^{d_1}\times\T_{2m}^{d_2}$, $m\ge 1$ integer, where $\T_{2m}$ is the cycle graph on $2m$ vertices (the path on 2 vertices if $m=1$), provided $d_1\ge 2$ and with $d:=d_1+d_2$ substituting for the dimension. The graph $\Z^{d_1}\times\T_{2m}^{d_2}$ may be viewed as a subset of $\Z^d$ in which the last $d_2$ coordinates are restricted to take value in $\{0,1,\ldots, 2m-1\}$ and are endowed with periodic boundary conditions. In this sense, it is only the local structure of $\Z^d$ which matters to the results. To keep the discussion focused, we present the proofs of the results only in the $\Z^d$ case. The extension to graphs of the above form require only minor modifications to the arguments beyond notational changes (analogous modifications are discussed in the companion paper~\cite{peledspinka2018colorings}).

\subsection{First applications}\label{sec:first applications}
We briefly describe our results in the context of some well-known models of statistical physics (see \cref{fig:graphs0}). Further applications are discussed in \cref{sec:applications}. The quantitative statements claimed below are derived in \cref{sec:first applications_revisited}.

\subsubsection{The antiferromagnetic Potts model}\label{sec:AF Potts model}
The \emph{$q$-state Potts model} is a classical model of statistical physics in which spins are in one of $q$ possible states. In its \emph{ferromagnetic} version, adjacent spins have a tendency to be equal, while in its \emph{antiferromagnetic} (AF) version they tend to be different. The ferromagnetic version has been studied extensively and is relatively well understood~\cite{grimmett2004random, duminil2017lectures} so our focus here is on the antiferromagnetic version for which understanding is still quite lacking.
On a finite $\Lambda\subset\Z^d$ and at inverse temperature $\beta>0$, the model assigns to each $f:\Lambda\to\{1,\ldots, q\}$ the probability
\begin{equation}\label{eq:AF Potts probability measure}
  \frac{1}{Z_{\Lambda,\beta}}\exp\bigg(-\beta\sum_{\{u,v\} \in E(\Lambda)}\1_{\{f(u)= f(v)\}}\bigg)
\end{equation}
with $Z_{\Lambda,\beta}$ a suitable normalization constant (the partition function). In the limit $\beta\to\infty$, one obtains the \emph{proper $q$-coloring model}, in which $f$ is uniformly sampled among functions satisfying that the spins at adjacent vertices are in different states. The proper $q$-coloring model is analyzed in our companion paper~\cite{peledspinka2018colorings}.

The AF Potts model may be described within our general setup by choosing
\begin{equation}\label{eq:AF Potts in framework}
\SS = \{1,\dots,q\},\qquad\lambda_i=1,\qquad\lambda_{i,j}=\1_{\{i\neq j\}}+e^{-\beta}\1_{\{i=j\}} .
\end{equation}
One then checks simply that in this model a pattern is any pair $(A,B)$ of disjoint subsets of $\SS$, and such a pattern is dominant when $\{|A|,|B|\}=\{\lfloor \frac q2 \rfloor, \lceil \frac q2 \rceil \}$, i.e., when $\{A,B\}$ is a partition of $\SS$ into sets of almost equal size. Thus, there are $\binom{q}{q/2}$ dominant patterns when $q$ is even and $2\binom{q}{(q-1)/2}$ when $q$ is odd.
As the AF Potts model is symmetric with respect to the elements of $\SS$, it is evident that the dominant patterns are all equivalent.

The question of understanding the type of emergent long-range order, or its absence, in the $q$-state AF Potts model, including proper $q$-colorings, has received significant attention. Berker--Kadanoff~\cite{berker1980ground} initially suggested in 1980 that a phase with algebraically decaying correlations may occur at low temperatures (including zero temperature) with fixed $q$ when $d$ is large. However, later numerical simulations and theoretical arguments of Banavar--Grest--Jasnow~\cite{banavar1980ordering} ($q=3,4$ and $d=3$) and Koteck\'y~\cite{kotecky1985long} ($q=3$, $d$ large) predicted instead an ordered phase (termed \emph{Broken-Sublattice-Symmetry (BSS) phase}) at low temperatures. This prediction was extended to general values of $q$ in sufficiently high dimensions by several authors including Salas--Sokal~\cite{salas1997absence}, Koteck{\'y}--Sokal--Swart~\cite[Section 1.4, (3)]{kotecky2014entropy}, Engbers--Galvin~\cite[Section 6.3]{engbers2012h2}, Galvin--Kahn--Randall--Sorkin~\cite[Conjecture 1.3]{galvin2012phase} and Feldheim and the authors~\cite[Section 8]{feldheim2013rigidity} and \cite[Section 1.3]{feldheim2015long} (with some of these works focusing on the zero-temperature case); see the companion paper~\cite{peledspinka2018colorings} for a more detailed review.

In the mathematically rigorous literature, Dobrushin's uniqueness condition implies that the AF Potts model is disordered when either
\begin{equation}\label{eq:Dobrushin uniqueness AF Potts}
  q>2d(2 - e^{-\beta})\qquad\text{or}\qquad\beta \le \frac{q}{4d+2}
\end{equation}
(the fact that $q>4d$ suffices is due to Koteck\'y~\cite[pp.~148-149,457]{georgii2011gibbs} and Salas--Sokal~\cite{salas1997absence}).
Results on long-range order were limited to the $3$-state model prior to this work and its companion. Specifically, long-range order was proved for the zero-temperature 3-state model (proper $3$-colorings) by the first author~\cite{peled2010high} and by Galvin--Kahn--Randall--Sorkin~\cite{galvin2012phase} (following closely related papers by Galvin--Randall~\cite{galvin2007torpid} and Galvin--Kahn~\cite{galvin2004phase}), who established a quantitative analogue of \cref{thm:long-range-order-NQ} for this model, showing that, in sufficiently high dimensions, each of the six dominant patterns gives rise to at least one extremal Gibbs state (a BSS phase) possessing a strong tendency to follow the pattern. In~\cite{galvin2012phase}, it was also shown that these Gibbs states have maximal entropy. Feldheim and the second author~\cite{feldheim2015long} extended this result to the positive-temperature 3-state model when $d \ge C$ and $\beta \ge C\log d$ for some absolute constant $C>0$, and further established an analogue of \cref{thm:existence_Gibbs_states-NQ} in this case.

Our results are the first to prove long-range order in the AF Potts model with $q\ge 4$ and the first to characterize the set of periodic Gibbs states when $q\ge 3$. This is established in the companion paper for the zero-temperature model and extended here to the low-temperature regime (the results here apply also to the zero-temperature model, but give a slightly worse dependence of $q$ and $d$ compared to the specialized analysis of~\cite{peledspinka2018colorings}). Our non-quantitative results, \cref{thm:long-range-order-NQ} and \cref{thm:existence_Gibbs_states-NQ}, show that, for any fixed $q \ge 3$ and $\beta>0$, in sufficiently high dimensions, each dominant pattern $(A,B)$ gives rise to an extremal Gibbs state, invariant to automorphisms of $\Z^d$ preserving the two sublattices, which is characterized by a strong tendency for spins at even vertices to take values in $A$ and for spins at odd vertices to take values in $B$ (a BSS phase). In particular, for any fixed $q$, the critical inverse temperature for the existence of long-range order (defined, e.g., as the infimum over $\beta$ for which the model has multiple Gibbs states) \emph{tends to zero} as the dimension tends to infinity. Moreover, by \cref{thm:characterization_of_Gibbs_states-NQ}, any periodic Gibbs state is a mixture of these $(1+\1_{\{q\text{ is odd}\}}) \binom{q}{\lfloor q/2 \rfloor}$ Gibbs states (it is easy to check that every periodic Gibbs state for the positive-temperature AF Potts model has maximal pressure).
The quantitative versions of our results given in \cref{sec:quantitative-results} show that the above description of the Gibbs states holds when
\begin{equation}\label{eq:potts_param_ineq}
d \ge C q^{12} \log^6 q \qquad\text{and}\qquad \beta \ge \frac{Cq^2 \log^{3/2} d}{d^{1/4}}
\end{equation}
for some absolute constant $C>0$. Comparing with the uniqueness regime~\eqref{eq:Dobrushin uniqueness AF Potts}, we see that power-law dependencies between the parameters are best possible though the precise powers are yet to be determined.

The effect of adding a magnetic field to the AF Potts model is analyzed in \cref{sec:AF Potts with magnetic field}. A variant of the proper coloring model allowing for uncolored sites is analyzed in \cref{sec:dilute proper colorings}.

\begin{figure}[!t]
	\tikzstyle{every state}=[circle,fill=gray!25,draw=black,minimum size=16pt,inner sep=0pt,text=black]

	\begin{subfigure}[t]{.3\textwidth}
		\centering
		\begin{tikzpicture}[auto,node distance=1cm]
		\node[state] (0) {1};
		\node[state] (1) [right of=0] {2};
		\node[state] (2) [above of=0, xshift=0.5cm] {3};
		\path (0) edge (1);
		\path (0) edge (2);
		\path (1) edge (2);
		\end{tikzpicture}\qquad\quad
		\begin{tikzpicture}[auto,node distance=1cm]
		\node[state] (0) {1};
		\node[state] (1) [right of=0] {2};
		\node[state] (3) [above of=0] {3};
		\node[state] (2) [right of=3] {4};
		\path (0) edge (1);
		\path (0) edge (2);
		\path (0) edge (3);
		\path (1) edge (2);
		\path (1) edge (3);
		\path (2) edge (3);
		\end{tikzpicture}
		\caption{The AF Potts model with $q=3$ and $q=4$}
		\label{fig:graph-proper-colorings}
	\end{subfigure}%
	\begin{subfigure}[t]{.04\textwidth}~\end{subfigure}%
	\begin{subfigure}[t]{.27\textwidth}
		\centering
		\begin{tikzpicture}[auto,node distance=1.25cm]
		\tikzstyle{every state}=[circle,fill=gray!25,draw=black,minimum size=20pt,inner sep=0pt,text=black,scale=0.8]
		\node[state] (1) {$+1$};
		\node[state] (2) [right of=1] {$+2$};
		\node[state] (-1) [left of=1] {$-1$};
		\node[state] (-2) [left of=-1] {$-2$};
		\path (1) edge (2);
		\path (-1) edge (1);
		\path (-2) edge (-1);
		\path (2) edge [loop above] (2);
		\path (1) edge [loop above] (1);
		\path (-1) edge [loop above] (-1);
		\path (-2) edge [loop above] (-2);
		\end{tikzpicture}
		\caption{The beach model}
		\label{fig:graph-beach-model}
	\end{subfigure}%
	\begin{subfigure}[t]{.04\textwidth}~\end{subfigure}%
	\begin{subfigure}[t]{.25\textwidth}
		\centering
		\begin{tikzpicture}
		\def \n {14}
		\def \radius {1.1cm}

		\tikzstyle{every state}=[circle,fill=gray!25,draw=black,minimum size=10pt,inner sep=0pt,text=black]
		\tikzstyle{every node}=[draw, circle, fill,minimum size=0.02cm,inner sep=0,scale=0.75]

		\draw (0:\radius) arc (0:360:\radius);

		\foreach \s in {1,...,\n}
		{
			\node[state] (\s) at ({360/\n * \s}:\radius) {};
		}

		\foreach \s in {1,...,\n}
		{
			\path (\s) edge [in={360/\n*\s-30},out={360/\n*\s+30},loop] ();
			\draw (\s) to [bend left = 45]  (\intcalcNum{\intcalcMod{\s+1}{\n}+1});
		}
		\end{tikzpicture}
		\caption{The clock model with $q=14$ and $m=2$}
		\label{fig:graph-clock-model}
	\end{subfigure}
	\caption{Graph representations of the three models discussed as first applications. The edges correspond to the pairs of states $\{i,j\}$ with maximal pair interaction $\lambda_{i,j}$.}
	\label{fig:graphs0}
\end{figure}

\subsubsection{The beach model}\label{sec:beach model}
The \emph{beach model} at activity $\lambda>0$ may be described within our general setup by choosing
	\[ \SS=\{-2,-1,1,2\},\qquad\lambda_i=\lambda^{|i|-1},\qquad\lambda_{i,j}=\mathbf{1}_{\{ij \ge -1\}} ,\]
	as illustrated by \cref{fig:graph-beach-model}.
Introduced by Burton--Steif~\cite{burton1994non} and named by H\"aggstr\"om~\cite{haggstrom1996phase} (in the context of shifts of finite type), it describes the $+$ (shore, land) and $-$ (beach, sea) species competing for space with the restriction that $+$'s and $-$'s may only meet at $\pm1$ (the beach and shore). The dominant patterns depend on the value of $\lambda$: When $\lambda>1$, there are two dominant patterns, $(\{1,2\},\{1,2\})$ and $(\{-1,-2\},\{-1,-2\})$, which are clearly equivalent. When $\lambda<1$, there is a unique dominant pattern $(\{-1,1\},\{-1,1\})$. When $\lambda=1$, all three of these patterns are dominant. However, as the latter pattern is not equivalent to the former two, our results are not applicable.

H\"aggstr\"om~\cite{haggstrom1996phase} has shown that there is a critical $\lambda_c(d)$ for phase transition in the model. That is, the model has a unique Gibbs state when $\lambda<\lambda_c(d)$ and multiple translation-invariant Gibbs states when $\lambda>\lambda_c(d)$ (in this model, all periodic Gibbs states have maximal pressure). The best known bounds on the critical threshold in dimensions $d\ge 2$ are
\begin{equation}\label{eq:beach_lambda_c_bounds}
\frac{2}{2d^2 + d - 1}\le\lambda_c(d)\le \min\big\{4e\cdot 28^d,(1+\sqrt{2})^{2^{2d-2}}\big\}-1.
\end{equation}
The lower bound is obtained by H\"aggstr\"om~\cite{haggstrom1996phase} who also obtains the upper bound with the second term in the minimum. The upper bound $\lambda_c(d)\le 4e\cdot 28^d-1$ is due to Burton--Steif~\cite{burton1994non} who also proved that when $\lambda$ exceeds this threshold the model has two translation-invariant extremal Gibbs states, any periodic Gibbs state is a mixture of these two and the two measures are distinguished by favoring the positive or the negative spin values (they give an equivalent description in terms of ergodic measures of maximal entropy).

Our results significantly improve the bounds in~\eqref{eq:beach_lambda_c_bounds}. Specifically, the non-quantitative results of \cref{sec:non-quantitative results} (applied separately for $\lambda<1$ and for $\lambda>1$) imply that $\lambda_c(d) = 1 \pm \varepsilon(d)$ with some $\varepsilon(d)\to0$ as $d\to\infty$. The quantitative versions of our results given in \cref{sec:quantitative-results} show more precisely that for all $d\ge 2$,
\begin{equation}\label{eq:beach_improved_lambda_c_bounds}
  |\lambda_c(d) - 1|\le \frac{C\log^{3/2} d}{d^{1/4}}
\end{equation}
for an absolute constant $C>0$. Moreover, when $\lambda-1$ exceeds the upper bound in~\eqref{eq:beach_improved_lambda_c_bounds}, each of the two dominant patterns gives rise to an extremal and invariant (under all automorphisms of $\Z^d$) Gibbs state which is characterized by a predominance of values of one sign, and any periodic Gibbs state is a mixture of these two. See \cref{sec:dim-dep-beach} for a conjecture about $\lambda_c(d)$.

A multi-type version of the beach model is analyzed in \cref{sec:multi type beach model}.

\subsubsection{Clock models and Lipschitz functions}\label{sec:clock-models}
The family of clock models comprises a particularly rich class of statistical physics models containing the Ising and Potts, Ashkin--Teller and discrete XY models as special cases. Spins in a $q$-state clock model take values in the $q$th roots of unity, with equal single-site activities and with the interaction energy of two spins depending only on their distance. In a typical ferromagnetic setting, the lowest-energy configurations of the model are simply the $q$ constant configurations and these govern the low-temperature behavior, as seen from a Peierls argument or an application of Pirogov--Sinai theory. This analysis does not carry over, however, to clock models with \emph{residual entropy}, where the number of lowest-energy configurations in a domain is exponential in its volume. Our results apply to most clock models, as the equivalence of dominant patterns is generally implied by the inherent rotational symmetry, and provide, in particular, an analysis of cases with residual entropy. This was demonstrated above for the antiferromagnetic Potts model (which may be viewed as a clock model). To give an example of a ferromagnetic flavor, we consider here a family of $q$-state clock models parameterized by a non-negative integer $m$, in which the interaction energy takes only two values: a lower interaction energy for spins separated by at most $m$ roots of unity and a higher interaction energy otherwise. The case $m=0$ is the standard ferromagnetic Potts model while for every $m\ge 1$ these models have residual entropy.

The above family of models may be described within our general setup by taking
\begin{equation}\label{eq:Lipschitz clock model in framework} \SS=\Z_q=\mathbb{Z}/q\Z,\qquad \lambda_i=1,\qquad \lambda_{i,j} = \1_{\{\text{dist}_{\Z_q}(i,j) \le m \}} + e^{-\beta} \1_{\{\dist_{\Z_q}(i,j) > m \}} ,
\end{equation}
where $\dist_{\Z_q}(i,j)$ is the minimal $k\ge 0$ such that $i=j+k$ or $j=i+k$ with addition modulo $q$ (see \cref{fig:graph-clock-model}). The parameter $0<\beta\le\infty$ represents inverse temperature, with zero temperature explicitly allowed. To keep the discussion focused, we restrict throughout to the case $1 \le m <\frac{q}{4}$, where one may check that the dominant patterns are $(i+A,i+A)$, $i\in \SS$, $A:=\{ 0,1,\dots,m \}$.

Dobrushin's uniqueness condition implies that the model is disordered when either $\beta$ is smaller than a threshold depending on $d$, $m$ and $q$, or when $\beta<\infty$ and $q$ is larger than a threshold depending on $d$, $m$ and $\beta$.

	The non-quantitative results of \cref{sec:non-quantitative results} show that, for any fixed $1 \le m <\frac{q}{4}$ and $0<\beta\le\infty$, in sufficiently high dimensions, each dominant pattern $(i+A,i+A)$ gives rise to an extremal and invariant (under all automorphisms of $\Z^d$) Gibbs state which is characterized by a strong tendency for spins to take values in $i+A$, and moreover, that any maximal-pressure (periodic) Gibbs state is a mixture of these $q$ Gibbs states (in fact, when $\beta<\infty$ all periodic Gibbs states have maximal pressure).
	The quantitative versions of our results given in \cref{sec:quantitative-results} show that this description is valid whenever $d\ge 2$ and
\begin{equation}\label{eq:quantitative_clock_cond}
	m^2 \log q \le \frac{c d^{1/4}}{\log^{3/2} d}\qquad\text{and}\qquad \beta \ge \frac{Cqm \log^{3/2} d}{d^{1/4}}
	\end{equation}
for some absolute constants $C,c>0$ (slightly sharper results are in \cref{sec:clock model revisited}). At zero temperature there is no essential dependence of the model on $q$ (under our assumption $q>4m$; see the following discussion) and this improves to
\begin{equation}\label{eq:quantitative_clock_cond2}
	m \le \frac{c d^{1/8}}{\log^{5/4} d}.
	\end{equation}

It is worthwhile noting that the \emph{zero-temperature} model admits a height function representation. The relevant height functions are called \emph{$m$-Lipschitz functions} and are integer-valued functions which change by at most $m$ between nearest neighbors. On $\Z^d$, such $m$-Lipschitz functions, considered modulo a global addition of a constant multiple of $q$, are in bijection with the admissible zero-temperature configurations by considering the heights modulo $q$ (the bijection relies on the assumption $q>4m$). This bijection may be used to transfer the results described below from the zero-temperature clock model to uniformly sampled $m$-Lipschitz functions (see \cref{sec:covering systems} for more details) and leads to the following: (i) For each Gibbs state $\mu_{(i+A,i+A)}$ obtained from \cref{thm:existence_Gibbs_states-NQ} and each integer $k$, a Gibbs state $\mu_{i,k}$ for $m$-Lipschitz functions is obtained by applying the bijection to the realizations of $\mu_{(i+A,i+A)}$ with the restriction that the heights on the unique infinite connected component of vertices in the $(i+A,i+A)$-pattern lift to heights in $\{i+kq, \ldots, i+kq+m\}$. The measures $\mu_{i,k}$ are extremal and translation invariant. (ii) Every (periodic) maximal-entropy Gibbs state for $m$-Lipschitz functions is a mixture of these ``lifted'' Gibbs states.

To our knowledge, the only previous result on long-range order in the clock model~\eqref{eq:Lipschitz clock model in framework} (with $1\le m< \frac{q}{4}$) is for the case $m=1$ and $\beta=\infty$. It was shown in~\cite{peled2010high} that $1$-Lipschitz functions in high dimensions have a Gibbs state under which spins have a strong tendency to take values in $\{0, 1\}$. The modulo $q$ of this Gibbs state is then a Gibbs state with a similar tendency for the $m=1$, $\beta=\infty$ clock model (when $q>4$). A related result of Aizenman~\cite{Aiz94} is a power-law lower bound on the rate of correlation decay in two dimensions for a continuous version of the zero-temperature clock model~\eqref{eq:Lipschitz clock model in framework} obtained in the limit when $q\to\infty$ and $\lim m/q \le \frac{1}{8}$. The proof of this result appears to extend to the zero-temperature clock model itself when $q$ is even and $m/q\le \frac{1}{8}$ (see also~\cite{pinson1998slow},~\cite[Section 2.8.2]{peled2019lectures} and~\cite[Section 2]{cohen2017rarity}).

We make a final remark regarding more general clock models. Examining the quantitative results of \cref{sec:quantitative-results} shows that the above description of the Gibbs states remains valid, under the same condition~\eqref{eq:quantitative_clock_cond}, for the class of $q$-state clock models given by
\[ \SS=\Z_q=\Z/q\Z,\qquad \lambda_i=1,\qquad \lambda_{i,j} = f(\text{dist}_{\mathbb{Z}_q}(i,j)), \]
with an arbitrary function $f$ satisfying that $f(r)=1$ when $0\le r\le m$ and $f(r)\le e^{-\beta}$ otherwise.

\subsubsection{Generic spin systems}\label{sec:generic-systems}
In this work, we study general discrete spin systems with nearest-neighbor isotropic interaction, as described in \cref{sec:the model}. The obtained results are applicable to the subset of these models satisfying the requirement that all dominant patterns are equivalent. At first glance, it may seem that only highly symmetric and non-generic models satisfy this requirement. In fact, the opposite is true: the requirement is satisfied in generic situations, in the sense explained now.

Consider \emph{any} discrete spin system with nearest-neighbor isotropic interaction, with single-site activities and pair interactions as in \cref{sec:the model}. Now consider a new system, obtained from the given one by making (small) generic perturbations to all the single-site activities or to all the pair interactions (in physical terms, perturbing either the magnetic fields or the coupling constants). One checks that the new system will have either a unique dominant pattern $(A,A)$ or precisely two dominant patterns of the form $(A,B)$ and $(B,A)$, for some $A,B\subset \SS$, $A\neq B$ (which of the two options will arise depends on the given spin system and the type of perturbation performed). Thus the new system will satisfy that all dominant patterns are equivalent, so that our results may be applied to it.

\subsection{Enumeration results}\label{sec:enum-intro}
Our previous results focused on the behavior of typical configurations and infinite-volume Gibbs states. Another question of enduring interest regards the weighted number of configurations, i.e., the \emph{(free) partition function} $Z^{\text{free}}_\Lambda := \sum_{f \in \SS^\Lambda} \omega_f$, where the weight $\omega_f$ is defined in~\eqref{eq:config-weight}. Of primary concern is the exponential rate of growth of $Z^{\text{free}}_\Lambda$ as $\Lambda$ grows, captured by the \emph{(free) topological pressure} $P_{\text{top}} := \lim \frac{1}{|\Lambda|} \log Z^{\text{free}}_\Lambda$, where the limit is taken along boxes $\Lambda$ increasing to $\Z^d$ (see \cref{sec:equilibrium-states}). For non-weighted homomorphism models (i.e., when the single-site activities are all 1 and the pair interactions are all 0 or 1), $P_{\text{top}}$ is called the topological entropy and is the exponential rate of growth of the number of homomorphisms, i.e.,
\[ |\Hom([n]^d,H)| = e^{(P_{\text{top}} + o(1))n^d} ,\]
where $\Hom(G,H)$ is the set of graph homomorphisms from a graph $G$ to $H$.

Similarly to $P_{\text{top}}$, one may also define the periodic topological pressure corresponding to periodic boundary conditions or the $(A,B)$ topological pressure corresponding to $(A,B)$ boundary conditions for a pattern $(A,B)$. We show, however, that all these notions coincide (see~\eqref{eq:partition-function-convergence-to-top-pres}).

It is convenient to assume that
\begin{equation}\label{eq:max-interaction-1}
  \max_{i,j \in \SS} \lambda_{i,j}=1.
\end{equation}
Under this normalization, it is easy to see that $P_{\text{top}} \ge \frac12 \log (\lambda_A \lambda_B)$ for any pattern $(A,B)$, where we write $\lambda_I := \sum_{i \in I} \lambda_i$ for a subset $I \subset \SS$. The best bound of this kind is then $P_{\text{top}} \ge \frac12 \log \omega_\text{dom}$, where $\omega_\text{dom} := \max_{(A,B)\text{ pattern}} \lambda_A \lambda_B$ is the weight of a dominant pattern. It turns out that this bound is sharp in the limit as the dimension tends to infinity, i.e.,
\begin{equation}\label{eq:topological pressure limit}
P_{\text{top}}  = \tfrac12 \log \omega_{\text{dom}} + o(1)\qquad\text{as $d\to\infty$}.
\end{equation}
Indeed, Galvin--Tetali~\cite{galvin2004weighted} showed that $P_{\text{top}} \le \tfrac12 \log \omega_{\text{dom}} + O(\frac 1\Delta)$ as $\Delta \to \infty$ for homomorphism models on any bipartite regular graph of degree $\Delta$ (in which case their bound is of the correct order), and this bound may be extended to non-homomorphism models using the results of Galvin~\cite{galvin2006bounding}. An analogous result for isotropic subshifts on $\Z^d$ in the limit $d\to\infty$ was shown by Meyerovitch--Pavlov~\cite{meyerovitch2014independence}.

More precise results are available for homomorphism models on the hypercube, or more generally, on tori of fixed even side length, as the dimension tends to infinity. Engbers--Galvin~\cite{engbers2012h2} studied the structure of homomorphisms, obtaining, in particular, that the log-partition-function per site is at most $\tfrac12 \log \omega_{\text{dom}} + e^{-\Omega(d)}$ as $d \to \infty$. Galvin~\cite{Galvin2003hammingcube} and Kahn--Park~\cite{kahn2020number} gave precise asymptotics for the number of proper colorings, with 3 and 4 colors respectively, of the hypercube. Refined structural and enumerative results on general homomorphisms were recently obtained by Jenssen--Keevash~\cite{jenssen2020homomorphisms}. Their results yield, in particular, asymptotic formulas for the number of proper $q$-colorings for all values of $q$.

It is natural to inquire about the rate of convergence in~\eqref{eq:topological pressure limit}.
Jenssen--Keevash~\cite[Conjecture~19.1]{jenssen2020homomorphisms} conjecture a refined formula for the (periodic) topological pressure for homomorphism models on $\Z^d$ in the limit as the dimension tends to infinity, which implies, in particular, that the $o(1)$ term in~\eqref{eq:topological pressure limit} is~$e^{-\Omega(d)}$. Roughly, the formula takes into account the effect of isolated single-site deviations from the $(A,B)$-pattern. Our methods yield a proof of this formula for homomorphism models having all dominant patterns equivalent. Moreover, we provide an extension of the formula to non-homomorphism models. To state the result, fix a dominant pattern $(A,B)$ and denote
\begin{equation}
\epsilon_{A,B} := \frac{1}{\lambda_A \lambda_B^{2d}}\sum_{i \in \SS \setminus A} \lambda_i \left(\sum_{b \in B} \lambda_b \lambda_{i,b} \right)^{2d} .
\end{equation}
It is straightforward to check that $\epsilon_{A,B}/(1+\epsilon_{A,B})$ (respectively, $\epsilon_{B,A}/(1+\epsilon_{B,A})$) is precisely the probability that an even (respectively, odd) vertex violates the $(A,B)$-pattern given that all other vertices within distance two from it are in the $(A,B)$-pattern. Note that $\epsilon_{A,B}$ can be zero, but otherwise, using that $(A,B)$ is a dominant pattern, one sees that it is $e^{-c_0d(1+o(1))}$ as $d \to\infty$ for some constant $c_0>0$ (depending on the fixed spin system).

\begin{thm}\label{thm:enum}
Fix a spin system in which all dominant patterns are equivalent, and normalized so that~\eqref{eq:max-interaction-1} holds. Fix a dominant pattern $(A,B)$. Then
\begin{equation}\label{eq:pressure with single-site deviations}
P_{\text{top}} = \tfrac12 \log \omega_\text{dom} + \tfrac12(\epsilon_{A,B}+\epsilon_{B,A})(1+\delta(d)),
\end{equation}
where $|\delta(d)|=e^{-\Omega(d)}$ as $d\to\infty$.
\end{thm}
The Jenssen--Keevash conjectured formula is the weaker statement that~\eqref{eq:pressure with single-site deviations} holds with $\delta(d)=o(1)$ as $d\to\infty$. Our methods allow, in principle, to calculate further order terms corresponding to larger deviations from the $(A,B)$-pattern. They also allow to obtain non-asymptotic estimates on $P_{\text{top}}$ (i.e., quantitative control for a fixed high dimension $d$).

\subsection{Methodology}
This paper is the companion paper to~\cite{peledspinka2018colorings} where long-range order for high-dimensional proper colorings of $\Z^d$ is established. The papers share a common methodology, with the arguments of the companion paper significantly extended here to apply to general spin systems. An overview of the basic methodology as well as a survey of connections to the existing literature can be found in~\cite{peledspinka2018colorings}, with lectures on proper colorings available at~\cite{peled2020three} and a review of the present work aimed towards a physics audience in~\cite{peled2017condition}. We thus content ourselves with mentioning here that the technique relies on a synthesis of entropy and contour techniques which are carefully adapted to the $\Z^d$ lattice geometry. Among previous works we emphasize the works of Kahn--Lawrentz~\cite{kahn1999generalized}, Kahn~\cite{kahn2001entropy, Kahn2001hypercube}, Galvin--Tetali~\cite{galvin2004weighted} and Engbers--Galvin~\cite{engbers2012h1,engbers2012h2} which pioneered the use of entropy techniques for the enumeration of graph homomorphisms on regular bipartite graphs and for the study of their long-range order on the hypercube. Galvin~\cite{galvin2006bounding} developed a method for extending the enumeration results to general spin systems. We further emphasize the works of Korshunov and Sapozhenko~\cite{korshunov1981number,Korshunov1983Th,sapozhenko1987onthen,sapozhenko1989number,sapozhenko1991number}, Galvin~\cite{Galvin2003hammingcube} and Galvin--Kahn~\cite{galvin2004phase} which introduced the contour techniques and applied them to the study of long-range order in specific models. A synthesis of the two techniques was introduced in~\cite{peled2017condition, peledspinka2018colorings} and was a main ingredient in the work of Kahn--Park~\cite{kahn2020number} on the asymptotic enumeration of proper $4$-colorings of the hypercube and in the recent work by Jenssen--Keevash~\cite{jenssen2020homomorphisms} on graph homomorphism models on tori of fixed side length.

\subsection{Organization}
The rest of the paper is organized as follows.
In \cref{sec:quantitative-results}, we formulate our quantitative results under various explicit quantitative conditions and under a more abstract condition. In \cref{sec:applications}, we discuss various applications and extensions of our results.
In \cref{sec:preliminaries}, notation and preliminary results which will be needed throughout the paper are given. In \cref{sec:high-level-proof}, we provide the main steps of the proof \cref{thm:long-range-order} (the quantitative version of \cref{thm:long-range-order-NQ}), including the definitions of breakup and approximation and the statements of several propositions which are then used to deduce \cref{thm:long-range-order}. In \cref{sec:breakup}, we prove the propositions about breakups (existence of non-trivial breakup, almost-sure absence of infinite breakups, bounds on the probability of breakups). In \cref{sec:shift-trans}, we prove \cref{lem:bound-on-pseudo-breakup} which provides a general bound on the probability of an event and which is used in the proofs in \cref{sec:prob-of-given-breakup} and \cref{sec:prob-of-approx}.
In \cref{sec:approx}, we prove \cref{prop:family-of-odd-approx} about the existence of a small family of approximations. In \cref{sec:model_on_K2d2d}, we analyze the model on a complete bipartite graph, showing that the explicit quantitative conditions given in \cref{sec:quantitative-results} imply the abstract condition. In \cref{sec:gibbs}, we prove results about the infinite-volume Gibbs states, namely, \cref{thm:existence_Gibbs_states} and \cref{thm:characterization_of_Gibbs_states} (the quantitative versions of \cref{thm:existence_Gibbs_states-NQ} and \cref{thm:characterization_of_Gibbs_states-NQ}), as well as \cref{thm:enum}. Finally, in \cref{sec:discussion}, we discuss open questions.

\subsection{Acknowledgments}  The research of both authors was supported in part by Israel Science Foundation grants 861/15 and 1971/19 and by the European Research Council starting grant 678520 (LocalOrder).
The research of R.P.\ was additionally supported by the National Science Foundation grant DMS-2451133.
The research of Y.S.\ was additionally supported by the Adams Fellowship Program of the Israel Academy of Sciences and Humanities, by NSERC of Canada, and by Israel Science Foundation grant 1361/22.
We thank Raimundo Brice\~no, Nishant Chandgotia, Ohad Feldheim and Wojciech Samotij for many pertinent discussions on proper colorings and other homomorphism models. We are grateful to Christian Borgs for valuable advice on the way to present the material of this paper and its companion~\cite{peledspinka2018colorings}.  We thank Michael Aizenman, Jeff Kahn, Roman Koteck\'y, Eyal Lubetzky, Dana Randall, Alan Sokal, Prasad Tetali and Peter Winkler for useful discussions and encouragement.

\section{Quantitative results}
\label{sec:quantitative-results}
In this section, we formulate our main results, which are quantitative versions of the results presented in \cref{sec:non-quantitative results}. Long-range order will be seen to emerge under a quantitative condition involving the dimension $d$ and the parameters of the spin system, i.e., the finite spin space $\SS$, single-site activities $(\lambda_i)_{i \in \SS}$ and pair interactions $(\lambda_{i,j})_{i,j \in \SS}$. The relevant information is summarized by four parameters $\rho_{\text{pat}}^{\text{bulk}}, \rho_{\text{pat}}^{\text{bdry}}, \rho_{\text{int}}\in[0,1)$ and $\rho_{\text{act}}\ge 1$, which will be defined in \cref{sec:the four parameters} below. Using these and defining
\begin{equation}\label{eq:alpha-def}
  \alpha_0 := -\log \max\left\{\rho_{\text{pat}}^{\text{bulk}},~1-(1-\rho_{\text{pat}}^{\text{bdry}})(1-\sqrt{\rho_{\text{int}}})\right\},
\end{equation}
our quantitative condition is that
\begin{equation}\label{eq:parameter-inequalities-simple}
\alpha_0\ge \frac{C|\SS|\log^{3/2}d}{d^{1/4}}\qquad\text{and}\qquad - \log \rho_{\text{int}} \ge \frac{|\SS| \log^2 (d\rho_{\text{act}})}{d^{3/4}} ,
\end{equation}
where $C>0$ is a universal constant. For homomorphism models, those models in which all pair interactions are zero or one, the parameter $\rho_{\text{int}}$ is zero (so that $- \log \rho_{\text{int}}=+\infty$) and the condition takes the simpler form
\begin{equation}\label{eq:parameter-inequalities-simple-hom}
\alpha_0 = -\log \max\big\{\rho_{\text{pat}}^{\text{bulk}},\rho_{\text{pat}}^{\text{bdry}} \big\} \ge \frac{C|\SS|\log^{3/2}d}{d^{1/4}}.
\end{equation}
These conditions are relatively simple to check and apply well in the examples that we consider.

Each of the four parameters deals with a potential obstruction to the system ordering according to its dominant patterns. Roughly speaking, $\rho_{\text{pat}}^{\text{bulk}}$ quantifies the entropic favorability of dominant patterns over other patterns, $\rho_{\text{pat}}^{\text{bdry}}$ quantifies the entropic loss incurred on the boundaries of ordered regions, $\rho_{\text{int}}$ measures the energetic cost of non-maximal pair interactions, so that $-\log \rho_{\text{int}}$ plays the role of an ``inverse-temperature'', and $\rho_{\text{act}}$ accounts for the possibility that very high single-site activites overcome the energetic cost of non-maximal pair interactions. Our analysis proves that once these are controlled, by verifying condition~\eqref{eq:parameter-inequalities-simple}, no other obstruction arises and long-range order indeed emerges.

Our quantitative condition may thus be viewed as saying that the inverse-temperature is sufficiently high and the dominant patterns are sufficiently favored over the other patterns. The required threshold decreases to \emph{zero} as a power-law with the dimension $d$ (thus, for instance, the temperature may be as high as a power of $d$). In particular, for a fixed spin system, our quantitative condition is satisfied in sufficiently high dimensions, so that the non-quantitative results given in \cref{sec:intro-results} follow directly from the quantitative versions given in this section.


\subsection{The four parameters: $\rho_{\text{pat}}^{\text{bulk}}$, $\rho_{\text{pat}}^{\text{bdry}}$, $\rho_{\text{int}}$, $\rho_{\text{act}}$}\label{sec:the four parameters}
\label{sec:four-params}
~

\smallskip
\noindent{\bf The interaction ratio $\rho_{\text{int}}$:}
This parameter is defined as the ratio of the second largest and largest pair interactions,
\begin{equation}
\rho_{\text{int}} := \frac{\max\{\lambda_{i,j}\colon ~i,j \in \SS,~ \lambda_{i,j}<\lambda^{\text{int}}_{\text{max}}\}}{\lambda^{\text{int}}_{\text{max}}}  , \qquad\text{where}\quad \lambda^{\text{int}}_{\text{max}}:=\max_{i,j \in \SS} \lambda_{i,j},
\end{equation}
and where it is understood that $\rho_{\text{int}} = 0$ if the set maximized over is empty.
As remarked, $-\log \rho_{\text{int}}$ measures a sort of inverse temperature for the spin system. For instance, it is equal to $\beta$ for the AF Potts model~\eqref{eq:AF Potts in framework} and clock model~\eqref{eq:Lipschitz clock model in framework} discussed in \cref{sec:first applications}. The condition~\eqref{eq:parameter-inequalities-simple} stipulates that $\rho_{\text{int}}$ is sufficiently small, in correspondence with the general principle that ordering phenomena do not occur at high temperature (where the quantitative threshold for ``high temperature'' depends on $d$ and in fact tends to \emph{infinity} as a power law with it). The parameter $\rho_{\text{int}}$ controls the \emph{energetic} contribution to the formation of order, but it is not sufficient by itself to ensure such order. Indeed, homomorphism models, such as the proper colorings, Widom--Rowlinson and beach models discussed in \cref{sec:first applications}, have $\rho_{\text{int}}=0$ and although this is the smallest possible value for it, we have seen that they may have a unique (disordered) Gibbs state. Indeed, \emph{entropic} contributions play an important role, as captured by the two pattern ratios, which we now describe.

\smallskip\noindent
{\bf The pattern ratios $\rho_{\text{pat}}^{\text{bulk}}$, $\rho_{\text{pat}}^{\text{bdry}}$:} Recall from~\eqref{eq:pattern def} that a pattern is an ordered pair $(A,B)$ of subsets of $\SS$ satisfying that each spin in $A$ has maximal interaction weight with each spin in $B$, i.e., $\lambda_{a,b}=\lambda^{\text{int}}_{\text{max}}$ for all $a\in A$ and $b\in B$.
As before, for a subset $I \subset \SS$, we denote
\begin{equation}\label{eq:lambda of a set}
\lambda_I := \sum_{i \in I} \lambda_i .
\end{equation}
Define the \emph{weight} of a pattern $(A,B)$ to be $\lambda_A \lambda_B$. Note that $(\SS,\emptyset)$ is a pattern with weight 0.
Recall that the pattern $(A,B)$ is called dominant when its weight is maximal, i.e., equal to
\begin{equation}
\omega_\text{dom} := \max_{(A,B)\text{ pattern}} \lambda_A \lambda_B .
\end{equation}
Both pattern ratios are defined via the ratios between the weights of dominant and non-dominant patterns. They may intuitively be thought of as the maximal ratio of this type, i.e.,
\begin{equation}\label{eq:naive ratio}
\max_{(A,B)\text{ non-dominant pattern}} \tfrac{\lambda_A \lambda_B}{\omega_\text{dom}} .
\end{equation}
Instead, however, we obtain smaller values for the pattern ratios by only maximizing the above ratio over suitable subsets of the non-dominant patterns. The improvement obtained in this manner is essential in some applications.\footnote{For example, in the Widom--Rowlinson model (discussed in \cref{sec:first applications}), one calculates that $\max\{\rho_{\text{pat}}^{\text{bulk}},\rho_{\text{pat}}^{\text{bdry}}\}$ is asymptotic to $2/\lambda$ as $\lambda \to \infty$ (see \cref{fig:parameter values of homomorphisms}), whereas the quantity~\eqref{eq:naive ratio} is asymptotic to $1$. Condition~\eqref{eq:parameter-inequalities-simple-hom} is thus verified for large $\lambda$ with the definitions~\eqref{eq:rho pat bulk and bdry def} but would be violated with the simpler definition~\eqref{eq:naive ratio}.}

Say that a pattern $(A,B)$ is \emph{maximal} if there is no other pattern $(A',B')$ with $A\subset A'$ and $B\subset B'$. Note that every dominant pattern is maximal, but not every maximal pattern is necessarily dominant.  Call $A\subset \SS$ a \emph{side} of a dominant (maximal) pattern if there exists $B\subset \SS$ with $(A,B)$ a dominant (maximal) pattern. We mention that $\emptyset$ may be a side of a maximal pattern.
Denote
\begin{equation}\label{eq:rho pat bulk and bdry def}
\rho_{\text{pat}}^\text{bulk} := \max_{\substack{(A,B)\text{ non-dominant}\\\text{maximal pattern}}} \tfrac{\lambda_A \lambda_B}{\omega_\text{dom}} \qquad\text{and}\qquad \rho_{\text{pat}}^\text{bdry} := \max_{\substack{A\text{ side of dominant pattern}\\A'\text{ side of maximal pattern}\\A' \subsetneq A}} \tfrac{\lambda_{A'}}{\lambda_A},
\end{equation}
where it is understood that these quantities are zero if the sets maximized over are empty.

To gain some intuition for the definition of the pattern ratios, observe that there are several alternatives to long-range order via a dominant pattern which need to be excluded. First, a bulk region of the system may be ordered according to a non-dominant pattern $(A,B)$ (i.e., a region where even (odd) vertices take values in $A$ ($B$)). In such a case, it makes sense for the pattern $(A,B)$ to be maximal as otherwise the system may gain entropy, with no energetic penalty, by extending $(A,B)$ to a maximal pattern containing it and ordering via that pattern. To rule out such behavior, the ratio of weights between maximal patterns and dominant patterns should be small, as captured by $\rho_{\text{pat}}^\text{bulk}$. Second, even when the system tends to order according to dominant patterns, one must rule out the possibility that several such orderings coexist in the domain, separated by interfaces (domain walls). At such an interface a side of a dominant pattern meets a side of another dominant pattern, or possibly meets a side of a maximal pattern if the interface is ``thick''.
Indeed, when the system tries to transition out of an ordered phase associated to a dominant pattern $(A,B)$, at its first step it will enter a subset $A'$ of $A$ (or perhaps instead a subset $B'$ of $B$) corresponding to the side of some maximal pattern $(A',B')$. In this case, it must be that $B \subset B'$, and the ratio $\frac{\lambda_{A \cap A'}\lambda_{B \cap B'}}{\lambda_A \lambda_B}$, which represents the loss of entropy incurred on the interface of such a transition, is simply $\frac{\lambda_{A'}}{\lambda_A}$.
To ensure that the formation of such interfaces is entropically unfavorable, we require the pattern ratio $\rho_{\text{pat}}^{\text{bdry}}$ to be small.

\cref{fig:parameter values of homomorphisms} shows the values of $\omega_{\text{dom}}$, $(\rho_{\text{pat}}^{\text{bulk}})^{-1}$ and $(\rho_{\text{pat}}^{\text{bdry}})^{-1}$ for various models, including those discussed in the first applications.

\smallskip\noindent
{\bf The activity ratio $\rho_{\text{act}}$:}
Behind the notion of patterns lies the heuristic idea that pair interactions play a more rudimentary role than single-site activities.
This idea stems from the fact that there are many more edges ($d$ times more) than vertices in $\Z^d$, so that ``bad interactions'' (i.e., pairs interacting with non-maximal interaction weight) are more costly than ``bad activities'' (i.e., spins of low activity). However, this heuristic is easily broken down in extreme cases: for fixed values of pair interactions, by increasing the single-site activities, one may cause entropic contributions to overwhelm any energetic barrier. Such a situation renders our notion of patterns irrelevant.
As an example, it is not hard to see that if for some $i_0 \in \SS$ with $\lambda_{i_0,i_0}>0$, one replaces $\lambda_{i_0}$ by a sufficiently large number (keeping all other parameters fixed, including the dimension), then regardless of whether or not $\lambda_{i_0,i_0}$ equals the largest interaction $\lambda^{\text{int}}_{\text{max}}$, the system will inevitably be forced into a phase which is significantly biased towards having spin $i_0$ at any lattice site. This indicates that in order for the system to order according to dominant patterns, the single-site activities should not be too large in comparison to the pair interactions.
The relevant quantity pertaining to the activities is captured by $\rho_{\text{act}}$, which is defined by
\begin{equation}\label{eq:rho-act-def}
\rho_{\text{act}} := \max_{\substack{A \neq \emptyset\\\text{side of maximal pattern}}} \frac{\lambda_\SS}{\lambda_A}.
\end{equation}
The second inequality in~\eqref{eq:parameter-inequalities-simple} then ensures that the dominant patterns are the relevant notion for ordering.
For an illustration of the relevance of this parameter, see the discussions on the antiferromagnetic Ising and Potts models with external field in \cref{sec:AF Ising and AF Potts with magnetic field}.

\subsection{Main quantitative results}
We describe here the quantitative versions of the theorems stated in \cref{sec:intro-results}. These results provide the existence and characterization of long-range order in spin systems satisfying the quantitative condition~\eqref{eq:parameter-inequalities-simple}.
As we will explain in \cref{sec:alternative-conditions} below, the results also hold under either of the quantitative conditions~\eqref{eq:parameter-inequalities-simple1}, \eqref{eq:parameter-inequalities-simple2} or~\eqref{eq:parameter-inequalities-simple3}, or under the less explicit \cref{main-cond}. When condition~\eqref{eq:parameter-inequalities-simple} is used we define
\[ \tilde\alpha := \alpha_0 \cdot \min\big\{1, \tfrac{\alpha_0}{|\SS|+\log d}\big\}, \]
and when either other condition is used we define
\begin{equation}\label{eq:alpha-tilde}
 \tilde\alpha := \alpha \cdot \min\big\{1, \tfrac{\alpha}{\fq+\log d}\big\} ,
\end{equation}
where $\alpha$ is $\alpha_1$, $\alpha_2$, $\alpha_3$ or $\alpha$, according to whether~\eqref{eq:parameter-inequalities-simple1}, \eqref{eq:parameter-inequalities-simple2}, \eqref{eq:parameter-inequalities-simple3} or \cref{main-cond} is used, and $\fq$ is defined in~\eqref{eq:fq-def}.
Recall the notation from \cref{sec:intro-results}.

\begin{thm}\label{thm:long-range-order}
	There exists a universal constant $c>0$ such that the following holds. Consider a spin system in which all dominant patterns are equivalent. Let $d \ge 2$ and suppose that either~\eqref{eq:parameter-inequalities-simple}, \eqref{eq:parameter-inequalities-simple1}, \eqref{eq:parameter-inequalities-simple2}, \eqref{eq:parameter-inequalities-simple3} or \cref{main-cond} holds.
	Then, for any dominant pattern $(A,B)$, any domain $\Lambda$, and any vertex $v \in \Lambda$, we have
	\[ \Pr_{\Lambda,(A,B)}\big(v\text{ is not in the $(A,B)$-pattern} \big) \le e^{-c\tilde\alpha d} .\]
\end{thm}

The reader may note that the assumptions imply that $\tilde\alpha\ge Cd^{-1/2}\log^2 d$ (when using either of the five conditions), so that the stated upper bound of $e^{-c\tilde\alpha d}$ is at most $e^{-c \sqrt{d}\log^2 d}$.
We mention that larger spatial violations of the boundary pattern are exponentially suppressed; see \cref{cor:prob-of-large-violation-of-P-pattern}.

\begin{thm}\label{thm:existence_Gibbs_states}
	Consider a spin system in which all dominant patterns are equivalent. Let $d \ge 2$ and suppose that either~\eqref{eq:parameter-inequalities-simple}, \eqref{eq:parameter-inequalities-simple1}, \eqref{eq:parameter-inequalities-simple2}, \eqref{eq:parameter-inequalities-simple3} or \cref{main-cond} holds. Then there exists a distinct extremal (periodic) maximal-pressure Gibbs state $\mu_{(A,B)}$ for each dominant pattern $(A,B)$. Moreover, for any sequence of domains $\Lambda_n$ increasing to~$\Z^d$, the measures $\Pr_{\Lambda_n,(A,B)}$ converge weakly to $\mu_{(A,B)}$ as $n \to \infty$. In particular, $\mu_{(A,B)}$ is invariant to automorphisms of $\Z^d$ preserving the two sublattices. Moreover, $\mu_{(A,B)}$ is supported on the set of configurations with an infinite connected component of vertices in the $(A,B)$-pattern, whose complement has only finite $(\Z^d)^{\otimes 2}$-connected components.
\end{thm}

We mention that the convergence of the finite-volume measures occurs at an exponential rate and that the limiting Gibbs state has exponential mixing properties; see \cref{sec:P-pattern-Gibbs-state}.

\begin{thm}\label{thm:characterization_of_Gibbs_states}
	Consider a spin system in which all dominant patterns are equivalent. Let $d \ge 2$ and suppose that either~\eqref{eq:parameter-inequalities-simple}, \eqref{eq:parameter-inequalities-simple1}, \eqref{eq:parameter-inequalities-simple2}, \eqref{eq:parameter-inequalities-simple3} or \cref{main-cond} holds. Then every (periodic) maximal-pressure Gibbs state is a mixture of the measures $\{\mu_{(A,B)}\}$.
\end{thm}

As with the qualitative results of \cref{sec:non-quantitative results}, the quantitative results stated here apply also when $\Z^d$ is replaced with a graph of the form $\Z^{d_1}\times\T_{2m}^{d_2}$, where $d_1\ge 2$ and $d:=d_1+d_2$ substitutes the dimension.

\subsection{Alternative conditions}
\label{sec:alternative-conditions}

Earlier in this section, we gave an explicit condition, namely~\eqref{eq:parameter-inequalities-simple}, under which our quantitative results hold.
In fact, our results also hold under various other, more relaxed, quantitative conditions. We describe three such conditions of varying complexity which will be helpful in some of the specific applications that we consider. All of these are, in fact, consequences (see \cref{sec:model_on_K2d2d}) of a more general, but less explicit, condition, which we subsequently describe. The latter condition encapsulates a conceptual part in our proof in which the behavior of the model on $\Z^d$ is reduced to the study of the model on a $2d$-regular complete bipartite graph. It is isolated here in order to make explicit the basic input required by our argument in case this may be useful in future applications.

\medskip
\noindent\textbf{Alternative condition 1:}
Defining
\begin{equation}\label{eq:alpha-def1}
  \alpha_1 := \alpha_0 - \tfrac{1+\frac13 \1_{\{\rho_{\text{int}} \neq 0\}}}{2d} \log |\phasemax|,
\end{equation}
the condition is that
\begin{equation}\label{eq:parameter-inequalities-simple1}
\alpha_1 \ge \frac{C(\fq + \log d)\sqrt{\log d}}{d^{1/4}}\qquad\text{and}\qquad \frac{-\log \rho_{\text{int}}}{4\log (d\rho_{\text{act}})} \ge \min\left\{1, \frac {|\SS|}{2d} + \frac{5|\SS|\log(2d\rho_{\text{act}})}{\alpha_1 d} \right\},
\end{equation}
where $C>0$ is a universal constant, $\alpha_0$ was defined in~\eqref{eq:alpha-def}, $\phasemax$ is the collection of maximal patterns, $\phasedom$ is the collection of dominant patterns, and $\fq$ is defined by
\begin{equation}\label{eq:fq-def}
\fq := \log_2 |\{ \phasedom(I) \colon I\subset\SS\}|, \qquad\text{where}\qquad \phasedom(I) := \{(A,B) \in \phasedom : I\subset A,\, |A|\le |B| \}.
\end{equation}
As before, the condition simplifies for homomorphism models, becoming
\begin{equation}\label{eq:parameter-inequalities-simple-hom2}
-\log \max\big\{\rho_{\text{pat}}^{\text{bulk}},\rho_{\text{pat}}^{\text{bdry}} \big\} \ge \frac{C(\fq + \log d)\sqrt{\log d}}{d^{1/4}} + \frac{\log |\phasemax|}{2d}.
\end{equation}

Roughly speaking, $2^\fq$ is the number of possible answers to the question ``which dominant patterns have their small side containing a given set $I$?''.
In particular,
\begin{equation}\label{eq:frak q inequalities}
\log_2(\tfrac12|\phasedom|+1) \le \fq \le \log_2 |\phasemax| \le |\SS|.
\end{equation}
Using the latter two of these inequalities, it is easy to check that~\eqref{eq:parameter-inequalities-simple} implies \eqref{eq:parameter-inequalities-simple1}.
This condition leads to a significant improvement over~\eqref{eq:parameter-inequalities-simple} when $|\SS|$ is much larger than $\fq$ (when $|\SS|$ and $\fq$ are of the same order, this condition still leads to a slight improvement if $|\SS|$ grows with $d$).
This condition will provide better results for the clock model discussed in \cref{sec:clock-models} and in the applications discussed in \cref{sec:applications-multitype}.

\medskip
\noindent\textbf{Alternative condition 2:}
Let $s \ge 1$ be an integer and define
\begin{equation}\label{eq:rho-bulk-hat-and-rho-act-hat-def}
\hat\rho_{\text{pat}}^{\text{bulk}} := \max_{\substack{(A,B)\text{ non-dominant maximal}\\\text{pattern having }A,B \neq \emptyset}} \tfrac{\lambda_A \lambda_B}{\omega_\text{dom}} \left(1 + \rho_{\text{int}}^s \cdot \tfrac{\lambda_\SS}{\lambda_A}\right) \left( 2d \cdot \tfrac{\lambda_\SS}{\lambda_B} \right)^{\frac{(s-1)|\SS|}{2d}}, \qquad
\hat\rho_{\text{act}} := \tfrac{\lambda_\SS^2}{\omega_{\text{dom}}},
\end{equation}
and
\begin{equation}\label{eq:alpha-def2}
  \alpha_2 := -\log \max\left\{\hat\rho_{\text{pat}}^{\text{bulk}},~1-(1-\rho_{\text{pat}}^{\text{bdry}})(1-\sqrt{\rho_{\text{int}}})\right\} - \tfrac{1+\frac13 \1_{\{\rho_{\text{int}} \neq 0\}}}{2d}\log |\phasemax|.
\end{equation}
The condition is that there exists $s$ such that
\begin{equation}\label{eq:s-ineq}
 \tfrac{2\log (d\hat\rho_{\text{act}})}{-\log \rho_{\text{int}}} \le s \le \min\left\{ \left\lceil \tfrac {2d}{|\SS|} \right\rceil, 1+ \tfrac{\alpha_2 d}{2|\SS| \log(2d\hat\rho_{\text{act}})} \right\}
\end{equation}
and
\begin{equation}\label{eq:parameter-inequalities-simple2}
\alpha_2 \ge \frac{C(\fq + \log d)\sqrt{\log d}}{d^{1/4}} ,
\end{equation}
where $C>0$ is a universal constant.

While clearly $\hat\rho_{\text{pat}}^{\text{bulk}} \ge \rho_{\text{pat}}^{\text{bulk}}$ (for any choice of $s$) and so $\alpha_2 \le \alpha_1$, we will see that this alternative condition holds whenever the previous alternative condition holds.
For homomorphism models, taking $s=1$ yields that $\hat\rho_{\text{pat}}^{\text{bulk}}=\rho_{\text{pat}}^{\text{bulk}}$, so that the condition is identical to the previous~\eqref{eq:parameter-inequalities-simple-hom2}. For non-homomorphism models, the role of the original $\rho_{\text{act}}$ is split into two here: $\hat\rho_{\text{act}}$ captures the part pertaining to the dominant patterns, whereas the part pertaining to other maximal patterns is instead incorporated into the definition of $\hat\rho_{\text{pat}}^{\text{bulk}}$.

We give yet another more relaxed condition, which is specialized to homomorphism models.

\medskip
\noindent\textbf{Alternative condition 3 (for homomorphism models only):}
\begin{equation}\label{eq:parameter-inequalities-simple3}
\alpha_3 := -\log \max\big\{\rho_{\text{pat}}^{\text{bulk*}},~ \rho_{\text{pat}}^{\text{bdry}}\big\} \ge \frac{C(\fq + \log d)\sqrt{\log d}}{d^{1/4}} ,
\end{equation}
where $C>0$ is a universal constant, $\mathfrak{q}$ is defined in~\eqref{eq:fq-def},
\begin{equation}\label{eq:rho pat bulk star}
  \rho_{\text{pat}}^{\text{bulk*}} := \frac{1}{\omega_\text{dom}} \Bigg(\sum_{\substack{(A,B)\text{ non-dominant}\\\text{maximal pattern}}} \lambda^{\langle 2d \rangle}_A \lambda_B^{2d}\Bigg)^{1/2d} ,
\end{equation}
and $\lambda_A^{\langle 2d \rangle}$ is the total weight $\sum_f \prod_{i \in [2d]} \lambda_{f(i)}$ of all functions $f \colon [2d] \to A$ whose image is not contained in any side of a maximal pattern that is strictly contained in $A$. Here and later, we use $[n]:=\{1,2,\ldots, n\}$ for integer $n\ge1$.

Note that $\lambda_A^{2d}$ is the total weight of all functions $f \colon [2d] \to A$.
In particular, $\lambda_A^{\langle 2d \rangle} \le \lambda_A^{2d}$, so that $\rho_{\text{pat}}^{\text{bulk*}} \le \rho_{\text{pat}}^{\text{bulk}} \cdot |\phasemax|^{1/2d}$.
Using this, it is straightforward that \eqref{eq:parameter-inequalities-simple-hom2} implies \eqref{eq:parameter-inequalities-simple3}. This latter condition may be more effective than the former one in models in which there are a large number of non-dominant maximal patterns whose weights are much smaller than the ones achieving the maximum defining $\rho_{\text{pat}}^{\text{bulk}}$ in \eqref{eq:rho pat bulk and bdry def}, or when $\lambda_A^{\langle 2d \rangle}$ is significantly smaller than $\lambda_A^{2d}$ for some non-dominant maximal patterns. For example, this condition will be used for the multi-type Widom--Rowlinson model (\cref{sec:multi type Widom Rowlinson}) in the regime of large number of types.

\medskip
\noindent\textbf{The general condition:}
We proceed to describe the general condition that our proof uses.
The condition is closely related to the behavior of the model on the complete bipartite $2d$-regular graph $K_{2d,2d}$.
For a collection $\Psi$ of functions $\psi \colon [2d] \to \SS$ and a set $I \subset \SS$, denote
\begin{equation}\label{eq:Z-Psi-I-def}
Z(\Psi,I) := \sum_{\psi \in \Psi} \left( \prod_{j=1}^{2d} \lambda_{\psi(j)} \right) \left( \sum_{i \in I} \lambda_i \prod_{j=1}^{2d} \lambda_{i,\psi(j)} \right)^{2d} .
\end{equation}
Note that $Z(\SS^{[2d]},\SS)$ is the partition function of the spin system on $K_{2d,2d}$.
In particular, for non-weighted homomorphism models, the partition function $Z(\SS^{[2d]},\SS)$ precisely equals $|\Hom(K_{2d,2d},H)|$. Thus, one may regard $Z(\Psi,I)$ as a restricted/partial partition function on $K_{2d,2d}$, where $\Psi$ and $I$ provide information about the values on the left and right sides of $K_{2d,2d}$ (i.e., its two partition classes).
In particular, $I$ will correspond to the possible values of a configuration $f$ at a vertex $v$ of $\Z^d$ and $\Psi$ to the possible values at its neighbors.
For example, note that $Z(A^{[2d]},B)=(\lambda_A \lambda_B)^{2d}$ for any pattern $(A,B)$.

For a set $I \subset \SS$, let $R(I)$ be the set of all $j \in \SS$ such that $\lambda_{i,j} = \lambda^{\text{int}}_{\text{max}}$ for all $i \in I$.
Given $I,J \subset \SS$, we write $I \simeq_R J$ if $R(I)=R(J)$.
Let $\Psi_J$ denote the collection of functions $\psi$ such that $\psi([2d]) \simeq_R J$. Let $\Psi^1_{J,\epsilon}$ denote the collection of functions $\psi \in \Psi_J$ for which there exists $I \subsetneq J$ such that $I$ is a side of a dominant pattern and $|\psi^{-1}(I)| > 2d-4\epsilon d$. Let $\Psi^2_{J,\bar\epsilon}$ denote the collection of functions $\psi \in \Psi_J$ for which there exists $I \subsetneq J$ such that $I \not\simeq_R J$ and $|\psi^{-1}(I)| > 2d-4\bar\epsilon d$.
Denote $\Psi_{J,\epsilon,\bar\epsilon} := \Psi_J \setminus (\Psi^1_{J,\epsilon} \cup \Psi^2_{J,\bar\epsilon})$.
For $\Psi \subset \Psi_J$, let $k_\Psi$ be the number of indices $j \in [2d]$ such that $\{ \psi(j) : \psi \in \Psi\} \not\simeq_R J$.

\begin{cond}\label{main-cond}
	We have $\lambda^{\text{int}}_{\text{max}}=1$ and there exist $\alpha,\gamma \ge 0$ and $\frac{1}{4d} \le \epsilon,\bar\epsilon \le \frac 12$ satisfying that
\begin{equation}\label{eq:alpha-cond}
c\alpha \ge \frac{(\fq+\log d)\sqrt{\log d}}{d^{1/4}} + \frac{(\fq+\log d)\log d}{\epsilon^2 d} + \gamma d + \sqrt{\gamma (\fq+\log d)d^{3/2} \log d}
\end{equation}
	with $c>0$ a universal constant,
	and such that for any side $J$ of a dominant pattern, we have
	\begin{align}
	 Z(\Psi,\SS) &\le \omega_{\text{dom}}^{2d} \cdot e^{2\gamma d - \alpha k_\Psi} \qquad\qquad\text{for any }\Psi \subset \Psi_{J,\epsilon,\bar\epsilon} , \label{eq:cond-restricted-left} \\
	 Z(\Psi_{J,\epsilon,\bar\epsilon},I) &\le \omega_{\text{dom}}^{2d} \cdot e^{2\gamma d - \alpha d} \qquad\qquad\text{for any }I \subset \SS\text{ such that }R(J) \not\subset R(R(I)) , \label{eq:cond-restricted-right}\\
	Z(\Psi_J \setminus \Psi_{J,\epsilon,\bar\epsilon},\SS) &\le \omega_{\text{dom}}^{2d} \cdot e^{2\gamma d - \alpha d} , \label{eq:cond-unbalanced}\\
	Z(\Psi_{J,\epsilon,\bar\epsilon},\SS \setminus R(J)) &\le \omega_{\text{dom}}^{2d} \cdot e^{2\gamma d - 3\alpha\epsilon d^2} . \label{eq:cond-highly-energetic}
 	\end{align}
Furthermore,
	\begin{equation}\label{eq:cond-non-dominant}
	\sum_{\substack{I \subset \SS\text{ side of maximal}\\\text{non-dominant pattern}}} Z(\Psi_I,\SS) \le \omega_{\text{dom}}^{2d} \cdot e^{2\gamma d - \alpha d} .
	\end{equation}
\end{cond}

We mention that that the assumption that $\lambda^{\text{int}}_{\text{max}}=1$ is merely a convenient normalization of the pair interactions.
Let us explain the various conditions. Let $J$ be a side of a dominant pattern and let $\Psi \subset \Psi_J$. Observe that
\[ Z(\Psi,R(J)) \le Z(\Psi_J,R(J)) \le Z(J^{[2d]},R(J)) = (\lambda_J \lambda_{R(J)})^{2d} = \omega_{\text{dom}}^{2d} .\]
For homomorphism models, this upper bound also applies to $Z(\Psi,\SS)$, as $Z(\Psi,\SS)=Z(\Psi,R(J))$. In fact, in high dimensions, $Z(\Psi_J,R(J))$ is very close to $\omega_{\text{dom}}^{2d}$ so that this bound is rather accurate.
For general models, this may not yield an upper bound on $Z(\Psi,\SS)$ due to the increase in possibilities coming from configurations having at least one edge whose spins interact via a non-maximal interaction weight (i.e., when some site on the right side takes a value in $\SS \setminus R(J)$).
The term $e^{2\gamma d}$ may be thought of as a factor which compensates for this. Thus, the bound $\omega_{\text{dom}}^{2d} \cdot e^{2\gamma d}$ may in some sense be thought of as the base to which we compare other bounds. The parameter $\alpha$ quantifies the improvement over this bound in various situations.
The term $e^{-\alpha k_{\Psi}}$ in~\eqref{eq:cond-restricted-left} serves to improve on this bound in the presence of information on the values of sites on the left side of $K_{2d,2d}$ (which is given in the form of a subset $\Psi$ of $\Psi_{J,\epsilon,\bar\epsilon}$) which restricts the set of possible values of some such sites (the smaller $\Psi$ is, the larger $k_{\Psi}$ is). Condition~\eqref{eq:cond-restricted-right} states that in the presence of information on the values of sites on the right side of $K_{2d,2d}$ (given in the form of a subset $I$ of $\SS$) which eliminates the possibility of some values in $R(J)$, one again has an improvement on the base bound. Condition~\eqref{eq:cond-unbalanced} says the same in the case when the values on the left side of $K_{2d,2d}$ are sufficiently imbalanced (where $\epsilon$ and $\bar\epsilon$ quantify this). Condition~\eqref{eq:cond-highly-energetic} controls the contribution to the partition function when the possible values on the right side of $K_{2d,2d}$ are restricted to those which interact with lower-than-maximum interaction weight with some site on the left side (and hence with many such sites, since the values on the left side are assumed to be balanced). Condition~\eqref{eq:cond-non-dominant} states that the contribution arising from all non-dominant maximal patterns is much smaller than the base bound.

For homomorphism models satisfying~\eqref{eq:parameter-inequalities-simple-hom}, \eqref{eq:parameter-inequalities-simple-hom2}, \eqref{eq:parameter-inequalities-simple2} or~\eqref{eq:parameter-inequalities-simple3}, it will not be very difficult to verify \cref{main-cond} with $\alpha$ as in the corresponding condition, $\gamma=0$, $\epsilon = \min \{\frac{\alpha}{64 \log d}, \frac18 \}$ and $\bar\epsilon = \frac1{4d}$. For general models satisfying~\eqref{eq:parameter-inequalities-simple} or~\eqref{eq:parameter-inequalities-simple1}, the verification of this condition with $\alpha$ as in~\eqref{eq:alpha-def} or~\eqref{eq:alpha-def1}, respectively, requires a more delicate analysis. This will be carried out in \cref{sec:model_on_K2d2d}.

\section{Applications and extensions}\label{sec:applications}

We have already described several applications of our results in \cref{sec:first applications}. In this section we go beyond these first applications, demonstrating the wide applicability of our general setup and the ease in which it may be used to obtain new results.

We begin in \cref{sec:hard-core-and-Widom-Rowlinson-models} by explaining our results in the context of two well-known models: the hard-core model (\cref{sec:hard-core-model}) which was studied by Dobrushin, Galvin, Kahn, Samotij and many others, and the Widom--Rowlinson model (\cref{sec:Widom-Rowlinson model}) which was studied by Burton, Gallavotti, Higuchi, Lebowitz, Steif, Takei and others. We also consider the hard-core model with unequal sublattice activities (\cref{sec:hard-core-unequal-sublattice-activities-model}), previously studied by van den Berg, H{\"a}ggstr{\"o}m and Steif.

In \cref{sec:applications-multitype}, we investigate various multi-type versions of the hard-core, Widom--Rowlinson and beach models: The multi-type Widom--Rowlinson model (\cref{sec:multi type Widom Rowlinson}) was studied by Lebowitz and Runnels, as was the anti-Widom--Rowlinson model (\cref{sec:dilute proper colorings}) which may also be viewed as a dilute proper coloring model. The multi-type beach model (\cref{sec:multi type beach model}) was studied by Burton, H{\"a}ggstr{\"o}m, Hallberg and Steif. The multi-occupancy hard-core model (\cref{sec:multi-occupancy hard-core model}) was studied by Mazel and Suhov.

In \cref{sec:AF Ising and AF Potts with magnetic field}, we discuss new phenomena, including a re-entrant phase, that arise when introducing an external magnetic field in the antiferromagnetic Ising model (\cref{sec:Ising antiferromagnet}) and the antiferromagnetic Potts model (\cref{sec:AF Potts with magnetic field}).

In \cref{sec:extensions}, we discuss various extensions of our results, which allow to analyze models which a priori do not satisfy the assumptions of our theorems (including certain models of homomorphisms to infinite graphs and cases of non-equivalent dominant patterns).

Finally, we revisit our first applications in \cref{sec:first applications_revisited} to elaborate on the results of \cref{sec:first applications}.

\smallskip

For each of the models discussed below, we divide the parameter space into regions in which different patterns are dominant. We then determine subsets of these regions in which our conditions are verified. In these subsets, one may apply each of our main theorems. Specifically, \cref{thm:long-range-order} (and its extensions in \cref{sec:large-violations}) quantifies the order present in the system when boundary conditions corresponding to a dominant pattern are imposed. \cref{thm:existence_Gibbs_states} shows that each dominant pattern gives rise to an extremal, maximal-pressure Gibbs state invariant under automorphisms of $\Z^d$ preserving the two sublattices. Lastly, \cref{thm:characterization_of_Gibbs_states} shows that all (periodic) maximal-pressure Gibbs states are a mixture of the Gibbs states arising from the dominant patterns.

In some of the applications, we also indicate the disordered regime which follows from an application of Dobrushin's uniqueness condition~\cite{Dobrushin1968TheDe}. In some cases this may be improved by an application of the van den Berg--Maes disagreement percolation condition~\cite{van1994disagreement}, which compares total variation distances to the Bernoulli site percolation threshold $p_c(\Z^d)$, though the improvement (necessarily) becomes less significant as the dimension grows~\cite[Section 3]{van1994disagreement}.

We write $C,c, C',c', C_d, c_d, C_q, c_q$ (and so on) for positive constants, which may change from line to line (with large constants only increasing, and small constants decreasing), and use the convention that subscripts indicate dependency (no subscript indicates a universal constant). We use the notation $a\approx b$ in the sense that $\frac{a}{b}$ is bounded away from zero and infinity by universal constants. We write $N(v)$ for the neighborhood of a vertex $v$ in a given graph.

\subsection{The hard-core and Widom--Rowlinson models}\label{sec:hard-core-and-Widom-Rowlinson-models}
In this section we apply our results to the closely related hard-core and Widom--Rowlinson models.

\subsubsection{The hard-core model}\label{sec:hard-core-model}\label{sec:hard-core model}
The hard-core model is a well-known model of a lattice gas and has been the subject of extensive study. In this model, every site may be occupied by a single particle or is otherwise vacant, and there is a constraint that adjacent sites cannot both be occupied (in combinatorial terms, configurations correspond to independent sets). The model at activity $\lambda>0$ may be described within our general setup by choosing
\begin{equation*}
  \SS=\{0,1\},\qquad \lambda_i=\lambda^i, \qquad \lambda_{i,j}=\1_{\{ij=0\}} ,
\end{equation*}
as illustrated by \cref{fig:graph-hard-core}.
It is straightforward to apply Dobrushin's uniqueness condition to the model and deduce that is disordered when $\lambda<\frac{1}{2d-1}$ (disagreement percolation yields the improved condition $\lambda<\frac{p_c(\Z^d)}{1-p_c(\Z^d)}$). The seminal work of Dobrushin~\cite{dobrushin1968problem} established that the model exhibits a phase transition, showing that $\lambda \ge C^d$ suffices for the existence of two Gibbs states, each characterized by unequal occupancy densities on the even and odd sublattices. Galvin--Kahn~\cite{galvin2004phase} significantly improved this, showing that $\lambda \ge Cd^{-1/4} \log^{3/4} d$ suffices, and thereby showing that the threshold decays to zero with the dimension. The best-known sufficient condition for large $d$ is $\lambda \ge Cd^{-1/3} \log^2 d$, which is due to Samotij and the first author~\cite{peled2014odd}. Our results yield the weaker condition $\lambda \ge Cd^{-1/4} \log^{3/2} d$, but also establish that, under this condition, any periodic Gibbs state is a mixture of the two ordered Gibbs states (any periodic Gibbs state has maximum pressure in this model). This is verified by condition~\eqref{eq:parameter-inequalities-simple-hom} upon checking that the maximal patterns are $(\{0\},\{0,1\}),(\{0,1\},\{0\})$, both of which are dominant, so that $\rho_{\text{pat}}^{\text{bulk}} = 0$ and $\rho_{\text{pat}}^{\text{bdry}} = \frac{1}{1+\lambda}$. We remark that it is not known whether there exists $\lambda_c(d)$ such that the model has a unique Gibbs state if $\lambda<\lambda_c(d)$ and multiple Gibbs states if $\lambda>\lambda_c(d)$ (but see~\cite{brightwell1999nonmonotonic, haggstrom2002monotonicity}).

A multi-occupancy extension of the hard-core model is discussed in \cref{sec:multi-occupancy hard-core model}.

The hard-core model can be obtained as a limiting case of the antiferromagnetic Ising model with external field. This viewpoint and the implication of our results for the antiferromagnetic Ising model are discussed in \cref{sec:Ising antiferromagnet}.

\begin{figure}[!t]
	\tikzstyle{every state}=[circle,fill=gray!25,draw=black,minimum size=16pt,inner sep=0pt,text=black]

	\begin{subfigure}[t]{.25\textwidth}
		\centering
		\begin{tikzpicture}[auto,node distance=1cm]
		\node[state] (0) {0};
		\node[state] (1) [right of=0] {1};
		\path (0) edge (1);
		\path (0) edge [loop above] ();
		\end{tikzpicture}
		\caption{The hard-core model}
		\label{fig:graph-hard-core}
	\end{subfigure}%
	\begin{subfigure}[t]{.06\textwidth}~\end{subfigure}%
	\begin{subfigure}[t]{.25\textwidth}
		\centering
		\begin{tikzpicture}[auto,node distance=1cm]
		\node[state] (1) {2};
		\node[state] (2) [right of=1] {3};
		\node[state] (-1) [left of=1] {1};
		\node[state] (-2) [left of=-1] {0};
		\path (1) edge (2);
		\path (-1) edge (1);
		\path (-2) edge (-1);
		\end{tikzpicture}
		\caption{Equivalent of the hard-core model}
		\label{fig:graph-hard-core-alt}
	\end{subfigure}%
	\begin{subfigure}[t]{.06\textwidth}~\end{subfigure}%
	\begin{subfigure}[t]{.26\textwidth}
		\centering
		\begin{tikzpicture}[auto,node distance=1cm]
		\node[state] (0) {0};
		\node[state] (1) [right of=0] {$+$};
		\node[state] (-1) [left of=0] {$-$};
		\path (0) edge (1);
		\path (0) edge (-1);
		\path (0) edge [loop above] ();
		\path (1) edge [loop above] ();
		\path (-1) edge [loop above] ();
		\end{tikzpicture}
		\caption{The Widom--Rowlinson model}
		\label{fig:graph-widom-rowlinson}
	\end{subfigure}
	\caption{Graph representations of the hard-core and Widom--Rowlinson models. The edges correspond to the pairs of states $\{i,j\}$ with maximal pair interaction $\lambda_{i,j}$.}
\end{figure}

\subsubsection{The hard-core model with unequal sublattice activities}\label{sec:hard-core-unequal-sublattice-activities-model}
The existence of Gibbs states with different occupancy densities on the two sublattices in the hard-core model is an example of spontaneous symmetry breaking --- the even/odd sublattice symmetry of the model is absent in the Gibbs state. Van den Berg and Steif~\cite{van1994percolation} consider a variant of the hard-core model which lacks this symmetry to begin with: two \emph{distinct} activities $\lambda_e,\lambda_o>0$ are given, and an independent set $I$ is chosen (in finite volume) with probability proportional to $\lambda_e^{|I \cap \Even|} \lambda_o^{|I \cap \Odd|}$. They conjecture that this model always has a \emph{unique} Gibbs state. For any $x \in \R$ they prove that this is the case when $(\lambda_e,\lambda_o)=(e^{x+h},e^{x-h})$ for all but at most countably many values of $h \in \R$. They also show that this is the case whenever $\max\{\lambda_e,\lambda_o\} < \frac{p_c(\Z^d)}{1-p_c(\Z^d)}$. H{\"a}ggstr{\"o}m~\cite{haggstrom1997ergodicity} establishes the conjecture in full in two dimensions. As we now explain, our results prove uniqueness of the Gibbs state whenever
\begin{equation}\label{eq:hard-core-unequal-sublattice-densitiy-threshold}
\frac{|\lambda_e - \lambda_o|}{1+\min\{\lambda_e,\lambda_o\}} \ge \frac{C \log^{3/2} d}{d^{1/4}} .
\end{equation}
The model is essentially equivalent to the spin system obtained in our setting when
\begin{equation}\label{eq:hard-core model unequal densities framework}
  \SS=\{0,1,2,3\},\qquad (\lambda_0,\lambda_1,\lambda_2,\lambda_3)=(\lambda_{\text{e}},1,1,\lambda_{\text{o}}), \qquad \lambda_{i,j}=\1_{\{|i-j|=1\}} ,
\end{equation}
as illustrated by \cref{fig:graph-hard-core-alt}.
We think of the states 0, 1, 2 and 3 as representing ``occupied-even'', ``vacant-odd'', ``vacant-even'' and ``occupied-odd'', respectively. The names are motivated by the fact that if a configuration assigns an even state to a single even site, then it must assign even states to all even sites. Restricting the model~\eqref{eq:hard-core model unequal densities framework} to such ``even configurations'' yields exactly the above hard-core model with unequal sublattice densities. Thus Gibbs states of the latter model can be identified with the Gibbs states of the model~\eqref{eq:hard-core model unequal densities framework} which are supported on even configurations. This is a particular instance of a bipartite covering system; see \cref{sec:covering systems}.

Let us explain how our results apply. The maximal patterns are $(\{1\},\{0,2\}),(\{2\},\{1,3\}), (\emptyset,\SS)$ and their reversals. Suppose for concreteness that $\lambda_e>\lambda_o$. Then the dominant patterns are $(\{0,2\},\{1\})$ and $(\{1\},\{0,2\})$, and we have $\rho_{\text{pat}}^{\text{bulk}} = \frac{1+\lambda_o}{1+\lambda_e}$ and $\rho_{\text{pat}}^{\text{bdry}}=\frac{1}{1+\lambda_e}$. Using~\eqref{eq:parameter-inequalities-simple-hom}, we see that when~\eqref{eq:hard-core-unequal-sublattice-densitiy-threshold} holds then the two dominant patterns give rise to two extremal Gibbs states, invariant under automorphisms preserving the two sublattices, and every periodic Gibbs state (necessarily of maximal pressure) is a mixture of these two. In this case, the hard-core model with unequal sublattice activities has a unique periodic Gibbs state, and therefore, by the anti-monotonicity of the model, it in fact has a unique Gibbs state (see~\cite[Lemma~3.2]{van1994percolation}).

\subsubsection{The Widom--Rowlinson model} \label{sec:Widom-Rowlinson model}
The lattice \emph{Widom--Rowlinson model} at activity $\lambda>0$ is obtained when
	\[ \SS=\{-1,0,1\},\qquad\lambda_i=\lambda^{|i|},\qquad\lambda_{i,j}=\mathbf{1}_{\{ij \neq -1\}} ,\]
	as illustrated by \cref{fig:graph-widom-rowlinson}.
	One may regard configurations in this model as describing the territory occupied by two competing species (represented by ``1'' and ``$-1$'') which cannot be adjacent to one another.

    Dobrushin's uniqueness theorem shows that the model has a unique Gibbs state when $\lambda<\frac{1}{2d-1}$ while disagreement percolation yields the condition $\lambda<\frac{p_c(\Z^d)}{2(1-p_c(\Z^d))}$. Both are improved by Higuchi~\cite{higuchi1983applications} who finds the condition $\lambda<\frac{p_c(\Z^d)}{1-p_c(\Z^d)}$. In the other direction, Lebowitz--Gallavotti~\cite{lebowitz1971phase} proved that, in any dimension $d\ge 2$, there are at least two Gibbs states $\mu_{-1}$ and $\mu_{1}$ when $\lambda$ is sufficiently large, with the measure $\mu_i$ featuring a prevalence of state $i$. Burton--Steif~\cite{burton1995new} show a similar result for the equivalent iceberg model (equivalence holds for rational $\lambda$). Both proofs appear to require $\lambda$ to grow exponentially with the dimension, though this is not stated explicitly. In two dimensions, Higuchi~\cite{higuchi1983applications} shows the existence of two Gibbs states when $\lambda>\frac{8p_c(\Z^2)}{1-p_c(\Z^2)}$.

    To apply our results, one notes that the maximal patterns of the model are $(\{-1,0\},\{-1,0\})$, $(\{0,1\},\{0,1\})$, $(\{0\},\{-1,0,1\})$ and $(\{-1,0,1\},\{0\})$, so that the first two patterns are always dominant, $\rho_{\text{pat}}^{\text{bulk}} = \frac{1+2\lambda}{(1+\lambda)^2}$ and $\rho_{\text{pat}}^{\text{bdry}} = \frac{1}{1+\lambda}$. Condition~\eqref{eq:parameter-inequalities-simple-hom} is then verified in the regime $\lambda \ge C d^{-1/8}\log^{3/4} d$ and shows the existence of two automorphism-invariant Gibbs states, assigning different densities to the states $-1$ and $1$, and further implies that every periodic Gibbs state is a mixture of these two (as all periodic Gibbs states have maximal pressure in this model). This characterization of periodic Gibbs states is new in all dimensions $d \ge 3$ (in two dimensions it was shown by Higuchi--Takei~\cite{higuchi2004some}; see also~\cite{carstens2012percolation}).

    However, it turns out that for the Widom--Rowlinson model an alternate approach leads to improved bounds. It has been noted~\cite[Section 5]{brightwell1999nonmonotonic} that the Widom--Rowlinson model on $\Z^d$ is equivalent to the hard-core model on $\Z^d\times\{0,1\}$ (by $\times$ we mean the Cartesian product of $G$ with $\{0,1\}$, sometimes denoted as $G\square\{0,1\}$) using the map $T(\sigma)_v = s_v (\sigma_{v,0} - \sigma_{v,1})$ with $s_v$ equalling $\pm1$ according to the parity of $v$ (see \cref{sec:projection-systems} for further discussion). Results which are valid for the hard-core model on $\Z^d\times\{0,1\}$ at activity $\lambda$ thus transfer to the Widom--Rowlinson model on $\Z^d$ at activity $\lambda$. Though not mentioned explicitly there, this is the case for the result of~\cite{peled2014odd}, which then implies that when $\lambda \ge Cd^{-1/3} \log^2 d$ the Widom--Rowlinson model on $\Z^d$ has at least two automorphism-invariant Gibbs states which differ in the densities of the states $-1$ and $1$. Our results also apply to the hard-core model on $\Z^d\times\{0,1\}$, and imply that when $\lambda \ge Cd^{-1/4} \log^{3/2} d$ any periodic Gibbs state of the Widom--Rowlinson model on $\Z^d$ is a mixture of the two ordered Gibbs states.

    The Widom--Rowlinson model has an asymmetric version, obtained by assigning \emph{distinct} single-site activities $\lambda_{-1}, \lambda_1$ to the states $-1, 1$, which is conjectured to always have a unique Gibbs state~\cite[Section 3.5]{GeoHagMae01}. A partial result of this type is shown by Higuchi--Takei~\cite{higuchi2004some} in two dimensions. It is straightforward to check that the asymmetric model on $\Z^d$ is equivalent to the hard-core model with unequal sublattice activities on $\Z^d\times\{0,1\}$ via the same map $T$ as above, with $\{\lambda_{-1},\lambda_1\} = \{\lambda_e, \lambda_o\}$. With the latter replacement, uniqueness in the asymmetric Widom--Rowlinson model follows under the condition~\eqref{eq:hard-core-unequal-sublattice-densitiy-threshold}.

    Multi-type extensions of the Widom--Rowlinson model are analyzed in \cref{sec:multi type Widom Rowlinson} and \cref{sec:dilute proper colorings}.

\subsection{Multi-type models}\label{sec:applications-multitype}

We discuss here multi-type variants of the Widom--Rowlinson, beach and hard-core models.
The multi-type Widom--Rowlinson model was introduced and studied by Runnels--Lebowitz~\cite{runnels1974phase}, and the multi-type beach model by Burton--Steif~\cite{burton1995new}.
The multi-occupancy hard-core model was studied by Mazel--Suhov~\cite{mazel1991random}. In all three cases, our results offer significant improvements over known results, pinpointing for every fixed number of types $q$, the critical activity as the dimension tends to infinity, whereas previously best-known results yielded bounds with large gaps (with the upper bound typically tending to infinity exponentially fast and the lower bound tending to zero).
In our analysis, to obtain a sharper dependence on the number of types, we will sometimes use the alternative conditions described in \cref{sec:alternative-conditions}.

\subsubsection{The multi-type Widom--Rowlinson model}\label{sec:multi type Widom Rowlinson}
Motivated by theories of liquid crystals, Runnels--Lebowitz~\cite{runnels1974phase} considered an extension of the Widom--Rowlinson model in which there are $q \ge 3$ competing species (the  $q=2$ case corresponds to the original model discussed in \cref{sec:first applications}). The model at activity $\lambda>0$ may be described within our general setup by choosing
\[ \SS = \{0,1,\dots,q\}, \qquad \lambda_i=\lambda^{\1_{\{i \neq 0\}}},\qquad\lambda_{i,j}=\mathbf{1}_{\{ij=0\text{ or }i=j \}},\]
as illustrated by \cref{fig:graph-widom-rowlinson-multi-type}. Dobrushin's uniqueness condition shows that the system is disordered when
\begin{equation*}
  \lambda<\frac{1}{(2d-1)q-2d}
\end{equation*}
(this is improved in low dimensions by the disagreement percolation method of van den Berg--Maes~\cite{van1994disagreement} which yields the condition $\frac{\lambda q}{1+\lambda q} < p_c(\Z^d)$).
Runnels--Lebowitz discuss two types of ordered phases: a \emph{demixed phase} in which the symmetry between the species is broken and where the order corresponds to a pattern of the form $(\{0,i\},\{0,i\})$ for some $1 \le i \le q$ (analogous to the ordered phase of the $q=2$ model), and a \emph{crystallized phase} in which the symmetry between the sublattices is broken and where the order corresponds to $(\{0\},\{0,1,\dots,q\})$ or $(\{0,1,\dots,q\},\{0\})$. They show that, for any fixed~$q$, the system has a demixed phase when $\lambda$ is sufficiently large, and that, for any fixed $\lambda$, the system has a crystallized phase when $q$ is sufficiently large.
Their proof of existence of a demixed phase seems to yield the condition $\lambda \ge q^{Cd}$, and their condition for the existence of a crystallized phase is that $q \min\{\lambda, \lambda^{-4} \}$ is sufficiently large as a function of $d$ (in particular, they require $q$ to be large).

Our results capture both types of ordered phases and significantly improve the conditions for their emergence. In particular, we show that for all fixed $q \ge 3$ the model undergoes both transitions in high dimensions, with the disorder-crystallized transition occurring around $\lambda = 0$ and the crystallized-demixed transition occurring around $\lambda = q-2$.

We now describe our results in detail. The maximal patterns of the system are $(\{0\},\{0,1,\dots,q\})$ or $(\{0,1,\dots,q\},\{0\})$ (crystallized phase), having weight $1+\lambda q$, and $(\{0,i\},\{0,i\})$ for some $1 \le i \le q$ (demixed phase), having weight $(1+\lambda)^2$. When $\lambda<q-2$, the patterns corresponding to the crystallized phase are dominant, $\rho_{\text{pat}}^{\text{bulk}}=\frac{(1+\lambda)^2}{1+\lambda q}$, $\rho_{\text{pat}}^{\text{bdry}}=\frac{1+\lambda}{1+\lambda q}$ and $\mathfrak{q}=1$.
Thus, condition~\eqref{eq:parameter-inequalities-simple-hom2} is verified when
\begin{equation}\label{eq:MTWR condition 1}
\log \left[ \frac{1+\lambda q}{(1+\lambda)^2} \right] \ge \frac{C\log^{3/2} d}{d^{1/4}} + \frac{\log (q+2)}{2d},
\end{equation}
whence our results are applicable and imply ordering according to the crystallized phase. For example, for fixed $q$ and large $d$ (it suffices that the second term on the right-hand side of~\eqref{eq:MTWR condition 1} does not exceed the first term), the condition is satisfied when
\begin{equation}
\frac{C \log^{3/2} d}{q d^{1/4}} \le \lambda \le q-2 - \frac{C q \log^{3/2} d}{d^{1/4}},
\end{equation}
while, for fixed $d$ and large $q$, the condition is satisfied when $C_d q^{-(1-1/2d)} \le \lambda \le c_d q^{1-1/2d}$. An improved lower bound may be obtained in the latter regime by appealing to another one of our conditions. The parameter $\rho_{\text{pat}}^{\text{bulk*}}$ of~\eqref{eq:rho pat bulk star} satisfies $\rho_{\text{pat}}^{\text{bulk*}} = \frac{q^{1/2d} (1+\lambda) ((1+\lambda)^{2d}-1)^{1/2d}}{1+\lambda q}$, so that condition~\eqref{eq:parameter-inequalities-simple3} is verified when
\begin{equation}
\log\left(\frac{1+\lambda q}{1+\lambda}\cdot \min\left\{\left(q\left((1+\lambda)^{2d}-1\right)\right)^{-1/2d}, 1\right\}\right) \ge \frac{C\log^{3/2} d}{d^{1/4}},
\end{equation}
In particular, for fixed $d$ and large $q$, one may check (using the inequality $(1+\lambda)^{2d}-1\le Cd\lambda$ in the regime $\lambda\le\frac{c}{d}$ and the inequality $(1+\lambda)^{2d}-1\le (1+\lambda)^{2d}$ in the remaining range) that the condition is satisfied when
\begin{equation}
\frac{C_d}{q}\le \lambda \le c_d q^{1-1/2d}.
\end{equation}

When $\lambda>q-2$, the patterns corresponding to the demixed phase are dominant, $\rho_{\text{pat}}^{\text{bulk}}=\frac{1+\lambda q}{(1+\lambda)^2}$, $\rho_{\text{pat}}^{\text{bdry}}=\frac{1}{1+\lambda}$ and $\fq=\log_2(q+2)$. Condition~\eqref{eq:parameter-inequalities-simple-hom2} is thus verified when
\begin{equation}
\log \left[ \frac{(1+\lambda)^2}{1+\lambda q} \right] \ge \frac{C\log(dq)\sqrt{ \log d}}{d^{1/4}},
\end{equation}
whence our results are applicable and imply ordering according to the demixed phase. For example, for fixed $q$ and large $d$ the condition is satisfied when
\begin{equation}
  \lambda\ge q - 2 + \frac{C_q \log^{3/2} d}{d^{1/4}},
\end{equation}
while for fixed $d$ and large $q$ the condition is satisfied when
\begin{equation}
  \lambda\ge q^{1+\frac{C\log d}{d^{1/4}}}.
\end{equation}

We add also that Chayes--Koteck\'y--Shlosman~\cite{chayes1995aggregation} considered a site-diluted version of the $q$-state ferromagnetic Potts model, for which the above multi-type Widom--Rowlinson model is the zero-temperature limit. The work~\cite{chayes1995aggregation} extends the results of Runnels--Lebowitz~\cite{runnels1974phase} to the positive temperature setting and further discusses the nature of the phase transition and considers other related models. Our analysis above is only for the original multi-type Widom--Rowlinson model but we remark that our main results can also be applied in the positive temperature setting. Other models with site dilution (equivalently, models to which a ``safe symbol'' is added) may also be handled, as the next example demonstrates, though we have not attempted a general discussion of this type of application.

\begin{figure}[!t]
	\tikzstyle{every state}=[circle,fill=gray!25,draw=black,minimum size=16pt,inner sep=0pt,text=black]
	\begin{subfigure}[t]{.49\textwidth}
		\centering
		\begin{tikzpicture}[auto,node distance=1.25cm]
		\node[state] (0) {0};
		\node[state] (1) [below of=0,xshift=-1.75cm] {1};
		\node[state] (2) [right of=1] {2};
		\node[state] (3) [right of=2] {3};
		\node[state] (4) [right of=3] {4};
		\path (0) edge [loop left] (0);
		\path (0) edge (1);
		\path (0) edge (2);
		\path (0) edge (3);
		\path (0) edge (4);
		\path (1) edge [loop left] (1);
		\path (2) edge [loop left] (2);
		\path (3) edge [loop left] (3);
		\path (4) edge [loop left] (4);
		\end{tikzpicture}
		\caption{The $q=4$ multi-type Widom--Rowlinson model. Single-site activities are $\lambda_0=1$ and $\lambda_i=\lambda$ for $i>0$.}
		\label{fig:graph-widom-rowlinson-multi-type}
	\end{subfigure}%
	\begin{subfigure}{10pt}
		\quad
	\end{subfigure}%
	\begin{subfigure}[t]{.46\textwidth}
		\centering

		\tikzstyle{every state}=[circle,fill=gray!25,draw=black,minimum size=16pt,inner sep=0pt,text=black]
		\begin{tikzpicture}[auto,node distance=1.25cm]
		\node[state] (1) {1};
		\node[state] (2) [right of=1] {2};
		\node[state] (3) [above of=1] {3};
		\node[state] (4) [right of=3] {4};
		\node[state] (1b) [left of=1] {$1'$};
		\node[state] (2b) [right of=2] {$2'$};
		\node[state] (3b) [left of=3] {$3'$};
		\node[state] (4b) [right of=4] {$4'$};
		\path (1) edge (2);
		\path (1) edge (3);
		\path (1) edge (4);
		\path (2) edge (3);
		\path (2) edge (4);
		\path (3) edge (4);
		\path (1) edge (1b);
		\path (2) edge (2b);
		\path (3) edge (3b);
		\path (4) edge (4b);
		\path (1) edge [out=150,in=120,looseness=8, ->, loop] node[above] {} (1);
		\path (2) edge [out=60,in=30,looseness=8, ->, loop] node[above] {} (2);
		\path (3) edge [loop above] (3);
		\path (4) edge [loop above] (4);
		\path (1b) edge [loop above] (1b);
		\path (2b) edge [loop above] (2b);
		\path (3b) edge [loop above] (3b);
		\path (4b) edge [loop above] (4b);
		\end{tikzpicture}

		\caption{The $q=4$ multi-type beach model.\\ Single-site activities are $\lambda_i=1$ and $\lambda_{i'}=\lambda$.}
		\label{fig:graph-beach-model-multi-type}
	\end{subfigure}

\vspace{10pt}

	\begin{subfigure}[t]{.49\textwidth}
		\centering
		\begin{tikzpicture}[auto,node distance=1.25cm]
		\node[state] (0) {1};
		\node[state] (1) [right of=0] {2};
		\node[state] (3) [above of=0] {3};
		\node[state] (2) [right of=3] {4};
		\node[state] (v) [left of=0, yshift=0.55cm] {0};
		\path (0) edge (1);
		\path (0) edge (2);
		\path (0) edge (3);
		\path (1) edge (2);
		\path (1) edge (3);
		\path (2) edge (3);
		\path (v) edge [loop above] (0);
		\path (v) edge (0);
		\path (v) edge [in=155,out=355] (1);
		\path (v) edge [in=205,out=5] (2);
		\path (v) edge (3);
		\end{tikzpicture}
		\caption{The $q=4$ anti-Widom--Rowlinson model.\\ Single-site activities are $\lambda_0=1$ and $\lambda_i=\lambda$ for $i>0$.}
		\label{fig:graph-anti-widom-rowlinson}
	\end{subfigure}
	\begin{subfigure}{10pt}
		\quad
	\end{subfigure}%
	\begin{subfigure}[t]{.46\textwidth}
		\centering
		\begin{tikzpicture}[auto,node distance=1cm]
		\node[state] (0) {0};
		\node[state] (1) [right of=0] {1};
		\node[state] (2) [right of=1] {2};
		\node[state] (3) [right of=2] {3};
		\path (0) edge [loop above] (0);
		\path (0) edge (1);
		\path (0) edge [bend right] (2);
		\path (0) edge [bend right] (3);
		\path (1) edge [loop above] (1);
		\path (1) edge (2);
		\end{tikzpicture}
		\caption{The $q=3$ multi-occupancy hard-core model.\\Single-site activities are $\lambda_i=\lambda^i$.}
		\label{fig:graph-hard-core-multi-type}
	\end{subfigure}

	\caption{The multi-type models.}
	\label{fig:graphs}
\end{figure}

\subsubsection{The anti-Widom--Rowlinson model / dilute proper colorings}\label{sec:dilute proper colorings}

Runnels--Lebowitz~\cite{runnels1976analyticity} considered a second extension of the Widom--Rowlinson model in which there are $q$ cooperating, but self-repelling, species. The model at activity $\lambda>0$ may be described within our general setup by choosing
\[ \SS = \{0,1,\dots,q\}, \qquad \lambda_i=\lambda^{\1_{\{i \neq 0\}}},\qquad\lambda_{i,j}=\mathbf{1}_{\{ij=0\text{ or }i \neq j \}}, \]
as illustrated by \cref{fig:graph-anti-widom-rowlinson}.
The model may alternatively be viewed as a site-diluted version of the proper coloring model by identifying the admissible configurations with partial proper colorings of $\Z^d$ with $q$ colors (where the state ``0'' is assigned to uncolored sites). A third interpretation is as a superposition of $q$ mutually exclusive configurations of the hard-core model (i.e., the union of $q$ disjoint independent sets).
When $q=1$, the model coincides with the hard-core model, and when $q=2$, it is equivalent to the usual Widom--Rowlinson model (this is true on any bipartite graph). We thus focus here on the case $q \ge 3$.

Runnels--Lebowitz showed that there is a unique Gibbs state when $q \ge C^d$ (with no restriction on $\lambda$). This is improved by Dobrushin's uniqueness condition which establishes uniqueness when
\begin{equation*}
  \lambda<\frac{1}{2d}\quad\text{or}\quad q>4d-\frac{1}{\lambda}.
\end{equation*}
(again, with some improvement possible in low dimensions with the disagreement percolation method of van den Berg--Maes~\cite{van1994disagreement}). The behavior of the model in the rest of the phase diagram was not analyzed and Runnels--Lebowitz~\cite{runnels1976analyticity} asked whether a phase transition occurs for high activity values. Our work resolves this in the affirmative, when $d$ is sufficiently large compared with $q$.

We proceed to analyze the model within our framework. The maximal patterns are of the form $(\{0\} \cup A, \{0\} \cup B)$, where $(A,B)$ is a partition of $\{1,\dots,q\}$, perhaps with $A$ or $B$ empty. The weight of such a pattern is $(1+\lambda|A|)(1+\lambda|B|)$ so that the dominant patterns are the ones satisfying $\{|A|,|B|\}=\{\lfloor \frac q2 \rfloor,\lceil \frac q2 \rceil \}$. It follows that $\rho_{\text{pat}}^{\text{bulk}} = \frac{(1 + \lambda(\lfloor \frac q2 \rfloor-1))(1+\lambda(\lceil \frac q2 \rceil+1))}{(1+\lambda\lfloor \frac q2 \rfloor)(1+\lambda\lceil \frac q2 \rceil)}$ and $\rho_{\text{pat}}^{\text{bdry}} = \frac{1 + \lambda(\lceil \frac q2 \rceil-1)}{1+\lambda\lceil \frac q2 \rceil}$. A calculation, using the fact that $1 - \rho_{\text{pat}}^{\text{bulk}}\approx \min\{\lambda, \frac{1}{q}\}^2$ and $1 - \rho_{\text{pat}}^{\text{bdry}}\approx \min\{\lambda, \frac{1}{q}\}$, shows that condition~\eqref{eq:parameter-inequalities-simple-hom} is verified in the regime
\begin{equation}
 d \ge Cq^{12}\log^6 q \qquad\text{and}\qquad \lambda \ge \frac{C\sqrt{q} \log^{3/4} d}{d^{1/8}},
\end{equation}
whence our results are applicable and imply ordering according to the dominant patterns.

\subsubsection{The multi-type beach model}\label{sec:multi type beach model}

Motivated by questions on subshifts of finite type, Burton--Steif~\cite{burton1995new} studied a multi-type extension of the beach model consisting of $q \ge 3$ species (the case $q=2$ is the original two-type model, also introduced by Burton--Steif~\cite{burton1994non}, which was discussed in \cref{sec:first applications}).
The model at activity $\lambda>0$ may be described within our general setup by choosing
\[ \SS = \{0,1\} \times \{1,\dots,q\}, \qquad \lambda_{(s,i)}=\lambda^s,\qquad\lambda_{(s,i),(t,j)}=\1_{\{s=t=0\text{ or }i=j \}} ,\]
as illustrated by \cref{fig:graph-beach-model-multi-type} (with states $(0,i)$ labeled $i$ and states $(1,i)$ labeled $i'$).
Burton--Steif showed that there are at least $q$ ergodic measures of maximal entropy (each characterized by the predominance of one species) when $\lambda > 2eq(7q^2)^d-1$, and exactly $q$ such measures when $\lambda \ge C^{d^2} q^{2d+1}$ (strictly speaking, the model in~\cite{burton1995new} is defined in a slightly different, but essentially equivalent, way than here and allows only for integer activities).
We note that Dobrushin's uniqueness condition does not directly apply to this model, and while we have not found a reference with an explicit condition guaranteeing a unique Gibbs state, the condition $\lambda < \frac{p_c(\Z^d)}{1-p_c(\Z^d)}$ follows from~\cite{hallberg2004gibbs} (see the proofs of Proposition~8.15 and Theorem~8.16 there).
In fact, like in the original two-type model, here too there exists a critical $\lambda_c(q,d)$ for phase transition in the sense that below $\lambda_c$ there is a unique Gibbs state and above $\lambda_c$ there are multiple translation-invariant Gibbs states (\cite[Section 6]{haggstrom1998random} or~\cite[Chapter 8]{hallberg2004gibbs}), and the above implies that
\begin{equation}\label{eq:multi-type beach existing bounds}
\tfrac{p_c(\Z^d)}{1-p_c(\Z^d)} \le \lambda_c(q,d) \le 2eq(7q^2)^d-1 .
\end{equation}
Our analysis narrows this gap considerably in high dimensions, proving that the phase transition point neither tends to zero nor to infinity with $d$ and in fact satisfies $\lambda_c(q,d)\to q-1$ as $d\to\infty$.

We proceed to analyze the model within our framework. The maximal patterns are of 3 types: (i) $(A,A)$ with $A=\{(0,1),\ldots, (0,q)\}$, corresponding to a disordered phase and having weight $q^2$, (ii) $(A_i,A_i)$ with $A_i = \{(0,i),(1,i)\}$ for $1 \le i \le q$, corresponding to a predominance of species $i$ over the other species and having weight $(1+\lambda)^2$, and (iii) $(A\cup \{(1,i)\},\{(0,i)\})$, $(\{(0,i)\}, A\cup \{(1,i)\})$ for $1\le i\le q$, having weight $q+\lambda$ so that these patterns are never dominant.

When $\lambda>q-1$, the dominant patterns are $(A_i, A_i)$ for $1\le i\le q$ and $\rho_{\text{pat}}^{\text{bulk}} = \frac{\max\{q^2, q+\lambda\}}{(1+\lambda)^2}$, $\rho_{\text{pat}}^{\text{bdry}} = \frac{1}{1+\lambda}$ and $\fq = \log_2(q+2)$. It is straightforward that $\rho_{\text{pat}}^{\text{bulk}}>\rho_{\text{pat}}^{\text{bdry}}$ and a calculation shows that condition~\eqref{eq:parameter-inequalities-simple-hom2} is verified in the regime
\begin{equation}
  \log\left(\frac{\lambda+1}{q}\right) \ge \frac{C\log(dq)\sqrt{\log d}}{d^{1/4}},
\end{equation}
whence our results are applicable and imply ordering according to the dominant patterns $(A_i, A_i)$.

When $\lambda<q-1$, the unique dominant pattern is $(A, A)$. In this case $\rho_{\text{pat}}^{\text{bulk}} = \frac{\max\{(1+\lambda)^2, q+\lambda\}}{q^2}$, $\rho_{\text{pat}}^{\text{bdry}} = \frac{1}{q}$ and $\fq = 1$. It is straightforward that $\rho_{\text{pat}}^{\text{bulk}}>\rho_{\text{pat}}^{\text{bdry}}$ and a calculation shows that condition~\eqref{eq:parameter-inequalities-simple-hom2} is verified in the regime
\begin{equation}\label{eq:MTB lambda small regime}
  -\log\left(\frac{\lambda+1}{q}\right) \ge \frac{C\log^{3/2} d}{d^{1/4}} + \frac{\log q}{4d},
\end{equation}
whence our results are applicable and imply ordering according to the dominant pattern $(A, A)$. Note that the condition~\eqref{eq:MTB lambda small regime} is non-vacuous when either $q$ or $d$ are large.

We conclude that
\begin{equation*}
  -\left(\frac{C\log^{3/2} d}{d^{1/4}} + \frac{\log q}{4d}\right) \le \log\left(\frac{\lambda_c(q,d) + 1}{q}\right)\le \frac{C\log(dq)\sqrt{\log d}}{d^{1/4}}.
\end{equation*}
In particular, for fixed $q$ and $d$ large we obtain
\begin{equation*}
  |\lambda_c(q,d) - (q-1)| \le \frac{C_q \log^{3/2} d}{d^{1/4}}
\end{equation*}
which generalizes condition~\eqref{eq:beach_improved_lambda_c_bounds} for the original two-type beach model, and shows that the critical activity tends to $q-1$ as the dimension tends to infinity.

\subsubsection{The multi-occupancy hard-core model}\label{sec:multi-occupancy hard-core model}
We consider here an extension of the hard-core model in which a site may be occupied by multiple particles, with the restriction that the total occupancy on adjacent sites is at most some fixed value~$q$. Two such models have been discussed in the literature~\cite{kelly1991loss, van1994disagreement, mazel1991random, galvin2011multistate}, partly motivated by communication networks. The models at activity $\lambda>0$ are obtained in our setup when
\[ \SS=\{0,1,\dots,q\},\qquad\lambda_i = \begin{cases}\frac{\lambda^i}{i!}&\text{(model 1)}\\\lambda^i&\text{(model 2)}\end{cases},\qquad \lambda_{i,j}=\1_{\{i+j \le q\}}, \]
as illustrated by \cref{fig:graph-hard-core-multi-type}.
The standard hard-core model discussed in \cref{sec:hard-core-model} is obtained in both models as the special case $q=1$ and we henceforth restrict to the case $q\ge 2$.

The condition of van den Berg--Maes may be applied to both models (see~\cite[Example 4]{van1994disagreement}; the derivation applies to any choice of $(\lambda_i)$ with $\lambda_0=1$) to deduce that they are in the disordered regime when
\begin{equation}\label{eq:van den Berg Maes multi occupancy hard core}
  \sum_{i=0}^q \lambda_i < \frac{1}{1 - p_c(d)}
\end{equation}
(this holds, in particular, when $\lambda \le p_c(d)$ in either of the two models). This is a sharper conclusion than that obtained from Dobrushin's condition (which yields~\eqref{eq:van den Berg Maes multi occupancy hard core} with $p_c(d)$ replaced by $\frac{1}{2d}$).

Mazel--Suhov~\cite{mazel1991random} considered model 2 and proved that in every dimension $d\ge 2$ there exists $\lambda_0(d)$ so that the following holds for every $q\ge 2$ and $\lambda \ge \lambda_0(d)$: (i) If $q$ is even, the model has a unique Gibbs state, samples of which are small perturbations of the constant configuration that takes the value $q/2$ everywhere. (ii) If $q$ is odd, the model has exactly two extremal periodic Gibbs states, samples of which are small perturbations of the ``chessboard'' configurations that take the value $\lfloor q/2 \rfloor$ on one sublattice and the value $\lceil q/2\rceil$ on the other sublattice. The threshold $\lambda_0(d)$ is not given explicitly. To our knowledge, no similar long-range order result has been shown for model~1 on $\Z^d$.

Our results may be applied to both models as follows: Let $A_k := \{0,1,\ldots, k\}$ and $\Lambda_k := \sum_{i=0}^k \lambda_i$. The maximal patterns are $P_k := (A_k, A_{q-k})$ for $0\le k\le q$, with $P_k$ having weight $W_k := \Lambda_k\Lambda_{q-k}$. A calculation shows that $W_k$ strictly increases with $k$ for $0\le k\le \frac{q}{2}$ and strictly decreases with $k$ for $\frac{q}{2}\le k\le q$. Thus for $q$ even the unique dominant pattern is $P_{q/2}$, while for $q$ odd the dominant patterns are $P_{\lfloor q/2\rfloor}$ and $P_{\lceil q/2 \rceil}$. It further follows that $\rho_{\text{pat}}^{\text{bulk}}=\frac{W_{\lfloor q/2\rfloor - 1}}{W_{\lfloor q/2\rfloor}}$ and $\rho_{\text{pat}}^{\text{bdry}} = \frac{\Lambda_{\lceil q/2\rceil - 1}}{\Lambda_{\lceil q/2\rceil}}$. Aiming to apply condition~\eqref{eq:parameter-inequalities-simple-hom2} we note that $\rho_{\text{pat}}^{\text{bdry}} = 1 - \frac{\lambda_{\lceil q/2\rceil}}{\Lambda_{\lceil q/2\rceil}}$, $\fq = 1$ and $|\phasemax| = q+1$. In addition, a calculation shows that
\begin{equation}\label{eq:rho pat bulk bounds multi occupancy hard core}
  1 - \rho_{\text{pat}}^{\text{bulk}} \ge \frac{\lambda_{\lfloor q/2\rfloor}}{W_{\lfloor q/2\rfloor}}
\end{equation}
in both models, while for model 1 we have also
\begin{equation}\label{eq:rho pat bulk bounds multi occupancy hard core2}
  1 - \rho_{\text{pat}}^{\text{bulk}} \ge \frac{c_q \lambda^q}{W_{\lfloor q/2\rfloor}}.
\end{equation}
We proceed to discuss separately the cases $\lambda \le1$ and $\lambda>1$.

Suppose first that $\lambda\le 1$. In this case $1\le \Lambda_k\le q$ for all $0\le k\le q$ so that $-\log \max\big\{\rho_{\text{pat}}^{\text{bulk}},\rho_{\text{pat}}^{\text{bdry}} \big\} \ge \frac{1}{q^2}\lambda_{\lceil q/2\rceil}$ using~\eqref{eq:rho pat bulk bounds multi occupancy hard core} and the expression for $\rho_{\text{pat}}^{\text{bdry}}$. Thus condition~\eqref{eq:parameter-inequalities-simple-hom2} is verified in the regime
\begin{equation}
  \lambda\le 1\quad\text{and}\quad \lambda_{\lceil q/2\rceil}\ge \frac{C q^2\log^{3/2} d}{d^{1/4}} + \frac{q^2\log (q+1)}{2d}.
\end{equation}
As this regime is empty in both models when $q\ge C\log d$ we see that the condition is verified when
\begin{alignat}{2}
  &\text{model 1:}\qquad&&1 \ge \lambda \ge q \left(\frac{C\log^{3/2} d}{d^{1/4}}\right)^{1/\lceil \frac q2 \rceil},\\
  &\text{model 2:}\qquad&&1 \ge \lambda \ge \left(\frac{C\log^{7/2} d}{d^{1/4}}\right)^{1/\lceil \frac q2 \rceil}.\label{eq:model 2 threshold}
\end{alignat}
In this regime the dominant patterns give rise to extremal automorphism-invariant Gibbs states and every periodic Gibbs state is a mixture of these states (as all periodic Gibbs states have maximal pressure here). Moreover, for even $q$, the uniqueness of periodic Gibbs states implies uniqueness among all Gibbs states as the model is anti-monotone (i.e., monotone after applying the map $i \mapsto q-i$ to the spins on one sublattice; this is shown similarly to~\cite[Lemma~3.2]{van1994percolation}).

Suppose now that $\lambda \ge 1$. The qualitative results given in \cref{sec:non-quantitative results} show that the models order according to the dominant patterns when $d$ is sufficiently large as a function of $q$ and $\lambda$. We omit the calculation of the regime where condition~\eqref{eq:parameter-inequalities-simple-hom2} is satisfied for model 2 but note that since $\rho_{\text{pat}}^{\text{bulk}}\to1$ as $\lambda\to\infty$ we would indeed require $d$ to grow as $\lambda$ tends to infinity. For model 1, using~\eqref{eq:rho pat bulk bounds multi occupancy hard core2} we have
\begin{equation*}
  -\log \max\big\{\rho_{\text{pat}}^{\text{bulk}},\rho_{\text{pat}}^{\text{bdry}} \big\} \ge \min\left\{\frac{c_q \lambda^q}{W_{\lfloor q/2\rfloor}}, \frac{\lambda_{\lceil q/2\rceil}}{\Lambda_{\lceil q/2\rceil}}\right\} \ge c_q
\end{equation*}
where the second inequality follows by noting that the two expressions in the minimum are monotone increasing in $\lambda$ and using that $\lambda\ge 1$. We conclude that condition~\eqref{eq:parameter-inequalities-simple-hom2} (or even~\eqref{eq:parameter-inequalities-simple-hom}) is satisfied in model 1 with $\lambda\ge 1$ whenever the dimension $d$ exceeds a threshold depending on $q$ but not on $\lambda$.

Comparing with the results of~\cite{mazel1991random} mentioned above, we see that the results there give a description of the Gibbs states in all dimensions $d\ge 2$ when $\lambda$ is sufficiently large (as a function of $d$ but not of $q$) while our results apply in a different regime: for any $\lambda$ and $q$ as long as $d$ is sufficiently large (in particular, our results do not give a description of the Gibbs states when $\lambda$ is large compared to $d$ and $q$). The results of~\cite{mazel1991random} are based on the theory of dominant ground states and rely on a delicate analysis of excitations which our parameters $\rho_{\text{pat}}^{\text{bulk}}$ and $\rho_{\text{pat}}^{\text{bdry}}$ do not capture. Nevertheless, our results do imply that model 2 orders according to the dominant patterns for values of $\lambda$ which stay bounded, or even tend to zero, with $d$. In this regard we note that on a rooted $d$-ary tree, model 2 has been shown to transition from a uniqueness regime to a regime of phase coexistence as $\lambda$ grows~\cite{galvin2011multistate}, with the transition occurring near $\lambda = \left(\frac{e}{d}\right)^{1/\lceil q/2\rceil}$ for odd $q$ and near $\lambda = \left(\frac{\log d}{d(q+2)}\right)^{2 / (q+2)}$ for even $q$. The dependence on $d$ in these expressions is vaguely reminiscent of the bound~\eqref{eq:model 2 threshold}.

\begin{figure}
\renewcommand{\arraystretch}{1.4}
\begin{tabular}{|l|c|c|c|}
\hline
Model
 & $\omega_{\text{dom}}$
 & $(\rho^{\text{bulk}}_{\text{pat}})^{-1}$
 & $(\rho^{\text{bdry}}_{\text{pat}})^{-1}$
\\
\hline
AF Potts
 & $\lfloor \frac q2 \rfloor \lceil \frac q2 \rceil$
 & $\left(1 + \frac 1{\lfloor \frac q2 \rfloor-1}\right)\left(1 - \frac 1{\lceil \frac q2 \rceil+1}\right)$
 & $1 + \frac 1{\lceil \frac q2 \rceil-1}$
\\
\hline
Beach; $\lambda>1$
 & $(1+\lambda)^2$
 & $\min\big\{\frac {(1+\lambda)^2}4,\frac{(1+\lambda)^2}{2+\lambda}\big\}$
 & $1+\lambda$
\\
\hline
Beach; $\lambda<1$
 & 4
 & $\min\big\{\frac 4{2+\lambda},\frac 4{(1+\lambda)^2}\big\}$
 & 2
\\
\hline
Clock
& $(m+1)^2$
& $1+\frac 1{m(m+2)}$
& $1+\frac 1m$
\\
\hline
Hard-core
 & $1+\lambda$
 & $\infty$
 & $1+\lambda$
\\
\hline
Widom--Rowlinson
 & $(1+\lambda)^2$
 & $1+ \frac{\lambda^2}{1+2\lambda}$
 & $1+\lambda$
\\
\hline
Multi-occupancy hard-core
 & $\frac{\left(1-\lambda^{\lfloor \frac q2 \rfloor + 1}\right)\left(1-\lambda^{\lceil \frac q2 \rceil + 1}\right)}{(1-\lambda)^2}$
 & $\frac{\left(1-\lambda^{\lfloor \frac q2 \rfloor + 1}\right)\left(1-\lambda^{\lceil \frac q2 \rceil + 1}\right)}{\left(1-\lambda^{\lfloor \frac q2 \rfloor}\right)\left(1-\lambda^{\lceil \frac q2 \rceil + 2}\right)}$
 & $\frac{1-\lambda^{\lceil \frac q2 \rceil + 1}}{1-\lambda^{\lceil \frac q2 \rceil}}$
\\
\hline
Multi-type WR; $\lambda<q-2$
 & $1+q\lambda$
 & $\frac{1+q\lambda}{(1+\lambda)^2}$
 & $\frac{1+q\lambda}{1+\lambda}$
\\
\hline
Multi-type WR; $\lambda>q-2$
 & $(1+\lambda)^2$
 & $\frac{(1+\lambda)^2}{1+q\lambda}$
 & $1+\lambda$
\\
\hline
Anti WR / diluted colorings
 & $(1+\lambda\lfloor \frac{q}{2} \rfloor)(1+\lambda\lceil \frac{q}{2} \rceil)$
 & $\frac{(1+\lambda\lfloor \frac{q}{2} \rfloor)(1+\lambda\lceil \frac{q}{2} \rceil)}{(1+\lambda(\lfloor \frac{q}{2} \rfloor-1))(1+\lambda(\lceil \frac{q}{2} \rceil+1))}$
 & $\frac{1+\lambda\lceil \frac{q}{2} \rceil}{1+\lambda(\lceil \frac{q}{2} \rceil-1)}$
\\
\hline
Multi-type beach; $\lambda>q-1$
 & $(1+\lambda)^2$
 & $\frac{(1+\lambda)^2}{\max\{q^2,q+\lambda\}}$
 & $1+\lambda$
\\
\hline
Multi-type beach; $\lambda<q-1$
 & $q^2$
 & $\frac{q^2}{\max\{(1+\lambda)^2, q+\lambda\}}$
 & $q$
\\
\hline
\end{tabular}
\caption{Some parameter values for various models.}
\label{fig:parameter values of homomorphisms}
\end{figure}

\subsection{The antiferromagnetic Ising and Potts models with external field}\label{sec:AF Ising and AF Potts with magnetic field}

The introduction of an external magnetic field to the AF Ising and Potts models gives rise to new phenomena. Below we apply our results to analyze the models in this extended phase diagram.

\subsubsection{The Ising antiferromagnet and the hard-core model at positive temperature}\label{sec:Ising antiferromagnet}
The \emph{Ising model}, perhaps the most basic of all classical statistical physics models, is equivalent to the Potts model with $q=2$ states. As for the Potts model, it is relatively well understood in its \emph{ferromagnetic} version (see, e.g.,~\cite[Chapter 3]{friedli2017statistical} and~\cite{duminil2017lectures}), when adjacent spins have a tendency to be equal, and less understood in its \emph{antiferromagnetic} (AF) version, when adjacent spins tend to be different, which is our focus here. On a finite $\Lambda\subset\Z^d$ and at inverse temperature $\beta>0$ and external magnetic field $h\in\R$, the AF Ising model assigns to each $f:\Lambda\to\{-1,1\}$ the probability
\begin{equation*}
  \frac{1}{Z_{\Lambda,\beta,h}}\exp\bigg(-\beta\bigg(\sum_{\{u,v\} \in E(\Lambda)}\sigma_u\sigma_v - h\sum_{u\in \Lambda}\sigma_u\bigg)\bigg)
\end{equation*}
with $Z_{\Lambda,\beta,h}$ a suitable normalization constant (the partition function). We proceed to discuss the model in dimensions $d\ge 2$.

By flipping $\sigma$ on the even sublattice (this works on any bipartite graph) one obtains the ferromagnetic Ising model with a staggered magnetic field ($+h$ on one sublattice and $-h$ on the other). In particular, the AF Ising model has two, possibly equal, periodic and extremal Gibbs states obtained by applying this flip operation to the $+$ and $-$ Gibbs states of the ferromagnetic Ising model. When these measures are equal the model has a unique Gibbs state and when they are distinct, they take a chessboard form in the sense that they exhibit different densities for the two states on the two sublattices.
A further consequence is that when $h=0$ the AF Ising model is equivalent to its ferromagnetic version. Thus, when $h=0$, the model has a unique Gibbs state for all $\beta\le \beta_c(d)$ and multiple Gibbs states when $\beta>\beta_c(d)$, where $\beta_c(d)$ is the critical inverse temperature of the ferromagnetic model on $\Z^d$ which is known to be asymptotic to $\frac{1}{2d}$ as $d\to\infty$~(see, e.g.,~\cite{duminil2017lectures} and~\cite[Remark 2.7]{peled2019lectures}). Below we assume that $h>0$ (noting that flipping the state of all spins is equivalent to changing the sign of $h$).

We proceed to describe the main existing results pertaining to the phase diagram of the AF Ising model at positive magnetic field~(see \cref{fig:AF Ising phase diagram}). On the one hand, Dobrushin's  uniqueness condition implies that the model has a unique Gibbs state when
\begin{equation}\label{eq:Dobrushin uniqueness for AF Ising}
  \text{either}\qquad \beta<\frac{1}{2}\log\left(1 + \frac{2}{2d-1}\right) \qquad\text{or}\qquad h \ge 2d +\frac{\log(2d)}{2\beta}.
\end{equation}
Disagreement percolation refines the second condition (see~\cite{van1993uniqueness} for the computation in the $d=2$ case), proving that there is a unique Gibbs state when
\begin{equation}\label{eq:disagreement percolation for AF Ising}
  h\ge 2d + \frac{\log\left(\frac{1}{p_c(d)} - 1\right)}{2\beta}.
\end{equation}
On the other hand, Dobrushin~\cite{dobrushin1968problem} introduced a variant of the classical Peierls argument involving a shift transformation to prove that the model has multiple Gibbs states in a certain regime of parameters (an alternative proof using reflection positivity is in~\cite[Model 3.1]{frohlich1980phase}). While Dobrushin did not write the regime explicitly, combining his argument with the contour counting estimates of Lebowitz--Mazel~\cite{lebowitz1998improved} and Balister--Bollob{\'a}s~\cite{balister2007counting} proves the existence of multiple Gibbs states when
\begin{equation}\label{eq:Dobrushin multiplicity for AF Ising}
  h\le 2d - \frac{C\log d}{\beta}.
\end{equation}

\begin{figure}
 \centering
 \includegraphics[scale=1]{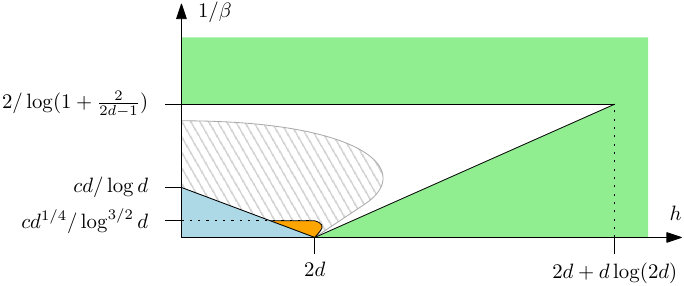}
 \caption{The phase diagram of the high-dimensional Ising antiferromagnet at positive external magnetic field. Uniqueness of the Gibbs state is known in the green region, defined by~\eqref{eq:Dobrushin uniqueness for AF Ising} and~\eqref{eq:disagreement percolation for AF Ising}. The existence of multiple Gibbs states is known in the blue region, defined by~\eqref{eq:Dobrushin multiplicity for AF Ising}, while our results prove it also in the orange region (see~\eqref{eq:new multiplicity regime for AF Ising}), establishing that for $h$ just above $2d$ the model undergoes two phase transitions as the temperature increases (the ``bulging phenomenon''). A possibility for the region of multiplicity is sketched by the diagonal lines. }
 \label{fig:AF Ising phase diagram}
\end{figure}

We now explain how our results may be used to extend the known region of multiplicity in the AF Ising model (see \cref{fig:AF Ising phase diagram}). As apparent from the above bounds, the transition between uniqueness and multiplicity of Gibbs states at low temperature occurs around the point $h=2d$. To further understand the behavior in this region, it is natural to take the zero-temperature limit. As it turns out, the limiting model is exactly the hard-core model discussed in \cref{sec:hard-core model}, in which the activity parameter $\lambda$ depends on the ``angle of approach'' of $h$ to the limiting value $2d$. Specifically, writing
\begin{equation}\label{eq:h scaling at zero temperature}
  h=2d - \frac{\log \lambda}{2\beta}
\end{equation}
and regarding the state $-1$ as occupied and the state $+1$ as vacant, one recovers the hard-core model with activity $\lambda$ in the
limit $\beta\to\infty$ with $\lambda>0$ constant. It is thus natural to expect that multiplicity in the low-temperature AF Ising model will occur when the angle of approach corresponds to a value of $\lambda$ for which the hard-core model has multiple Gibbs states. Indeed, our results on the hard-core model extend to the positive temperature regime, proving that there are multiple Gibbs states when
\begin{equation}\label{eq:new multiplicity regime for AF Ising}
  \beta\ge \frac{C\log^{3/2} d}{d^{1/4}}\qquad\text{and}\qquad 2d - c\sqrt{d}\le h\le 2d + \frac{1}{2\beta}\log\left(\frac{cd^{1/4}\min\{\beta, 1\}}{\log^{3/2} d}\right).
\end{equation}
The results further imply that in this range every periodic Gibbs state is a mixture of the two chessboard Gibbs states discussed above.

Our results shed light on a ``bulging phenomenon'' (or re-entrant phase phenomenon) which has received some attention in the literature. The question discussed is whether the AF Ising model exhibits two phase transitions for some $h>2d$ as the temperature increases. Indeed, as noted above (see~\eqref{eq:Dobrushin uniqueness for AF Ising}) the model is disordered for each $h>2d$ both when the temperature is sufficiently low and when it is sufficiently high. The question is thus whether there exists an intermediate temperature for which the model has multiple Gibbs states (so that the critical curve in the phase diagram ``bulges'' above the $h=2d$ point; see \cref{fig:AF Ising phase diagram}). R\`acz~\cite{racz1980phase} predicted the absence of such an intermediate phase for the square lattice. This was supported by Dobrushin--Kolafa--Shlosman~\cite{dobrushin1985phase} who proved that there exist $\beta_0>0$ and $\lambda_0>1$ such that the two-dimensional model is disordered for $\beta>\beta_0$ and $\lambda<\lambda_0$ (using the parametrization~\eqref{eq:h scaling at zero temperature}). Van den Berg~\cite{van1993uniqueness} then used disagreement percolation to fully rule out the phenomenon in the two-dimensional model (see~\eqref{eq:disagreement percolation for AF Ising} and note that $p_c(2)>\frac{1}{2}$~\cite{higuchi1982coexistence}). Absence of the bulging phenomenon is further predicted for the $\Z^3$ lattice~\cite{racz1980phase, yamagata1995absence} (though the phenomenon is predicted to occur on a different three-dimensional lattice -- the body-centered-cubic lattice~\cite{racz1980phase}).
From the above-mentioned relation between the AF Ising and hard-core models it is natural to expect the bulging phenomenon to occur on $\Z^d$ when the critical activity for the hard-core model is smaller than $1$. In particular, by the result of Galvin--Kahn~\cite{galvin2004phase} the phenomenon is expected for large $d$. Our results verify that it indeed occurs (see~\eqref{eq:new multiplicity regime for AF Ising} and \cref{fig:AF Ising phase diagram}).

Lastly, we explain how to apply our results in the regime~\eqref{eq:new multiplicity regime for AF Ising}. The AF Ising model may be described within our general setup by choosing
\[ \SS=\{-1,+1\},\qquad \lambda_i=e^{\beta hi}, \qquad \lambda_{i,j}=e^{-\beta ij} .\]
To treat the low-temperature regime around the point $h=2d$ we use a technique, explained in \cref{sec:reweighting activities} below, of expressing the model in terms of different activities and pair interactions. Applying \eqref{eq:model-reparameterization} with $m_{\pm 1}= e^{\mp 2\beta d}$ (and then scaling the single-site activities and pair interactions) allows to equivalently describe the model as
\[ \SS=\{0,1\},\qquad \lambda_i=\lambda^i, \qquad \lambda_{i,j}=\1_{\{ij = 0\}} + e^{-4\beta} \1_{\{ij=1\}} ,\]
where we have used the relation~\eqref{eq:h scaling at zero temperature} and have replaced the spin values $(-1,+1)$ with $(1,0)$. In this parametrization the model may be thought of as a positive temperature version of the hard-core model. As in the hard-core model, the maximal patterns are $(\{0\},\{0,1\})$ and $(\{0,1\},\{0\})$, both of which are dominant, and we have $\rho_{\text{pat}}^{\text{bulk}}=0$ and $\rho_{\text{pat}}^{\text{bdry}}=\frac{1}{1+\lambda}$. Here we also have $\rho_{\text{int}}=e^{-4\beta}$ and $\rho_{\text{act}}=1+\lambda$. It follows that the parameter $\alpha_0$ of~\eqref{eq:alpha-def} satisfies $\alpha_0 \approx \min\{\log(1+\lambda), \beta, \lambda \beta\}$. Thus, condition~\eqref{eq:parameter-inequalities-simple1} shows that when
\begin{equation}\label{eq:postive-temp-hard-core}
\min\{\lambda,\beta,\lambda\beta\} \ge \frac{C\log^{3/2} d}{d^{1/4}} \qquad\text{and}\qquad \beta \ge \frac{C \log (1+\lambda)}{\sqrt{d}} ,
\end{equation}
the two dominant patterns give rise to two ordered Gibbs states. Substituting the relation~\eqref{eq:h scaling at zero temperature} in~\eqref{eq:postive-temp-hard-core} yields~\eqref{eq:new multiplicity regime for AF Ising}.

As a final remark we note that the combination of the multiplicity regime~\eqref{eq:new multiplicity regime for AF Ising} yielded by our results with the previously known multiplicity regime~\eqref{eq:Dobrushin multiplicity for AF Ising} shows that the high-dimensional low-temperature AF Ising model exhibits multiple Gibbs states for all $0\le h\le 2d$. Taken on its own, our result limits $h$ to be at least $2d-c\sqrt{d}$ but we point out (without elaborating on the details) that this limitation may be improved by relying on the more involved \cref{main-cond} instead of checking condition~\eqref{eq:parameter-inequalities-simple1}.

\subsubsection{The antiferromagnetic Potts model with external magnetic field}\label{sec:AF Potts with magnetic field}
We consider the antiferromagnetic (AF) Potts model with an external magnetic field applied to the first state. Precisely, let $q\ge 3$ be the number of states (the $q=2$ case is the Ising antiferromagnet considered in the previous section), $h\in\R$ be the external magnetic field and $\beta>0$ be the inverse temperature. On a finite $\Lambda\subset\Z^d$, the model assigns to each $f:\Lambda\to\{1,\ldots, q\}$ the probability
\begin{equation}\label{eq:AF Potts with external field probability measure}
  \frac{1}{Z_{\Lambda,\beta, h}}\exp\bigg(-\beta\bigg(\sum_{\{u,v\} \in E(\Lambda)}\1_{\{f(u)= f(v)\}} - h \sum_{v\in \Lambda} \1_{\{f(v)=1\}}\bigg)\bigg)
\end{equation}
with $Z_{\Lambda,\beta, h}$ a suitable normalization constant (the partition function). The case $h=0$ was considered in \cref{sec:AF Potts model} so we assume that $h\neq 0$. We remark that the limit $h\to-\infty$ corresponds to the AF Potts model with $q-1$ states while the limit $\beta\to\infty$ corresponds to the proper $q$-coloring model with an external magnetic field applied to the first color. The model may be described within our general setup by choosing
\[ \SS = \{1,\dots,q\},\qquad\lambda_i= \1_{\{i \neq 1\}}+\lambda \1_{\{i=1\}},\qquad\lambda_{i,j}=\1_{\{i\neq j\}}+e^{-\beta}\1_{\{i=j\}} \]
with $\lambda = e^{\beta h}$.

To keep the discussion focused, we fix $q$ and $\beta_0>0$ and consider the model in the low-temperature regime $\beta\ge \beta_0$ and in high dimensions $d \ge C_{q,\beta_0}$. We aim to study the effect of varying $\lambda$ in this setup. As will be presented, this effect is rather pronounced, with the model admitting at least $\lceil \frac{q}{2}\rceil + 2$ different phases (see \cref{fig:af_potts_magnetic_field}).

\begin{figure}
	\includegraphics[scale=1]{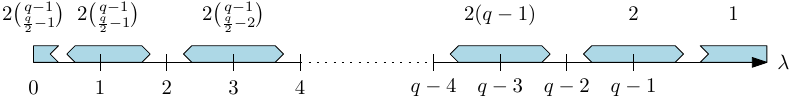}
	\vspace{12pt}\\
	\includegraphics[scale=1]{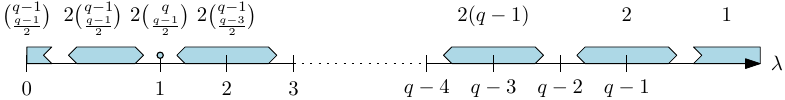}
	\caption{A partial phase diagram for the high-dimensional $q$-state antiferromagnetic Potts model with an external magnetic field $\frac{1}{\beta}\log \lambda$ applied to one of the states. The top depicts the case of even $q$ and the bottom that of odd $q$. The blue segments depict the regimes studied in \cref{sec:AF Potts with magnetic field}. The numbers above the segments indicate the number of maximal-pressure extremal Gibbs states.}
	\label{fig:af_potts_magnetic_field}
\end{figure}

Dobrushin's uniqueness condition implies that the system is disordered when
\begin{equation*}
  \lambda \ge (2d-1)(q-1) e^{2\beta d}.
\end{equation*}
The unique Gibbs state in this regime favors the first state at every site.

We proceed to analyze the model within our framework. As in the usual AF Potts model, here too the maximal patterns are pairs $(A,B)$ which partition $\{1,\dots,q\}$. For $1\le k\le q-1$, let $\cP^k$ be the set of maximal patterns having exactly $k$ states in the side containing state $1$. Note that the patterns in $\cP^k$ are all equivalent and have weight $(k-1+\lambda)(q-k)$. In addition, $\cP^k$ contains exactly $2\binom{q-1}{k-1}$ patterns. One checks that the patterns in $\cP^1$ are dominant exactly when $\lambda\ge q-2$, that for $1<k<\lceil \frac{q}{2}\rceil$ the patterns in $\cP^k$ are dominant exactly when $q-2k\le \lambda \le q-2(k-1)$ and that the patterns in $\cP^{\lceil \frac{q}{2}\rceil}$ are dominant exactly when $\lambda \le 1 + \1_{\{q\text{ even}\}}$.

We proceed to check condition~\eqref{eq:parameter-inequalities-simple} to deduce regimes in which the model is ordered according to the various dominant patterns. The above discussion shows that when $\lambda\notin \{q-2, q-4,\ldots, q-2\lfloor\frac{q}{2}\rfloor\}$ then all dominant patterns are equivalent (otherwise they are not, taking into account our assumption that $\lambda\neq 1$). It is straightforward that $\rho_{\text{int}} = e^{-\beta}$ and $\rho_{\text{act}} = \frac{q-1+\lambda}{\min\{\lambda, 1\}}$.
\begin{itemize}[leftmargin=15pt]
  \item $k=1$: When $\lambda>q-2$, one calculates that $\rho_{\text{pat}}^{\text{bdry}} = \frac{q-2}{q-1}$ and $\rho_{\text{pat}}^{\text{bulk}} = 1 - \frac{\lambda - (q-2)}{\lambda(q-1)}$. Condition~\eqref{eq:parameter-inequalities-simple} is thus verified when $d\ge C_{q,\beta_0}$ and
\begin{equation}
   q-2 + \frac{C_q\log^{3/2} d}{d^{1/4}}\le \lambda \le \exp\left(\frac{1}{2}\sqrt{\frac{\beta d^{3/4}}{q}}\right),
\end{equation}
whence our results are applicable and imply ordering according to the patterns in $\cP^1$.
  \item $1<k<\lceil \frac{q}{2}\rceil$: When $q-2k\le \lambda \le q-2(k-1)$, one calculates that $\rho_{\text{pat}}^{\text{bdry}} \le \frac{q-k}{q-k+1}$ and $\rho_{\text{pat}}^{\text{bulk}} = 1-\min\left\{\frac{q-2(k-1) - \lambda}{(\lambda+k-1)(q-k)}, \frac{\lambda - (q-2k)}{(\lambda+k-1)(q-k)}\right\}$. Condition~\eqref{eq:parameter-inequalities-simple} is thus verified when $d\ge C_{q,\beta_0}$ and
\begin{equation}
   q-2k + \frac{C_q\log^{3/2} d}{d^{1/4}}\le \lambda \le q-2(k-1) - \frac{C_q\log^{3/2} d}{d^{1/4}},
\end{equation}
whence our results are applicable and imply ordering according to the patterns in $\cP^k$.
  \item $k = \lceil\frac{q}{2}\rceil$: When $\lambda \le 1 + \1_{\{q\text{ even}\}}$, one calculates that $\rho_{\text{pat}}^{\text{bdry}} = 1 - \frac{\min\{\lambda, 1\}}{\lambda + k-1}$ and $\rho_{\text{pat}}^{\text{bulk}}$ is given by the same expression as in the case $1<k<\lceil \frac{q}{2}\rceil$. Condition~\eqref{eq:parameter-inequalities-simple} is thus verified when $d\ge C_{q,\beta_0}$ and
\begin{equation}
   \frac{C_{q,\beta_0} \log^{3/2}d}{d^{1/4}}\le \lambda\le 1 + \1_{\{q\text{ even}\}} - \frac{C_q\log^{3/2} d}{d^{1/4}},
\end{equation}
whence our results are applicable and imply ordering according to the patterns in $\cP^{\lceil\frac{q}{2}\rceil}$.
\item There is an additional regime of $\lambda$ to which our results apply. As mentioned above, in the limit $\lambda\to 0$, the model becomes the AF Potts model (without external field) with $q-1$ states, for which long-range order was discussed in \cref{sec:AF Potts model}. Our results may be used to show that for small positive $\lambda$ the model continues to order in a similar manner as the $\lambda\to 0$ limit. To treat this regime we again use the technique, explained in \cref{sec:reweighting activities} below, of expressing the model in terms of different activities and pair interactions. Specifically, for $m>1$, we replace $\lambda_1$ by $\lambda_1' := m\lambda_1 = m\lambda$, $\lambda_{1,1}$ by $\lambda_{1,1}':= m^{-\frac{1}{d}} \lambda_{1,1} = m^{-\frac{1}{d}} e^{-\beta}$ and $\lambda_{1,i}$ by $\lambda_{1,i}' := m^{-\frac{1}{2d}} \lambda_{1,i} = m^{-\frac{1}{2d}}$ for $2\le i\le q$. Note that a maximal pattern is now a non-trivial partition of $\{2,\dots,q\}$ (along with the trivial patterns $(\emptyset, \SS), (\SS,\emptyset)$) and that such a pattern is dominant if the sides have sizes $\{\lfloor \frac {q-1}2 \rfloor,\lceil \frac {q-1}2 \rceil\}$. Note also that $\rho_{\text{pat}}^{\text{bulk}}$ and $\rho_{\text{pat}}^{\text{bdry}}$ depend only on $q$, and that $\rho_{\text{act}}=q-1+\lambda m$ and $\rho_{\text{int}}=\max\{e^{-\beta},m^{-\frac 1{2d}}\}$. Choosing $m=e^{C_q d^{3/4} \log^{3/2} d}$, one checks that condition~\eqref{eq:parameter-inequalities-simple} holds if $d\ge C_{q,\beta_0}$ and
\begin{equation}
  \lambda \le \exp(-C_q d^{3/4} \log^{3/2} d),
\end{equation}
whence our results are applicable and imply ordering according to the above dominant patterns (i.e., the partitions of $\{2,\dots,q\}$ into sets of sizes $\{\lfloor \frac {q-1}2 \rfloor,\lceil \frac {q-1}2 \rceil\}$).
\end{itemize}

\subsection{Extensions of the results}\label{sec:extensions}
There are models to which the results described in \cref{sec:non-quantitative results} and \cref{sec:quantitative-results} do not directly apply due to one of the following reasons: the dominant patterns are not all equivalent, the spin space $\SS$ is infinite, or the quantitative conditions are not satisfied. Nevertheless, in some situations our results may still be applied indirectly to these systems via specialized ``tricks''.
In this section, we describe four tools which may be used to this end: reweighting the activities and pair interactions (\cref{sec:reweighting activities}), product systems (\cref{sec:product_systems}), projections from $\Z^d\times\{0,1\}$ (\cref{sec:projection-systems}) and covering systems (\cref{sec:covering systems}).

\subsubsection{Reweighting the activities and pair interactions}\label{sec:reweighting activities}
Suppose we are given a spin system in the form of a triplet $(\SS,(\lambda_i)_{i \in \SS}, (\lambda_{i,j})_{i,j \in \SS})$ as described in \cref{sec:the model}. Evidently, for any constants $m,m'>0$, the triplet $(\SS,(m\lambda_i)_{i \in \SS}, (m'\lambda_{i,j})_{i,j \in \SS})$ describes the same spin system as it simply introduces a constant multiplicative factor to the weight of a configuration $f$ in~\eqref{eq:config-weight}, which is independent of $f$. Note that this trivial modification does not change the set of maximal/dominant patterns for the model or any of the quantities appearing in~\eqref{eq:parameter-inequalities-simple} (our definitions are invariant to this operation). A more useful modification is as follows: let $(m_i)_{i \in \SS}$ be an arbitrary vector of positive numbers, and consider the triplet $(\SS,(\lambda'_i)_{i \in \SS}, (\lambda'_{i,j})_{i,j \in \SS})$ given by
\begin{equation}\label{eq:model-reparameterization}
 \lambda'_i := m_i \lambda_i \qquad\text{and}\qquad \lambda'_{i,j} := (m_i m_j)^{-\frac{1}{2d}} \lambda_{i,j} .
\end{equation}
It is straightforward to check that this operation preserves the weights in~\eqref{eq:config-weight}, so that the new triplet describes the same spin system as the original one. The usefulness of this new form is that it may change the patterns, raising the possibility of applying our theorems to obtain additional results. Applications of this idea are presented for the AF Ising and AF Potts models with external field in \cref{sec:AF Ising and AF Potts with magnetic field}.

\subsubsection{Product systems}\label{sec:product_systems}
Two spin systems $(\SS^{(m)},(\lambda_i^{(m)})_{i \in \SS^{(m)}}, (\lambda_{i,j}^{(m)})_{i,j \in \SS^{(m)}})$, $m=1,2$, may be combined into a product system, given by the triple
\begin{equation}
  \SS=\SS^{(1)}\times\SS^{(2)},\qquad \lambda_{(i,j)}=\lambda_i^{(1)}\lambda_j^{(2)},\qquad \lambda_{(i,j),(k,\ell)}=\lambda_{i,k}^{(1)}\lambda_{j,\ell}^{(2)}.
\end{equation}
This definition implies that a configuration sampled from the product system on a finite domain, with some boundary conditions, is distributed as independent samples from the two given systems. Consequently, the extremal Gibbs states of the product system are exactly the products of the extremal Gibbs states of the two systems. In addition, it is not difficult to check that the maximal patterns of the product system are exactly the ``product pairs'' $(A^{(1)}\times A^{(2)},B^{(1)}\times B^{(2)})$ where $(A^{(m)}, B^{(m)})$, $m=1,2$, are maximal patterns of the initial systems. Consequently, the dominant patterns of the product system are exactly the products of the dominant patterns of the two systems. However, even if the two systems satisfy the assumption required for our results, namely, that the dominant patterns in each of them are equivalent, this assumption may be violated for the product system. Indeed, if $(A^{(m)}, B^{(m)})$, $m=1,2$, are dominant patterns of the initial systems then both $(A^{(1)}\times A^{(2)},B^{(1)}\times B^{(2)})$ and $(A^{(1)}\times B^{(2)},B^{(1)}\times A^{(2)})$ are dominant patterns of the product system which, however, may fail to be equivalent. In such a case, our results may not be applied directly to the product system but may be applied to each of the initial systems separately and the conclusions may then be transferred to the product system.

\begin{figure}[!t]
	\tikzstyle{every state}=[circle,fill=gray!25,draw=black,minimum size=16pt,inner sep=0pt,text=black]
	\begin{subfigure}{.4\textwidth}
		\centering
		\begin{tikzpicture}[auto,node distance=1.25cm]
		\node[state] (0) {1};
		\node[state] (1) [above of=0] {-1};
		\node[state] (v) [left of=0, yshift=0.55cm] {0};
		\node[state] (2) [left of=v] {2};
		\path (v) edge [loop above] (1);
		\path (v) edge (0);
		\path (v) edge (1);
		\path (v) edge (2);
		\path (0) edge (1);
		\end{tikzpicture}
		\label{fig:hardcore-product}
	\end{subfigure}%
	\begin{subfigure}{.4\textwidth}
		\centering

		\tikzstyle{every state}=[circle,fill=gray!25,draw=black,minimum size=16pt,inner sep=0pt,text=black]
		\begin{tikzpicture}[auto,node distance=1cm]
		\node[state] (01) {01};
		\node[state] (22) [right of=01] {12};
		\node[state] (10) [right of=22] {20};
		\node[state] (12) [below of=01] {22};
		\node[state] (00) [right of=12] {00};
		\node[state] (21) [right of=00] {11};
		\node[state] (20) [below of=12] {10};
		\node[state] (11) [right of=20] {21};
		\node[state] (02) [right of=11] {02};
		\path (01) edge (22);
		\path (22) edge (10);
		\path (12) edge (00);
		\path (00) edge (21);
		\path (20) edge (11);
		\path (11) edge (02);
		\path (01) edge (12);
		\path (22) edge (00);
		\path (10) edge (21);
		\path (12) edge (20);
		\path (00) edge (11);
		\path (21) edge (02);
		\draw (10) -- (2.6,0) [dotted,thick];
		\draw (21) -- (2.6,-1) [dotted,thick];
		\draw (02) -- (2.6,-2) [dotted,thick];
		\draw (01) -- (-0.6,0) [dotted,thick];
		\draw (12) -- (-0.6,-1) [dotted,thick];
		\draw (20) -- (-0.6,-2) [dotted,thick];
		\draw (01) -- (0,0.6) [dotted,thick];
		\draw (22) -- (1,0.6) [dotted,thick];
		\draw (10) -- (2,0.6) [dotted,thick];
		\draw (20) -- (0,-2.6) [dotted,thick];
		\draw (11) -- (1,-2.6) [dotted,thick];
		\draw (02) -- (2,-2.6) [dotted,thick];
		\end{tikzpicture}

		\label{fig:T^2_3}
	\end{subfigure}

	\caption{Two examples of product systems.}\label{fig:product constructions}
\end{figure}

Let us describe two examples where the above observations are useful (see \cref{fig:product constructions}). As a first example consider the spin system which is described within our framework as
\begin{equation}\label{eq:product hard core system}
  \SS = \{-1,0,1,2\},\qquad \lambda_i = \lambda^{|i|},\qquad\lambda_{i,j} = \1_{\{|i|+|j|\le 2 \text{ and } ij\neq 1\}}
\end{equation}
for a parameter $\lambda>0$. The model bears similarity with the multi-occupancy hard-core model of \cref{sec:multi-occupancy hard-core model} (model 2 with $q=2$ in the notation there): Each vertex may be occupied by $0$, $1$ or $2$ particles (of activity $\lambda$) with the sum of occupancies of neighbors not exceeding $2$, and with the additional feature that there are two types of single-occupancy states and neither one can be adjacent to itself. The maximal patterns of this system are $(\{0\},\{-1,0,1,2\}), (\{0,1\},\{-1,0\})$ and their reversals, all of which have weight $(1+\lambda)^2$ and are dominant. Thus our results do not (directly) apply to the model as the dominant patterns are not all equivalent. Nonetheless, by identifying the states $(-1,0,1,2)$ with $((0,1),(0,0),(1,0),(1,1))$ one sees that the given system is a product of the hard-core model (of \cref{sec:hard-core-model}) with itself. Consequently, when $\lambda$ is sufficiently large, each of the four dominant patterns gives rise to a distinct (ordered) Gibbs state and each periodic Gibbs state (necessarily of maximal pressure) is a mixture of these four states.

As a second example, consider the system corresponding to graph homomorphisms to a two-dimensional torus of side length $3$, i.e., the system given by
\begin{equation*}
  \SS = \T^2_3, \qquad \lambda_i = 1, \qquad \lambda_{i,j} = \1_{\{\|i-j\|_1=1\}}
\end{equation*}
(where $\|i-j\|_1$ denotes the graph distance on the torus).
This system has two types of maximal patterns (besides the trivial $(\emptyset,\SS), (\SS,\emptyset)$): $(\{v\},\{v-e_1, v+e_1, v-e_2, v+e_2\})$ and $(\{v,v+e_1+e_2\},\{v+e_1, v+e_2\})$ (and their reversals), with $v\in\T^2_3$ and where $e_1 = (1,0)$, $e_2 = (0,1)$ and additions are performed modulo $3$. All the non-trivial maximal patterns have weight $4$ and are dominant, so that our results cannot be applied directly as the dominant patterns are not all equivalent. It turns out, however, that the system is equivalent to the product of the system of graph homomorphisms to $\T_3$ with itself (this is not straightforward, as $\T^2_3$ is the so-called box product of graphs while the products used in this section are the so-called tensor product. Indeed, it is not true that the system corresponding to graph homomorphisms to $\T_4^2$ is a product of graph homomorphisms to $\T_4$). This may be seen by identifying the vertex $(x,y)\in\T^2_3$ with the vertex $(x+y, x-y)$ in the (tensor) product of $\T_3$ with itself (see \cref{fig:product constructions}). Consequently, each of the four dominant patterns gives rise to a distinct Gibbs state and each periodic Gibbs state (necessarily of maximal pressure) is a mixture of these four states.

We remark that the usefulness of product constructions was first pointed out to the authors in a discussion with Martin Tassy who observed that the identification used in the second example shows that the system of graph homomorphisms to $\Z^2$ is the product of the system of graph homomorphisms to $\Z$ with itself.

\subsubsection{Projections from $\Z^d\times\{0,1\}$}\label{sec:projection-systems}
We write $\Z^d\times\{0,1\}$ for the induced subgraph of $\Z^{d+1}$ on the set of vertices whose last coordinate is $0$ or $1$. As stated after \cref{thm:characterization_of_Gibbs_states-NQ}, our results remain valid for spin systems on $\Z^d\times\{0,1\}$. This turns out to be useful even for studying spin systems on $\Z^d$, as it turns out that certain spin systems on $\Z^d$ are equivalent to simpler spin systems on $\Z^d\times\{0,1\}$. For instance, we have seen in \cref{sec:Widom-Rowlinson model} that the Widom--Rowlinson model on $\Z^d$ is equivalent to the hard-core model on $\Z^d\times\{0,1\}$ (and also, the asymmetric Widom--Rowlinson model on $\Z^d$ is equivalent to the hard-core model on $\Z^d\times\{0,1\}$ with unequal sublattice activities) and that our results yield stronger conclusions when first applied to the hard-core model and then transferred through the equivalence than when applied directly to the Widom--Rowlinson model. Let us explain the general mechanism behind this equivalence.

Consider a spin system on $\Z^d\times\{0,1\}$ described in our framework by the triple $(\SS,(\lambda_i)_{i \in \SS}, (\lambda_{i,j})_{i,j \in \SS})$. The configurations of the system may also be viewed as configurations on $\Z^d$ which assign a state in $\SS^2$ to every vertex. In this way, the given spin system on $\Z^d\times\{0,1\}$ is equivalent to the spin system on $\Z^d$ described by the triple $(\SS',(\lambda_i')_{i \in \SS'}, (\lambda_{i,j}')_{i,j \in \SS'})$ defined by
\begin{equation}
  \SS':=\{(i,j)\in\SS^2\colon \lambda_{i,j}>0\},\qquad \lambda_{(i,j)}':=\lambda_i\lambda_j\lambda_{i,j},\qquad \lambda_{(i,j),(k,\ell)}':=\lambda_{i,k}\lambda_{j,\ell}.
\end{equation}
This equivalence naturally translates to an equivalence between the Gibbs states of $(\SS,(\lambda_i), (\lambda_{i,j}))$ on $\Z^d\times\{0,1\}$ and the Gibbs states of $(\SS',(\lambda_i'), (\lambda_{i,j}'))$ on $\Z^d$.

The hard-core model on $\Z^d\times\{0,1\}$ is equivalent by the above mechanism to the $q=2$ anti-Widom--Rowlinson model on $\Z^d$, which is itself equivalent to the standard Widom--Rowlinson model; see \cref{sec:dilute proper colorings}. Additional examples of this mechanism include the equivalence of the beach model on $\Z^d$ (\cref{sec:beach model}) to weighted graph homomorphisms from $\Z^d\times\{0,1\}$ to the path of length~5, and the equivalence of graph homomorphisms from $\Z^d$ to $\Z$ to graph homomorphisms from $\Z^d\times\{0,1\}$ to the graph $\Z$ with loops added at every vertex (1-Lipschitz functions; see also \cref{sec:hom-to-path}). The models in the latter equivalence lie outside the framework of this paper as their spin spaces $\SS$ are infinite (but see \cref{sec:covering systems} below); Their equivalence was pointed out by Ariel Yadin~\cite[Yadin bijection]{peled2010high}.

We remark that one may consider similar projections from $\Z^d\times G$ to $\Z^d$ for other finite graphs $G$ though we do not discuss specific models for which this may be useful.

\subsubsection{Covering systems}\label{sec:covering systems}
We describe here a natural notion of ``covering'' between spin systems on $\Z^d$ which implies their equivalence.

Let $H$ be a graph, without multiple edges but possibly with self loops. A \emph{cover} of $H$ is a graph $\cH$ and a surjective map $\varphi:V(\cH)\to V(H)$ such that for each $v\in V(\cH)$ the map $\varphi$ restricted to the neighbors of $v$ in $\cH$ is a bijection onto the neighbors of $\varphi(v)$ in $H$. This notion will be used to ``lift'' graph homomorphisms from $\Z^d$ to $H$ to graph homomorphisms from $\Z^d$ to $\cH$. Precisely, given a graph $G$, a graph homomorphism $\hat{f}:V(G)\to V(\cH)$ is called a \emph{lift} of a graph homomorphism $f:V(G)\to V(H)$ if $f = \varphi\circ \hat{f}$. The definition of a cover implies that, when $G$ is connected, any two lifts of $f$ which coincide at one vertex must coincide everywhere. Lifts need not exist in general: for instance, $\Z$ covers the $4$-cycle $C_4$ via the modulo $4$ map, but any graph homomorphism $f:\Z^2\to V(C_4)$ in which some $4$-cycle of $\Z^2$ is mapped to all of $C_4$ has no lift.

We now restrict attention to the special case that $G=\Z^d$ with $d \ge 2$ (results in the one-dimensional case may differ as $\Z$ is not one-ended).
A cover $(\cH,\varphi)$ of $H$ is \emph{lift-permitting} if for every graph homomorphism $f \colon \Z^d \to V(H)$ and vertices $o \in \Z^d, v \in V(\cH)$ such that $f(o)=\varphi(v)$ there exists a lift $\hat{f} \colon \Z^d \to V(\cH)$ of $f$ such that $\hat{f}(o)=v$. From the above discussion, such a lift is unique. A necessary and sufficient condition for a cover to be lift-permitting is that for any path $(v_0,v_1,v_2,v_3,v_4)$ in $\cH$ with $v_0 \neq v_4$, we have that $\varphi(v_0) \neq \varphi(v_4)$. One may verify this using the fact that the cycle space of $\Z^d$ is generated by its basic 4-cycles. For example, the cover of the $q$-cycle by $\Z$ (via the modulo $q$ map) is lift-permitting for all $q \ge 3$ except $q=4$.

We now extend the above notions to that of a covering spin system. Consider a spin system described in our framework by the triple $(\SS,(\lambda_i)_{i \in \SS}, (\lambda_{i,j})_{i,j \in \SS})$. Let $H_{\text{pos}}$ be the graph with vertex set $\SS$ and edge set $\{\{i,j\}\colon\lambda_{i,j}>0\}$. Let $(\cH,\varphi)$ be a lift-permitting cover of $H_{\text{pos}}$.
Define a new spin system $(\SS',(\lambda_i')_{i \in \SS'}, (\lambda_{i,j}')_{i,j \in \SS'})$ by
\begin{equation}
  \SS':=V(\cH),\qquad \lambda'_i:=\lambda_{\varphi(i)},\qquad \lambda_{i,j}':=\lambda_{\varphi(i),\varphi(j)}.
\end{equation}
Observe that the edges of $\cH$ are exactly the pairs $\{i,j\}\subset\SS'$ with $\lambda_{i,j}'>0$. A spin system $(\SS',(\lambda_i'), (\lambda_{i,j}'))$ formed in this way is called a \emph{covering system} of $(\SS,(\lambda_i), (\lambda_{i,j}))$ (via the map $\varphi$).

Covering systems are useful as there is a close connection between the Gibbs states of the two systems and it may be the case that one system is technically easier to analyze than the other. Let us elaborate on the connection between the Gibbs states: Let $(\SS',(\lambda_i'), (\lambda_{i,j}'))$ be a covering system of $(\SS,(\lambda_i), (\lambda_{i,j}))$, via a map $\varphi:\SS'\to\SS$. As usual, we assume that $\SS$ is finite. However, it is useful to note that the above definitions also make sense when $\SS'$ is (countably) infinite and to explicitly allow this in our discussion (as an example, the proper $3$-coloring model is covered by graph homomorphisms to $\Z$ via the modulo $3$ map). The following properties hold:
\begin{enumerate}
  \item Any Gibbs state for the covering system is pushed forward by $\varphi$ to a Gibbs state for the covered system (regardless of whether $\SS'$ is finite or infinite). This push-forward operation preserves periodicity, extremality and pressure (with a suitable definition of pressure when $\SS'$ is infinite).
  \item If $\SS'$ is finite, then any Gibbs state $\mu$ for the covered system is the push-forward by $\varphi$ of some Gibbs state $\mu'$ for the covering system. Indeed, such a $\mu'$ is obtained by sampling a configuration from $\mu$ and then (independently) sampling uniformly among its (finitely many) lifts. In particular, $\mu'$ is (periodic) of maximal pressure if and only if $\mu$ is.
  \item If $\SS'$ is infinite, it may happen that there is a Gibbs state for the covered system which is not the push-forward by $\varphi$ of any Gibbs state for the covering system. We mention as an example that the unique maximal-entropy Gibbs state for proper $3$-colorings of $\Z^2$ is not the push-forward (by the modulo 3 map) of any Gibbs state for graph homomorphisms from $\Z^2$ to $\Z$ (this can be deduced from the results of~\cite{chandgotia2018delocalization,duminil2019logarithmic,ray2020proper}).
\end{enumerate}
Thus, when $\SS'$ is finite, the push-forward of the set of maximal-pressure Gibbs states for the covering system is precisely the set of maximal-pressure Gibbs states for the covered system. This gives a certain equivalence between the two systems. When $\SS'$ is infinite, such an equivalence does not necessarily hold, and more care is needed when trying to compare the systems.
\begin{enumerate}
 \setcounter{enumi}{3}
  \item Say that a Gibbs state $\mu$ is ``flat'' if, for some pattern $P$, samples from $\mu$ almost surely have a unique infinite connected component in the $P$-pattern.
  Then any flat Gibbs state $\mu$ for the covered system is the push-forward by $\varphi$ of some flat Gibbs state $\mu'$ for the covering system (regardless of whether $\SS'$ is finite or infinite).
  Indeed, such a $\mu'$ is obtained by fixing a pattern $P'$ (in the covering system) that is mapped by $\varphi$ to $P$, sampling a configuration $f$ from $\mu$ and then lifting it to the unique configuration $f'$ for which the unique infinite $P$-cluster in $f$ is lifted to a $P'$-cluster in $f'$. Moreover, if $\mu$ is extremal then so is $\mu'$ and every extremal Gibbs state for the covering system that is pushed-forward by $\varphi$ to $\mu$ is obtained in this way.
\end{enumerate}
The significance of this is that when our results apply to the covered system, the Gibbs states corresponding to dominant patterns are flat, so that they can be lifted to the covering system (even when $\SS'$ is infinite), resulting in an analogous characterization of the maximal-pressure Gibbs states for the covering system (i.e., each dominant pattern gives rise to a Gibbs state and any maximal-pressure Gibbs state is a mixture of these).
When our results instead apply to the covering system (and not to the covered system), the existence of flat Gibbs states for the covered system follows, as does the characterization of maximal-pressure Gibbs states for the covered system when $\SS'$ is finite. This latter characterization does not immediately follow when $\SS'$ is infinite from what we have said above, but are implied by the following:
\begin{enumerate}
 \setcounter{enumi}{4}
 \item Suppose that $(\tilde\SS,(\tilde\lambda_i), (\tilde\lambda_{i,j}))$ is a second system covered by $(\SS',(\lambda_i'), (\lambda_{i,j}'))$ via a map $\tilde\varphi$.
  If every (periodic) maximal-pressure Gibbs state for the first covered system is a mixture of flat Gibbs states, then so are those of the second covered system. To see this, suppose that $\tilde\mu$ is a Gibbs state for the second covered system. Fix a set in $\SS'$ containing exactly one preimage of each element in $\tilde \SS$ under $\tilde\varphi$. Given $\tilde f \in \tilde\SS^{\Z^d}$, we lift $\tilde f$ to the unique $f'$ whose value at the origin belongs to the fixed set, and then project it to $f \in \SS^{\Z^d}$ via $\varphi$. Let $\Phi$ denote the map $\tilde f \mapsto f$. Denote $\mu_0$ be the push-forward of $\tilde\mu$ by $\Phi$. We use $\mu_0$ to construct a Gibbs state for the first covered system (we do not claim that $\mu_0$ itself is a Gibbs state). For each $n \ge 1$, let $\mu_n$ be the average of the translates of $\mu_0$ by elements in $[-n,n]^d$. Let $\mu$ be any subsequential limit of $\mu_n$. It is straightforward that $\mu$ is a translation-invariant measure, and one may also check that it is a Gibbs state, and that it has the same pressure as $\tilde\mu$ when $\tilde\mu$ is periodic. Suppose now that $\tilde\mu$ is a maximal-pressure Gibbs state, and that every maximal-pressure Gibbs state for the first covered system is a mixture of flat Gibbs states. The previous items imply that $\mu$ also has maximal pressure. Thus, it is a mixture of flat Gibbs states. To conclude that $\tilde\mu$ is a mixture of flat Gibbs states, it remains only to observe that, for any two vertices $u,v \in \Z^d$, the property that ``there exists a dominant pattern $P$ such that $u$ and $v$ are connected by a path of vertices in the $P$-pattern'' either holds for both $\tilde f$ and $f=\Phi(\tilde f)$ or for neither. Since this occurs with probability bounded from below (over the choice $u$ and $v$) under $\mu$, it is also so under $\tilde\mu$.
\end{enumerate}

We now discuss some applications. For simplicity, we mostly focus on homomorphism models in the examples below (we emphasize that in the context of homomorphisms to a graph $H$, when we talk about a covering of $H$, we mean in the sense of a covering system, so that the covering is lift-permitting).

\smallskip
\noindent\textbf{The bipartite cover.}
The biparitite covering is the ``smallest'' (non-trivial) covering and is defined by taking $\cH$ to be the graph with vertex set $\SS \times \{0,1\}$ and edge set $\{ \{(i,0),(j,1)\} : \lambda_{i,j}>0 \}$, and taking $\varphi$ to be the projection $\varphi(i,p)=i$. Observe that $\cH$ is necessarily bipartite.

For example, the bipartite covering of the standard hard-core model is the 4-path with suitable activities (this was used in \cref{sec:hard-core-unequal-sublattice-activities-model} to allow for unequal sublattice activities), and the bipartite covering of the proper 3-coloring model (equivalently, the 3-cycle) is the 6-cycle.
More interestingly for our purposes, there are systems which do not satisfy our assumptions, but whose bipartite covering does. For instance, when $H$ is a path on $\{-1,0,1\}$ with loops at $-1$ and $1$, there are non-equivalent dominant patterns (e.g., $(\{1\},\{0,1\})$ and $(\{0\},\{-1,1\})$), but the bipartite covering of $H$ is the 6-cycle to which our results apply.

We mention that the properties listed above imply that when our assumptions are satisfied for either the system or its bipartite covering, the bipartite covering system has exactly twice as many extremal (periodic) maximal-pressure Gibbs states.

\smallskip
\noindent\textbf{Homomorphisms to infinite graphs.}
We have already seen that $\Z$ is a covering of the $q$-cycle (via the modulo $q$ map) for any $q \ge 3$ except $q=4$. While our results do not apply directly to the model of homomorphisms to $\Z$, the properties mentioned above allow us to obtain results about homomorphisms to $\Z$ by applying our results to homomorphisms to any of these $q$-cycles. Specifically, applying our results to the $q$-cycle yields that in dimensions $d \ge d_0(q)$ (for a function $d_0(q)$ increasing to infinity as $q \to \infty$), each of the $2q$ dominant patterns gives rise to an ordered Gibbs state and that all (periodic) maximal-entropy Gibbs states are mixtures of these. We then deduce that for the $\Z$-homomorphism model in dimensions $d \ge d_0(3)$, there exists an ordered Gibbs state for each pattern $(\{i\},\{i-1,i+1\})$, $i \in \Z$, and its reversal, and moreover, that every (periodic) maximal-entropy Gibbs state is a mixture of these. Interestingly, we may now transfer this result back down to the $q$-cycle to conclude that we can take $d_0(q)$ above to be independent of $q$ (this can be deduced from the results in \cref{sec:equilibrium-states} as discussed above, or more simply in this case by noting that the 3-cycle and $q$-cycle have a common covering by a finite cycle).

In a similar manner, the graph on $\Z$ whose edge set is $\{\{i,j\} : |i-j| \le m \}$ is a covering of the $m$-Lipschitz $q$-clock model for $1 \le m < \frac q4$ and inverse temperature $\beta=\infty$ (see \cref{sec:clock-models}), and this can be used to remove the dependency on $q$ as in~\eqref{eq:quantitative_clock_cond2}.

We note that both examples above generalize to higher dimensions. For instance, the lattice $\Z^k$ covers the torus $\T_q^k$ (via the coordinate-wise modulo $q$ map) for any $k \ge 1$ and $q \ge 3$ except $q=4$. When $k=2$, not all dominant patterns are equivalent, but together with the analysis for $\T_3^2$ given in \cref{sec:product_systems}, this covering enables us to deduce results for $\T_q^2$ in all cases other than $q=4$. The case of $\T_4^2$ remains open (see \cref{sec:hypercube}).

Another family of infinite graphs of interest are the regular trees, which we discuss next.

\smallskip
\noindent\textbf{Homomorphisms to regular trees and four-cycle free graphs.}
Consider the infinite $k$-regular tree $T_k$, where $k \ge 3$, and consider also the class of finite $k$-regular graphs which contain no four-cycle and no self-loops. It is not hard to check that any such graph is covered by $T_k$ (in the sense of covering system; note that the four-cycle free property ensures that the covering is lift-permitting). This observation was pointed out to us by Nishant Chandgotia (see the related~\cite{chandgotia2017four}). We also note that the maximal/dominant patterns in such a graph are $(\{v\},N(v))$ and its reversal, with $v$ being any vertex. In particular, all dominant patterns are equivalent if and only if the graph is transitive.

As with any infinite graph, our results cannot be directly applied to the model of homomorphisms from $\Z^d$ to $T_k$.
To analyze this model using the ideas of this section, we would need to find a finite graph that is covered by $T_k$ and to which our results can be applied. Fortunately, it is known~\cite{imrich1984explicit} that there exist finite, transitive, $k$-regular, four-cycle free, loopless graphs (which are in fact Cayley graphs with large girth). Let $H$ be such a graph and note that (due to its transitivity) our results can be applied for homomorphisms to $H$ yielding that in high dimensions, each of the $2|V(H)|$ dominant patterns gives rise to an ordered Gibbs state and that all (periodic) maximal-entropy Gibbs states are mixtures of these. We then deduce that for the $T_k$-homomorphism model in high dimensions, there exists an ordered Gibbs state for each pattern $(\{v\},N(v))$, $v \in T_k$, and its reversal, and moreover, that every (periodic) maximal-entropy Gibbs state is a mixture of these.
It may be worthwhile to mention that it is also possible to prove the existence of ordered states for homomorphisms to $T_k$ (and perhaps also the characterization of all maximal-entropy Gibbs states) by a more direct approach (something of this sort is done for homomorphisms to $\Z$ in~\cite{peled2010high}).

Now consider a finite $k$-regular graph $H$ with no four-cycles or self-loops.
If $H$ happens to be transitive, then we can directly apply our results for homomorphisms to $H$. Otherwise, our results cannot be applied directly, but we can use the fact that $T_k$ is a covering of $H$ in order to transfer the results from the $T_k$ model to the $H$ model. We conclude that, as in the case of transitive $H$, in high dimensions, each of the $2|V(H)|$ dominant patterns gives rise to an ordered Gibbs state and that all (periodic) maximal-entropy Gibbs states are mixtures of these.

\subsection{Revisiting our first applications}\label{sec:first applications_revisited}
We briefly revisit the first applications discussed in \cref{sec:first applications} in order to explain the quantitative bounds stated there. See also \cref{fig:parameter values of homomorphisms}.

\subsubsection{The AF Potts model}
We first show that Dobrushin's uniqueness condition is satisfied when~\eqref{eq:Dobrushin uniqueness AF Potts} holds. Let $\tau_1,\tau_2 \in \SS^{N(v)}$ be two boundary conditions which differ only at some $u \in N(v)$, and let $\mu_1$ and $\mu_2$ denote the distributions of the state of $v$ with these boundary conditions. We need to bound the total variation distance $\distTV(\mu_1,\mu_2)$. Let $\tau \in \SS^{N(v) \setminus \{u\}}$ be given by the common values of $\tau_1$ and $\tau_2$, and let $\mu^\tau$ be the distribution of the state of $v$ under $\tau$ boundary condition when the vertex $u$ is removed from the graph. Then, as in~\cite[(3.9)]{salas1997absence}, we have that $\distTV(\mu_1,\mu_2) \le \max_i \frac{\mu^\tau(1-f_i)}{\mu^\tau(f_i)}$, where $f_i(s) := e^{-\beta \1_{\{s=\tau_i(u)\}}}$. In particular, $\distTV(\mu_1,\mu_2) \le \max_i \frac{(1-e^{-\beta})\mu^\tau(\tau_i(u))}{1-(1-e^{-\beta})\mu^\tau(\tau_i(u))}\le \frac{(1-e^{-\beta})p}{1-(1-e^{-\beta})p}$, where $p:=\max_{\eta} \max_{s\in \SS} \mu^\eta(s)$. Thus, Dobrushin's condition holds when
\begin{equation}\label{eq:Dobrushin uniqueness AF Potts calculation}
  \tfrac{1}{p} > (1-e^{-\beta})(2d+1).
\end{equation}
By the arithmetic-geometric mean inequality,
\[ \tfrac 1p = \min_{n_1+\cdots+n_q=2d-1} \tfrac{e^{-\beta n_1} +\cdots + e^{-\beta n_q}}{e^{-\beta n_1}} = \min_{n_2+\cdots+n_q=2d-1} (1+e^{-\beta n_2} +\cdots + e^{-\beta n_q}) \ge q e^{-\beta(2d-1)/q}.\]
Hence, using also that $1-\exp(-\beta)\le\beta$, Dobrushin's condition holds when
\begin{equation*}
  \tfrac{q}{\beta(2d+1)} e^{-\beta(2d+1)/q}> 1
\end{equation*}
It thus suffices that $\beta(2d+1)/q < c$, where $c>0$ is the unique solution to $c=e^{-c}$. Using that $c>1/2$, we obtain the second condition in~\eqref{eq:Dobrushin uniqueness AF Potts}. For the first condition, note that if $q\ge 2d$ then
\[ \tfrac 1p = \min_{n_2+\cdots+n_q=2d-1} (1+e^{-\beta n_2} +\cdots + e^{-\beta n_q}) > 1+q-2d.\]
In particular, \eqref{eq:Dobrushin uniqueness AF Potts calculation} is satisfied when $q>2d(2-e^{-\beta})$.

We check that the model is in the ordered regime when~\eqref{eq:potts_param_ineq} holds by checking condition~\eqref{eq:parameter-inequalities-simple}.
Recall that the dominant patterns are partitions $(A,B)$ of $\{1,\dots,q\}$ such that $\{|A|,|B|\}=\{\lfloor \frac q2 \rfloor,\lceil \frac q2 \rceil \}$, so that $\omega_{\text{dom}} = \lfloor \frac q2 \rfloor \lceil \frac q2 \rceil$. The maximum defining $\rho_{\text{pat}}^\text{bulk}$ in~\eqref{eq:rho pat bulk and bdry def} is obtained for patterns $(A,B)$ with $|A|=\lfloor \frac q2 \rfloor -1$ and $|B|=\lceil \frac q2 \rceil +1$, and the maximum defining $\rho_{\text{pat}}^{\text{bdry}}$ is obtained when $|A|=\lceil \frac q2 \rceil$ and $|A'|=\lceil \frac q2 \rceil-1$. In particular, $1-\rho_{\text{pat}}^{\text{bulk}}\approx q^{-2}$ and $1-\rho_{\text{pat}}^{\text{bdry}} \approx q^{-1}$. It is immediate that $\rho_{\text{act}}=q$ and $\rho_{\text{int}} = e^{-\beta}$. Thus, $\alpha_0 \approx \min\{\frac 1{q^2}, \frac{1-\exp(-\beta/2)} q \}\approx \min\{\frac 1{q^2}, \frac{\beta} q \}$.
Plugging this into condition~\eqref{eq:parameter-inequalities-simple} yields~\eqref{eq:potts_param_ineq}.

\subsubsection{The beach model}
The maximal patterns are $(\{-1,1\},\{-1,1\})$, the pattern $(\{1,2\},\{1,2\})$ and its negation, the pattern $(\{1\},\{-1, 1,2\})$, its negation and their reversals, and the pattern $(\emptyset,\SS)$ and its reversal.  Recall that the dominant patterns depend on whether $\lambda>1$ or $\lambda<1$. We check in each case when condition~\eqref{eq:parameter-inequalities-simple-hom} is satisfied.

Let us first consider the case $\lambda>1$. The dominant patterns are $(\{1,2\},\{1,2\})$ and its negation so that $\omega_{\text{dom}}=(1+\lambda)^2$. One checks that $\rho_{\text{pat}}^{\text{bulk}} = \max\{\frac{4}{(1+\lambda)^2}, \frac{2+\lambda}{(1+\lambda)^2}\}\le \frac{1}{\lambda}$ and that $\rho_{\text{pat}}^{\text{bdry}} = \frac{1}{1+\lambda}$. Plugging this into~\eqref{eq:parameter-inequalities-simple-hom} yields the upper bound on $\lambda_c$ in~\eqref{eq:beach_improved_lambda_c_bounds}.

Let us now consider the case $\lambda<1$.
The unique dominant pattern is $(\{-1,1\},\{-1,1\})$ so that $\omega_{\text{dom}}=4$. One checks that $\rho_{\text{pat}}^{\text{bulk}} = \max\{\frac{(1+\lambda)^2}{4}, \frac{2+\lambda}{4}\}$ so that $ -\log(\rho_{\text{pat}}^{\text{bulk}})\approx 1 - \lambda$ (when $\lambda\le 1$) and $\rho_{\text{pat}}^{\text{bdry}} = \frac{1}{2}$. Plugging this into~\eqref{eq:parameter-inequalities-simple-hom} yields the lower bound on $\lambda_c$ in~\eqref{eq:beach_improved_lambda_c_bounds} (since \cref{thm:characterization_of_Gibbs_states} implies that there is a unique translation invariant Gibbs state in this case).

\subsubsection{Clock models}\label{sec:clock model revisited}
We check that the model is in the ordered regime when~\eqref{eq:quantitative_clock_cond} holds by checking condition~\eqref{eq:parameter-inequalities-simple1}. First note that when $q>4m$ all maximal patterns of the model have the form $(A,B)$ where both $A$ and $B$ are intervals in $\Z_q$ with the same midpoint which satisfy $|A| + |B| = 2(m+1)$ (and the trivial patterns $(\emptyset,\SS)$, $(\SS,\emptyset)$). The dominant patterns arise when $|A|=|B|=m+1$. A simple calculation shows that $\rho_{\text{pat}}^{\text{bulk}} = 1 - \frac{1}{(m+1)^2}$, $\rho_{\text{pat}}^{\text{bdry}} = 1 - \frac{1}{m+1}$, $\rho_{\text{int}} = e^{-\beta}$, $\rho_{\text{act}} = q$ and $|\phasemax|\le q^2$. Recalling~\eqref{eq:frak q inequalities} it follows that $\fq \approx \log q$. This implies that $\alpha_0 \approx \min\{\frac 1{m^2}, \frac{\beta}{m}\}$. Instead of~\eqref{eq:parameter-inequalities-simple1} we require the slightly stronger condition
\begin{equation*}
  \alpha_1 \ge \frac{C\fq \log^{3/2} d}{d^{1/4}}\qquad\text{and}\qquad \frac{-\log \rho_{\text{int}}}{4\log (d\rho_{\text{act}})} \ge \min\left\{1, \frac {|\SS|}{2d} + \frac{|\SS|\log(2d\rho_{\text{act}})}{d^{3/4}\fq \log^{3/2} d} \right\}.
\end{equation*}
This is satisfied when
\begin{equation}\label{eq:clock model sharper condition}
m^2 \log q \le \frac{c d^{1/4}}{\log^{3/2} d}\qquad\text{and}\qquad \beta \ge \frac{Cm \log q \log^{3/2} d}{d^{1/4}} + C\log(dq) \cdot \min\left\{\frac{q}{d^{3/4}},1\right\} .
\end{equation}
and it is not difficult to see that this condition is implied by~\eqref{eq:quantitative_clock_cond}.

\section{Preliminaries}
\label{sec:preliminaries}

In this section we give some notation and preliminary results which will be used throughout the proof in the following sections.

\subsection{Notation}\label{sec:notation}
Let $G=(V,E)$ be a graph.
For vertices $u,v \in V$ such that $\{u,v\} \in E$, we say that $u$ and $v$ are {\em adjacent} and write $u \sim v$.
We denote the graph-distance between $u$ and $v$ by $\text{dist}(u,v)$.
For two non-empty sets $U,W \subset V$, we denote by $\dist(U,W)$ the minimum graph-distance between a vertex in $U$ and a vertex in $W$.
For a subset $U \subset V$, denote by $N(U)$ the {\em neighbors} of $U$, i.e., vertices in $V$ adjacent to some vertex in $U$, and define for $t>0$,
\[ N_t(U) := \{ v \in V : |N(v) \cap U| \ge t \} .\]
In particular, $N_1(U)=N(U)$.
Denote the {\em external boundary} and the \emph{internal boundary} of $U$ by
\[ \extB U := N(U) \setminus U \qquad\text{and}\qquad \intB U := \extB U^c ,\]
respectively.
Denote also
\[ \intextB U := \intB U \cup \extB U, \qquad U^+ := U \cup \extB U \qquad\text{and}\qquad \Int(U) := U \setminus \intB U .\]
For a positive integer $r$, we denote
\[ U^{+r} := \{ v \in V : \dist(v,U) \le r \} .\]
The set of edges between two sets $U$ and $W$ is denoted by
\[ \partial(U,W):=\{ \{u,w\} \in E : u \in U,~ w \in W \} .\]
The {\em edge-boundary} of $U$ is denoted by $\partial U := \partial(U,U^c)$. We also define the set of out-directed boundary edges of $U$ to be
\[ \dpartial U := \{ (u,v) : u \in U,~v \in U^c,~u \sim v \} .\]
We write $\dpartialrev U := \dpartial (U^c)$ for the in-directed boundary edges of~$U$. We also use the shorthands $u^+ := \{u\}^+$, $\partial u := \partial \{u\}$ and $\dpartial u := \dpartial \{u\}$. Occasionally we write $(\cdot)^{++}$ for $(\cdot)^{+2}$.
The \emph{diameter} of $U$, denoted by $\diam U$, is the maximum graph-distance between two vertices in $U$, where we follow the convention that the diameter of the empty set is $-\infty$.
For a positive integer $r$, we denote by $G^{\otimes r}$ the graph on $V$ in which two vertices are adjacent if their distance in $G$ is at most $r$.

We consider the graph $\Z^d$ with nearest-neighbor adjacency, i.e., the edge set $E(\Z^d)$ is the set of $\{u,v\}$ such that $u$ and $v$ differ by one in exactly one coordinate.
A vertex of $\Z^d$ is called {\em even (odd)} if it is at even (odd) graph-distance from the origin.
We denote the set of even and odd vertices of $\Z^d$ by $\Even$ and $\Odd$, respectively.

For $t>0$ and an integer $n \ge 1$, we denote $\binom{n}{\le t}:=\sum_{k=0}^{\lfloor t \rfloor} \binom nk$ and note that $\binom{n}{\le t} \le (en/t)^t$.

\subsection{Odd sets and regular odd sets}
\label{sec:odd-sets}

We say that a set $U \subset \Z^d$ is {\em odd} (even) if its internal boundary consists solely of odd (even) vertices, i.e., $U$ is odd if and only if $\intB U \subset \Odd$ and it is even if and only if $\intB U \subset \Even$. We say that an odd or even set $U$ is \emph{regular} if both it and its complement contain no isolated vertices. Observe that $U$ is odd if and only if $(\Even \cap U)^+ \subset U$ and that $U$ is regular odd if and only if $U=(\Even \cap U)^+$ and $U^c=(\Odd \cap U^c)^+$.

An important property of odd sets is that the size of their edge-boundary is directly related to the difference between the number of odd and even vertices it contains.

\begin{lemma}[{\cite[Lemma~1.3]{feldheim2016growth}}]\label{lem:boundary-size-of-odd-set}
	Let $A \subset \Z^d$ be finite and odd. Then $\frac{1}{2d} |\partial A| = |\Odd \cap A|-|\Even \cap A|$. In particular, if $A$ contains an even vertex then $|\partial A| \ge 2d(2d-1)$.
\end{lemma}

\subsection{Co-connected sets}
\label{sec:co-connected-sets}

In this section, we fix an arbitrary connected graph $G=(V,E)$.
A set $U \subset V$ is called {\em co-connected} if its complement $V \setminus U$ is connected.
For a set $U \subset V$ and a vertex $v \in V$, we define the {\em co-connected closure} of $U$ with respect to $v$ to be the complement of the connected component of $V \setminus U$ containing $v$, where it is understood that this results in $V$ when $v \in U$.
We say that a set $U' \subset V$ is a co-connected closure of a set $U \subset V$ if it is its co-connected closure with respect to some $v \in V$.
Evidently, every co-connected closure of a set $U$ is co-connected and contains $U$.
The following simple lemma summarizes some basic properties of the co-connected closure (see~\cite[Lemma~2.5]{feldheim2015long} for a proof).

\begin{lemma}\label{lem:co-connect-properties}
Let $A,B \subset V$ be disjoint and let $A'$ be a co-connected closure of $A$. Then
\begin{enumerate}[label=(\alph*)]
\item \label{it:co-connect-reduces-boundary} $\intB A' \subset \intB A$, $\extB A' \subset \extB A$ and $\partial A' \subset \partial A$.
\item \label{it:co-connect-reduces-boundary2} $\intB (B \setminus A') \subset \intB B$ and $\extB (B \setminus A') \subset \extB B$.
\item \label{it:co-connected-minus-set-is-co-connected} If $B$ is co-connected then $B \setminus A'$ is also co-connected.
\item \label{it:co-connect-kills-components} If $B$ is connected then either $B \subset A'$ or $B \cap A' = \emptyset$.
\end{enumerate}
\end{lemma}

The following lemma, taken from~\cite[Proposition~3.1]{feldheim2013rigidity} and based on ideas of Tim{\'a}r \cite{timar2013boundary}, establishes the connectivity of the boundary of subsets of $\Z^d$ which are both connected and co-connected.

\begin{lemma}\label{lem:int+ext-boundary-is-connected}
Let $A \subset \Z^d$ be connected and co-connected. Then $\intextB A$ is connected.
\end{lemma}

\subsection{Graph properties}

In this section, we gather some elementary combinatorial facts about graphs.
Here, we fix an arbitrary graph $G=(V,E)$ of maximum degree $\Delta$.

\begin{lemma}\label{lem:sizes}
	Let $U \subset V$ be finite and let $t>0$. Then
	\[ |N_t(U)| \le \frac{\Delta}{t} \cdot |U| .\]
\end{lemma}
\begin{proof}
	This follows from a simple double counting argument.
	\[ t |N_t(U)|
		\le \sum_{v \in N_t(U)} |N(v) \cap U|
		= \sum_{u \in U} \sum_{v \in N_t(U)} \1_{N(u)}(v)
		= \sum_{u \in U} |N(u) \cap N_t(U)|
		\le \Delta |U| . \qedhere \]
\end{proof}

The next lemma follows from a classical result of Lov{\'a}sz~\cite[Corollary~2]{lovasz1975ratio} about fractional vertex covers,
applied to a weight function assigning a weight of $\frac1t$ to each vertex of $S$.

\begin{lemma}\label{lem:existence-of-covering2}
Let $S \subset V$ be finite and let $t \ge 1$. Then there exists a set $T \subset S$ of size $|T|~\hspace{-4pt}\le~\hspace{-4pt}\frac{1+\log \Delta}{t} |S|$ such that $N_t(S) \subset N(T)$.
\end{lemma}

The following standard lemma gives a bound on the number of connected subsets of a graph.
\begin{lemma}[{\cite[Chapter~45]{Bol06}}]\label{lem:number-of-connected-graphs}
The number of connected subsets of $V$ of size $k+1$ which contain the origin is at most $(e(\Delta-1))^k$.
\end{lemma}

\subsection{Isoperimetry}
The following simple isoperimetric inequalities are useful for small sets. We use them only in \cref{sec:enum}.

\begin{lemma}\label{lem:isoperimetry}
Let $K \subset \Z^d$ be finite. Then
\begin{enumerate}
 \item $|\extB K| \ge 2d|K| - 2|K|^2$.
 \item $|K^+| \ge |K| (2d-\diam K)/(\diam K+1)$.
\end{enumerate}
\end{lemma}
\begin{proof}
The first inequality follows from the observation that each vertex in $K$ has $2d$ neighbors of which at least $2d-2|K|$ are in $\extB K$ and unique to it, as any other vertex in $K$ excludes at most 2 of these neighbors (see also~\cite[Lemma~6.2]{Galvin2003hammingcube}).

Let us turn to the second inequality.
Denote $D:=\diam K$.
Let $L_i$ be the set of vertices of $\Z^d$ at distance $i$ from some fixed vertex of $K$. Each vertex in $K \cap L_i$ has at least $2d-i$ neighbors in $L_{i+1}$, and each vertex in $L_{i+1}$ has at most $i+1$ neighbors in $L_i$. It follows that $|N(K \cap L_i) \cap L_{i+1}| \ge |K \cap L_i| (2d-i)/(i+1)$. Thus,
\[ |K^+| = 1+\sum_{i=0}^D |K^+ \cap L_{i+1}| \ge \sum_{i=0}^D |N(K \cap L_i) \cap L_{i+1}| \ge \sum_{i=0}^D |K \cap L_i| \frac{2d-i}{i+1} \ge |K| \frac{2d-D}{D+1} . \qedhere \]
\end{proof}

\subsection{Entropy}\label{sec:entropy}

In this section, we give a brief background on entropy (see, e.g., \cite{mceliece2002theory} for a more thorough discussion). Let $Z$ be a discrete random variable and denote its support by $\supp Z$.
The \emph{Shannon entropy} of $Z$ is
\[ \Ent(Z) := -\sum_z \Pr(Z=z) \log \Pr(Z=z) ,\]
where we use the convention that such sums are always over the support of the random variable in question.
Given another discrete random variable $Y$, the conditional entropy of $Z$ given $Y$ is
\[ \Ent(Z \mid Y) := \E\big[ \Ent(Z \mid Y=y) \big] = - \sum_y \Pr(Y=y) \sum_z \Pr(Z=z \mid Y=y) \log \Pr(Z=z \mid Y=y) .\]
This gives rise to the following chain rule:
\begin{equation}\label{eq:entropy-chain-rule}
\Ent(Y,Z) = \Ent(Y) + \Ent(Z \mid Y) ,
\end{equation}
where $\Ent(Y,Z)$ is shorthand for the entropy of $(Y,Z)$.
A simple application of Jensen's inequality gives the following two useful properties:
\begin{equation}\label{eq:entropy-support}
\Ent(Z) \le \log |\supp Z|
\end{equation}
and
\begin{equation}\label{eq:entropy-jensen}
\Ent(Z \mid Y) \le \Ent(Z \mid \phi(Y))\qquad\text{for any function }\phi .
\end{equation}
Equality holds in~\eqref{eq:entropy-support} if and only if $Z$ is a uniform random variable.
Together with the chain rule, \eqref{eq:entropy-jensen} implies that entropy is subadditive. That is, if $Z_1,\dots,Z_n$ are discrete random variables, then
\begin{equation}\label{eq:entropy-subadditivity}
\Ent(Z_1,\dots,Z_n) \le \Ent(Z_1) + \cdots + \Ent(Z_n) .
\end{equation}
The following is an extension of this inequality.

\begin{lemma}[Shearer's inequality~\cite{chung1986some}]\label{lem:shearer}
Let $Z_1,\dots,Z_n$ be discrete random variables. Let $\cI$ be a collection of subsets of $\{1,\dots,n\}$ such that $|\{I \in \cI : i \in I\}| \ge k$ for every~$i$.
Then
\[ \Ent(Z_1,\dots,Z_n) \le \frac{1}{k} \sum_{I \in \cI} \Ent((Z_i)_{i\in I}) .\]
\end{lemma}

\section{Main steps of proof}
\label{sec:high-level-proof}

In this section, we give the main steps of the proof of \cref{thm:long-range-order} (the quantitative version of \cref{thm:long-range-order-NQ}), providing definitions, stating lemmas and propositions, and concluding \cref{thm:long-range-order} from them.
The proofs of the technical lemmas and propositions are given in subsequent sections.
The reader may also find it helpful to consult the overview given in the companion paper~\cite{peledspinka2018colorings}, as the high-level proofs are rather similar.
The assumption of \cref{thm:long-range-order} is that either~\eqref{eq:parameter-inequalities-simple}, \eqref{eq:parameter-inequalities-simple1}, \eqref{eq:parameter-inequalities-simple2}, \eqref{eq:parameter-inequalities-simple3} or \cref{main-cond} holds. It will be shown in \cref{sec:model_on_K2d2d} that the four former conditions are particular cases of the latter \cref{main-cond}. We thus assume the latter and let $\alpha,\gamma,\epsilon,\bar\epsilon$ be as in \cref{main-cond}. We also let $\fq$ be as in~\eqref{eq:fq-def} and $\tilde\alpha$ as in~\eqref{eq:alpha-tilde}.

\subsection{Notation}
\label{sec:overview-notation}
Throughout \cref{sec:high-level-proof}, we fix a domain $\Lambda \subset \Z^d$ and a dominant pattern
\begin{equation}\label{eq:P_0_def}
P_0=(A_0,B_0) \qquad\text{such that}\qquad |A_0| \le |B_0| .
\end{equation}
Recall that
\begin{equation}\label{eq:finite_volume_measure}
\Pr_{\Lambda,P_0}\text{ is supported on configurations satisfying that $\intB \Lambda$ is in the $P_0$-pattern}.
\end{equation}
In proving statements for this finite-volume measure, it will be technically convenient to work in an infinite-volume setting as follows. Sample $f$ from $\Pr_{\Lambda,P_0}$ and extend it to a configuration on $\Z^d$ by requiring that
\begin{equation}\label{eq:prob_outside_Lambda_def1}
\{ f(v) \}_{v \in \Z^d \setminus \Lambda}\quad \text{are independent random variables, independent also from }f|_\Lambda,
\end{equation}
and
\begin{equation}\label{eq:prob_outside_Lambda_def2}
\quad\Pr(f(v)=i) = \begin{cases}
 \frac{\lambda_i}{\lambda_{A_0}}&\text{if $v$ is even and $i \in A_0$}\\
 \frac{\lambda_i}{\lambda_{B_0}}&\text{if $v$ is odd and $i \in B_0$}\\
 0&\text{otherwise}
\end{cases} \qquad\text{ for any }v \notin \Lambda ,
\end{equation}
where we recall from~\eqref{eq:lambda of a set} the notation $\lambda_I = \sum_{i \in I} \lambda_i$ for a set $I \subset \SS$.
With a slight abuse of notation, we continue to denote the distribution of the configuration $f$ obtained as such by $\Pr_{\Lambda,P_0}$.

Recall the notion of equivalent patterns given in \cref{sec:the model}.
We say that two patterns $(A,B)$ and $(A',B')$ are \emph{direct-equivalent} if there is a bijection $\varphi \colon \SS \to \SS$ such that
\begin{equation}\label{eq:direct-equivalent-patterns-def}
\varphi(A)=A',\quad \varphi(B)=B',\quad \lambda_{\varphi(i)}=\lambda_i,\quad \lambda_{\varphi(i),\varphi(j)}=\lambda_{i,j}\qquad\text{for all }i,j \in \SS.
\end{equation}
Thus, $(A,B)$ and $(A',B')$ are equivalent if and only if $(A,B)$ is direct-equivalent to either $(A',B')$ or $(B',A')$.
Denote the set of dominant patterns by $\phasedom$.
Let $\phase_0$ be the set of dominant patterns which are direct-equivalent to $P_0$ and let $\phase_1 := \phasedom \setminus \phase_0$ be the set of those which are not.
The difference between dominant patterns in $\phase_0$ and $\phase_1$ plays an important role. For this reason, it will be convenient to use a notation which distinguishes the two.
For $P=(A,B) \in \phasedom$, denote
\begin{equation}\label{eq:P_bdry_inner_def}
(P_\bdry,P_\inner) := \begin{cases}(A,B) &\text{if }P \in \phase_0\\(B,A) &\text{if }P \in \phase_1\end{cases} ,
\end{equation}
so that, for any $P \in \phasedom$,
\begin{equation}\label{eq:P_bdry_P_inner}
(P_\bdry,P_\inner)\text{ is direct-equivalent to $P_0$}\qquad\text{and}\qquad |P_\bdry| \le |P_\inner| .
\end{equation}
We think of $P_\bdry$ as the ``small side'' of the dominant pattern and of $P_\inner$ as the ``large side'' (although they may happen to have the same size). The reason for the names ``bdry'' and ``int'' (short for boundary and internal) will become more apparent in the next section (vertices on the boundary of $P$-ordered regions will have values in $P_\bdry$).
We also use the term $P$-even to mean ``even'' when $P \in \phase_0$ and ``odd'' when $P \in \phase_1$, and similarly, $P$-odd to mean ``odd'' when $P \in \phase_0$ and ``even'' when $P \in \phase_1$. With this terminology, for any $P \in \phasedom$ and $v \in \Z^d$,
\begin{equation}\label{eq:P-even-odd-in-P-phase}
\begin{aligned}
	&\text{$v$ is in the $P$-pattern}~\iff~ f(v) \in P_\bdry &&\qquad\text{when $v$ is $P$-even,}\\
	&\text{$v$ is in the $P$-pattern}~\iff~ f(v) \in P_\inner &&\qquad\text{when $v$ is $P$-odd.}
\end{aligned}
\end{equation}
Note that $P_0$-even is even and $P_0$-odd is odd.
We denote by $\Even_P$ and $\Odd_P$ the set of $P$-even and $P$-odd vertices of $\Z^d$, respectively.

As the notion of patterns suggests and as was explained in \cref{sec:four-params}, pairs of spins interacting with maximum interaction weight are of primary importance. It is therefore convenient to consider the graph $H$ of maximum interaction pairs -- this is the graph on $\SS$ whose edge set is
\[ E(H) = \{ \{i,j\} \subset \SS : \lambda_{i,j} = \lambda^{\text{int}}_{\text{max}} \} .\]
Note that $H$ may have self-loops. Following our graph notation (see \cref{sec:notation}), we write $i \sim j$ for adjacent $i$ and $j$ in $H$, and $N(i)$ for the neighborhood of $i$ in $H$.
Another notion which plays a role in the proofs is that of common neighbors in $H$.
For a set $I \subset \SS$, we denote the \emph{common neighbors} of $I$ in $H$ by
\[ R(I) := \bigcap_{i \in I} N(i) = \{ j \in \SS : j \sim i\text{ for all }i \in I \} ,\]
where it is understood that $R(\emptyset):=\SS$.
Thus, $R(I)$ represents the set of spins which interact with maximum interaction weight with all spins in $I$.
Observe that $I \subset R(R(I))$ for any $I \subset \SS$ and that $(A,B)$ is a pattern if and only if $A \subset R(B)$ and $B \subset R(A)$. In particular, $(I,R(I))$ is always a pattern.
Say that
\[ I \subset \SS \text{ is an \emph{$R$-set} if it satisfies }I=R(R(I)) .\]
Observe that $R^3=R$ so that $R(I)$ is an $R$-set for any $I \subset \SS$. In particular, $\SS=R(\emptyset)$ is always an $R$-set, whereas $\emptyset$ is an $R$-set if and only if no spin is adjacent to all spins (including itself).
Note also that a pattern $(A,B)$ is maximal if and only if $A=R(B)$ and $B=R(A)$ if and only if $A$ and $B$ are $R$-sets. In particular, every dominant pattern is maximal and the mapping $(A,B) \mapsto A$ is a bijection between the set of maximal patterns and the collection of $R$-sets.
Note that any pattern $(A,B)$ extends to a maximal pattern $(A',B')=(R(R(A)),R(A))$ in which $A \subset A'$ and $B \subset B'$.
Consider the equivalence relation $\simeq_R$ on subsets of $\SS$ in which
\[ I \simeq_R J \quad\iff\quad R(I)=R(J) \qquad[\iff R(R(I))=R(R(J))] .\]
Then the $R$-set $R(R(I))$ is a canonical representative of the equivalence class of~$I$.
Throughout the proof, we will mostly be interested in the $\simeq_R$ equivalence class of a given set.

\smallskip\noindent
{\bf Policy on constants:} We continue to use the policy described in the beginning of~\cref{sec:applications}.

\subsection{Identification of ordered and disordered regions}\label{sec:unlikeliness-of-interfaces}

Let $f \colon \Z^d \to \SS$ be a configuration.
We wish to identify ``ordered'' and ``disordered'' regions in the configuration. That is, to each dominant pattern $P$, we aim to associate a subset of $\Z^d$, with the idea that the vertices in this subset are ``ordered'' according to the $P$-pattern.
A first naive idea is to consider the set
\[ S_P(f) := \big\{ v \in \Z^d : v\text{ is in the $P$-pattern}\big\} .\]
However, in many models (e.g., the AF Potts model), every vertex is in several different dominant patterns and this will not lead to a useful notion of ordering. It turns out to be more useful to look at whether all the neighbors of the vertex are in the $P$-pattern. It also turns out to be important to distinguish between $P$-even and $P$-odd vertices at this point. Specifically, we define
\begin{align*}
 T_P(f) &:= \big\{ v \in \Z^d : v\text{ is $P$-odd},~ N(v) \subset S_P(f) \big\} \\&\phantom{:}= \big\{ v \in \Z^d : v\text{ is $P$-odd},~ f(N(v)) \subset P_\bdry \big\}
\end{align*}
and
\begin{equation}\label{eq:Z-def}
Z_P=Z_P(f) := T_P(f)^+ \qquad\text{and}\qquad Z'_P=Z'_P(f) := (T_P(f) \setminus S_P(f))^+ .
\end{equation}
One should regard $Z_P$ as the region that is ordered according to the $P$-pattern (more precisely, the $P$-even vertices there are in the $P$-pattern), with the exception that $Z'_P \subset Z_P$ marks the sub-region where this order is partially violated (more precisely, the $P$-even vertices there are adjacent to a vertex that is not in the $P$-pattern).
In some special cases, including the zero-temperature AF Potts model (but not all homomorphism models), the $Z_P$ are sufficient for us to carry out our arguments. However, in the general case, we also require the additional information encoded by the $Z'_P$.
Note that the $Z_P$ may overlap each other and that their union may not cover the entire space.
This motivates the following notation:
\[ Z_\bad := \bigcap_P (Z_P)^c,\qquad Z_\overlap := \bigcup_{P \neq Q} (Z_P \cap Z_Q),\qquad Z_\hole := \bigcup_P Z'_P .\]
As we shall show, regions of this type, along with the boundaries of $Z_P$, are regions where the configuration $f$ does not achieve its maximal entropy per vertex, in a way which is quantified later. It will be our task to prove that such regions are not numerous and this will lead to a proof of \cref{thm:long-range-order}. To this end, we define
\begin{equation}\label{eq:Z_*-def}
Z_* := \bigcup_P \intextB Z_P \cup Z_\bad \cup Z_\overlap \cup Z_\hole.
\end{equation}
\cref{fig:breakup} depicts these sets for a configuration of the $5$-state AF Potts model (\cref{sec:AF Potts model}).
\begin{figure}
	\centering
	\includegraphics[scale=0.35]{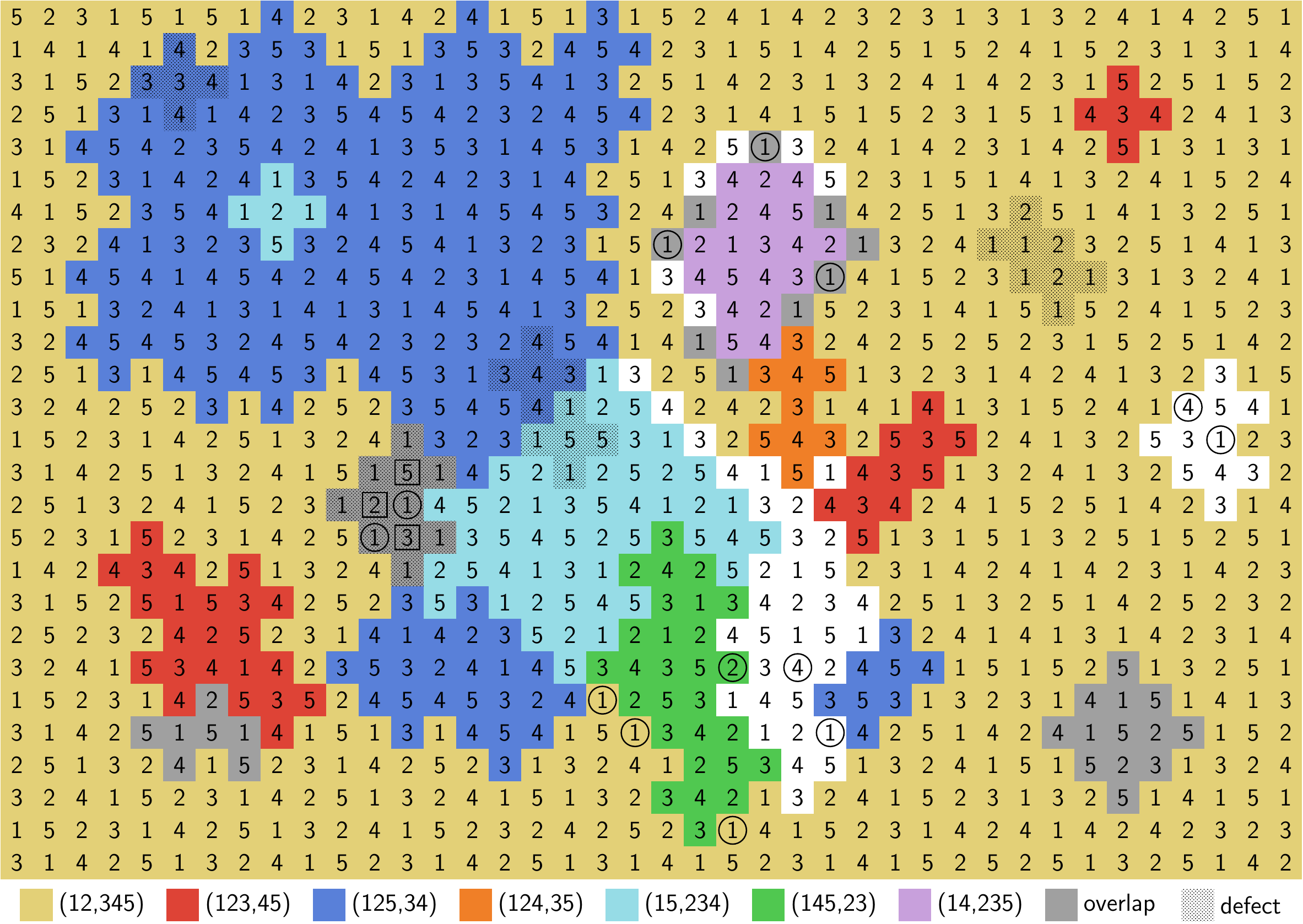}
	\captionsetup{width=0.95\textwidth, font=small}
	\caption{The sets $(Z_P)$, $Z_\overlap$, $Z_\hole$ and $Z_\bad$ for a configuration of the 5-state AF Potts model (each $Z_P$ has a different color, gray indicates $Z_\overlap$, white indicates $Z_\bad$, and dotted background indicates $Z_\hole$). Non-dominant vertices (defined in~\cref{sec:Shearer_overview}) are depicted: squares indicate vertices whose neighbors are assigned less than $\lfloor \frac q2 \rfloor$ different values, whereas circles indicate vertices whose neighbors are assigned more than $\lceil \frac q2 \rceil$ different values. This figure is taken and modified from the overview section of the companion paper~\cite{peledspinka2018colorings}, where the reader may find additional figures illustrating subsequent proof steps.}
	\label{fig:breakup}
\end{figure}

\subsection{Breakups -- definition and existence}

With \cref{thm:long-range-order} in mind, our goal is to show that $v$ is typically in the $P_0$-pattern. One checks that $Z_P \setminus Z'_P$ is in the $P$-pattern, and therefore it suffices to show that, with high probability, $v$ belongs to $Z_{P_0} \setminus Z_*$. This, in turn, follows by showing that there is a path from $v$ to infinity avoiding $Z_*$. If no such path exists, there needs to be a \emph{connected} component of $Z_*^+$ which disconnects~$v$ from infinity. Our focus is then on these connected components and this motivates the notions of a breakup and a breakup seen from $v$, which we now aim to define.

The geometric structure of a breakup is captured by the following notion of an atlas.
An \emph{atlas} is a collection $X = (X_P,X'_P)_{P \in \phasedom}$ of pairs of subsets of $\Z^d$ such that, for every $P$,
\begin{equation}\label{eq:def-atlas}
 X'_P \subset X_P\qquad\text{and}\qquad\text{$X_P$ and $X'_P$ are regular $P$-even sets}.
\end{equation}
For an atlas $X$, we define
\[ X_\bad := \bigcap_P (X_P)^c,\qquad X_\overlap := \bigcup_{P \neq Q} (X_P \cap X_Q),\qquad X_\hole := \bigcup_P X'_P ,\]
and
\[ X_* := \bigcup_P \intextB X_P \cup X_\bad \cup X_\overlap \cup X_\hole .\]
We say that an atlas $X$ is \emph{non-trivial} if $X_*$ is non-empty and that it is \emph{finite} if $X_*$ is finite. For a set $V \subset \Z^d$, we also say that an atlas $X$ is \emph{seen from $V$} if every finite connected component of $X_*^{+5}$ disconnects some vertex $v \in V$ from infinity.

Let $f \colon \Z^d \to \SS$ be a configuration.
An atlas $X$ is a \emph{breakup of $f$} (with respect to the fixed domain $\Lambda$ and boundary pattern $P_0$) if it satisfies that
\begin{equation}\label{eq:breakup-0}
 \Lambda^c \subset X_{P_0}
\end{equation}
and the following two conditions hold for every dominant pattern $P$ and every vertex $v \in X_*^{+5}$:
\begin{align}
&\text{If $v$ is $P$-odd then}&&v \in X_P ~\iff~ N(v)\text{ is in the $P$-pattern.} \label{eq:breakup-1}\\
&\text{If $v$ is $P$-even then}&&v \in X'_P ~\iff~ N(v) \cap X_P\text{ is not in the $P$-pattern} \label{eq:breakup-2}
\end{align}
(recalling that we say that a set is not in the $P$-pattern when \emph{at least one} of its vertices is not in the $P$-pattern).

The latter two properties are stated in terms of the values on the neighbors of a vertex $v$. It is convenient to note their implication on the value at $v$ itself.
Suppose that $X$ is a breakup of $f$ and let $P$ be a dominant pattern.
Then, by~\eqref{eq:P-even-odd-in-P-phase}, \eqref{eq:def-atlas}, \eqref{eq:breakup-1} and~\eqref{eq:breakup-2},
\begin{align}
&f(v) \in P_\bdry &&\text{for any $P$-even }v \in X_*^{+5} \cap X_P, \label{eq:breakup-prop-even}\\
&f(v) \in P_\inner &&\text{for any $P$-odd }v \in X_*^{+5} \cap X_P \setminus X'_P. \label{eq:breakup-prop-odd}
\end{align}
Observe also that
\begin{align}
 &f(N(v)) \not\subset P_\bdry &&\qquad\text{for any $P$-odd }v \in X_\bad , \label{eq:breakup-prop-bad}\\
 f(u) \in P_\bdry\quad\text{and}\quad &f(N(v)) \not\subset P_\bdry &&\qquad\text{for any }(u,v) \in \dpartial X_P , \label{eq:breakup-prop-bdry-X_P}\\
f(u) \in P_\bdry\quad\text{and}\quad &f(N(u)) \not\subset P_\inner &&\qquad\text{for any $P$-even }u \in X'_P. \label{eq:breakup-prop-bdry-X'_P}
\end{align}

The above definition allows $f$ to have multiple breakups. A trivial example of a breakup, for which $X_* = \emptyset$, is obtained when $X_{P_0}=\Z^d$ while $X_P = \emptyset$ for all $P\neq P_0$ and all $X'_P = \emptyset$. A second example of a breakup is $X = (Z_P, Z'_P)_P$, for which $X_* = Z_*$ (see~\cref{fig:breakup}). More generally, the idea behind the definition is that some subset of the connected components of $Z_*$ is selected (though not every choice is possible) and then $X$ is set up in such a way that $X_*$ is exactly the union of the selected components, and each $(X_P, X'_P)$ coincides with $(Z_P, Z'_P)$ in a suitable neighborhood of $X_*$.

The following lemma, whose proof is given in \cref{sec:breakup-construction}, shows that whenever there is a violation of the boundary pattern, there exists a breakup that ``captures'' that violation.
Let $Z_*^{+5}(f,V)$ denote the union of connected components of $Z_*(f)^{+5}$ that are either infinite or disconnect some vertex in $V$ from infinity.

\begin{lemma}[existence of breakups seen from a vertex/set]\label{lem:existence-of-breakup}
	Let $f\colon \Z^d \to \SS$ be such that $\Int(\Lambda)^c$ is in the $P_0$-pattern and let $V \subset \Lambda$. Then there exists a breakup $X$ of $f$ such that $X_*^{+5}=Z_*^{+5}(f,V)$.
	In particular,
	\begin{itemize}
	 \item $X$ is seen from $V$.
	 \item $X$ is non-trivial if $V^{+5}$ either intersects $Z_*(f)$ or is not in the $P_0$-pattern.
	 \item $V^{+5} \cap X_{P_0} \setminus X'_{P_0}$ is in the $P_0$-pattern.
	\end{itemize}
\end{lemma}

\subsection{Unlikeliness of breakups}\label{sec:unlikeliness-of-breakups}
Now that we have a definition of breakup and we know that any violation of the boundary pattern creates a non-trivial breakup, it remains to show that breakups are unlikely.

The main part of the proof consists of obtaining a quantitative bound on the probability of a large breakup.
Nevertheless, formally one also needs to rule out the existence of an infinite breakup. As this does not require a quantitative bound, it is actually rather simple to do so. The following lemma is proved in \cref{sec:no-infinite-breakups}.
\begin{lemma}\label{lem:no-infinite-breakups}
	$\Pr_{\Lambda,P_0}$-almost surely, every breakup seen from a finite set is finite.
\end{lemma}

We now discuss the quantitative bound on finite breakups.
To this end, denote by $\breakups$ the collection of atlases which have a positive probability of being a breakup and denote
\[ \breakups_{L,M,N} := \left\{ X \in \breakups ~:~ \Big|\bigcup_P \partial X_P\Big|=L,~|X_\overlap \cup X_\hole|=M,~|X_\bad|=N \right\} .\]

\begin{prop}\label{prop:prob-of-breakup-associated-to-V}
	For any finite $V \subset \Z^d$ and any integers $L,M,N \ge 0$, we have
	\[ \Pr_{\Lambda,P_0}(\text{there exists a breakup in }\breakups_{L,M,N}\text{ seen from }V) \le 2^{|V|} \cdot \exp\left(- c\tilde\alpha \big( \tfrac{L}{d}+M+ \epsilon N \big) \right) .\]
\end{prop}
This is the main technical proposition of this paper. An overview of the tools used to prove the proposition is given in the rest of \cref{sec:high-level-proof}, with the detailed proofs appearing in \cref{sec:breakup}, \cref{sec:shift-trans} and \cref{sec:approx}.

Based on the above, it is now a simple matter to deduce \cref{thm:long-range-order}.

\begin{proof}[Proof of \cref{thm:long-range-order}]
	Suppose that $v$ is not in the $P_0$-pattern.
	\cref{lem:existence-of-breakup} implies the existence of a non-trivial breakup $X$ seen from $v$.
	By \cref{lem:no-infinite-breakups}, we may assume that $X$ is finite so that $X \in \breakups_{L,M,N}$ for some $L,M,N \ge 0$. Since $X$ is also non-trivial, some set in $\{ X_P,X_P^c,X'_P,(X'_P)^c \}_P$ is both non-empty and not $\Z^d$. Recalling~\eqref{eq:def-atlas} and applying \cref{lem:boundary-size-of-odd-set} (or its analogue for even sets) to any such set shows that $L\ge d^2$.
	Therefore, by \cref{prop:prob-of-breakup-associated-to-V},
	\[ \Pr_{\Lambda,P_0}\big(v\text{ is not in the $P_0$-pattern}\big) \le \sum_{\substack{L \ge d^2,\,M,N \ge 0}}2\exp\left(-c\tilde\alpha \big( \tfrac{L}{d}+M+ \epsilon N \big) \right) .\]
	Using~\eqref{eq:alpha-cond} and perhaps decreasing the universal constant there, the desired inequality follows.
\end{proof}

\subsection{Unlikeliness of specific breakups}

Towards establishing \cref{prop:prob-of-breakup-associated-to-V}, it is natural to first prove that a specific atlas is unlikely to be a breakup. Precisely, we would like to show that, for any $X \in \breakups_{L,M,N}$, we have
	\begin{equation}\label{eq:given_breakup_bound}
	\Pr_{\Lambda,P_0}(X\text{ is a breakup}) \le \exp\left(- c\alpha \big( \tfrac{L}{d}+M+\epsilon N \big)\right) .
	\end{equation}
Proving this bound (or rather the mildly stronger \cref{prop:prob-of-given-breakup} below) is one of two main technical parts of our paper and involves as a key step the use of Shearer's inequality in order to quantitatively estimate the loss of entropy (and energy) on the regions of the configuration in $X_*$. The parts pertaining to Shearer's inequality are developed in \cref{sec:Shearer_overview} and \cref{sec:shift-trans}, while \cref{prop:prob-of-given-breakup} is deduced in \cref{sec:breakup}.

It is temping to conclude that breakups are unlikely by summing the bound~\eqref{eq:given_breakup_bound} over all atlases in $\breakups_{L,M,N}$ that are seen from $V$. This approach applies to spin systems in which the parameter $\alpha$ is sufficiently large, as the bound~\eqref{eq:given_breakup_bound} is sufficiently small in this case. Unfortunately, this approach fails in many models of interest to us, such as the AF Potts model (even at zero temperature, where it becomes the proper coloring model), as the size of the above collection of atlases exceeds the reciprocal of the bound~\eqref{eq:given_breakup_bound}.
Overcoming this obstacle forms the second main technical part of our paper and requires an analysis of the structure of atlases. This idea is developed in detail in \cref{sec:approx-overview} and \cref{sec:approx}.

It is convenient to state already here a bound which is stronger than~\eqref{eq:given_breakup_bound}, involving a certain coarse-graining of the sets $X_P'$, as the details of this step are rather distinct and simpler than those used in the general coarse-graining strategy discussed in \cref{sec:approx-overview} below. Thus, we prove the following statement in which, rather than the specifying the specific sets $X'_P$, one needs only to specify a single set $H$ containing all of these sets (and not too much more).
Precisely, given an atlas $X \in \breakups_{L,M,N}$ and a set $H \subset \Z^d$, denote
\begin{equation}\label{eq:H-approx-def}
\breakups_{L,M,N}(X,H) := \Big\{ \hat{X} \in \breakups_{L,M,N} : (\hat{X}_P)_P=(X_P)_P\text{ and }\hat{X}_\hole \subset H \subset \hat{X}_*^{+3} \Big\}.
\end{equation}

\begin{prop}\label{prop:prob-of-given-breakup}
	For any $X \in \breakups_{L,M,N}$ and $H \subset \Z^d$ such that $|H| \le M\sqrt{d}$, we have
	\[ \Pr_{\Lambda,P_0}\big(\text{there exists a breakup }\hat{X} \in \breakups_{L,M,N}(X,H)\big) \le \exp\left(- c\alpha \big( \tfrac{L}{d}+M+\epsilon N \big)\right) .\]
\end{prop}

\subsection{Approximations}\label{sec:approx-overview}
As mentioned, the standard union bound does not allow to upgrade the bound of \cref{prop:prob-of-given-breakup} to that of \cref{prop:prob-of-breakup-associated-to-V}.
Instead, we employ a delicate coarse-graining scheme for the possible breakups according to their rough features.
For notational convenience, we write $Q \simeq P$ when $Q$ and $P$ are direct-equivalent dominant patterns.

Let $A=((A_P)_{P \in \phasedom},A_\hole,A^*,A^{**})$ be a collection of subsets of $\Z^d$ such that each $A_P$ is $P$-even and $A^* \subset A^{**}$.
We say that $A$ is an \emph{approximation of an atlas $X \in \breakups_{L,M,N}$} if the following conditions hold for all $P$:
\begin{enumerate}[label=\text{(A\arabic*)},ref=\text{(A\arabic*)}]
  \item \label{it:approx-A_P} $A_P \subset X_P \subset A_P \cup (\Odd_P \cap A^*) \cup (\Even_P \cap A^{**})$.
 \item \label{it:approx-A*_P} $\Odd_P \cap A^* \subset N_d(\bigcup_{Q \simeq P} A_Q)$.
 \item \label{it:approx-hole} $A_\hole \subset X_\hole \cup N_{2d}(X_\hole) \subset A_\hole \cup A^{**}$.
 \item \label{it:approx-unknown-location} $A^{**} \subset X_*^{+3}$.
  \item \label{it:approx-unknown-size} $|A^{**}| \le \tfrac{C(L+dM)\log d}{\sqrt{d}}$.
\end{enumerate}

Since $A^*\subset A^{**}$, property \ref{it:approx-A_P} implies that $A_P\subset X_P\subset A_P\cup A^{**}$ for all $P$. In words, the sets $A_P$ indicate vertices which are guaranteed to be in $X_P$ while the set $A^{**}$ indicates vertices whose classification into the various $X_P$ is not fully specified by the approximation. The distinguished subset $A^*\subset A^{**}$ conveys additional information through \ref{it:approx-A_P} and \ref{it:approx-A*_P}: A $P$-odd vertex is either guaranteed to belong to $X_P$ (if it belongs to $A_P$), is guaranteed not to belong to $X_P$ (if it does not belong to $A_P \cup A^*$), or at least half of its neighbors belong to $\bigcup_{Q \simeq P} A_Q$.
In light of~\ref{it:approx-hole}, we think of $A_\hole$ as a rough approximation of $X_\hole$.
The other two properties further restrict the ``missing information'', with~\ref{it:approx-unknown-location} ensuring that $A^{**}$ is only present near~$X_*$ and \ref{it:approx-unknown-size} ensuring that $A^{**}$ is not too large.

The following proposition shows that one may find a small family which contains an approximation of every atlas seen from a given set.

\begin{prop}\label{prop:family-of-odd-approx}
	For any integers $L,M,N \ge 0$ and any finite set $V \subset \Z^d$, there exists a family $\cA$ of approximations of size
	\[ |\cA| \le 2^{|V|} \cdot \exp\left(CL \tfrac{(\fq+\log d)\log d}{d^{3/2}} + C(M+N) \tfrac{ \log^2 d}{d} \right) \]
	such that any $X \in \breakups_{L,M,N}$ seen from $V$ is approximated by some element in $\cA$.
\end{prop}

Of course, working with approximations, finding a suitable modification of~\eqref{eq:given_breakup_bound} becomes a more complicated task.
The following proposition provides a similar bound on the probability of having a breakup which is approximated by a given approximation

\begin{prop}\label{prop:prob-of-odd-approx}
	For any approximation $A$ and any integers $L,M,N \ge 0$, we have
	\[ \Pr_{\Lambda,P_0}(A\text{ approximates some breakup in }\breakups_{L,M,N}) \le \exp\left(- c\tilde\alpha \big( \tfrac{L}{d}+M+ \epsilon N \big) \right) .\]
\end{prop}

We are now ready to complete the proof of \cref{prop:prob-of-breakup-associated-to-V}.

\begin{proof}[Proof of \cref{prop:prob-of-breakup-associated-to-V}]\label{proof:prob-of-breakup-associated-to-V}
Let $\cA$ be a family of approximations as guaranteed by \cref{prop:family-of-odd-approx}. Let $\Omega$ be the event that there exists a breakup in $\breakups_{L,M,N}$ seen from $V$ and let $\Omega(A)$ be the event that there exists a breakup in $\breakups_{L,M,N}$ seen from $V$ and approximated by $A$. Then, by \cref{prop:family-of-odd-approx} and \cref{prop:prob-of-odd-approx},
\[ \Pr_{\Lambda,P_0}(\Omega) \le \sum_{A \in \cA} \Pr_{\Lambda,P_0}(\Omega(A)) \le 2^{|V|} \cdot \exp\left(\tfrac{CL(\fq+\log d)\log d}{d^{3/2}} + \tfrac{C(M+N)\log^2 d}{d} - c\tilde\alpha \big( \tfrac{L}{d}+M+ \epsilon N \big) \right) .\]
The proposition now follows using~\eqref{eq:alpha-cond}.
\end{proof}

The proofs of \cref{lem:existence-of-breakup}, \cref{lem:no-infinite-breakups}, \cref{prop:prob-of-given-breakup} and \cref{prop:prob-of-odd-approx} are given in \cref{sec:breakup}. The proofs of the latter two propositions rely on \cref{lem:bound-on-pseudo-breakup}, which is stated below and proven in \cref{sec:shift-trans}. The proof of \cref{prop:family-of-odd-approx} is given in \cref{sec:approx}.

\subsection{The repair transformation and upper bounds on entropy}\label{sec:Shearer_overview}

Here we explain the main ideas behind the proofs of \cref{prop:prob-of-given-breakup} and \cref{prop:prob-of-odd-approx}, providing also important definitions and a key lemma that will be used in the proof.
As the main ideas are already present in the proof that a given breakup is unlikely, we focus on explaining~\eqref{eq:given_breakup_bound}. Thus, we let $X=(X_P,X'_P)_P$ be given and aim to bound the probability that $X$ is a breakup of $f$, when $f$ is sampled from $\Pr_{\Lambda,P_0}$.

Let $\Omega$ be the set of configurations having $X$ as a breakup.
To establish the desired bound on $\Pr_{\Lambda,P_0}(\Omega)$, we apply the following one-to-many operation to every $f \in \Omega$ (see \cref{fig:repairmap}): (i) Erase the spin values at all vertices of $X_*$. (ii) For each dominant pattern $P$, apply a permutation of $\SS$ which takes $P$ to $P_0$ (given by the assumption that all dominant patterns are equivalent) to the spin values of $f$ on $X_P\setminus X_*$, and also, in the case that $P$ and $P_0$ are not direct-equivalent, shift the configuration in $X_P \setminus X_*$ by a single lattice site in some fixed direction (such a shift was first used by Dobrushin for the hard-core model~\cite{dobrushin1968problem}). (iii) Arbitrarily assign spin values in the $P_0$-pattern at all remaining vertices (making the transformation multiple valued).

\begin{figure}
	\centering
	\captionsetup{font=small}
	\begin{subfigure}[t]{0.5\textwidth}
		\includegraphics[scale=0.204]{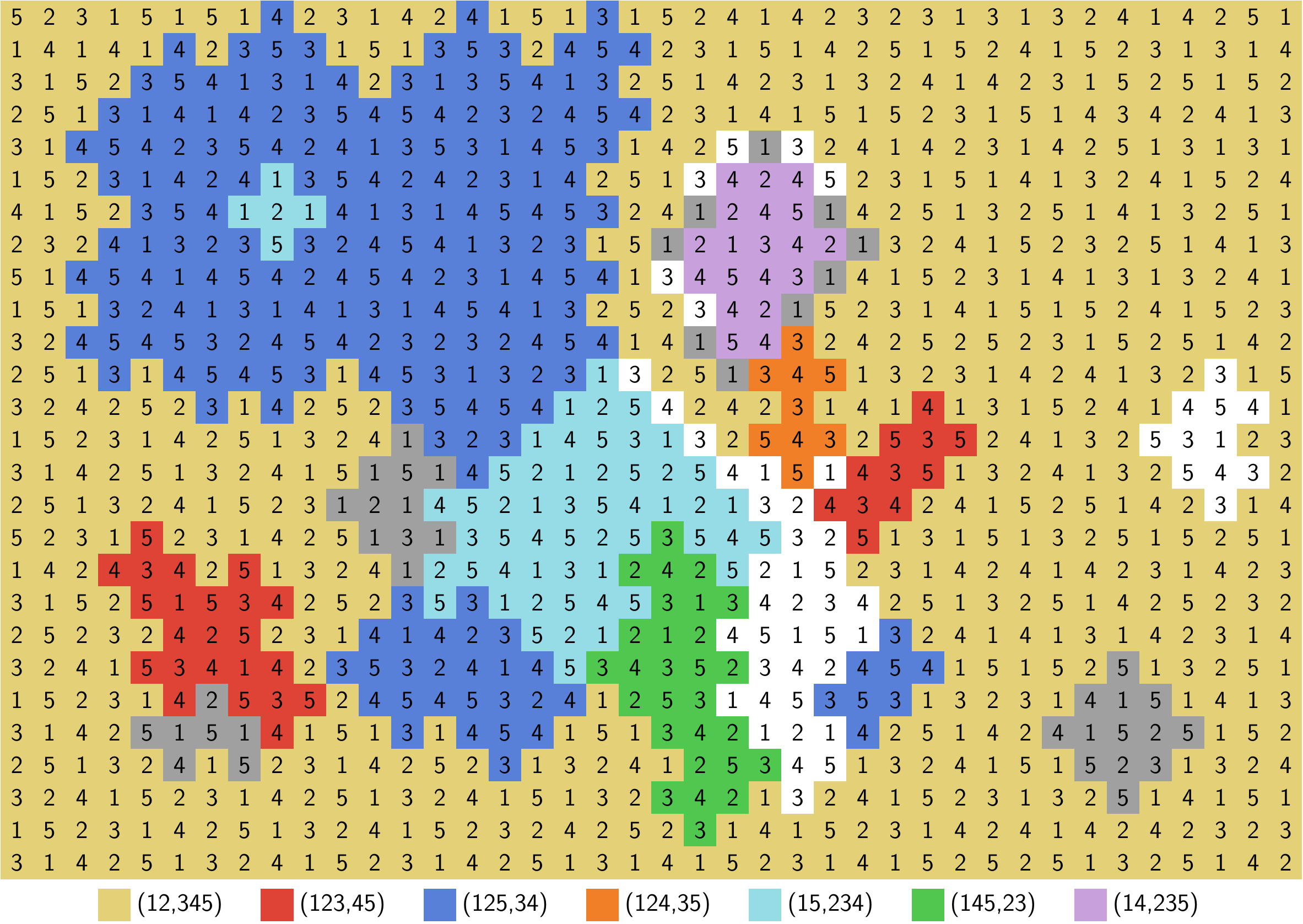}
		\caption{A coloring having a breakup $X$.}
	\end{subfigure}\,\,%
	\begin{subfigure}[t]{0.5\textwidth}
		\includegraphics[scale=0.204]{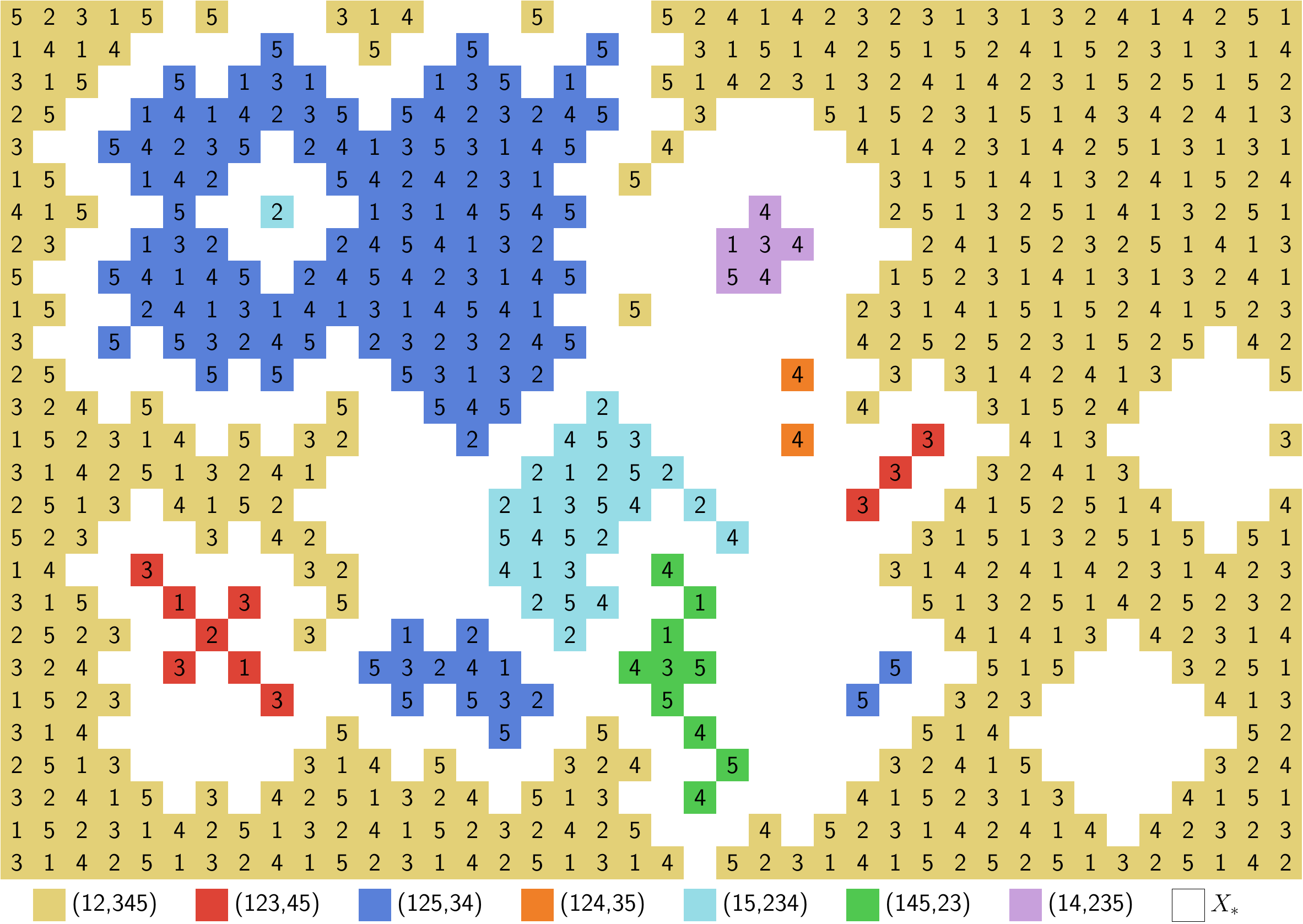}
		\caption{Step (i): colors in $X_*$ are erased.}
	\end{subfigure}
	\vspace{5pt}

	\begin{subfigure}[t]{0.5\textwidth}
		\includegraphics[scale=0.204]{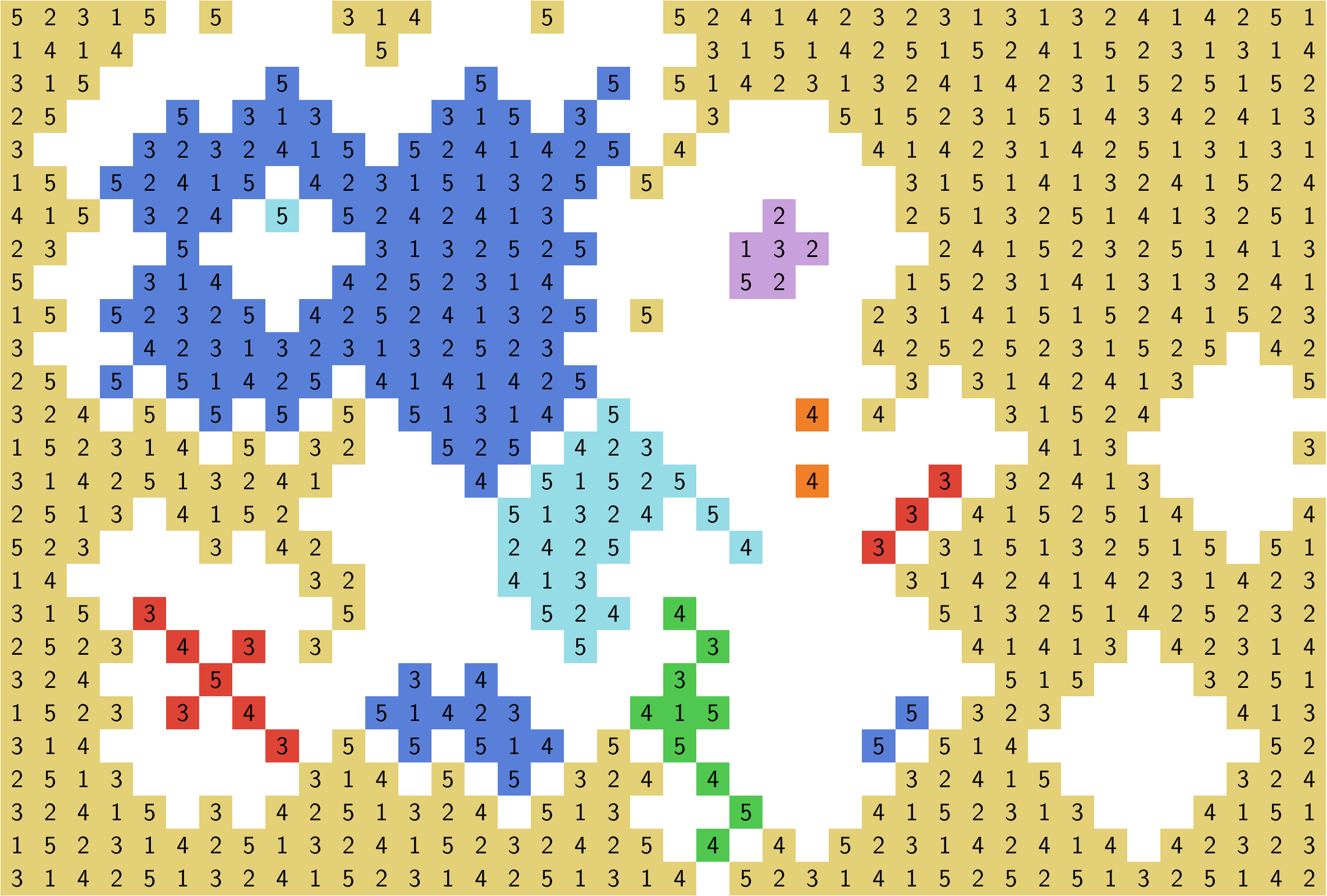}
		\caption{Step (ii): colors in $X_P$ are permuted and shifted.}
	\end{subfigure}\,\,%
	\begin{subfigure}[t]{0.5\textwidth}
		\includegraphics[scale=0.204]{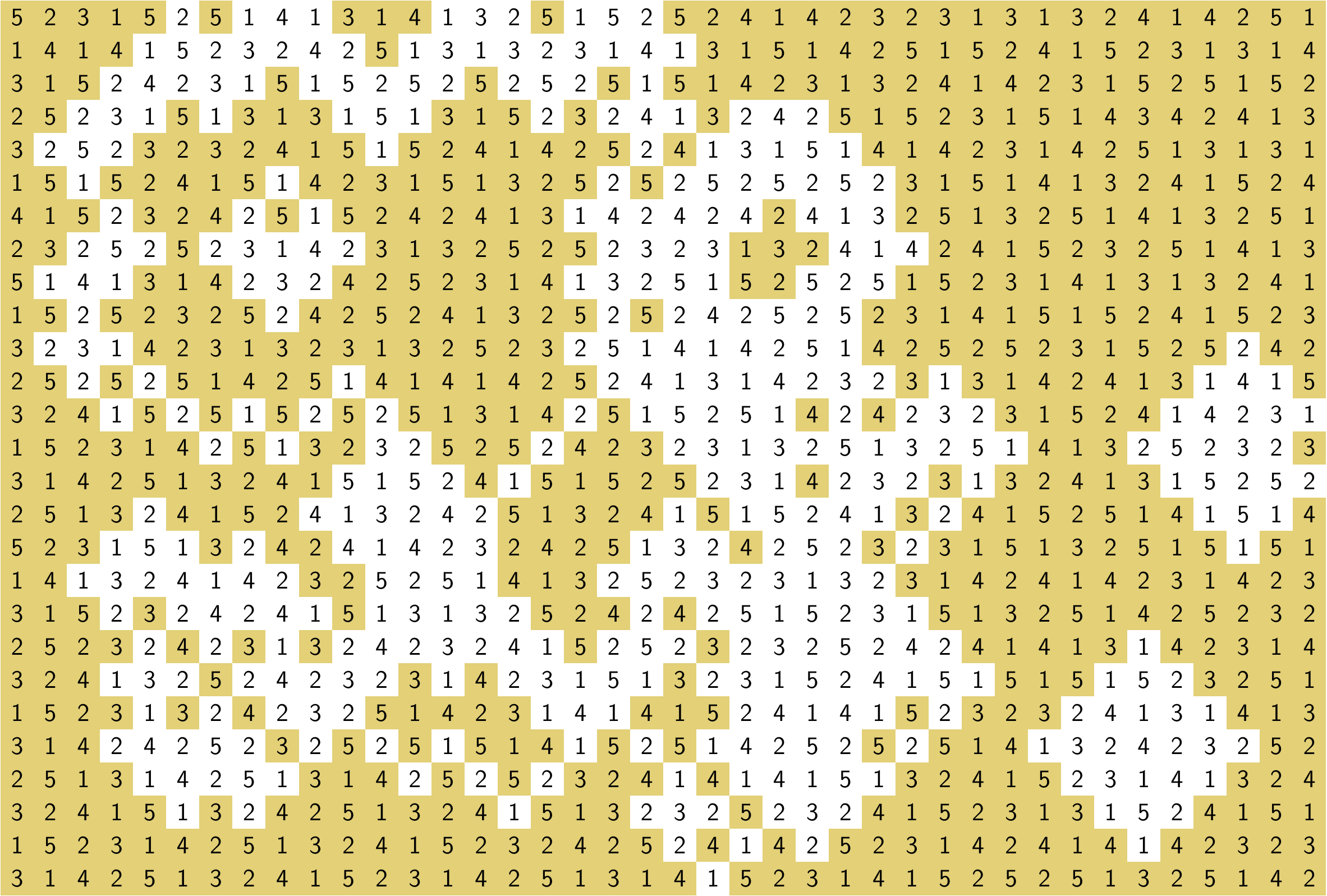}
		\caption{Step (iii): empty sites are colored in the $P_0$-pattern.}
	\end{subfigure}
	\captionsetup{font=normal}
	\caption{An illustration of the three steps of the repair transformation in the context of the $5$-coloring model (taken from~\cite{peledspinka2018colorings}).}
	\label{fig:repairmap}
\end{figure}

Noting that the resulting configuration is always of larger or equal probability than the original, and that no entropy is lost in step (ii), it remains to show that the entropy gain in step (iii) is much larger than the entropy loss in step (i). The gain in step (iii) is either $\log|A_0|$ or $\log|B_0|$ per vertex according to its parity, making the entropy gain an easily computable quantity. The main challenge is thus to bound the loss in step (i), and the method used for this purpose is described next.

Our bound relies on entropy methods (see \cref{sec:entropy} for the definition of entropy).
Specifically, we make use Shearer's inequality (\cref{lem:shearer}), first used in a similar context by Kahn~\cite{kahn2001entropy}, followed by Galvin--Tetali~\cite{galvin2004weighted}.

To get an idea of how one may use Shearer's inequality, assume for the moment that the model under consideration is a non-weighted homomorphism model (i.e., all single-site activities $\lambda_i$ are 1 and all pair interactions $\lambda_{i,j}$ are 0 or 1), and let $f\in\Omega$ be uniformly chosen (in general, we should consider the conditional distribution of $f$ given that $f \in \Omega$, and we should replace entropy a suitable counterpart). Let $F$ be the configuration coinciding with $f$ on $X_*$ and equaling a fixed symbol $\star$ on $X_*^c$. Thus, $F$ has the same entropy as $f|_{X^*}$, so that our goal is to bound the entropy of $F$.
Applying Shearer's inequality to the collection of random variables $(F_v)_{v \in \Even}$ with the collection of covering sets $\cI = \{ N(v) \}_{v \in \Odd}$ yields
\[ \Ent(F)
 = \Ent(F|_{\Even}) + \Ent(F|_{\Odd} \mid F|_{\Even})
 \le \sum_{v \in \Odd} \left[ \tfrac{\Ent(F|_{N(v)})}{2d}  + \Ent\big(F(v) \mid F|_{N(v)}\big) \right]. \]
Averaging this with the inequality obtained by reversing the roles of $\Even$ and $\Odd$ yields that
\[ \Ent(F) \le \frac{1}{2}\sum_{v} \bigg[ \underbrace{\tfrac{\Ent\big(F(N(v))\big)}{2d}}_{\textup{I}} + \underbrace{\tfrac{\Ent\big(F|_{N(v)}~\mid~ F(N(v))\big)}{2d} + \Ent\big(F(v) \mid F(N(v))\big)}_{\textup{II}} \bigg], \]

The advantage of this bound is that it is local, with each term involving only the values of $F$ on a vertex and its neighbors.
The terms corresponding to vertices $v$ at distance $2$ or more from $X_*$ equal zero as $F$ is deterministic in their neighborhood.
The boundary terms corresponding to vertices $v$ in $\intextB X_*$ need to be handled with careful bookkeeping, which we do not elaborate on here.
Each of the remaining terms admits the simple bounds $\textup{I}\le \frac{|\SS|\log 2}{2d}$ and $\textup{II}\le\log \omega_{\text{dom}}$, which only take into account the fact that $F(v)\in R(F(N(v))) \subset \SS$.
Equality in the second bound is achieved when $(F(v),F|_{N(v)})$ is uniformly distributed in $A \times B^{2d}$ for some dominant pattern $(A,B)$ (and in certain mixtures of such distributions). To obtain stronger bounds, we use additional information implied by the knowledge that $f \in \Omega$.

The type of additional information we shall use in order to improve the naive bounds is based on five notions --- \emph{non-dominant vertices}, vertices having \emph{unbalanced neighborhoods}, \emph{restricted edges}, \emph{highly energetic vertices} and vertices having a \emph{unique pattern} --- all of which we now define. These notions are somewhat abstract (and not directly related to a specific breakup) in order to allow sufficient flexibility for the proof of both \cref{prop:prob-of-given-breakup} and \cref{prop:prob-of-odd-approx}. While these notions will be used for all models, some of our heuristic explanations below are still geared toward the non-weighted homomorphism case.

Let $f \colon \Z^d \to \SS$ be a configuration and let $\Omega$ be a collection of configurations.
The five notions implicitly depend on $f$ and/or~$\Omega$. Let $v \in \Z^d$ be a vertex and let $u$ be adjacent to $v$. Recall that $(v,u) \in \dpartial v$ is the directed edge from $v$ to $u$. We say that

\smallskip
\begin{itemize}[leftmargin=15pt]

\item $v$ is \emph{non-dominant} (in $f$) if
	\[ f(N(v)) \not\simeq_R P_\bdry, P_\inner \qquad\text{for every dominant pattern }P.\]
	Thus, $v$ is non-dominant if $f(N(v))$ is not $R$-equivalent to any side of a dominant pattern. Otherwise, we say that $v$ is dominant. Non-dominant vertices are indicated by squares and circles in \cref{fig:breakup}. They yield an immediate entropy loss as they reduce the simple bound on term \textup{II} above to $\log \omega_{\text{dom}} + \log \rho_{\text{pat}}^{\text{bulk}}$.

\smallskip
\item
	$(v,u)$ is \emph{restricted} (in $(f,\Omega)$) if $v$ is non-dominant or
	$f(N(v)) \not\simeq_R A$ or $R(f(N(v))) \not\simeq_R B$, where
	\begin{align}
	A &:= \big\{g(u) : g \in \Omega,~ g(N(v)) \simeq_R f(N(v)) \big\},\label{eq:restricted_A_def}\\
	B &:= \big\{g(v) : g \in \Omega,~ g(N(v)) \simeq_R f(N(v)) \big\} \cap R(f(N(v)) .\label{eq:restricted_B_def}
	\end{align}

Recall that $R(f(N(v)))$ is exactly the set of values that $v$ may take and still interact with highest interaction weight with all of its neighbors.
Thus, roughly speaking, $(v,u)$ is restricted if upon inspection of the set of values (up to $R$-equivalence) which appears on the neighbors of $v$, either this set is not a side of a dominant pattern, or one is guaranteed that either $u$ or $v$ cannot ``legally'' (i.e., without incurring an energetic cost) take all possible values which they should typically take.
More precisely, we consider the set of $g\in\Omega$ whose set of values at the neighbors of $v$ coincide, up to $R$-closure, with the set of values of $f$ on the neighbors of $v$. We let $A$ be the set of values realizable by such $g$ at the vertex $u$. We similarly let $B$ be the set of values realizable by such $g$ at the vertex $v$, but from this set we discard those values which do not interact with highest interaction weight with all elements of $f(N(v))$. Then, the edge $(v,u)$ is restricted if the $R$-closure of $A$ together with the $R$-closure of $B$ does not form a dominant pattern (this is equivalent to the stated definition). See \cref{sec:restricted-edges} for some sufficient conditions for an edge to be restricted.

Note that if $v$ is a non-dominant vertex, then all its outgoing edges $\dpartial v$ are restricted. We remark that we have incorporated the notion of non-dominant vertices into that of restricted edges in order to reduce notation later on, but the reader may find it instructive to regard these as separate situations: either $v$ is non-dominant, or $v$ is dominant, say $f(N(v)) \simeq_R A'$ for some dominant pattern $(A',B')$, in which case $A \subset A'$ and $B \subset B'$, but either $A \not\simeq_R A'$ or $B \not\simeq_R B'$, so that the values of either $u$ or $v$ are restricted.
The latter case yields a loss in entropy as each such edge reduces the simple bound $\log \omega_{\text{dom}}$ on term \textup{II} above by $\frac{1}{2d}\log \rho_{\text{pat}}^{\text{bdry}}$.

\smallskip
\item
	$v$ has an \emph{unbalanced neighborhood} (in $f$) if it is dominant and there exists $A \subset f(N(v))$ such that either
    \[ f(N(v)) \not\simeq_R A \qquad\text{and}\qquad |\{ u\in N(v) : f(u) \in A \}| > 2d-4\bar\epsilon d ,\]
    or $A$ is of the form $A=P_\bdry$ for some dominant pattern $P$ and
    \[ f(N(v)) \not\simeq_R A \qquad\text{and}\qquad |\{ u\in N(v) : f(u) \in A \}| > 2d-4\epsilon d .\]

As $f(N(v))$ increases, the set $R(f(N(v)))$ is reduced, resulting in a trade-off in the entropy contribution at $v$ quantified by the two terms in \textup{II} above.
In order to have high entropy, if some neighbor of $v$ takes a value that causes a reduction to $R(f(N(v)))$, many other neighbors of $v$ should take advantage of this as well. The neighborhood of $v$ is therefore deemed unbalanced if there is a subset $A$ which represents a stricter smaller $R$-set than $f(N(v))$ and in which almost all (but not all) neighbors of $v$ take values in.

The precise definition is given by two separate conditions in order to allow different thresholds (given by $\epsilon$ and $\bar\epsilon$) in the case that $A$ has the special form $A=P_\bdry$ and when it has no special form.
For homomorphism models, the first condition will not play a role, but we include this flexibility as it may lead to better results in some applications for non-homomorphism models.

\smallskip
\item
	$v$ is \emph{highly energetic} (in $(f,\Omega)$) if it is dominant, it has a balanced neighborhood, and $B$, as in~\eqref{eq:restricted_B_def}, is empty.

As its name suggests, a highly energetic vertex comes at a substantial energetic cost.
Indeed, the fact that $B$ is empty means that any realizable choice of value for $v$ interacts with some neighbor of $v$ with a lower (than the maximum possible) interaction weight (since it does not belong to $R(f(N(v)))$). As we require that $v$ has a balanced neighborhood, $v$ interacts in this manner with many of its neighbors, thus leading to a large interaction cost. This notion will only be relevant for non-homomorphism models.

\smallskip
\item
	$v$ has a \emph{unique pattern} (in $\Omega$) if there exists $A \subset \SS$ such that, for every $g \in \Omega$, either $g(N(v)) \simeq_R A$ or $v$ has an unbalanced neighborhood in $g$ or all edges in $\dpartial v$ are restricted in~$(g,\Omega)$.

We may more appropriately term this notion as a unique high-entropy pattern or unique unrestricted pattern, the reason being that there is at most one choice for $R(g(N(v)))$ which does not lead to a reduction of entropy at $v$ by making it non-dominant or its neighborhood unbalanced or causing all edges in $\dpartial v$ to be restricted. For such vertices we will be able to bound term \textup{I} above more effectively (roughly improving the $1/d$ bound to an exponentially small bound in $d$). Note that this notion does not depend on $f$.
\end{itemize}
\medskip

The following lemma, which is proved in \cref{sec:shift-trans}, provides a general upper bound on the probability of certain events in terms of the above notions.
Given a configuration $f$, a collection of configurations $\Omega$ and a subset $S \subset \Z^d$, let $S^{\Omega,f}_\rest$ be the set of directed edges $(v,u)$ with $v \in S$ which are restricted in $(f,\Omega)$, let $S^f_\unbal$ be the set of vertices in $S$ which have unbalanced neighborhoods in $f$, let $S^{\Omega,f}_\highlyrest$ be the set of vertices in $S$ which are highly energetic in $(f,\Omega)$, and let $S^\Omega_\unique$ be the set of vertices in $S$ which have a unique pattern in $\Omega$.
Recall $\alpha$ and $\gamma$ from the beginning of \cref{sec:high-level-proof}. Denote
\begin{equation}\label{eq:gamma-bar-def}
\bar\gamma : = \gamma+e^{-\alpha d/25} .
\end{equation}

\begin{lemma}\label{lem:bound-on-pseudo-breakup}
	Let $S \subset \Z^d$ be finite and let $\{ S_P \}_{P \in \phasedom}$ be a partition of $S^c$ such that $\extB S_P \subset S$ for all~$P$.
	Suppose that $S \cup S_{P_0}$ contains $(\Lambda^c)^+$.
	Let $\Omega$ be an event on which $(\intB S_P)^+$ is in the $P$-pattern for every $P$ and denote
	\[ k(\Omega) := \min_{f \in \Omega} \left( \big|S^f_\unbal\big| + \tfrac{1}{d}\big|S^{\Omega,f}_\rest\big| + \epsilon d \big|S^{\Omega,f}_\highlyrest\big| \right) .\]
	Then
	\[ \Pr_{\Lambda,P_0}(\Omega) \le \exp\Big[-\tfrac{\alpha}{32}k(\Omega) + \tfrac{\fq}{d}\big|S \setminus S^\Omega_\unique\big| +\bar\gamma|S| \Big] .\]
\end{lemma}

Thus, roughly speaking, if $\Omega$ is an event on which there are almost surely many ill-behaved vertices/edges (i.e., vertices having unbalanced neighborhoods, restricted edges or highly energetic vertices), then it must be an unlikely event.
We conclude with a short outline as to how \cref{lem:bound-on-pseudo-breakup} is used to prove~\eqref{eq:given_breakup_bound}. To this end, we take $S$ to be $X_*$ and $S_P$ to be $X_P \setminus X_*$ and, as a first attempt, we take $\Omega$ to be the event that $X$ is a breakup.
Concluding~\eqref{eq:given_breakup_bound} from \cref{lem:bound-on-pseudo-breakup} is still not straightforward, as the latter, when applied directly to $\Omega$, gives an insufficient bound on its probability.
The difficulty here is that, while $k(\Omega)$ is large in comparison to $L$ and~$M$, it is not necessarily large in comparison to~$N$.
Indeed, as we will show (see \cref{lem:lower-bound-on-size-of-restricted}), every edge in $\dpartialrev X_P$ and every edge (in at least one of its two directions) incident to $X_\overlap \cup X_\hole$ is necessarily restricted in $f$, so that
\[ \tfrac{1}{d}\big|S^{\Omega,f}_\rest\big| \ge \tfrac{L}{2d} + \tfrac{M}{2} \qquad\text{for every }f \in \Omega .\]
Unfortunately, $X_\bad$ need not contain enough ill-behaved vertices/edges -- the main reason being that $X_\bad$ may contain $P$-even vertices $v$ for which $N(v)$ is in the $P$-pattern (that is, there is no analogue of~\eqref{eq:breakup-prop-bad} for $P$-even vertices; see \cref{fig:breakup}).
Instead, to obtain a good bound, we shall apply \cref{lem:bound-on-pseudo-breakup} to subevents $\Omega' \subset \Omega$ on which we have additional information about the configuration on the set $X_\bad$.
For suitably chosen subevents (see \cref{lem:lower-bound-on-size-of-restricted-with-V}), the number of ill-behaved vertices/edges in $X_\bad$ increases enough to ensure that
\[ k(\Omega') \ge c \big( \tfrac{L}{d}+ M +\epsilon N \big) .\]
As the entropy of this additional information is negligible with our assumptions (see \cref{lem:family-of-strong-odd-approx}), this will allow us to conclude~\eqref{eq:given_breakup_bound} by taking a union bound over the subevents~$\Omega'$. In addition, we also show that the ``loss'' terms $\frac{\fq}{d}\big|S \setminus S^\Omega_\unique\big|$ and $\bar\gamma|S|$ are controlled by the ``gain'' term $\frac{\alpha}{32}k(\Omega')$.
This is carried out in detail in \cref{sec:prob-of-given-breakup}, where \cref{prop:prob-of-given-breakup} is proved. The proof of \cref{prop:prob-of-odd-approx} is given in \cref{sec:prob-of-approx}.

\section{Breakups}
\label{sec:breakup}

In this section, we prove \cref{lem:existence-of-breakup} about the existence of a non-trivial breakup, we prove \cref{lem:no-infinite-breakups} about the absence of infinite breakups, we prove \cref{prop:prob-of-given-breakup} about the probability of a given breakup, and we prove \cref{prop:prob-of-odd-approx} about the probability of an approximated breakup.

\subsection{Constructing a breakup seen from a given vertex/set}
\label{sec:breakup-construction}

	Here we prove \cref{lem:existence-of-breakup}.
	Recall $Z_P(f)$ and $Z'_P(f)$ from~\eqref{eq:Z-def}.
	As we shall see, the slightly modified $(Z_P(f),Z'_P(f) \cup N_{2d}(Z'_P(f)))_P$ is always a breakup as long as $\Int(\Lambda)^c$ is in the $P_0$-pattern. The main difficulty is therefore to construct a breakup that is seen from a given set. For this, we require the following lemma which allows to ``close holes''.

	\begin{lemma}[{\cite[Lemma~4.1]{peledspinka2018colorings}}]\label{lem:closing-holes}
		Let $V,W \subset \Z^d$ and let $B$ be the union of connected components of $W$ that are either infinite or disconnect some vertex in $V$ from infinity. Let $A$ be a connected component of $B^c$. Then $\extB A$ is contained in a connected component of $(W^c)^+$.
	\end{lemma}

The next lemma shows that an atlas can be ``localized'' into an atlas which is seen from $V$.

\begin{lemma}\label{lem:existence-of-adapted-atlas}
	Let $\Lambda$ be a domain, let $V \subset \Lambda$, let $P_0$ be a dominant pattern and let $Z$ be an atlas such that $\Lambda^c \subset Z_{P_0}$. Then there exists an atlas $X$ which is seen from $V$ and satisfies that
	\begin{equation}\label{eq:adapted}
	X_*^{+5} \cap X_P = X_*^{+5} \cap Z_P, \quad X_*^{+5} \cap X'_P = X_*^{+5} \cap Z'_P \qquad\text{for any dominant pattern }P.
	\end{equation}
	Moreover, $\Lambda^c \subset X_{P_0}$ and $X_*^{+5}$ is the union of connected components of $Z_*^{+5}$ that are either infinite or disconnect some vertex in $V$ from infinity.
\end{lemma}
\begin{proof}
	Let $B$ be the union of connected components of $Z_*^{+5}$ that are either infinite or disconnect some vertex in $V$ from infinity. Let $\cA$ be the set of connected components of $B^c$. We claim that
	\[ \text{for every $A \in \cA$, there exists a unique dominant pattern $P_A$ such that $A^{+5} \setminus A \subset Z_{P_A} \setminus Z_*$} .\]
	Indeed, it follows from the definition of $Z_*$ that for every $a \in A^{+5} \setminus A \subset Z_*^c$, there exists a unique dominant pattern $P_a$ such that $a \in Z_{P_a}$. Since \cref{lem:closing-holes} applied with $W:=Z_*^{+5}$ yields that $\extB A$ is contained in a connected component of $(W^c)^+ \subset (Z_*^{+4})^c$, we see that $P_a=P_{a'}$ for all $a,a' \in \extB A$.
	The claim follows. Note also that, since $\Lambda^c \subset Z_{P_0}$, we have $P_A=P_0$ for all $A \in \cA$ such that $A \not\subset \Lambda$.

	We now define $X=(X_P, X'_P)_P$ by
	\[ X_P := (Z_P \cap B) \cup \bigcup \{A \in \cA : P_A=P \} \qquad\text{and}\qquad X'_P := Z'_P \cap B, \qquad P \in \phasedom .\]
	Let us show that $X$ satisfies the conclusion of the lemma.
	Note first that $X_P \cap B = Z_P \cap B$ and $X_* \subset B$, so that $X_* = Z_* \cap B$ and $X_*^{+5} = B$. It easily follows that $X$ is an atlas satisfying~\eqref{eq:adapted}.
	Let us check that $X$ is seen from $V$. Indeed, every finite connected component of $X_*^{+5}=B$ is by definition a connected component of $Z_*^{+5}$ that disconnects some vertex in $V$ from infinity. Finally, $\Lambda^c \subset X_{P_0}$, since $\Lambda^c \subset Z_{P_0}$ and $P_A=P_0$ for all $A \in \cA$ such that $A \not\subset \Lambda$.
\end{proof}

\begin{proof}[Proof of \cref{lem:existence-of-breakup}]
	Let $f \colon \Z^d \to \SS$ and $V \subset \Lambda$ be as in the lemma.
	Recall the definition of $Z_P(f)$ and $Z'_P(f)$ from~\eqref{eq:Z-def}.
	Define $\tilde{Z}_P := Z_P(f)$ and $\tilde{Z}'_P := Z'_P(f) \cup N_{2d}(Z'_P(f))$.
	It is straightforward to check that $\tilde{Z}=(\tilde{Z}_P,\tilde{Z}'_P)_P$ is an atlas and that, for any $P$-odd vertex $v$, we have $v \in \tilde{Z}_P$ if and only if $N(v)$ is in the $P$-pattern and, for any $P$-even vertex $v$, we have $v \in \tilde{Z}'_P$ if and only if $N(v) \cap \tilde{Z}_P$ is not in the $P$-pattern.
	Thus, the lemma follows from \cref{lem:existence-of-adapted-atlas}.
\end{proof}

\subsection{No infinite breakups}\label{sec:no-infinite-breakups}

Here we prove \cref{lem:no-infinite-breakups}.
As mentioned above, our main argument (namely, \cref{prop:prob-of-given-breakup} and \cref{prop:prob-of-odd-approx}) is concerned only with finite breakups. However, it is easy to rule out the existence of an infinite breakup in a random configuration. In doing so, there are two possibilities to have in mind: either there exists an infinite component of $Z_*^{+5}$ or infinitely many finite components surrounding a vertex.

\begin{proof}[Proof of \cref{lem:no-infinite-breakups}]
	To give some intuition, let us assume for a moment that we are working under the explicit condition~\eqref{eq:parameter-inequalities-simple}, rather than the abstract \cref{main-cond}. By~\eqref{eq:prob_outside_Lambda_def1}, \eqref{eq:prob_outside_Lambda_def2} and the definitions of $\rho_{\text{pat}}^{\text{bdry}}$ and $\alpha_0$, for any $u \notin \Lambda^+$ and dominant pattern $P \neq P_0$ for which $u$ is $P$-even,
	\[ \Pr\big(u\text{ is in the $P$-pattern} \mid (f(v))_{v \neq u}\big) \le \rho_{\text{pat}}^{\text{bdry}} \le e^{-\alpha_0} .\]
	Say that $u$ is in a \emph{double pattern} if $u$ is $P$-odd and $N(u)$ is in the $P$-pattern for some dominant pattern $P \neq P_0$. Then
	\[ \Pr\big(u\text{ is in a double pattern} \mid (f(v))_{v \notin N(u)}\big) \le |\phasedom| \cdot e^{-2\alpha_0 d} \le 2^q e^{-2\alpha_0 d} .\]
	When \cref{main-cond} is assumed instead, \eqref{eq:prob_outside_Lambda_def1}, \eqref{eq:prob_outside_Lambda_def2} and~\eqref{eq:cond-non-dominant} yield that
	\[ \Pr\big(u\text{ is in a double pattern} \mid (f(v))_{v \notin N(u)}\big) \le e^{2\gamma d-\alpha d} .\]
	Note that if a vertex $u \in \Z^d \setminus \Lambda^{+5}$ belongs to $Z_*$, then some vertex in $u^{++}$ is in a double pattern.

		We wish to show that, almost surely, every breakup seen from $V$ is finite.
		For $v \in \Z^d$, let $E_v$ be the event that $v$ is in an infinite connected component of $Z_*^{+5}$. Let $E'_v$ be the event that $v$ is disconnected from infinity by infinitely many connected components of $Z_*^{+5}$. It suffices to show that $\Pr(E_v)=\Pr(E'_v)=0$ for any $v \in \Z^d$. Let us show that $\Pr(E'_v)=0$; the proof that $\Pr(E_v)=0$ is very similar.
		On the event $E'_v$, for any $m$, there exists a set $B \subset \Z^d \setminus \Lambda^{+5}$ of size at least $m$ such that $B^{+5}$ is connected and disconnects $v$ from infinity and such that for every vertex $u \in B$ there exists a vertex in $u^{++}$ which is in a double pattern. In particular, for any $m$, there exists a path $\gamma$ in $(\Z^d)^{\otimes 50}$ of some length $n \ge m$ such that $\{ \gamma_i^+ \}_{i=0}^n$ are pairwise disjoint, $\dist(v,\gamma_0) \le Cn$ and all vertices $\{ \gamma_i \}_{i=0}^n$ are in a double pattern. Since $\Pr(\gamma) \le e^{(2\gamma d-\alpha d)n}$ for any such fixed $\gamma$, and since the number of simple paths $\gamma$ in $(\Z^d)^{\otimes 50}$ of length $n$ with $\dist(v,\gamma_0) \le Cn$ is at most $d^{Cn}$, the lemma follows using that $c\alpha d \ge \gamma d+\log d$ by~\eqref{eq:alpha-cond}.
\end{proof}

\subsection{Which edges are restricted?}
\label{sec:restricted-edges}

In this section, we discuss several scenarios in which a directed edge $(v,u)$ is restricted (this notion was defined in \cref{sec:Shearer_overview}). Some conditions do not involve $u$ and thus imply that all edges in $\dpartial v$ are restricted in which case we say that $\dpartial v$ is restricted.

When $v$ is a non-dominant vertex, then $\dpartial v$ is restricted.
When $v$ is a dominant vertex, the question of whether or not $(v,u)$ is restricted depends on the ambient set of configurations $\Omega$.
In some cases, we have certain information about $g(v)$ and/or $g(u)$ for all $g \in \Omega$, which can be used to deduce that $(v,u)$ is restricted. However, the definition of restricted edge allows to deduce that $(v,u)$ is restricted even when we have such information only for $g$ in a certain subset of $\Omega$, namely, the set
\[ \Omega_{f,v} := \{ g \in \Omega : g(N(v)) \simeq_R f(N(v)) \}.\]
In other words, we are allowed to first examine the $R$-closure of the set $f(N(v))$, and only then decide whether $(v,u)$ is restricted based on this information.
As we require this flexibility in some cases, we formulate all conditions below with $\Omega_{f,v}$. The reader may also wish to consider the weaker conditions in which $\Omega_{f,v}$ is replaced by $\Omega$.

Denote $D := R(f(N(v)))$ and let $A=\{ g(u) : g \in \Omega_{f,v} \}$ and $B=\{ g(v) : g \in \Omega_{f,v} \} \cap D$ be as in~\eqref{eq:restricted_A_def} and~\eqref{eq:restricted_B_def}. Note that $A \subset R(D)$ and $B \subset D$.
In particular, $D \subset R(A)$ and $R(D) \subset R(B)$.
Note that $(v,u)$ is restricted if and only if $v$ is non-dominant or $D \neq R(A)$ or $R(D) \neq R(B)$.

\smallskip
\noindent{\bf Scenario 1.}
Let $P$ be a dominant pattern.
\begin{equation}\label{eq:restricted-edge-cond-f-P-inner-g-P-bdry}
\dpartial v\text{ is restricted if}\qquad f(N(v)) \not\simeq_R P_{\inner} \quad\text{and}\quad g(v) \in P_\bdry \quad\text{for all }g\in\Omega_{f,v}.
\end{equation}
Indeed, the condition implies that $R(D) \neq P_\inner$ and $B \subset P_\bdry$. Thus, $B \subset D \cap P_\bdry$ and $D \neq P_\bdry$.
If $D \not\supset P_\bdry$ then $R(B) \supsetneq P_\inner$ and so $R(B)$ can not be the side of a dominant pattern by~\eqref{eq:P_bdry_P_inner}.
Otherwise, $D \supsetneq P_\bdry$ so that $R(A) \supsetneq P_\bdry$ and $R(B) \supset P_\inner$. It follows from~\eqref{eq:P_bdry_P_inner} that $(R(A),R(B)) \notin \phasedom$ (since $|R(A)|>|P_\bdry|$ and $|R(B)| \ge |P_\inner|$).

\smallskip
\noindent{\bf Scenario 2.}
Let $P$ and $Q$ be distinct direct-equivalent dominant patterns.
\begin{equation}\label{eq:restricted-edge-cond-g-PQ-bdry-v}
\dpartial v\text{ is restricted if}\qquad g(v) \in P_\bdry \cap Q_\bdry \quad\text{for all }g\in\Omega_{f,v}.
\end{equation}
Indeed, the condition implies that $B \subset P_\bdry \cap Q_\bdry$, so that $R(B) \supset P_\inner \cup Q_\inner \supsetneq P_\inner$. Thus, $|R(B)|>|P_\inner| \ge |P_\bdry|$ by~\eqref{eq:P_bdry_P_inner}, so that $R(B)$ is not a side of a dominant pattern and hence does not equal $R(D)$.

\medskip

The above conditions do not involve $u$ and thus imply that all edges in $\dpartial v$ are restricted. The next two conditions take into account information about the value at $u$ and thus apply to a particular edge $(v,u)$. These conditions may be seen as counterparts of the previous two conditions.

\smallskip
\noindent{\bf Scenario 3.}
Let $P$ be a dominant pattern.
\begin{equation}\label{eq:restricted-edge-cond-f-P-bdry-g-P-bdry}
(v,u)\text{ is restricted if}\qquad f(N(v)) \not\simeq_R P_{\bdry} \quad\text{and}\quad g(u) \in P_\bdry \quad\text{for all }g\in\Omega_{f,v}.
\end{equation}
Indeed, the condition implies that $A \subset P_\bdry$ and $R(D) \neq P_\bdry$. In particular, $R(A) \supset D \cup P_\inner$ and $D \neq P_\inner$. If $D \not\subset P_\inner$ then $R(A) \supsetneq P_\inner$ and so $R(A)$ can not be the side of a dominant pattern by~\eqref{eq:P_bdry_P_inner}.
Otherwise, $D \subsetneq P_\inner$ so that $R(A) \supset P_\inner$ and $R(B) \supsetneq P_\bdry$. It follows from~\eqref{eq:P_bdry_P_inner} that $(R(A),R(B)) \notin \phasedom$.

\smallskip
\noindent{\bf Scenario 4.}
Let $P$ and $Q$ be distinct direct-equivalent dominant patterns and let $T$ be any dominant pattern (perhaps $P$ or $Q$).
\begin{equation}\label{eq:restricted-edge-cond-f-T-inner-g-PQ-inner}
(v,u)\text{ is restricted if}\qquad f(N(v)) \simeq_R T_{\inner} \quad\text{and}\quad g(u) \in P_\inner \cap Q_\inner \quad\text{for all }g\in\Omega_{f,v}.
\end{equation}
Indeed, the condition implies that $A \subset P_\inner \cap Q_\inner$ and $R(D)=T_\inner$. Thus, $R(A) \supset T_\bdry \cup P_\bdry \cup Q_\bdry \supsetneq T_\bdry$ and $R(B) \supset T_\inner$. It follows from~\eqref{eq:P_bdry_P_inner} that $(R(A),R(B)) \notin \phasedom$.

\subsection{The probability of a given breakup}
\label{sec:prob-of-given-breakup}

In this section, we prove \cref{prop:prob-of-given-breakup}.
Recall that $\breakups_{L,M,N}(X,H)$ represents a set of atlases which encompasses the information of $(X_P)$ precisely, but only a certain approximation of $X_\hole$ via $H$.
We mention that even if $X_\hole$ was specified, the $X'_P$ themselves would still be left unspecified (so that~\eqref{eq:given_breakup_bound} would still be slightly weaker).
Nonetheless, the proof of \cref{prop:prob-of-given-breakup} is no more complicated than would be a proof of~\eqref{eq:given_breakup_bound}.
In particular, if $X_\hole$ was specified, the only place where the distinction between $X$ and $\hat{X}$ would come into play is in the last paragraph in the proof of \cref{lem:lower-bound-on-size-of-restricted} below. In general, there is much similarity between $\hat{X}$ and $X$. Specifically, we note that if $\hat{X} \in \breakups_{L,M,N}(X,H)$ then
\[ (\hat{X}_P)_P=(X_P)_P, \quad \hat{X}_\bad=X_\bad, \quad \hat{X}_\overlap=X_\overlap, \quad \hat{X}_* \cup H = X_* \cup H. \]
Thus, for the most part, we do not need to worry about the difference between $\hat{X}$ and $X$, except when discussing properties related to $(X'_P)_P$ and $X_\hole$.

Fix $X \in \breakups_{L,M,N}$ and $H \subset \Z^d$ such that $|H| \le M\sqrt{d}$, and let $\Omega$ be the set of configurations $f$ having some breakup $\hat{X} \in \breakups_{L,M,N}(X,H)$.
In order to bound the probability of $\Omega$, we aim to apply \cref{lem:bound-on-pseudo-breakup} with
\[ S:=X_* \cup H \qquad\text{and}\qquad S_P := X_P \setminus (X_* \cup H) .\]
The definition of $X_*$ implies that $\{ S_P^+ \}_P$ are pairwise disjoint so that, in particular, $\{ S_P \}_P$ is a partition of $S^c$.
By~\eqref{eq:breakup-0}, $S \cup S_{P_0}$ contains $(\Lambda^c)^+$.
By~\eqref{eq:H-approx-def}, \eqref{eq:breakup-prop-even}, \eqref{eq:breakup-prop-odd} and~\eqref{eq:def-atlas}, $S_P^+ \cap S^{+2}$ is in the $P$-pattern on the event $\Omega$.
Thus, the assumptions of \cref{lem:bound-on-pseudo-breakup} are satisfied.

The following lemma guarantees that there are many restricted edges in $(f,\Omega)$.
Recall the definitions of $S^f_\unbal$, $S^{\Omega,f}_\rest$, $S^{\Omega,f}_\highlyrest$ and $S^\Omega_\unique$ from \cref{sec:Shearer_overview}.

\begin{lemma}\label{lem:lower-bound-on-size-of-restricted}
	For any $f \in \Omega$, we have
	\[ \big|S^{\Omega,f}_\rest\big| \ge \tfrac12 (L + dM) .\]
\end{lemma}
\begin{proof}
	Fix $f \in \Omega$ and write $S_\rest$ for $S^{\Omega,f}_\rest$. In particular, the notion of restricted edge is with respect to $\Omega$ and~$f$.
	It suffices to show that $|S_\rest| \ge L$ and $|S_\rest| \ge dM$.

	To show that $|S_\rest| \ge L$, it suffices to show that
\begin{equation}\label{eq:partial_X_P_restricted}
\dpartialrev X_P \subset S_\rest \qquad\text{for any }P .
\end{equation}
To this end, let $(v,u) \in \dpartialrev X_P$. Then $g(u) \in P_\bdry$ and $g(N(v)) \not\subset P_\bdry$ for any $g \in \Omega$ by~\eqref{eq:breakup-prop-bdry-X_P}, from which it follows that $(v,u)$ is restricted by~\eqref{eq:restricted-edge-cond-f-P-bdry-g-P-bdry}.

	It remains to show that $|S_\rest| \ge dM$.
	Let $S^2_\rest$ denote the set of vertices, all of whose incident edges are restricted in one or the other direction (that is, $v\in S^2_\rest$ if for every $u\sim v$ either $(u,v)$ or $(v,u)$ is in $S_\rest$). Note that $|S_\rest| \ge d|S^2_\rest|$ so that it suffices to show that $|S^2_\rest| \ge M$. This will follow once we show that $X_\overlap \subset S^2_\rest$ and that $\hat{X}_\hole \subset S^2_\rest$ for some $\hat{X} \in\breakups_{L,M,N}(X,H)$.

	Let us first show that $X_\overlap \subset S^2_\rest$, i.e., that $X_P \cap X_Q \subset S^2_\rest$ for any $P \neq Q$. Since $X_P$ is $P$-even, it suffices to show that
	\begin{equation}\label{eq:X_P_X_Q_restricted}
	\dpartial (\Even_P \cap X_P \cap X_Q) \subset S_\rest \qquad\text{for any }P \neq Q .
	\end{equation}
	Towards showing this, let $v \in X_P \cap X_Q$ be $P$-even.
	If $v$ is also $Q$-even then $g(v) \in P_\bdry \cap Q_\bdry$ for any $g \in \Omega$ by~\eqref{eq:breakup-prop-even}, and it follows from~\eqref{eq:restricted-edge-cond-g-PQ-bdry-v} that all edges in $\dpartial v$ are restricted.
Otherwise, $v$ is $Q$-odd so that $v^+ \subset X_Q$ since $X_Q$ is $Q$-odd. Thus, $g(v) \in P_\bdry$ and $g(N(v)) \subset Q_\bdry$ for any $g \in \Omega$ by~\eqref{eq:breakup-prop-even}.
	In particular, $g(N(v)) \not\simeq_R P_\inner$ for any $g \in \Omega$ (note that $P_\inner \not\subset Q_\bdry$).
	It follows from~\eqref{eq:restricted-edge-cond-f-P-inner-g-P-bdry} that all edges in $\dpartial v$ are restricted. We remark that (in either case) the edges in $\dpartialrev v$ are also restricted, but we do not need this.

	Let us now show that $\hat{X}_\hole \subset S^2_\rest$ for some $\hat{X} \in \breakups_{L,M,N}(X,H)$. To this end, let $\hat{X} \in \breakups_{L,M,N}(X,H)$ be a breakup of $f$ (which exists by the definition of $\Omega$) and note that $\hat{X}_\hole \subset S$ by~\eqref{eq:H-approx-def}.
	Since $\hat{X}'_P$ is $P$-even, it suffices to show that
	\begin{equation*}\label{eq:partial_X'_P_restricted}
	\dpartial (\Even_P \cap \hat{X}'_P) \subset S_\rest \qquad\text{for any }P .
	\end{equation*}
	To see this, let $v \in \hat{X}'_P$ be $P$-even. Since $\hat{X}'_P \subset X_P$, we have that $v \in X_P$.
	Then $g(v) \in P_\bdry$ for any $g \in \Omega$ and $f(N(v)) \not\subset P_\inner$ by~\eqref{eq:breakup-prop-bdry-X'_P}, from which it follows that all edges in $\dpartial v$ are restricted by~\eqref{eq:restricted-edge-cond-f-P-inner-g-P-bdry}.
\end{proof}

As explained in \cref{sec:Shearer_overview}, applying \cref{lem:bound-on-pseudo-breakup} directly for $\Omega$ does not produce the bound stated in \cref{prop:prob-of-given-breakup}. This bound will instead follow by applying \cref{lem:bound-on-pseudo-breakup} to subevents of $\Omega$ on which we have additional information about the configuration on the set $X_\bad$ and then summing the resulting bounds.
To explain the reason for this and to motivate the definitions below, we note that, although~\eqref{eq:breakup-prop-bad} prohibits the possibility that the neighborhood $N(v)$ of a $P$-odd vertex $v \in X_\bad$ is in the $P$-pattern, this is possible for a $P$-even vertex. That is, it cannot happen that $f(N(v)) \simeq_R P_\bdry$ for a $P$-odd vertex, but it may happen that $f(N(v)) \simeq_R P_\inner$ for a $P$-even vertex. A vertex for which the latter occurs is problematic as it does not immediately reduce the entropy of the configuration (since it may also have a balanced neighborhood and no or few restricted edges incident to it). However, if many (perhaps even all or almost all) of the vertices in $X_\bad$ are of this type, then by recording a small subset of these vertices and the dominant patterns in their neighborhoods, we may ensure that most vertices in $X_\bad$ become restricted in some manner (unbalanced neighborhood, many incident restricted edges, or highly energetic).
We now describe the structure of this additional information.

For $f \in \Omega$ and a dominant pattern $P$, define
\begin{equation}\label{eq:U_P_def}
U_P(f) := \big\{ u \in X_\bad \setminus S^f_{\unbal} : \text{$u$ is $P$-even, }f(N(u)) \simeq_R P_\inner \big\} .
\end{equation}
Note that the sets $\{ U_P(f) \}_P$ are pairwise disjoint.
Note also that $u \in U_P(f)$ implies that $N(u)$ is in the $P$-pattern and is not in the $Q$-pattern for any $Q \neq P$. On the other hand, it is not necessarily the case that $u$ itself is in the $P$-pattern (though it is for homomorphism models). For this reason, we introduce also
\begin{equation}\label{eq:U'_P_def}
U'_P(f) := \big\{ u \in U_P(f) : f(u) \notin P_\bdry \big\} .
\end{equation}

The collection $(U_P(f),U'_P(f))_P$ contains the relevant information on $f$ beyond that which is given by $\Omega$. However, it contains more information than is necessary and this comes at a large enumeration cost. Instead, we wish to specify only a certain approximation of this information.
Given a collection $V=(V_P,V'_P)_P$ of subsets of $\Z^d$, let $\Omega(V)$ denote the set of $f \in \Omega$ satisfying that, for every dominant pattern $P$,
\begin{align}
V_P &\subset U_P(f) &\text{and}&& N_{\epsilon d}\Bigg(\bigcup_{Q \neq P} U_Q(f)\Bigg) &\subset N\Bigg(\bigcup_{Q \neq P} V_Q\Bigg), & \label{eq:iso-approx} \\
V'_P &\subset U'_P(f) &\text{and}&& N_{\epsilon d}\Big(U'_P(f)\Big) &\subset N(V'_P). & \label{eq:iso-bad-approx}
\end{align}
Thus, $V$ is a kind of approximation of $(U_P(f),U'_P(f))_P$.
With this definition at hand, there are now two goals. The first is to show that the additional information given by $V$ is enough to improve the bound given in \cref{lem:lower-bound-on-size-of-restricted}.
The second is to show that the cost of enumerating $V$ is not too large.

\begin{lemma}\label{lem:lower-bound-on-size-of-restricted-with-V}
	For any $V$ and any $f \in \Omega(V)$, we have
	\[ \big|S^f_\unbal\big| + \tfrac{1}{d} \big|S^{\Omega(V),f}_\rest\big| + \epsilon d\big|S^{\Omega(V),f}_\highlyrest\big| \ge \tfrac{L}{4d} + \tfrac{M}{4} + \tfrac{\epsilon N}{8} .\]
\end{lemma}
\begin{proof}
	We fix $V$ and $f \in \Omega(V)$ and suppress these in the notation of $S^f_\unbal$, $S^{\Omega(V),f}_\rest$, $S^{\Omega(V),f}_\highlyrest$, $U_P(f)$ and $U'_P(f)$.
	It suffices to show that
	\[ |S_\unbal| + \tfrac{1}{d} |S_\rest| + \epsilon d |S_\highlyrest| \ge \tfrac {\epsilon N}4 ,\]
	as the lemma then follows by averaging this bound with the one given by \cref{lem:lower-bound-on-size-of-restricted}. In fact, we will show the slightly stronger inequality
	\[ N \le (1+\tfrac 2\epsilon)(|S_\unbal|+\tfrac1{2d}|S_\rest|) + \tfrac 2{\epsilon d} |S_\rest| + 2d |S_\highlyrest| .\]
	Let $S_\nondom$ denote the set of vertices in $S$ that are non-dominant in $f$.
	Let $S^\epsilon_\rest$ denote the set of vertices which are incident to at least $\epsilon d$ edges in $S_\rest$. Note that $|S_\rest| \ge 2d |S_\nondom|$ and $|S_\rest| \ge \tfrac12 \epsilon d|S^\epsilon_\rest|$, so that it suffices to show that
	\begin{equation}\label{eq:X_bad_restricted}
	X_\bad \subset S_\nondom \cup S_\unbal \cup S^\epsilon_\rest \cup N_{\epsilon d}(S_\nondom \cup S_\unbal) \cup N(S_\highlyrest) .
	\end{equation}

	Let us first show that
	\begin{equation}\label{eq:X_bad_restricted1}
	X_\bad \setminus U \subset S_\nondom \cup S_\unbal, \qquad\text{where }U := \bigcup_P U_P .
	\end{equation}
	To this end, let $u \in X_\bad \setminus (U \cup S_\unbal)$. Suppose first that $f(N(u)) \simeq_R A$ for some $P=(A,B)$. By replacing $P$ with $(B,A)$ if needed, we may assume that $u$ is $P$-odd and $f(N(u)) \simeq_R P_\bdry$ or that $u$ is $P$-even and $f(N(u)) \simeq_R P_\inner$. However, the former case is impossible by~\eqref{eq:breakup-prop-bad} and the latter case is impossible by~\eqref{eq:U_P_def} as $u \notin U_P$. Thus, $f(N(u)) \not\simeq_R A$ for all $(A,B) \in \phasedom$, so that $u \in S_\nondom$. This establishes~\eqref{eq:X_bad_restricted1}.

	Next, we show that
	\begin{equation}\label{eq:X_bad_restricted2}
	\bigcap_P N_{\epsilon d}(U \setminus U_P) \subset S^\epsilon_\rest .
	\end{equation}
	To see this, let $u \in \bigcap_P N_{\epsilon d}(U \setminus U_P)$ and note that, by~\eqref{eq:iso-approx}, $u \in N(V_P)$ for some $P$. Since $u \in N_{\epsilon d}(U \setminus U_P)$, another application of~\eqref{eq:iso-approx} yields that $u \in N(V_Q)$ for some $Q \neq P$.
	Since $V_P \subset U_P$ and $V_Q \subset U_Q$ by~\eqref{eq:iso-approx}, it follows from~\eqref{eq:U_P_def} that $u$ is in both the $P$-pattern and the $Q$-pattern, so that $g(u) \in P_\inner \cap Q_\inner$ for all $g \in \Omega(V)$. Since $u \in N_{\epsilon d}(U)$, in order to show that $u \in S^\epsilon_\rest$, it suffices to show that if $v \in N(u) \cap U_T$ for some $T$, then the directed edge $(v,u)$ is restricted. Indeed, this follows since $f(N(v)) \simeq_R T_\inner$ by~\eqref{eq:U_P_def}, which implies that $(v,u)$ is restricted by~\eqref{eq:restricted-edge-cond-f-T-inner-g-PQ-inner}.
	This establishes~\eqref{eq:X_bad_restricted2}.

	Next, we show that
	\begin{equation}\label{eq:X_bad_restricted3}
	\bigcup_P N_{\epsilon d}(U'_P) \subset N(S_\highlyrest) .
	\end{equation}
	Indeed, if $v \in N_{\epsilon d}(U'_P)$ then, by~\eqref{eq:iso-bad-approx} and~\eqref{eq:U'_P_def}, there exists $u \in N(v) \setminus S_\unbal$ such that $g(N(u)) \simeq_R P_\inner$ and $g(u) \notin P_\bdry$ for all $g \in \Omega(V)$. Thus, by the definition of highly energetic vertex, $u \in S_\highlyrest$.
	This establishes~\eqref{eq:X_bad_restricted3}.

	Finally, towards showing~\eqref{eq:X_bad_restricted}, let $u \in X_\bad$ and assume that
	\[ u \notin S_\nondom \cup S^\epsilon_\rest \cup N_{\epsilon d}(S_\nondom \cup S_\unbal) \cup N(S_\highlyrest) .\]
	We must show that $u \in S_\unbal$. By~\eqref{eq:X_bad_restricted1}, $u \notin N_{\epsilon d}(X_\bad \setminus U)$ so that $u \in N_{2d-\lceil\epsilon d \rceil+1}(\bigcup_P X_P \cup U)$.
	Since $X_\bad \cap N_{\epsilon d}(\bigcup_P X_P) \subset S^\epsilon_\rest$ by~\eqref{eq:partial_X_P_restricted}, it follows that $u \in N_{2d-2\lceil\epsilon d \rceil+2}(U)$.
	Hence, by~\eqref{eq:X_bad_restricted2}, we have that $u \in N_{2d-3\lceil\epsilon d \rceil+3}(U_P)$ for some~$P$. Using~\eqref{eq:X_bad_restricted3}, we conclude that $u \in N_{2d-4\lceil\epsilon d \rceil+4}(U_P \setminus U'_P)$. In particular, $|N(u) \cap f^{-1}(P_\bdry)| \ge 2d-4(\lceil\epsilon d \rceil-1)$ by~\eqref{eq:U'_P_def}. Since $f(N(u)) \not\subset P_\bdry$ by~\eqref{eq:breakup-prop-bad} (note that $u$ is $P$-odd as it is adjacent to $U_P$) and $\lceil \epsilon d \rceil - 1 < \epsilon d$, it follows that $u \in S_\unbal$.
\end{proof}

\begin{lemma}\label{lem:family-of-strong-odd-approx}
	There exists a family $\cV$ satisfying that
	\[ |\cV| \le \exp\left(\tfrac{CN(\fq+\log d)\log d}{\epsilon d}\right) \qquad\text{and}\qquad \Omega \subset \bigcup_{V \in \cV} \Omega(V) .\]
\end{lemma}
\begin{proof}
	Let $\cV$ be the collection of all $(V_P,V'_P)_P$ such that $\{V_P\}_P$ are disjoint subsets of $X_\bad$ having $\sum_P |V_P| \le 3rN$, where $r := (1+\log 2d)/\epsilon d$, and similarly, $\{V'_P\}_P$ are disjoint subsets of $X_\bad$ having $\sum_P |V'_P| \le 3rN$. Let us check that $\cV$ satisfies the requirements of the lemma. Since $|X_\bad|=N$, we have
	\[ |\cV| \le \binom{N}{\le 3rN}^2 \cdot |\phasedom|^{6rN} \le \left(\frac{e|\phasedom|}{3r}\right)^{6rN} \le e^{CN(\fq+\log d)(\log d)/\epsilon d} .\]

	Fix $f \in \Omega$. We must find a collection $(V_P,V'_P)_P \in \cV$ for which~\eqref{eq:iso-approx} and~\eqref{eq:iso-bad-approx} hold. We begin with~\eqref{eq:iso-bad-approx}. For each $P$, by \cref{lem:existence-of-covering2}, we may find a set $V'_P \subset U'_P(f) \subset X_\bad$ such that $N_{\epsilon d}(U'_P(f)) \subset N(V'_P)$ and $|V'_P| \le r|U'_P(f)|$. Note also that $\sum_P |U'_P(f)| \le |X_\bad|=N$ so that $\sum_P |V'_P| \le rN$.

	We now construct the sets $\{ V_P \}_P$. We write $U_P$ for $U_P(f)$, and we denote $U_I := \bigcup_{P \in I} U_P$ for $I \subset \phasedom$ and $U := U_{\phasedom}$.
	Define a bipartite graph $G$ with vertex set $(\Z^d \times \{0,1\}) \cup U$ as follows. For each $v \in \Z^d$, let $I_v$ be a minimal set of dominant patterns for which $|N(v) \cap U_{I_v}| \ge \tfrac13 |N(v) \cap U|$, and place an edge between $(v,i) \in \Z^d \times \{0,1\}$ and $u \in U$ if and only if $v \sim u$ and $\1(u \in U_{I_v})=i$.
	Note that $G$ has maximum degree at most $2d$.

	By \cref{lem:existence-of-covering2} applied to $G$ with $t=\epsilon d/3$, we obtain a set $W \subset U$ of size $|W| \le 3rN$ such that
	\[ v \in N_{\epsilon d/3}(U_I) \implies v \in N(W \cap U_I) \qquad\text{for any }v \in \Z^d\text{ and }I \in \{I_v, \phasedom \setminus I_v \}.\]
	Set $V_P := W \cap U_P$ for all $P$ and note that $W = \bigcup_P V_P$. Towards showing~\eqref{eq:iso-approx}, let $P \in \phasedom$ and $v \in N_{\epsilon d}(U \setminus U_P)$.
	Suppose first that $P \notin I_v$.
	Then
	\[ v \in N_{\epsilon d}(U) \subset N_{\epsilon d/3}(U_{I_v}) \subset N(W \cap U_{I_v}) \subset N(W \setminus U_P) = N(W \setminus V_P) .\]
    Suppose next that $P \in I_v$. By the minimality of $I_v$, either $I_v=\{P\}$ or $|N(v) \cap U_{I_v}| < \tfrac23 |N(v) \cap U|$. In either case, we have $|N(v) \cap U_{\phasedom \setminus I_v}| \ge \epsilon d/3$ so that
    \[ v \in N_{\epsilon d/3}(U_{\phasedom \setminus I_v}) \subset N(W \cap U_{\phasedom \setminus I_v}) \subset N(W \setminus U_P) = N(W \setminus V_P) . \qedhere \]
\end{proof}

\begin{lemma}\label{lem:vertices-with-unique-pattern}
	$S \setminus X_\bad \subset S^\Omega_\unique$.
\end{lemma}
\begin{proof}
	Let $v \in S \setminus X_\bad$ and note that there exists $P$ such that $v \in X_P$.
	Assume first that $v$ is $P$-even. Then, by~\eqref{eq:breakup-prop-even}, $g(v) \in P_\bdry$ for all $g \in \Omega$, so that if $g(N(v)) \not\simeq_R P_\inner$ then all edges in $\dpartial v$ are restricted in $g$ by~\eqref{eq:restricted-edge-cond-f-P-inner-g-P-bdry}. Hence, $v$ has a unique pattern.
	Assume next that $v$ is $P$-odd. Then, since $X_P$ is $P$-even, $v^+ \subset X_P$ so that, by~\eqref{eq:breakup-prop-even}, $g(N(v)) \subset P_\bdry$ for all $g \in \Omega$. Thus, either $g(N(v)) \simeq_R P_\bdry$ or all edges in $\dpartial v$ are restricted by~\eqref{eq:restricted-edge-cond-f-P-bdry-g-P-bdry}. In particular, $v$ has a unique pattern.
\end{proof}

\begin{proof}[Proof of \cref{prop:prob-of-given-breakup}]\label{proof:prob-of-given-breakup}
Note that $|X_*| \le 2L+M+N$ so that $|S| \le 2L+2M\sqrt{d}+N$.
Thus, \cref{lem:bound-on-pseudo-breakup}, \cref{lem:lower-bound-on-size-of-restricted-with-V} and \cref{lem:vertices-with-unique-pattern} imply that, for any $V$,
	\[ \Pr(\Omega(V)) \le \exp\Big(-\tfrac{\alpha}{32} \left(\tfrac{L}{4d}+\tfrac{M}{4}+ \tfrac{\epsilon N}{8}\right)+ \tfrac{\fq N}{d}+\bar\gamma (2L+2M\sqrt{d}+N)\Big) ,\]
	where $\bar\gamma$ was defined in~\eqref{eq:gamma-bar-def}.
Therefore, by \cref{lem:family-of-strong-odd-approx},
\[ \Pr(\Omega) \le \exp\left(\tfrac{CN(\fq+\log d)\log d}{\epsilon d} + \bar\gamma (2L+2M\sqrt{d}+N) - c\alpha \big(\tfrac{L}{d}+M+\epsilon N\big) \right) .\]
Finally, using~\eqref{eq:alpha-cond}, noting also that it implies that $c\alpha \ge e^{-\alpha d/25}d$, we obtain that
\[ \Pr(\Omega) \le \exp\left(- c\alpha \big( \tfrac{L}{d}+M+\epsilon N \big)\right) . \qedhere \]
\end{proof}

\subsection{The probability of an approximated breakup}
\label{sec:prob-of-approx}

In this section, we prove \cref{prop:prob-of-odd-approx}.
Fix integers $L,M,N \ge 0$ and an approximation $A$.
Denote
\[ A_\bad := \bigcap_P (A_P \cup A^{**})^c ,\quad A_\overlap := \bigcup_{P \neq Q} (A_P \cap A_Q), \quad U := A^{**} \cup A_\bad \cup A_\overlap \cup A_\hole .\]
Further define
\[  S_P := \Int(A_P \setminus U) \qquad\text{and}\qquad S := \bigcap_P (S_P)^c .\]
Note that $U^+ \subset S$, that $\{S_P\}_P$ is a partition of $S^c$ and that $\{ S_P^+ \}_P$ are pairwise disjoint.
Let $X$ be an atlas which is approximated by $A$. Note that, by~\ref{it:approx-A_P} and~\ref{it:approx-hole},
\[ A_\bad \subset X_\bad, \quad A_\overlap \subset X_\overlap, \quad A_\hole \subset X_\hole, \quad U = A^{**} \cup X_\bad \cup X_\overlap \cup X_\hole .\]
\begin{lemma}\label{lem:S}
\[ S = X_* \cup (A^{**})^+ .\]
\end{lemma}
\begin{proof}
	Let us first show that $S \subset X_* \cup (A^{**})^+$.
	Let $v \in S$ and note that $v \notin \Int(A_P \setminus U)$ for all $P$. Thus, for any $P$, there exists $u \in v^+$ such that $u \notin A_P$ or $u \in U$. If the latter occurs for some $P$, then $u \in U \subset A^{**} \cup X_*$ and we are done. Otherwise, for every $P$, there exists $u \in v^+$ such that $u \notin A_P$. That is, $u \in \bigcap_P \Int(A_P)^c$. Suppose that $u \notin X_*$ so that $u \in \Int(X_P)$ for some $P$. By~\ref{it:approx-A_P}, $u \in \Int(A_P \cup A^{**})$. Since $u \notin \Int(A_P)$, it must be that $u \in (A^{**})^+$.

	Let us now show that $X_* \cup (A^{**})^+ \subset S$.
	Since $A^{**} \subset U$ and $U^+ \subset S$, we see that $(A^{**})^+ \subset S$. Similarly, $X_\bad \cup X_\overlap \cup X_\hole \subset U \subset S$. It remains to show that $\bigcup_P \intextB X_P \subset S$.
	Let $v \in \intextB X_P$ for some $P$ and suppose towards a contradiction that $v \in S_Q$ for some $Q$. Then~\ref{it:approx-A_P} implies that $v \in \Int(X_Q \setminus X_\overlap)$, which clearly contradicts the fact that $v \in \intextB X_P$.
\end{proof}

Thus, using~\ref{it:approx-unknown-location}, we see that $S \subset X_*^{+4}$.
Then, since $A_P \subset X_P$ and $\intextB S_P \subset \intextB S \cap X_P \setminus X'_P$, \eqref{eq:breakup-prop-even} and~\eqref{eq:breakup-prop-odd} imply that, for any $f$ having $X$ as a breakup,
\begin{equation}\label{eq:prob-of-approx-1}
\Even_P \cap A_P \cap S^+\text{ and }\intextB S_P\text{ are in the $P$-pattern}.
\end{equation}
Finally, by~\eqref{eq:breakup-0}, \ref{it:approx-A_P} and the fact that $X_* \subset S$, we have that $S \cup S_{P_0}$ contains $(\Lambda^c)^+$.
We have thus established that the assumptions of \cref{lem:bound-on-pseudo-breakup} are satisfied for the event $\Omega$ that $A$ approximates some breakup in $\breakups_{L,M,N}$.

\begin{lemma}\label{lem:vertices-with-unique-pattern-approx}
Every vertex in $S \setminus U$ has a unique pattern. That is, $S \setminus U \subset S^\Omega_\unique$.
\end{lemma}
\begin{proof}
	The proof is essentially the same as that of \cref{lem:vertices-with-unique-pattern}.
Let $v \in S \setminus U$ and note that there exists $P$ such that $v \in A_P$.
Assume first that $v$ is $P$-even. Then, by~\eqref{eq:prob-of-approx-1}, we have $g(v) \in P_\bdry$ for all $g \in \Omega$. Thus, by~\eqref{eq:restricted-edge-cond-f-P-inner-g-P-bdry}, if $g(N(v)) \not\simeq_R P_\inner$ then all edges in $\dpartial v$ are restricted in $g$. Hence, $v$ has a unique pattern.
Assume next that $v$ is $P$-odd. Then, since $A_P$ is $P$-even, $v^+ \subset A_P$ so that $g(N(v)) \subset P_\bdry$ for all $g \in \Omega$ by~\eqref{eq:prob-of-approx-1}. Thus, either $g(N(v)) \simeq_R P_\bdry$ or all edges in $\dpartial v$ are restricted by~\eqref{eq:restricted-edge-cond-f-P-bdry-g-P-bdry}. In particular, $v$ has a unique pattern.
\end{proof}

Define an equivalence relation $\cong$ on $\breakups_{L,M,N}$ by declaring $\hat{X} \cong X$ if $(\hat{X}_P)_P=(X_P)_P$.
The proof of \cref{prop:prob-of-odd-approx} is based on the idea that one of two situations can occur: either there are enough restricted edges so that one may directly apply \cref{lem:bound-on-pseudo-breakup} to obtain the desired bound, or there are not many possible breakups (up to $\cong$ equivalence) so that one may apply \cref{prop:prob-of-given-breakup} together with a union bound. At the heart of this approach lies the following lemma which informally states that an unknown vertex (of a certain type) either incurs an entropic loss (in the sense that it is adjacent to many restricted edges) or there is a unique way to determine to which $X_P$'s it belongs.
We now make this precise.

For an atlas $X$, let $\Omega_X$ denote the event that there exists a breakup $\hat{X} \in \breakups_{L,M,N}(X,H)$ that is approximated by $A$, where $H:=A_\hole \cup A^{**}$. We note that $\Omega_X$ depends only on the $\cong$ equivalence class of $X$.
With a slight abuse of notation, denote
\[ S^{\Omega,f,1/2}_\rest := \Big\{ v : \big|\dpartial v \cap S^{\Omega,f}_\rest\big| \ge \tfrac{d}{2} \Big\} \qquad\text{and}\qquad S^{\Omega,X,1/2}_\rest := \bigcap_{f \in \Omega_X} S^{\Omega,f,1/2}_\rest .\]

\begin{lemma}\label{lem:breakup-recovery}
Let $X$ be an atlas which is approximated by $A$, let $P$ be a dominant pattern and let $v \in A^*$ be a $P$-odd vertex. Then
	\[ \text{either}\qquad v \in S^{\Omega,X,1/2}_\rest \qquad\text{or}\qquad  v \in X_P \iff v \in N_{d/2}(A_P) .\]
\end{lemma}
\begin{proof}
	Recall that $v^+ \subset (A^{**})^+ \subset S \subset X_*^{+4}$.
	Denote $I := \{ Q \simeq P : v \in X_Q \}$.
	We first show that $|I|>1$ implies that $v \in S^{\Omega,X,1/2}_\rest$.
	Indeed, if $Q,T \in I$ are distinct, then $f(N(v)) \subset Q_\bdry \cap T_\bdry$ for any $f \in \Omega_X$ by~\eqref{eq:def-atlas} and~\eqref{eq:breakup-prop-even}, and it follows that $v$ is a non-dominant vertex in $f$ (so that $\dpartial v$ is restricted in $f$).
	Next, we show that, for any $Q \simeq P$ and $u \in N(v) \cap A_Q$, if $Q \notin I$ then $(v,u) \in S^{\Omega,f}_\rest$ for all $f \in \Omega_X$.
	Indeed, $g(u) \in Q_\bdry$ for all $g \in \Omega$ by~\eqref{eq:prob-of-approx-1}, and $f(N(v)) \not\subset Q_\bdry$ by~\eqref{eq:breakup-prop-bdry-X_P}, so that $(v,u)$ is restricted by~\eqref{eq:restricted-edge-cond-f-P-bdry-g-P-bdry}.

Suppose now that $v \notin S^{\Omega,X,1/2}_\rest$. Note that $v \in N_d(\bigcup_{Q \simeq P} A_Q)$ by~\ref{it:approx-A*_P}. It therefore follows from what we have just shown that $I=\{Q\}=\{T \simeq P : v \in N_{d/2}(A_T)\}$ for some $Q \simeq P$. In particular, $v \in X_P$ if and only if $P=Q$ if and only if $v \in N_{d/2}(A_P)$.
\end{proof}

\begin{proof}[Proof of \cref{prop:prob-of-odd-approx}]\label{proof:prob-of-odd-approx}
Let $c_1$ denote the universal constant from \cref{prop:prob-of-given-breakup}. Let $C_1$ denote a universal constant, which will be determined later. Denote
\[ \tilde{\epsilon}:=\frac{c_1\alpha}{2C_1(\fq+\log d)} .\]
Recall the definition of $\bar\gamma$ from~\eqref{eq:gamma-bar-def}.
Observe that by requiring~\eqref{eq:alpha-cond} to hold with a smaller universal constant if needed, we may ensure that for some universal constant $c$, we have
\begin{equation}\label{eq:alpha-cond2}
c\alpha\tilde\epsilon \ge \frac{\fq \log d}{\sqrt{d}} + \frac{\fq}{\epsilon d} + \bar\gamma d^{3/2} \log d .
\end{equation}

Consider the event
\[ \Omega' := \Omega \cap \Big\{ \big|S^{\Omega,f,1/2}_\rest\big| \ge \tilde{\epsilon} \big(\tfrac{L}{d} + M + \epsilon  N\big) \Big\} .\]
We bound separately the probabilities of $\Omega'$ and $\Omega \setminus \Omega'$.
Let us begin with $\Omega'$.
Note that
\[ \big|S^{\Omega',f}_\rest\big| \ge \big|S^{\Omega,f}_\rest\big| \ge \tfrac14 d\big|S^{\Omega,f,1/2}_\rest\big| \qquad\text{for any }f \in \Omega .\]
By~\ref{it:approx-unknown-size} and \cref{lem:S},
\begin{equation}\label{eq:A**-U-S-sizes}
|A^{**}| \le \tfrac{C(L+dM)\log d}{\sqrt{d}}, \qquad |U| \le M+N+|A^{**}|, \qquad |S| \le 2L+M+N+(2d+1)|A^{**}| .
\end{equation}
Applying \cref{lem:bound-on-pseudo-breakup} together with \cref{lem:vertices-with-unique-pattern-approx} and~\eqref{eq:alpha-cond2} yields that
\[ \Pr(\Omega') \le \exp\Big(\bar\gamma|S| + \tfrac{\fq}{d}|U| -\tfrac{\alpha \tilde{\epsilon}}{32}\big(\tfrac{L}{d} +M+ \epsilon N\big)\Big) \le e^{- c\alpha\tilde{\epsilon} ( \frac{L}{d}+M+\epsilon N)} .\]

We now bound the probability of $\Omega \setminus \Omega'$.
To do this, as explained above, we recover the breakup (up to $\cong$ equivalence) and then apply \cref{prop:prob-of-given-breakup}.
Formally, let $\cB'$ be the collection of atlases $X \in \breakups_{L,M,N}$ which are approximated by $A$ and have
\begin{equation}\label{eq:prob-of-approx-size-of-restricted}
 \big|S^{\Omega,X,1/2}_\rest\big| < \tilde{\epsilon}\big(\tfrac{L}{d}+M+\epsilon N\big).
\end{equation}
Let $\cB$ be a set of $\cong$ representatives of $\cB'$ and note that $\Omega \setminus \Omega' \subset \bigcup_{X \in \cB} \Omega_X$.
We shall show that
\begin{equation}\label{eq:prob-of-approx-size-of-B}
|\cB| \le e^{C_1\tilde{\epsilon} (\frac{L}{d}+M+\epsilon N)(\fq+\log d)}.
\end{equation}
Using \cref{prop:prob-of-given-breakup}, this will then yield by a union bound that
\[ \Pr(\Omega \setminus \Omega') \le \sum_{X\in\cB} \Pr(\Omega_X) \le \exp\Big[\big(C_1\tilde{\epsilon}(\fq+\log d) - c_1\alpha\big) \big( \tfrac{L}{d}+M+\epsilon N \big)\Big] = e^{-\frac12 c_1\alpha(\frac{L}{d}+M+\epsilon N)} . \]
Together with the above bound on $\Pr(\Omega')$, this yields the proposition.

\smallskip
It remains to establish~\eqref{eq:prob-of-approx-size-of-B}. Towards showing this, we first show that the mapping
\begin{equation}\label{eq:prob-of-approx-mapping}
X ~\longmapsto~ \left(S^{\Omega,X,1/2}_\rest,~ (I_X(v))_{v \in S^{\Omega,X,1/2}_\rest}\right)
\end{equation}
is injective on $\cB$, where
\[ I_X(v) := \big\{ P\in \phasedom : v \in \Odd_P \cap X_P \big\} .\]
By~\eqref{eq:def-atlas} and~\ref{it:approx-A_P}, we have
\[ X_P = (\Odd_P \cap X_P)^+ = (\Odd_P \cap (A_P \cup (X_P \cap A^*)))^+ \qquad\text{for all }X \in \cB\text{ and all }P .\]
Thus, to determine $X_P$, we only need to know the set $\Odd_P \cap X_P \cap A^*$. In other words, we only need to know for each vertex $v \in \Odd_P \cap A^*$, whether it belongs to $X_P$ or not. If $v \in S^{\Omega,X,1/2}_\rest$ then this is given by $I_X(v)$, and otherwise, \cref{lem:breakup-recovery} implies that this is determined by the approximation. This establishes the claimed injectivity.

Let $\cR$ be the image of the mapping in~\eqref{eq:prob-of-approx-mapping} as $X$ ranges over $\cB$. As this mapping is injective, we have $|\cB| = |\cR|$.
The bound~\eqref{eq:prob-of-approx-size-of-B} will then easily follow once we show that
\begin{equation}\label{eq:possible-patterns-a-vertex-can-be-in}
|\{I_X(v) : X \in \breakups,~ v\text{ even}\}| \le 2^{\fq} \qquad\text{and}\qquad |\{I_X(v) : X \in \breakups,~ v\text{ odd}\}| \le 2^{\fq}.
\end{equation}
Indeed, \eqref{eq:A**-U-S-sizes}, \eqref{eq:prob-of-approx-size-of-restricted} and~\eqref{eq:possible-patterns-a-vertex-can-be-in} imply that
\[ |\cB| = |\cR| \le \binom{|S|}{\le \tilde{\epsilon} (\frac{L}{d}+M+\epsilon N)} 2^{\fq \tilde{\epsilon} (\frac{L}{d}+M+\epsilon N)} \le e^{C_1\tilde{\epsilon} (\frac{L}{d}+M+\epsilon N) (\fq+\log d)} .\]

To show~\eqref{eq:possible-patterns-a-vertex-can-be-in}, recalling the definition of $\fq$ from~\eqref{eq:fq-def}, it suffices to show that, for any $X \in \breakups$ and $v \in \Z^d$, there exits $I \subset \SS$ such that
\[ I_X(v) = \{ P \in \phase_i : I \subset P_\bdry\} , \qquad\text{where }i := \1_{\{v\text{ is even}\}} .\]
To see this, let $f \colon \Z^d \to \SS$ be any configuration such that $X$ is a breakup of $f$, and set $I := f(N(v))$. By~\eqref{eq:breakup-1}, for $P \in \phase_i$, we have $v \in X_P$ if and only if $I \subset P_\bdry$. For $P \in \phasedom \setminus \phase_i$, we clearly have $P \notin I_X(v)$, since $v$ is $P$-even.
\end{proof}

\section{Repair transformation and Shearer's inequality}
\label{sec:shift-trans}

In this section, we prove the following generalization of \cref{lem:bound-on-pseudo-breakup}.
Recall from \cref{sec:overview-notation} that $\Lambda \subset \Z^d$ is a fixed domain outside which the configuration is forced to be in the $P_0$-pattern.

\begin{lemma}\label{lem:bound-on-pseudo-breakup-expectation}
	Let $S \subset \Z^d$ be finite and let $\{ S_P \}_{P \in \phasedom}$ be a partition of $S^c$ such that $\intB S_P \subset \extB S$ for all $P$.
	Suppose that $S \cup S_{P_0}$ contains $(\Lambda^c)^+$.
	Let $E$ be an event which is determined by the values of $f$ on $S^+$. Let $\Omega$ be the event that $E$ occurs and $(\intB S_P)^+$ is in the $P$-pattern for every~$P$.
	Then
	\[ \Pr_{\Lambda,P_0}(\Omega) \le \exp\left[-\tfrac{\alpha}{16} \E\left(\big|S^f_\unbal\big| + \tfrac{1}{d}\big|S^{\Omega,f}_\rest\big| + \epsilon d \big|S^{\Omega,f}_\highlyrest\big| \right) + \tfrac{\fq}{d}\big|S \setminus S^\Omega_\unique\big| +\bar\gamma|S| \right] ,\]
	where the expectation is taken with respect to a random function $f$ chosen from $\Pr_{\Lambda,P_0}(\cdot \mid \Omega)$.
\end{lemma}

Let us show how this yields \cref{lem:bound-on-pseudo-breakup}.

\begin{proof}[Proof of \cref{lem:bound-on-pseudo-breakup}]
Let $E$ be the event that $f|_{S^+}=\phi|_{S^+}$ for some $\phi \in \Omega$ and let $\Omega'$ be the event that $E$ occurs and $(\intB S_P)^+$ is in the $P$-pattern for all $P$.
	Note that $E$ is determined by $f|_{S^+}$, $\Omega \subset \Omega'$, $k(\Omega)=k(\Omega')$ and $S^\Omega_\unique = S^{\Omega'}_\unique$. Thus, \cref{lem:bound-on-pseudo-breakup} follows from \cref{lem:bound-on-pseudo-breakup-expectation}.
\end{proof}

The proof of \cref{lem:bound-on-pseudo-breakup-expectation} is based on a general upper bound on the total weight of configurations in an event, given in \cref{prop:shearer-for-bad-set} below. Recall from~\eqref{eq:config-weight} that $\omega_f$ is the weight of a function $f \colon \Lambda \to \SS$.
For a collection $\cF$ of such functions, we denote
\[ \omega(\cF) := \sum_{f \in \cF} \omega_f .\]
Note that since the measure $\Pr_{\Lambda,P_0}$ is defined through ratios of such weights, it is unaffected by a global multiplicative scaling of the pair interactions $(\lambda_{i,j})_{i,j \in \SS}$. Thus, without loss of generality, we may assume throughout this entire section that
\begin{equation}\label{eq:max-edge-weight-is-1}
\lambda^{\text{int}}_{\text{max}}=\max_{i,j \in \SS} \lambda_{i,j} = 1 ,
\end{equation}
Recall that $\omega_{\text{dom}}$ is the weight of a dominant pattern.

For a set $U \subset \Z^d$, we denote $U^\even := \Even \cap U$ and $U^\odd := \Odd \cap U$.
For two sets $U,V \subset \Z^d$, we denote
\[ \partial^\even(U,V) := \partial(U^\even,V^\odd) \qquad\text{and}\qquad \partial^\odd(U,V) := \partial(U^\odd,V^\even) ,\]
so that $\partial(U,V)=\partial^\even(U,V) \cup \partial^\odd(U,V)$. We also write $\partial^\even U := \partial^\even(U,U^c)$ and $\partial^\odd U := \partial^\odd(U,U^c)$.
For a dominant pattern $P=(A,B)$, we write $\lambda^\even_P := \lambda_A$ and $\lambda^\odd_P := \lambda_B$.
Recall the notions of unbalanced neighborhood, restricted edge, highly energetic vertex and unique pattern defined in \cref{sec:Shearer_overview}. Note that although those notions were defined for functions $f \colon \Z^d \to \SS$, they are well defined for any $v \in S$ when $f \colon S^+ \to \SS$.

\begin{prop}\label{prop:shearer-for-bad-set}
	Let $S \subset \Z^d$ be finite and let $\{ S_P \}_{P \in \phasedom}$ be a partition of $S^c$. Let $\cF$ be a set of functions $f \colon S^+ \to \SS$ satisfying that $S^+ \cap (\intB S_P)^+$ is in the $P$-pattern for every $P$.
	Then
	\[ \begin{aligned} \omega(\cF) \le \omega_{\text{dom}}^{|S^+|/2} &\cdot \exp\left[-\tfrac{\alpha}{16} \E\Big(\big|S^f_\unbal\big| + \tfrac{1}{d}\big|S^{\cF,f}_\rest\big| + \epsilon d\big|S^{\cF,f}_\highlyrest\big|\Big)+ \tfrac{\fq}{d}\big|S \setminus S^{\cF}_\unique\big| +\bar\gamma|S|\right] \\ &\cdot \prod_P \left(\tfrac{\lambda^\even_P}{\lambda^\odd_P}\right)^{\frac{1}{4d}(|\partial^{\even} (S^+,S_P \setminus S^+)|-|\partial^{\odd} (S^+,S_P \setminus S^+)|)} , \end{aligned}\]
	where the expectation is taken with respect to a random element $f \in \cF$ chosen with probability proportional to its weight~\eqref{eq:config-weight}.
\end{prop}

Before proving the proposition, let us show how it implies \cref{lem:bound-on-pseudo-breakup-expectation}, and thus \cref{lem:bound-on-pseudo-breakup}. The proof is based on a ``repair transformation'' similar to the one described in \cref{sec:Shearer_overview}.

\begin{proof}[Proof of \cref{lem:bound-on-pseudo-breakup-expectation}]
	Note that $\Omega$ is determined by the values of $f$ on $S^{++}$. Let $\bar{\Lambda}$ be a finite subset of $\Z^d$ that contains $\Lambda \cup S^{++}$.
	Let $\bar{\Omega}$ be the support of the marginal of $\Pr_{\Lambda,P_0}$ on $\SS^{\bar{\Lambda}}$.
	We henceforth view $\Omega$ as a subset of $\bar{\Omega}$.
	Denote
	\[ \Omega_0 := \Big\{ f|_{\bar\Lambda \setminus S^+} : f \in \Omega \Big\} \subset \SS^{\bar{\Lambda} \setminus S^+} \quad\text{and}\quad \Omega_1 := \Big\{ f|_{S^+} : f \in \Omega \Big\} \subset \SS^{S^+} .\]
	Let $T \colon \Omega_0 \to 2^{\bar{\Omega}}$ be a map which satisfies $T(f) \cap T(f') = \emptyset$ for distinct $f,f' \in \Omega_0$.
	Recalling \eqref{eq:finite_volume_measure},~\eqref{eq:prob_outside_Lambda_def1}, \eqref{eq:prob_outside_Lambda_def2}, \eqref{eq:config-weight} and using the fact that $\Lambda\subset\bar{\Lambda}$ we see that $\Pr_{\Lambda,P_0}(f|_{\bar{\Lambda}}=g)$ is proportional to $\omega_g \1_{g\in\bar{\Omega}}$. Therefore,
	\[ \Pr_{\Lambda,P_0}(\Omega)
	= \frac{\omega(\Omega)}{\omega(\bar{\Omega})}
	\le \frac{\omega(\Omega_0) \cdot \omega(\Omega_1)}{\sum_{f \in \Omega_0} \omega(T(f))} \le \frac{\omega(\Omega_1)}{\min_{f \in \Omega_0} \frac{\omega(T(f))}{\omega_f}}, \]
	where we used that $\lambda^{\text{int}}_{\text{max}}\le 1$ by~\eqref{eq:max-edge-weight-is-1} in the first inequality.

	Before defining $T$, let us bound $\omega(\Omega_1)$.
	To this end, we aim to apply \cref{prop:shearer-for-bad-set} with $\cF=\Omega_1$ (and $S$ and $(S_P)$ as here).
	Observe that, since $(\intB S_P)^+$ is in the $P$-pattern on $\Omega$ and since $E$ is determined by $f|_{S^+}$, the collection $\cF$ satisfies the assumption of the proposition and, moreover, a random element of $\cF$ chosen with probability proportional to its weight~\eqref{eq:config-weight} has the same distribution as $\Pr_{\Lambda,P_0}(f|_{S^+} \in \cdot \mid \Omega)$.
	For $i \in \{0,1\}$, denote $S_i := \bigcup_{P \in \cP_i} S_P \setminus S^+$, where $\phase_0$ and $\phase_1$ were defined in \cref{sec:overview-notation}.
	Then, by \cref{prop:shearer-for-bad-set},
	\begin{align*}
	\omega(\Omega_1) &\le \omega_{\text{dom}}^{|S^+|/2} \cdot e^{-\frac{\alpha}{16} \Big(\big|S^f_\unbal\big| + \tfrac{1}{d}\big|S^{\cF,f}_\rest\big| + \epsilon d\big|S^{\cF,f}_\highlyrest\big|\Big)+\tfrac{\fq}{d}\big|S \setminus S^{\cF}_\unique\big| +\bar\gamma|S|} \\&\quad \cdot \left(\tfrac{\lambda_{A_0}}{\lambda_{B_0}}\right)^{\frac{1}{4d}(|\partial^\even (S^+,S_0)|-|\partial^\odd (S^+,S_0)|-|\partial^\even (S^+,S_1)|+|\partial^\odd (S^+,S_1)|)} .
	\end{align*}
	Thus, the lemma will follow if we find a map $T$ such that $T(f) \cap T(f') = \emptyset$ for distinct $f,f' \in \Omega_0$ and which also satisfies
	\begin{equation}\label{eq:lower-bound-for-psuedo-breakup}
	\min_{f \in \Omega_0} \frac{\omega(T(f))}{\omega_f} \ge \omega_{\text{dom}}^{|S^+|/2} \cdot \left(\tfrac{\lambda_{A_0}}{\lambda_{B_0}}\right)^{\frac{1}{4d}(|\partial^\even(S^+,S_0)|-|\partial^\even(S^+,S_1)| - |\partial^\odd (S^+,S_0)|+|\partial^\odd (S^+,S_1)|)} .
	\end{equation}

	We now turn to the definition of $T$.
	Fix a unit vector $e \in \Z^d$.
	For $u \in \Z^d$, we denote $u^{\up} := u+e$ and $u^{\down} := u-e$. For a set $U \subset \Z^d$, we also write $U^{\up} := \{ u^{\up} : u \in U\}$ and $U^{\down} := \{ u^{\down} : u \in U\}$.
	For each $P \in \phasedom$, let $\psi_P$ be a $\lambda$-weight-preserving permutation of $\SS$ (as in~\eqref{eq:direct-equivalent-patterns-def}) taking $P$ to $P_0=(A_0,B_0)$ if $P \in \phase_0$ or to $(B_0,A_0)$ otherwise (for $P=P_0$, we take $\psi_{P_0}$ to be the identity). Let $\cH$ be the set of all functions $h \colon S_* \to \SS$ which are in the $P_0$-pattern, where
	\[ S_* := (S_0 \cup S_1^{\down})^c .\]
	For $f \in \Omega_0$ and $h \in \cH$, define $\phi_{f,h} \colon \bar\Lambda \to \SS$ by
	\[ \phi_{f,h}(v) := \begin{cases}
	\psi_P(f(v)) &\text{if }v \in S_P \setminus S^+\text{ for }P \in \cP_0 \\
	\psi_P(f(v^{\up})) &\text{if }v \in (S_P \setminus S^+)^{\down}\text{ for }P \in \cP_1 \\
	h(v) &\text{if }v \in S_*
	\end{cases} .\]
	Note that $\phi_{f,h}$ is well defined, since the assumption that $\intB S_P \subset \extB S$ for all $P$ implies that
	\begin{equation}\label{eq:S_P-distance-3}
	\dist(S_P \setminus S^+,S_Q \setminus S^+) \ge 3 \qquad\text{for distinct }P\text{ and }Q
	\end{equation}
	so that, in particular, $\{ S_0, S_1^{\down}, S_* \}$ is a partition of $\Z^d$.

	Let us check that $\phi:=\phi_{f,h} \in \bar{\Omega}$. By~\eqref{eq:finite_volume_measure}, \eqref{eq:prob_outside_Lambda_def1} and~\eqref{eq:prob_outside_Lambda_def2}, we must check that $\bar\Lambda \setminus \Int(\Lambda) = \bar\Lambda \cap (\Lambda^c)^+$ is in the $P_0$-pattern in $\phi$.
	Let $v \in \bar\Lambda \setminus \Int(\Lambda)$ and recall that, by assumption, $(\Lambda^c)^+ \subset S \cup S_{P_0}$.
	If $v \in S \subset S_*$ then $\phi(v)=h(v)$ and it is clear that $v$ is in the $P_0$-pattern in $\phi$.
	Otherwise, $v \in (S \cup S_{P_0}) \setminus S_* \subset S_{P_0} \setminus S^+$ so that $\phi(v)=f(v)$, and this is also clear.

	Finally, define
	\[ T(f) := \{ \phi_{f,h} : h \in \cH \} .\]
	To see that the desired property that $T(f) \cap T(f') = \emptyset$ for distinct $f,f' \in \Omega_0$ holds, we now show that the mapping $(f,h) \mapsto \phi_{f,h}$ is injective on $\Omega_0 \times \cH$. To this end, we show how to recover $(f,h)$ from a given $g$ in the image of this mapping. Indeed, it is straightforward to check that
	\[ f(v) =
	\begin{cases}
	\psi_P^{-1}(g(v)) &\text{if }v \in S_P \setminus S^+\text{ for }P \in \cP_0 \\
	\psi_P^{-1}(g(v^{\down})) &\text{if }v \in S_P \setminus S^+\text{ for }P \in \cP_1
	\end{cases}
	\qquad\text{and}\qquad
	h(v) = g(v)\text{ for }v \in S_* .\]

	It remains to check that~\eqref{eq:lower-bound-for-psuedo-breakup} holds.
	We begin by showing that $\omega_{\phi_{f,h}} = \omega_f \cdot \omega_h$ for any $f \in \Omega_0$ and $h \in \cH$.
	Denote $\phi := \phi_{f,h}$ and observe that
	\[ \omega_\phi = \omega_{\phi|_{\bar\Lambda \setminus S_*}} \cdot \omega_{\phi|_{S_*}} \cdot \prod_{\{u,v\} \in \partial S_*} \lambda_{\phi(u),\phi(v)} = \omega_{\phi|_{\bar\Lambda \cap S_0}} \cdot \omega_{\phi|_{S_1^{\down}}} \cdot \omega_{\phi|_{S_*}} \cdot \prod_{\{u,v\} \in \partial S_0 \cup \partial S_1^{\down}} \lambda_{\phi(u),\phi(v)} .\]
	Let us check that $\omega_{\phi|_{\bar\Lambda \cap S_0}} = \omega_{f|_{\bar\Lambda \cap S_0}}$. Since $f$ and $\phi$ coincide on $S_{P_0} \setminus S^+$, we have $\lambda_{\phi(u)} = \lambda_{f(u)}$ for any vertex $u \in \bar\Lambda \cap S_{P_0} \setminus S^+$ and $\lambda_{\phi(u),\phi(v)} = \lambda_{f(v),f(u)}$ for any edge $\{u,v\} \subset \bar\Lambda \cap S_{P_0} \setminus S^+$. Similarly, for $P \in \phase_0 \setminus \{P_0\}$, since $\psi_P$ preserves weights, we have $\lambda_{\phi(u)} = \lambda_{\psi_P(f(u))} = \lambda_{f(u)}$ for any $u \in S_P \setminus S^+$ and $\lambda_{\phi(u),\phi(v)} = \lambda_{\psi_P(f(v)),\psi_P(f(u))} = \lambda_{f(v),f(u)}$ for any edge $\{u,v\} \subset S_P \setminus S^+$. Thus, by~\eqref{eq:S_P-distance-3}, $\omega_{\phi|_{\bar\Lambda \cap S_0}} = \omega_{f|_{\bar\Lambda \cap S_0}}$.
	Similarly, one may check that $\omega_{\phi|_{(S_1)^{\down}}} = \omega_{f|_{S_1}}$.
	Since $\partial(S_0,S_1)=\emptyset$, we conclude that $\omega_{\phi|_{\bar\Lambda \cap S_0}} \cdot \omega_{\phi|_{S_1^{\down}}} = \omega_{f|_{\bar\Lambda \cap S_0}} \cdot \omega_{f|_{S_1}} = \omega_{f|_{\bar\Lambda \cap (S_0 \cup S_1)}} = \omega_{f|_{\bar\Lambda \setminus S^+}} = \omega_f$.

	Since $\phi|_{S_*} = h$, it remains only to check that $\lambda_{\phi(u),\phi(v)} = 1$ for all $\{u,v\} \in \partial S_0 \cup \partial S_1^{\down}$, i.e., that $\{ \phi(u),\phi(v) \} \in E(H)$.
	In fact, we show that $\intextB S_0$ and $\intextB S_1^{\down}$ are in the $P_0$-pattern in $\phi$. By definition, every vertex in $S_*$ is in the $P_0$-pattern in $\phi$, so that it suffices to check that $\intB S_0$ and $\intB S_1^{\down}$ are in the $P_0$-pattern in $\phi$.
	Indeed, if $w \in \intB S_0$ then $w \in \intB (S_P \setminus S^+) \subset (\intB S_P)^+$ for some $P \in \phase_0$. By the assumption of the lemma, $w$ is in the $P$-pattern in $f$, and thus, by the definition of $\psi_P$, $w$ is in the $P_0$-pattern in $\psi_P \circ f$ and hence also in $\phi$.
	Similarly, if $w \in \intB S_1^{\down}$ then $w^{\up} \in \intB (S_P \setminus S^+) \subset (\intB S_P)^+$ for some $P \in \phase_1$, so that $w^{\up}$ is in the $P$-pattern in $f$, and thus, $w^{\up}$ is in the $(B_0,A_0)$-pattern in $\psi_P \circ f$ so that $w$ is in the $P_0$-pattern in $\phi$.

	We therefore conclude that
	\[ \frac{\omega(T(f))}{\omega_f} = \omega(\cH) \qquad\text{for all }f \in \Omega_0 .\]
	Since the definition of $\cH$ immediately implies that
	\[ \omega(\cH) = (\lambda_{A_0})^{|S^\even_*|} \cdot (\lambda_{B_0})^{|S^\odd_*|} ,\]
	concluding~\eqref{eq:lower-bound-for-psuedo-breakup} is essentially just a computation. To see this, using the fact (which we prove below) that, for any finite set $U \subset \Z^d$,
	\begin{equation}\label{eq:even-odd-diff}
	|U^\even| - |U^\odd| = \tfrac{1}{2d} (|\partial^\even U| - |\partial^\odd U|) ,
	\end{equation}
	and writing $|S_*^\even| = \tfrac12 (|S_*| + |S_*^\even|-|S_*^\odd|)$, and similarly for $|S_*^\odd|$, we have
	\[ \omega(\cH) = \omega_{\text{dom}}^{|S_*|/2} \cdot \left(\tfrac{\lambda_{A_0}}{\lambda_{B_0}}\right)^{\frac{1}{4d}(|\partial^\even S_*|-|\partial^\odd S_*|)} .\]
Noting that $|S_*|=|S^+|$, it thus suffices to show that
\begin{align*}
|\partial^\even S_*| &= |\partial^\even(S^+,S_0)|+|\partial^\odd (S^+,S_1)|,\\
|\partial^\odd S_*| &= |\partial^\odd (S^+,S_0)|+|\partial^\even(S^+,S_1)|.
\end{align*}
Since $\partial S_* = \partial(S^+,S_0) \cup \partial((S^+)^{\down},S_1^{\down})$, this easily follows.

It remains to prove~\eqref{eq:even-odd-diff}. To see this, first observe that $u \mapsto u^{\down}$ is a bijection between $U^\even \cap U^{\up}$ and $U^\odd \cap U^{\down}$, so that
\[ |U^\even| - |U^\odd| = |(U \setminus U^{\up})^\even|-|(U \setminus U^{\down})^\odd| .\]
As the same holds for any direction $\up$, summing over the $2d$ possible directions yields~\eqref{eq:even-odd-diff}.
\end{proof}

\subsection{Proof of Proposition~\ref{prop:shearer-for-bad-set}}

The proof of \cref{prop:shearer-for-bad-set} relies on the following lemma which provides a bound on the total weight of a collection of configurations. An important feature of the bound is that it is factorized into ``local terms'' involving the values of the configuration on a vertex and its neighbors. The proof of the lemma, which is based on Shearer's inequality, is given in \cref{sec:shearer} below.
Recall the definition of $Z(\Psi,I)$ from~\eqref{eq:Z-Psi-I-def}.

\begin{lemma}\label{lem:weighted-shearer}
	Let $S \subset \Z^d$ be finite and even and let $\{ \SS_u \}_{u \in \intB S}$ be a collection of subsets of $\SS$. Let $\cF \subset \SS^S$ be such that $f(u) \in \SS_u$ for every $f \in \cF$ and $u \in \intB S$. Let $f$ be an random element of $\cF$ chosen with probability proportional to its weight~\eqref{eq:config-weight}. For each odd vertex $v \in S$, let $X_v$ be a random variable which is measurable with respect to $f|_{N(v)}$.
	Then
	\[ \omega(\cF) \le \prod_{v \in S^\odd} \prod_x \left[\frac{Z(\Psi_{v,x},I_{v,x})}{\Pr(X_v=x)}\right]^{\frac{1}{2d} \Pr(X_v=x)} \cdot \prod_{u \in \intB S} (\lambda_{\SS_u})^{\frac{1}{2d}|\partial u \cap \partial S|} ,\]
	where the second product is over $x$ in the support of $X_v$, and $\Psi_{v,x}$ and $I_{v,x}$ are the supports of $f|_{N(v)}$ and $f(v)$ on the event $\{X_v=x\}$, respectively.
\end{lemma}

Thus, besides factorizing the bound on $\omega(\cF)$ over the odd vertices in $S$, \cref{lem:weighted-shearer} allows exposing some information about $f|_{N(v)}$, which can then be used to bound $Z(\Psi,I)$. One could theoretically expose $f|_{N(v)}$ completely (i.e., by taking $X_v$ to equal $f|_{N(v)}$ above), but this would increase the number of terms in the second product above. One would therefore like to expose as little information as possible, which still suffices to obtain good bounds on $Z(\Psi,I)$.
Recalling the notions of non-dominant vertex, restricted edge, unbalanced neighborhood and highly energetic vertex introduced in \cref{sec:Shearer_overview}, we aim to expose just enough information as to allow to determine the occurrence of these. For this, it will suffice to reveal the information of whether $v$ is a dominant vertex, and if so, also the $R$-closure of $f(N(v))$ and whether the neighborhood of $v$ is unbalanced.
We will then make use of the inequality given in \cref{lem:weighted-shearer}, accompanying it with the bounds on $Z(\Psi,I)$ given in \cref{main-cond}. At this point, we remind the reader that the fact that \cref{main-cond} holds under either of the explicit quantitative conditions~\eqref{eq:parameter-inequalities-simple}, \eqref{eq:parameter-inequalities-simple1} or \eqref{eq:parameter-inequalities-simple2} is shown in \cref{sec:model_on_K2d2d}.

\begin{proof}[Proof of \cref{prop:shearer-for-bad-set}]
	We prove something slightly stronger than the inequality stated in the lemma. Namely, we show that
	\begin{equation}\label{eq:shearer-ub}
	\begin{aligned}
	\omega(\cF) \le \omega_{\text{dom}}^{|(S^+)^\odd|} & \cdot e^{-\frac{\alpha}{8} \E\big(\big|S^{f,\odd}_\unbal\big| + \frac{1}{d}\big|S^{\cF,f,\odd}_\rest\big| + \epsilon d\big|S^{\cF,f,\odd}_\highlyrest\big|\big) + \bar\gamma|S| + \frac{\fq}{d}|S \setminus S^\cF_\unique|} \\ &\cdot \prod_P (\lambda^\even_P)^{\frac{1}{2d}(|\partial^\even (S^+,S_P \setminus S^+)|-|\partial^{odd} (S^+,S_P \setminus S^+)|)} ,
	\end{aligned}
	\end{equation}
	where $S^{f,\odd}_\unbal := (S^f_\unbal)^\odd$, $S^{\cF,f,\odd}_\highlyrest = (S^{\cF,f}_\highlyrest)^\odd$ and $S^{\cF,f,\odd}_\rest$ is the set of restricted edges $(v,u)$ with $v \in S^\odd$.
	The lemma then follows by taking the geometric average of the above bound and its symmetric version in which the roles of odd and even are exchanged.

	In proving~\eqref{eq:shearer-ub}, instead of working directly with $S^+$, it is convenient to work with its even expansion, defined as
	\[ S' := S^+ \cup (\extB S^+)^\even = S^{++} \setminus (\extB S^+)^\odd .\]
	Note that $S^+ \subset S' \subset S^{++}$ and $(S^+)^\odd = (S')^\odd$. Let $\cF'$ be the set of functions $f' \in \SS^{S'}$ satisfying that $f'|_{S^+} \in \cF$ and for which $S_P \cap S' \setminus S^+$ is in the $P$-pattern for every $P$.
	Observe that if one samples a random element $f' \in \cF'$ with probability proportional to its weight~\eqref{eq:config-weight}, then $f'|_{S^+}$ has the same distribution as $f$, and the random variables $\{f'(u)\}_{u \in S' \setminus S^+}$ are independent (of one another and of $f'|_{S^+}$), with each $f'(u)$ distributed according to the single-site activities $(\lambda_i)$ on $A$ (analogously to~\eqref{eq:prob_outside_Lambda_def2}), where $P=(A,B)$ is the unique dominant pattern such that $u \in S_P$. Recalling~\eqref{eq:max-edge-weight-is-1}, it follows that
	\[ \omega(\cF') = \omega(\cF) \cdot \prod_P (\lambda^\even_P)^{|S_P \cap S' \setminus S^+|} .\]
	Thus, noting that $2d|S_P \cap S' \setminus S^+|=|\partial^\odd(S^+,S_P \setminus S^+)|+|\partial^\odd((S^+)^c,S_P \cap \extB S^+)|$, we see that~\eqref{eq:shearer-ub} is equivalent to
	\begin{equation}\label{eq:shearer-ub2}
	\begin{aligned}
	\omega(\cF') \le \omega_{\text{dom}}^{|(S^+)^\odd|} &\cdot e^{-\frac{\alpha}{8} \E\big(\big|S^{f,\odd}_\unbal\big| + \frac{1}{d}\big|S^{\cF,f,\odd}_\rest\big| + \epsilon d\big|S^{\cF,f,\odd}_\highlyrest\big|\big) +\bar\gamma|S| + \frac{\fq}{d}|S \setminus S^\cF_\unique|} \\ &\cdot \prod_P (\lambda^\even_P)^{\frac{1}{2d}(|\partial^\even (S^+,S_P \setminus S^+)| + |\partial^\odd ((S^+)^c,S_P \cap \extB S^+)|)} .
	\end{aligned}
	\end{equation}
	We also note at this point that $S^{\cF,f,\odd}_\rest = S^{\cF',f,\odd}_\rest$, $S^{\cF,f,\odd}_\highlyrest = S^{\cF',f,\odd}_\highlyrest$ and $S^{\cF}_\unique = S^{\cF'}_\unique$.

	We now aim to apply \cref{lem:weighted-shearer} with $S'$ and $\cF'$. For $u \in \intB S'$, define
	\[ \SS_u := \begin{cases} A &\text{if }u \in \extB S^+ \cap S_{(A,B)}\\ \bigcap \{ A : (A,B) \in \phasedom \text{ and } u \in N(S_{(A,B)}) \} &\text{if }u \in \intB S^+ \end{cases} .\]
	Note that, by the assumption on $\cF$ and by the definition of $\cF'$, we have $\phi(u) \in \SS_u$ for all $\phi \in \cF'$ and all $u \in \intB S'$.
	For an odd vertex $v \in \intB S^+$, define $X_v := 0$, and for an odd vertex $v \in S$, define
	\[ X_v := \begin{cases}
		(R(R(f'(N(v)))),\1_{\{v\text{ has an unbalanced neighborhood in }f'\}}) &\text{if $v$ is dominant in }f' \\
		(\emptyset,0) &\text{otherwise}
	\end{cases}.\]
	Then, by \cref{lem:weighted-shearer},
	\[ \omega(\cF') \le \prod_{v \in (S^+)^\odd} \prod_x \left[\frac{Z(\Psi_{v,x},I_{v,x})}{\Pr(X_v=x)}\right]^{\frac{1}{2d} \Pr(X_v=x)} \cdot \prod_{u \in \intB S'} (\lambda_{\SS_u})^{\frac{1}{2d}|\partial u \cap \partial S'|} ,\]
	where $\Psi_{v,x}$ and $I_{v,x}$ are the supports of $f'|_{N(v)}$ and $f'(v)$ on the event $\{X_v=x\}$, respectively. We stress that the probabilities above are with respect to $f'$, but we also remind that $f'|_{S^+}$ equals $f$ in distribution so that these probabilities are the same when taken with respect to~$f$.

	We first show that
	\[ \prod_{u \in \intB S'} (\lambda_{\SS_u})^{\frac{1}{2d}|\partial u \cap \partial S'|} \le \prod_P (\lambda^\even_P)^{\frac{1}{2d}(|\partial^\even (S^+,S_P \setminus S^+)| + |\partial^\odd ((S^+)^c,S_P \cap \extB S^+)|)} .\]
	Since $\{ \partial u \cap \partial S'\}_{u \in \intB S'}$ and $\{ \partial^\even (S^+,S_P \setminus S^+), \partial^\odd ((S^+)^c,S_P \cap \extB S^+) \}_P$ are two partitions of $\partial S'$, it suffices to show an inequality for each edge separately, namely, that $\lambda_{\SS_u} \le \lambda^\even_P$ for any $\{u,w\} \in \partial S'$ and $P=(A,B)$ such that $u \in S'$ and $\{u,w\} \in \partial^\even (S^+,S_P \setminus S^+) \cup \partial^\odd ((S^+)^c,S_P \cap \extB S^+)$. Note that $u$ is even and that $w \notin S^+$. If $u \notin S^+$ then $u \in S_P$ so that $\SS_u = A$ and $\lambda_{\SS_u} = \lambda^\even_P$. If $u \in S^+$ then $w \in S_P$ so that $\SS_u \subset A$ and $\lambda_{\SS_u} \le \lambda^\even_P$.

	Thus, to obtain~\eqref{eq:shearer-ub2}, it suffices to show that, for any $v \in (S^+)^\odd$,
	\begin{equation}\label{eq:shearer-ub-3}
	\prod_x \left[\frac{Z(\Psi_{v,x},I_{v,x})}{\Pr(X_v=x)}\right]^{\frac{1}{2d} \Pr(X_v=x)} \le \omega_{\text{dom}} \cdot \begin{cases} e^{-\frac{\alpha}{8} p_v + \bar\gamma + \frac{\fq}{d}\1_{v \notin S^\cF_\unique}} &\text{if }v \in S\\1&\text{if }v \notin S\end{cases} ,
	\end{equation}
	where
	\[ p_v := \Pr\big(v \in S^f_\unbal\big) + \tfrac{1}{d} \cdot \E\big|\vec\partial v \cap S^{\cF,f}_\rest\big| + \epsilon d \cdot \Pr\big(v\in S^{\cF,f}_\highlyrest\big) .\]
	Suppose first that $v \notin S$. By the assumption on $\cF$ and by definition of $\cF'$, we have that $\Psi_{v,x} \subset A^{N(v)}$ and $I_{v,x} \subset B$, where $P=(A,B)$ is the unique dominant pattern such that $v \in S_P$. Thus,
	\[ \prod_x \left[\frac{Z(\Psi_{v,x},I_{v,x})}{\Pr(X_v=x)}\right]^{\frac{1}{2d} \Pr(X_v=x)} = Z(\Psi_{v,0},I_{v,0})^{1/2d} \le Z(A^{[2d]},B)^{1/2d} = \lambda_A \lambda_B = \omega_{\text{dom}} .\]
	Suppose now that $v \in S$.
	Taking logarithms in~\eqref{eq:shearer-ub-3}, the inequality becomes
	\[ \tfrac{1}{2d} \Ent(X_v) + \tfrac{1}{2d} \E\big[ \log Z(\Psi_{v,X_v},I_{v,X_v}) \big] \le \log \omega_{\text{dom}} -\tfrac{\alpha}{8} \cdot p_v + \bar\gamma + \tfrac{\fq}{d}\1_{v \notin S^\cF_\unique} .\]
	Recalling that $\bar\gamma=\gamma+e^{-\alpha d/25}$, this will follow if we show that
	\begin{align}
	\tfrac{1}{2d} \Ent(X_v) &\le \tfrac{\alpha}{24} \cdot p_v +e^{-\alpha d/25} + \tfrac{\fq}{d}\1_{v \notin S^\cF_\unique} ,\label{eq:shearer-ub-entropy}\\
	\tfrac{1}{2d} \E\big[ \log Z(\Psi_{v,X_v},I_{v,X_v}) \big] &\le \log \omega_{\text{dom}} -\tfrac{\alpha}{6} \cdot p_v + \gamma. \label{eq:shearer-ub-Z}
	\end{align}

	\smallskip
	We begin by showing~\eqref{eq:shearer-ub-entropy}.
	By~\eqref{eq:entropy-support} and~\eqref{eq:frak q inequalities}, we always have the trivial bound
	\[ \Ent(X_v) \le \log |\supp(X_v)| \le \log(1+2|\phasedom|) \le 2\fq.\]
	Thus, it suffices to show that, for any $v \in S^\cF_\unique$,
	\[ \Ent(X_v) \le \tfrac{\alpha d}{12} \cdot p_v + 2d e^{-\alpha d/25} .\]
	Fix $v \in S^\cF_\unique$ and denote $p:=p_v$. When $p \ge 1/2$, the above bound is clear from the trivial bound on $\Ent(X_v)$ as $\alpha d\ge C\fq$ by~\eqref{eq:alpha-cond}. Thus, we may assume that $p<1/2$.
	By the definition of unique pattern, there exists some $J$ for which $X_v \neq (J,0)$ implies that $v \in S^f_\unbal$ or $\vec\partial v \subset S^{\cF,f}_\rest$.
	In particular, $\Pr(X_v \neq (J,0)) \le p < 1/2$.
	Hence, using the chain rule for entropy~\eqref{eq:entropy-chain-rule},
	\begin{align*}
	\Ent(X_v) &= \Ent(\1_{\{X_v=(J,0)\}}) + \Ent(X_v \mid \1_{\{X_v=(J,0)\}}) \\& \le p\log \tfrac{2^{\fq+2}}{p} + (1-p)\log\tfrac{1}{1-p} \le 2p\log \tfrac{2^{\fq+2}}{p} ,
	\end{align*}
	where we used~\eqref{eq:entropy-support} and that binary entropy is increasing on $[0,\frac12]$ in the first inequality, and we used that $x\log\frac1x \ge (1-x)\log\frac1{1-x}$ for $0<x<\frac12$.
	Thus, using that $x \log \frac ax \le\frac ae \le \frac a2$ for $0<x<1$ and that $\alpha d\ge C\fq$ by~\eqref{eq:alpha-cond}, we obtain
	\[ \Ent(X_v) - \tfrac{\alpha d}{12} p \le 2p \log \left( \tfrac{2^{\fq+2}}{p} \cdot e^{-\alpha d /24} \right) \le 2^{\fq+2} e^{-\alpha d /24} \le e^{-\alpha d /25} \le 2d e^{-\alpha d /25} .\]

	\smallskip
	It remains to show~\eqref{eq:shearer-ub-Z}.
	Denote $X_v=(J_v,U_v)$ and let $\cR_v$ be the set of $u \sim v$ for which $\{g(u) : g \in \Psi_{v,X_v} \} \not\simeq_R J_v$.
	By~\eqref{eq:cond-restricted-left}-\eqref{eq:cond-non-dominant} of \cref{main-cond},
	\[ \frac{\log Z(\Psi_{v,X_v},I_{v,X_v})}{2d} \le \log \omega_{\text{dom}} + \gamma - \frac \alpha 2
	\begin{cases}
	\tfrac{1}{d} |\cR_v| &\text{if }J_v \neq \emptyset,~U_v=0,~R(I_{v,X_v}) \subset J_v\\
	3\epsilon d &\text{if }J_v \neq \emptyset,~U_v=0,~I_{v,X_v} \subset \SS \setminus R(J_v)\\
	1 &\text{otherwise}\\
	\end{cases}.
	\]
	Then~\eqref{eq:shearer-ub-Z} will follow if we show that
\[ \1_{\{v \in S^f_\unbal\}} + \tfrac{1}{d} \big|\vec\partial v \cap S^{\cF,f}_\rest\big| + \epsilon d \1_{\{v \in S^{\cF,f}_\highlyrest\}} \le \begin{cases}
	\tfrac{3}{d} |\cR_v| &\text{if }J_v \neq \emptyset,~U_v=0,~R(I_{v,X_v}) \subset J_v\\
	9 \epsilon d &\text{if }J_v \neq \emptyset,~ U_v=0,~ I_{v,X_v} \subset \SS \setminus R(J_v)\\
	3 &\text{otherwise}\\
	\end{cases} .\]
	This follows from the definitions of balanced neighborhood, restricted edge and highly energetic (recall these from \cref{sec:Shearer_overview}):
	If $J_v=\emptyset$ or $U_v=1$, then $v \notin S^{\cF,f}_\highlyrest$ and the inequality is clear. Otherwise, $J_v \neq \emptyset$ and $U_v=0$, so that $v \notin S^f_\unbal$. If $v \in S^{\cF,f}_\highlyrest$ then $I_{v,X_v} \subset \SS \setminus R(J_v)$ and the inequality is also clear since $\epsilon \ge \frac{1}{4d}$. Otherwise, as $(v,u) \in S^{\cF,f}_\rest$ implies that either $u \in \cR_v$ or $R(I_{v,X_v}) \not\subset J_v$, the inequality follows.
\end{proof}

\subsection{Proof of Lemma~\ref{lem:weighted-shearer}}
\label{sec:shearer}
The proof proceeds by first proving a generalized non-weighted version of the lemma using Shearer's inequality, and then reducing the weighted case to the non-weighted case by a method introduced by Galvin~\cite{galvin2006bounding}.

Let $\mathbb{H}$ be a finite set and let $\cH=(H_e)_{e \in E(\Z^d)}$ be a collection of subsets of $\mathbb{H}^2$. Let $S \subset \Z^d$ be even and denote by $\Hom_{S,\cH}$ the set of all functions $F \colon S \to \mathbb{H}$ such that $(F(v),F(u)) \in H_{\{v,u\}}$ whenever $v \in S$ is odd and $u \sim v$. Thus, we may regard $H_e$ as specifying a constraint on the possible values appearing on the endpoints of the edge $e$.
For $e \in E(\Z^d)$ and $h \in \mathbb{H}$, denote
\[ H_{e,h} := \big\{ h' \in \mathbb{H} : (h',h) \in H_e \big\} .\]
Thus, $H_{e,h}$ is the set of allowed values on the odd endpoint of $e$ given that the even endpoint of $e$ takes the value $h$. We extend this definition to allow taking into account the value at all neighbors of an odd vertex.
Namely, for an odd vertex $v \in S$ and a function $\psi \colon N(v) \to \mathbb{H}$, we write
\[ H_{v,\psi} := \bigcap_{u \sim v} H_{\{v,u\},\psi(u)} .\]
For a set $\Psi$ of functions $\psi \colon N(v) \to \mathbb{H}$ and a subset $I \subset \mathbb{H}$, denote
\[ Z^\cH_v(\Psi,I) := \sum_{\psi \in \Psi} |H_{v,\psi} \cap I|^{2d} .\]
The above is the non-weighted analogue of~\eqref{eq:Z-Psi-I-def}. In particular, $Z^\cH_v(\Psi,I)$ counts the number of functions $\varphi \colon K_{2d,2d} \to \SS$, which obey the constraints given by $\cH$ at $v$, and whose restrictions to the two sides belong to $\Psi$ and $I^{[2d]}$.
The following is a non-weighted analogue of \cref{lem:weighted-shearer}.

\begin{lemma}\label{lem:generalized-uniform-shearer-for-bad-set}
	Let $\mathbb{H}$ be a finite set, let $S \subset \Z^d$ be finite and even and let $\{ \mathbb{H}_u \}_{u \in \intB S}$ be a collection of subsets of $\mathbb{H}$. Let $\cF \subset \Hom_{S,\cH}$ be such that $F(u) \in \mathbb{H}_u$ for every $F \in \cF$ and $u \in \intB S$. Let $F$ be an element of $\cF$ chosen uniformly at random. For each odd vertex $v \in S$, let $X_v$ be a random variable which is measurable with respect to $F|_{N(v)}$.
	Then
	\[ |\cF| \le \prod_{v \in S^\odd} \prod_x \left(\frac{Z^\cH_v(\Psi_{v,x},I_{v,x})}{\Pr(X_v=x)}\right)^{\Pr(X_v=x)/2d} \cdot \prod_{u \in \intB S} |\mathbb{H}_u|^{\frac{1}{2d}|\partial u \cap \partial S|} ,\]
	where the second product is over $x$ in the support of $X_v$, and $\Psi_{v,x}$ and $I_{v,x}$ are the supports of $F|_{N(v)}$ and $F(v)$ on the event $\{X_v=x\}$, respectively.
\end{lemma}

\begin{proof}
	Let $F$ be a uniformly chosen element in $\cF$ and note that $\Ent(F)=\log |\cF|$. Hence, the desired inequality becomes
	\begin{equation}\label{eq:shearer-ub3}
	\begin{aligned}
	\Ent(F) &\le \tfrac{1}{2d} \sum_{v \in S^\odd} \sum_x \Pr(X_v=x) \cdot \log \frac{Z^\cH_v(\Psi_{v,x},I_{v,x})}{\Pr(X_v=x)} + \tfrac{1}{2d} \sum_{u \in \intB S} |\partial u \cap \partial S| \cdot \log |\mathbb{H}_u| .
	\end{aligned}
	\end{equation}
	We make use of~\eqref{eq:entropy-chain-rule}-\eqref{eq:entropy-subadditivity} throughout the proof.
	To prove~\eqref{eq:shearer-ub3}, we begin by writing
	\[ \Ent(F) = \Ent(F^\even) + \Ent(F^\odd \mid F^\even) .\]
	By the sub-additivity of entropy, we have
	\[ \Ent(F^\odd \mid F^{\even}) \le \sum_{v \in S^\odd} \Ent\big(F(v) \mid F|_{N(v)}\big) .\]
	We use Shearer's inequality to bound $\Ent(F^\even)$.
	Namely, \cref{lem:shearer} applied with the random variables $(Z_i) = (F(v))_{v \in S^\even}$, the collection $\cI=\{N(v)\}_{v \in S^\odd} \cup \{N(v) \cap S\}_{v \in \extB S}$ and $k=2d$, yields
	\[ \Ent(F^\even) \le \tfrac{1}{2d} \sum_{v \in S^\odd} \Ent\big(F|_{N(v)}\big) + \tfrac{1}{2d} \sum_{v \in \extB S} \Ent\big(F|_{N(v) \cap S}\big) .\]
	Note that, by the assumption on $\cF$,
	\[ \sum_{v \in \extB S} \Ent\big(F|_{N(v) \cap S}\big) \le \sum_{v \in \extB S} \sum_{u \in N(v) \cap S} \Ent(F(u)) = \sum_{u \in \intB S} |\partial u \cap \partial S| \cdot \Ent(F(u)) .\]
	Thus, \eqref{eq:shearer-ub3} will follow once we show that
	\[ E_v := \tfrac{1}{2d} \cdot \Ent\big(F|_{N(v)}\big) + \Ent\big(F(v) \mid F|_{N(v)}\big) \le \tfrac{1}{2d} \sum_x \Pr(X_v=x) \cdot \log \frac{Z^\cH_v(\Psi_{v,x},I_{v,x})}{\Pr(X_v=x)} .\]
	Let $\cX_v$ be the support of $X_v$ and note that the support of $F|_{N(v)}$ is the disjoint union $\bigcup_{x \in \cX_v} \Psi_{v,x}$. Then, by Jensen's inequality,
	\begin{align*}
	E_v
	&= \sum_{x \in \cX_v} \sum_{\psi \in \Psi_{v,x}} \Pr(F|_{N(v)}=\psi) \cdot \Big[ \tfrac{1}{2d} \cdot \log \tfrac{1}{\Pr(F|_{N(v)}=\psi)} + \Ent(F(v) \mid F|_{N(v)}=\psi) \Big] \\
	&\le \tfrac{1}{2d} \sum_{x \in \cX_v} \sum_{\psi \in \Psi_{v,x}} \Pr(F|_{N(v)}=\psi) \cdot \log \frac{|H_{v,\psi} \cap I_{v,x}|^{2d}}{\Pr(F|_{N(v)}=\psi)} \\
	&\le \tfrac{1}{2d} \sum_{x \in \cX_v} \Pr(X_v=x) \cdot \log \frac{Z^\cH_v(\Psi_{v,x},I_{v,x})}{\Pr(X_v=x)} . \qedhere
	\end{align*}
\end{proof}

\begin{proof}[Proof of \cref{lem:weighted-shearer}]
	We first observe that, since $S$ is even, $|S^\odd|=|E(S)|/2d$ and, by \cref{lem:boundary-size-of-odd-set}, $|S|=2|S^\odd|+|\partial S|/2d$.
	Using this, one easily checks that the inequality we wish to prove is invariant with respect to scaling of both the vertex-weights $\{ \lambda_i \}_{i \in \SS}$ and the edge-weights $\{ \lambda_{i,j} \}_{i,j \in \SS}$. Thus, by the continuity and the scale-invariance of the quantities in the weights, we may assume that $\{ \lambda_i \}_i$ are rational numbers in $(0,1]$ and $\max_{i,j} \lambda_{i,j} = 1$.
	Let $M$ be a large integer such that $\lambda_i M$ is an integer for all $i \in \SS$, and denote
	\[ \mathbb{H} := \big\{ (i,m) :  i \in \SS,~ 1 \le m \le \lambda_i M \big\} .\]
	We aim to apply \cref{lem:generalized-uniform-shearer-for-bad-set} with $\mathbb{H}$ and a suitably chosen $\cH=(H_e)_{e \in E(\Z^d)}$, which we construct randomly as follows. For each odd $v \in S$ and each $g=(i,m,j,n) \in \mathbb{H}^2$, we put $g$ in $H_e$ with probability $\lambda_{i,j}$, with all choices made independently.

	We first show that $\omega(\cF)$, the quantity we wish to bound, is given up to an explicit factor by an expectation over the random choice of $\cH$. Namely,
	\[ \omega(\cF) = \E|\cG| \cdot M^{-|S|} , \qquad\text{where}\quad\cG := \big\{ (f,\phi) \in \HH^S : f \in \cF,~ (f,\phi) \in \Hom_{S,\cH} \big\} .\]
	Here we write $(f,\phi)$ for the function from $S$ to $\HH$ which maps $s \in S$ to $(f(s),\phi(s)) \in \HH$.
	We regard $F=(f,\phi)$ as a ``lift'' of an element in $\SS^S$ to an element in $\HH^S$.
	For $f \in \cF$, we denote the set of all lifts of $f$ and the random set of all lifts of $f$ that obey $\cH$ by
	\[ \cG_0(f) := \big\{ F \in \HH^S : F=(f,\phi)\text{ for some }\phi \big\} \quad\text{and}\quad \cG(f) := \cG_0(f) \cap \Hom_{S,\cH} .\]
	Note that
	\[ |\cG_0(f)| = \prod_{v \in S} (\lambda_{f(v)}M) = M^{|S|} \cdot \prod_{v \in S} \lambda_{f(v)} ,\]
	and that, for any fixed $F \in \cG_0(f)$,
	\[ \Pr(F \in \cG) = \Pr\big((F(v),F(u)) \in H_{\{v,u\}}\text{ for all }v \in S^\odd\text{ and }u \sim v\big) = \prod_{v \in S^\odd,~u \sim v} \lambda_{f(v),f(u)} .\]
	Thus,
	\[ \E|\cG(f)| = \sum_{F \in \cG_0(f)} \Pr(F \in \cG) = \omega_f \cdot M^{|S|} ,\]
	from which it follows that $\E|\cG| = \sum_{f \in \cF} \E|\cG(f)| = \omega(\cF) \cdot M^{|S|}$.

	Let $F^\cH=(f^\cH,\phi^\cH)$ be a random element of $\cG$ sampled uniformly and let $\mu_\cH$ be the (random) distribution of $f^\cH$ conditioned on $\cH$.
	Applying \cref{lem:generalized-uniform-shearer-for-bad-set} with $\mathbb{H}_u := \{ (i,m) \in \mathbb{H} : i \in \SS_u \}$, we obtain
	\[ |\cG| \le \prod_{v \in S^\odd} \prod_x \left(\frac{Z^{\cH}_v(\Psi^\cH_{v,x},I^\cH_{v,x})}{\mu_\cH(X_v=x)}\right)^{\mu_\cH(X_v=x)/2d} \cdot \prod_{u \in \intB S} \big|\lambda_{\SS_u} M \big|^{\frac{1}{2d}|\partial u \cap \partial S|} ,\]
	where $\Psi^\cH_{v,x}$ and $I^\cH_{v,x}$ are the conditional supports of $F^\cH|_{N(v)}$ and $F^\cH(v)$ given $\cH$ on the event $\{X_v=x\}$, respectively.
	By the scale-invariance, this can equivalently be written as
	\begin{equation}\label{eq:Z-bound}
	\frac{|\cG|}{M^{|S|}} \le \prod_{v \in S^\odd} \prod_x \left(\frac{Z^{\cH}_v(\Psi^\cH_{v,x},I^\cH_{v,x})}{M^{4d} \cdot \mu_\cH(X_v=x)}\right)^{\mu_\cH(X_v=x)/2d} \cdot \prod_{u \in \intB S} |\lambda_{\SS_u}|^{\frac{1}{2d}|\partial u \cap \partial S|} .
	\end{equation}

	Let $\mu$ denote the distribution on $\cF$ given by the weights~\eqref{eq:config-weight}.
	Let us show that
	\begin{align}
	\distTV(\mu_{\cH},\mu) \Rightarrow 0 \qquad\text{as }M \to \infty , \label{eq:Z-convergence} \\
	\max_{v \in S^\odd} \max_x \frac{Z^{\cH}_v(\Psi^\cH_{v,x},I^\cH_{v,x})}{Z(\Psi_{v,x},A_{v,x})} \cdot M^{-4d} \Rightarrow 1 \qquad\text{as }M \to \infty , \label{eq:Z-comparison}
	\end{align}
	where $\distTV$ denotes total-variation distance and $\Rightarrow$ denotes convergence in distribution.
	Towards proving these claims, let us first show that, for any fixed $\epsilon>0$, with probability tending to 1 as $M \to \infty$, we have
	\begin{equation}\label{eq:Z-degree-event}
	1 - \epsilon \le \frac{|H_{v,\psi} \cap (\{i\} \times \N)|}{\lambda_i M \prod_{u \sim v} \lambda_{i,\psi_0(u)}} \le 1+\epsilon \qquad\text{for all }v \in S^\odd,~\psi=(\psi_0,\psi_1) \in \HH^{N(v)},~i \in \SS ,
	\end{equation}
	where, for notational convenience, we regard ``$0/0$'' to be 1.
	Indeed, for any such $(v,\psi,i)$, since $|H_{v,\psi} \cap (\{i\} \times \N)| \sim$ Bin($\lambda_i M,\prod_{u \sim v} \lambda_{i,\psi_0(u)}$), a standard Chernoff bound yields that
	\[ \Pr\left(\left|\frac{|H_{v,\psi} \cap (\{i\} \times \N)|}{\lambda_i M \prod_{u \sim v} \lambda_{i,\psi_0(u)}}-1\right| > \epsilon \right) \le 2e^{-\delta \lambda_i M} ,\]
	for some $\delta>0$ which depends only on $\epsilon$.
	Hence, by a union bound over the choices of $(v,\psi,i)$, the probability that~\eqref{eq:Z-degree-event} does not occur is at most $|S| (|\SS|M)^{2d} |\SS| \cdot 2e^{-\delta \min_i \lambda_i M}$.

	We proceed to show~\eqref{eq:Z-convergence}. Since $\mu_\cH$ and $\mu$ are supported on the same finite set, \eqref{eq:Z-convergence} is equivalent to the fact that $\mu_{\cH}(f) \Rightarrow \mu(f)$ as $M \to \infty$ for any fixed $f \in \cF$.
	Since $\mu_{\cH}(f) = |\cG(f)|/|\cG|$ and $\mu(f) = \omega(f)/\omega(\cF)$, it suffices to show that, on the event that~\eqref{eq:Z-degree-event} occurs,
	\begin{equation}\label{eq:Z-conv}
	(1-\epsilon) M^{|S|} \cdot \omega(f) \le |\cG(f)| \le (1+\epsilon) M^{|S|} \cdot \omega(f) \qquad\text{for all }f \in \cF .
	\end{equation}
	Recall that $|\cG(f)|$ is the number of $\phi$ such that $(f,\phi) \in \Hom_{S,\cH}$. To count the number of such $\phi$, we first choose $\phi|_{S^\even}$ and then choose $\phi(v)$ for each $v \in S^\odd$, noting that the possible choices for such an odd vertex are not affected by the choices for other odd vertices.
	The number of choices for $\phi|_{S^\even}$ is precisely $M^{|S^\even|} \prod_{u \in S^\even} \lambda_{f(u)}$. Given any such choice of $\phi|_{S^\even}$ and any vertex $v \in S^\odd$, the number of choices for $\phi(v)$ is precisely $|H_{v,(f,\phi)|_{N(v)}} \cap (\{f(v)\} \times \N)|$, which, by~\eqref{eq:Z-degree-event}, is between $(1-\epsilon)\lambda_{f(v)} M \prod_{u \sim v} \lambda_{f(v),f(u)}$ and $(1+\epsilon)\lambda_{f(v)} M \prod_{u \sim v} \lambda_{f(v),f(u)}$. This yields~\eqref{eq:Z-conv}.

	Towards proving~\eqref{eq:Z-comparison}, let $v \in S^\odd$ and let $x$ be in the support of $X_v$. Note that, since $X_v$ is measurable with respect to $f^\cH|_{N(v)}$,
	\[ \Psi^\cH_{v,x} = \big\{ (\psi,\phi) \in \mathbb{H}^{N(v)} : \psi \in \Psi_{v,x} \big\} \qquad\text{and}\qquad I^\cH_{v,x} := \big\{ (i,m) \in \mathbb{H} : i \in I_{v,x} \big\} .\]
	Thus, when~\eqref{eq:Z-degree-event} holds, we have
	\[ Z_v^\cH(\Psi^\cH_{v,x},I^\cH_{v,x}) = \sum_{\psi \in \Psi^\cH_{v,x}} |H_{v,\psi} \cap I^\cH_{v,x}|^{2d} \le (1+\epsilon) M^{4d} \sum_{\psi \in \Psi} \left(\prod_{u \sim v} \lambda_{\psi(u)}\right) \left( \sum_{i \in I} \lambda_i \prod_{u \sim v} \lambda_{i,\psi(u)} \right)^{2d} .\]
	Together with the analogous lower bound and recalling~\eqref{eq:Z-Psi-I-def}, we obtain
	\[ (1-\epsilon) M^{4d} \cdot Z(\Psi_{v,x},I_{v,x}) \le Z_v^\cH(\Psi^\cH_{v,x},I^\cH_{v,x}) \le (1+\epsilon) M^{4d} \cdot Z(\Psi_{v,x},I_{v,x}) .\]
	Since, for any fixed $\epsilon>0$, this holds with probability tending to 1 as $M \to \infty$, we obtain~\eqref{eq:Z-comparison}.

	Finally, we conclude from~\eqref{eq:Z-bound}, \eqref{eq:Z-convergence} and~\eqref{eq:Z-comparison} that, for any fixed $\epsilon>0$, with probability tending to 1 as $M \to \infty$,
	\[ \frac{|\cG|}{M^{|S|}} \le (1+\epsilon) \prod_{v \in S^\odd} \prod_x \left(\frac{Z(\Psi_{v,x},I_{v,x})}{\Pr(X_v=x)}\right)^{\Pr(X_v=x)/2d} \cdot \prod_{u \in \intB S} |\SS_u|^{\frac{1}{2d}|\partial u \cap \partial S|} .\]
	Noting that $|\cG| \cdot M^{-|S|} \le |\SS|^{|S|}$ is uniformly bounded and recalling that $\E|\cG| \cdot M^{-|S|} = \omega(\cF)$, the desired inequality follows by taking expectation.
\end{proof}

\section{Approximations}
\label{sec:approx}

In this section, we prove \cref{prop:family-of-odd-approx}.
That is, we show that there exists a small family of approximations which contains an approximation of every atlas in $\breakups_{L,M,N}$ that is seen from a given set.
Conceptually, our goal is to obtain an approximation of the following four objects:
\[ (X_P)_{P \in \phase_0}, \qquad (X_P)_{P \in \phase_1}, \qquad (X'_P)_{P \in \phase_0}, \qquad (X'_P)_{P \in \phase_1} .\]
In practice, however, there is an important difference between how the first two and last two objects are approximated. While the first two are approximated as collections (see below), the last two are only approximated by their (regularized) unions, namely,
\[ X'_\even := \bigcup_{P \in \phase_0} X'_P \cup N_{2d}\Big(\bigcup_{P \in \phase_0} X'_P\Big), \qquad X'_\odd := \bigcup_{P \in \phase_1} X'_P \cup N_{2d}\Big(\bigcup_{P \in \phase_1} X'_P\Big) .\]
Note that $X'_\even$ and $X'_\odd$ are regular even and odd sets, respectively.

The construction of the family of approximations is done in two steps, the first of which is to construct a small family of small sets which contains a tightly separating set of every atlas, where we say that a set $W$ \emph{separates} an atlas $X$ if every edge in $\bigcup_P \partial X_P \cup \partial X'_\even \cup \partial X'_\odd$ has an endpoint in $W$, and that it \emph{tightly separates} $X$ if also $W \subset X_*^{+2}$.

\begin{lemma}\label{lem:family-of-separating-sets}
	For any integers $d \ge 2$ and $L,M,N \ge 0$ and any finite set $V \subset \Z^d$, there exists a family $\cW$ of subsets of $\Z^d$, each of size at most $C(L+dM) (\log d)/\sqrt{d}$,
	such that
	\[ |\cW| \le 2^{|V|} \cdot \exp\Big(\tfrac{C(L+dM) \log^2 d}{d^{3/2}} + \tfrac{CN \log^2 d}{d} \Big) \]
	and any atlas $X \in \breakups_{L,M,N}$ seen from $V$ is tightly separated by some set in $\cW$.
\end{lemma}

The definition of an atlas does not require any relation between $X_P$ for different $P$. In particular, the set of $P$ for which a given vertex belongs to $X_P$ could be any subset of the dominant patterns. Since there are exponentially in $|\phasedom|$ many such subsets (which is potentially exponentially large in $\fq$), this would not lead to the correct dependency on $\fq$. In light of this, we require an additional property of atlases, satisfied by any breakup, namely, \eqref{eq:possible-patterns-a-vertex-can-be-in}.
In order to keep this section as independent as possible, we introduce some abstract definitions. For simplicity of writing, we fix the parity of the sets we work with here to be odd, even sets being completely analogous.

Let $S=(S_i)_i$ be a collection of regular odd sets (we do not explicitly specify the index set as it has no significance in what follows). A \emph{rule} is a family $\cQ$ of subsets of indices. We say that such a rule $\cQ$ has \emph{rank at most $\fq$} if $|\cQ| \le 2^{\fq}$.
We say that $S$ is an odd \emph{$\cQ$-collection} if it obeys the rule $\cQ$ in the following sense:
\[ \{ i : v \in S_i \} \in \cQ \qquad\text{for any even vertex }v .\]
An \emph{approximation} of $S$ is a collection $A=((A_i)_i,A_*)$ such that $A_i \subset S_i \subset A_i \cup A_*$ and $A_i$ is odd for all $i$ and such that $\Even \cap A_* \subset N_d(\bigcup_i A_i)$.
We say that $A$ is \emph{controlled by} a set $W$ if $|A_*| \le C|W|$ and $A_* \subset W^+$, and that $W$ \emph{separates} $S$ if every edge in $\bigcup_i \partial S_i$ has an endpoint in $W$.

\begin{lemma}[{\cite[Lemma~7.2]{peledspinka2018colorings}}]\label{lem:family-of-approx-from-separating-set}
	For any integers $d \ge 2$ and $\fq \ge 1$, any rule $\cQ$ of rank at most $\fq$ and any finite set $W \subset \Z^d$, there exists a family $\cA$ of approximations, each of which is controlled by~$W$, such that
	\[ |\cA| \le \exp\Big(\tfrac{C|W|(\fq+\log d)}{d} \Big) \]
	and any odd $\cQ$-collection which is separated by $W$ is approximated by some element in $\cA$.
\end{lemma}

Before proving \cref{lem:family-of-separating-sets}, let us show how the above two lemmas yield \cref{prop:family-of-odd-approx}.

\begin{proof}[Proof of \cref{prop:family-of-odd-approx}]
	Applying \cref{lem:family-of-separating-sets}, we obtain a family $\cW$ of subsets of $\Z^d$, each of size at most $r := C(L+dM) (\log d) / \sqrt{d}$, such that every $X \in \breakups_{L,M,N}$ seen from $V$ is separated by some set in $\cW$.
	By~\eqref{eq:def-atlas} and~\eqref{eq:possible-patterns-a-vertex-can-be-in}, there exists a rule $\cQ$ of rank at most $\fq$ such that $(X_P)_{P \in \phase_1}$ is an odd $\cQ$-collection for any $X \in \breakups$.
	Now, for each $W \in \cW$, we apply \cref{lem:family-of-approx-from-separating-set} to obtain a family $\cA^1_W$ of approximations, each of which is controlled by $W$, such that $|\cA^1_W| \le \exp( Cr(\fq+\log d)/d)$ and satisfying that any odd $\cQ$-collection which is separated by $W$ is approximated by some element in $\cA^1_W$.
	Similarly, applying \cref{lem:family-of-approx-from-separating-set} for the trivial rank 1 rule, we obtain a family $\cA^{'1}_W$ of approximations, each of size at most $Cr$, such that $|\cA^{'1}_W| \le \exp( Cr(\fq+\log d)/d)$ and satisfying that any regular odd set which is separated by $W$ is approximated by some element in $\cA^{'1}_W$.
	Reversing the roles of even and odd, we also obtain families $\cA^0_W$ and $\cA^{'0}_W$ in a similar manner.
	Finally, define $\cA := \bigcup_{W \in \cW} \phi(\cA_W)$, where $\cA_W := \cA^0_W \times \cA^1_W \times \cA^{'0}_W \times \cA^{'1}_W$ and where
	\[ \phi(A,B,C,D) := \big((A_P)_{P \in \phase_0} \cup (B_P)_{P \in \phase_1},C_0 \cup D_0, (\Odd \cap A_*) \cup (\Even \cap B_*),A_* \cup B_* \cup C_* \cup D_*\big) .\]
	It is straightforward to verify that $\cA$ satisfies the requirements of the lemma.
\end{proof}

For the proof of \cref{lem:family-of-separating-sets}, we require two lemmas from~\cite{peledspinka2018colorings}, the first of which states that for every collection $S=(S_i)_i$ of regular odd sets, there exists a small set $U$ such that $N(U)$ separates $S$.
For such a collection, denote $\partial S := \bigcup_i \partial S_i$ and $\intextB S := \bigcup_i \intextB S_i$.

\begin{lemma}[{\cite[Lemma~7.3]{peledspinka2018colorings}}]\label{lem:existence-of-U}
	Let $S=(S_i)_i$ be a collection of regular odd sets.
	Then there exists $U \subset (\intextB S)^+$ of size at most $|\partial S| \cdot Cd^{-3/2}\log d$ such that $N(U)$ separates $S$.
\end{lemma}

\begin{lemma}[{\cite[Lemma~7.4]{peledspinka2018colorings}}]\label{lem:number-of-disconnecting-sets}
The number of sets $U \subset \Z^d$ of size at most $n$ such that $U^{+10}$ is connected and disconnects the origin from infinity is at most $\exp(Cn \log d)$.
\end{lemma}

\begin{proof}[Proof of \cref{lem:family-of-separating-sets}]
	Let $L,M,N \ge 0$ be integers and let $V \subset \Z^d$ be finite.
	Let $\cU$ be the collection of all subsets $U$ of $\Z^d$ of size at most
	\[ r := C(L+dM) d^{-3/2}\log d + CN d^{-1} \log d \]
	such that every connected component of $U^{+7}$ disconnects some vertex $v \in V$ from infinity. Define
	\[ \cW := \big\{ N(U') : U \in \cU,~ U' \subset U,~|U'| \le C(L+dM) d^{-3/2} \log d \big\} .\]
	Let us show that $\cW$ satisfies the requirements of the lemma. Note first that every $W \in \cW$ has $|W| \le C(L+dM) d^{-1/2} \log d$. Next, to bound the size of $\cW$, observe that $|\cW| \le |\cU| \cdot 2^r$. Consider a set $U \in \cU$ and let $\{U_l\}_{l=1}^n$ be the connected components of $U^{+7}$ and denote $r_l := |U \cap U_l|$. For each~$l$, choose a vertex $v_l \in V$ such that $U_l$ disconnects $v_l$ from infinity. There are at most $2^{|V|}$ choices for $\{ v_l \}_{l=1}^n$, and given such a choice, there are then at most $\binom{r+n}{n} \le 4^r$ choices for $(v_l,r_l)_l$.
	Thus, \cref{lem:number-of-disconnecting-sets} implies that
	\[ |\cU| \le 2^{|V|} \cdot 4^r \cdot \exp(Cr \log d) \le 2^{|V|} \cdot \exp\left(\tfrac{C(L+dM) \log^2 d}{d^{3/2}} + \tfrac{CN \log^2 d}{d} \right) .  \]

	It remains to show that any $X \in \breakups_{L,M,N}$ seen from $V$ is tightly separated by some set in $\cW$. Let $X$ be such an atlas and denote $S^0 := (X_P)_{P \in \phase_0}$, $S^1 := (X_P)_{P \in \phase_1}$, $S^2 := X'_\even$ and $S^3 := X'_\odd$ and $L^i := |\partial S^i|$ for $i \in \{0,1,2,3\}$.
	By \cref{lem:existence-of-U}, there exists a set $U^i \subset (\intextB S^i)^+ \subset X_*^+$ such that $|U^i| \le CL^i d^{-3/2} \log d$ and $N(U^i)$ separates $S^i$. Denote $U' := U^0 \cup U^1 \cup U^2 \cup U^3$ and note that $|U'| \le C(L+dM) d^{-3/2} \log d$ and $N(U')$ tightly separates $X$. Hence, to obtain that $N(U') \in \cW$ and thus conclude the proof, it remains to show that $U'$ is contained in some $U \in \cU$.

	By \cref{lem:existence-of-covering2}, there exists $U'' \subset X_\bad \cup X_\overlap \cup X_\hole$ such that $|U''| \le C(M+N)d^{-1}\log d$ and $N_{2d}(X_\bad \cup X_\overlap \cup X_\hole) \subset N(U'')$.
	Denote $U := U' \cup U''$ and note that $X_* \subset U^{++}$, $U \subset X_*^+$ and $|U| \le r$. In particular, every connected component of $U^{+7}$ disconnects some vertex $v \in V$ from infinity so that $U \in \cU$.
\end{proof}

\section{The model on a complete bipartite graph}
\label{sec:model_on_K2d2d}

As we have explained, the quantitative results given in \cref{sec:quantitative-results} hold under either one of the explicit conditions~\eqref{eq:parameter-inequalities-simple}, \eqref{eq:parameter-inequalities-simple1}, \eqref{eq:parameter-inequalities-simple2}, \eqref{eq:parameter-inequalities-simple3}, or alternatively, under the more abstract \cref{main-cond} pertaining to the behavior of the spin system on the complete bipartite graph $K_{2d,2d}$. In this section, we show that the former conditions are special cases of the latter. We have already explained in \cref{sec:alternative-conditions} that~\eqref{eq:parameter-inequalities-simple} is a special case of~\eqref{eq:parameter-inequalities-simple1}, and that, for homomorphism models, \eqref{eq:parameter-inequalities-simple1} and~\eqref{eq:parameter-inequalities-simple2} are equivalent and are special cases of~\eqref{eq:parameter-inequalities-simple3}.

Let us show that~\eqref{eq:parameter-inequalities-simple1} is a special case of~\eqref{eq:parameter-inequalities-simple2} for non-homomorphism models. Suppose that~\eqref{eq:parameter-inequalities-simple1} holds and set $s := \left\lceil 4\log (d\rho_{\text{act}})/(-\log \rho_{\text{int}}) \right\rceil$.
Note that $\rho_{\text{int}}^s \rho_{\text{act}} \le \frac1d$ and that~\eqref{eq:parameter-inequalities-simple1} implies that $(2d\rho_{\text{act}})^{(s-1)|\SS|} \le e^{\alpha_1 d/5}$, so that $\hat\rho_{\text{pat}}^{\text{bulk}} \le \rho_{\text{pat}}^{\text{bulk}} (1+\frac1d)e^{\alpha_1/10}$. Thus, $\alpha_2 \ge \alpha_1 - \log(1+\frac1d) - \frac{\alpha_1}{10} \ge \frac{9\alpha_1}{10} - \frac1d$, so that~\eqref{eq:s-ineq} and~\eqref{eq:parameter-inequalities-simple2} follow from~\eqref{eq:parameter-inequalities-simple1} using also that $\hat\rho_{\text{act}} \le \rho_{\text{act}}^2$.

It remains to show that for homomorphism models, \eqref{eq:parameter-inequalities-simple3} is a special case of \cref{main-cond}, and that for non-homomorphism models, \eqref{eq:parameter-inequalities-simple2} is a special case of \cref{main-cond}.
We will handle both cases in parallel, showing that \cref{main-cond} is satisfied with the parameters $\alpha$, $\gamma$, $\epsilon$ and $\bar\epsilon$ chosen as follows.

\noindent{\bf Homomorphism case}: Assume a homomorphism model and~\eqref{eq:parameter-inequalities-simple3}, and set
\[ \alpha := \alpha_3, \qquad \epsilon := \min \left\{\frac{\alpha}{64 \log d}, \frac18 \right\} , \qquad \bar\epsilon := \frac{1}{4d}, \qquad \gamma := 0 .\]

\noindent{\bf Non-homomorphism case}: Assume that~\eqref{eq:parameter-inequalities-simple2} holds with some $s$ as in~\eqref{eq:s-ineq} and set
\[ \alpha := \alpha_2, \qquad \epsilon := \min \left\{\frac{\alpha}{64 \log d}, \frac18 \right\} , \qquad \bar\epsilon := \max\left\{ \frac{s}{4d}, \frac{\alpha\epsilon}{-\log \rho_{\text{int}}} \right\}, \qquad \gamma := \rho_{\text{act}} \cdot \rho_{\text{int}}^s .\]
Formally, the analysis in this case is valid also for homomorphism models, where it is then understood that $s=1$ so that $\bar\epsilon=\frac1{4d}$ and $\gamma=0$ as before.

\begin{lemma}\label{lem:main-cond-1}
	In both cases we have~\eqref{eq:alpha-cond},
\[ \tfrac{1}{4d} \le \bar\epsilon \le \epsilon \le \tfrac18 \qquad\text{and}\qquad 2^{\fq+1} (\tfrac e{2\epsilon})^{4\epsilon d} (\rho_{\text{pat}}^{\text{bdry}})^{2d-4\epsilon d} \le \tfrac14 e^{-\alpha d} .\]
Furthermore, in the non-homomorphism case,
\[ |\phasemax| \cdot (\tfrac e{2\bar\epsilon})^{4\bar\epsilon d} (\rho_{\text{pat}}^{\text{bdry}})^{2d-4\bar\epsilon d} \le \tfrac14 e^{-\alpha d} .\]
\end{lemma}

\begin{lemma}\label{lem:main-cond-2}
	In both cases we have~\eqref{eq:cond-restricted-left}-\eqref{eq:cond-non-dominant}.
\end{lemma}

\begin{proof}[Proof of \cref{lem:main-cond-1}]
We make use of~\eqref{eq:parameter-inequalities-simple3} or~\eqref{eq:parameter-inequalities-simple2} and~\eqref{eq:s-ineq} throughout the proof.

To establish~\eqref{eq:alpha-cond}, it suffices to show that $c\alpha$ is greater than each term on the right-hand side of~\eqref{eq:alpha-cond} separately. Using that $\alpha \ge C d^{-1/4} \fq (\fq + \log d)\log d$, this is immediate for the first term, and since $\gamma \le d^{-2}$, it also easily follows for the last two terms. It remains to check that $c\alpha \ge \frac{(\fq+\log d)\log d}{\epsilon^2 d}$.
Note that either $\epsilon=\frac18$ or $\epsilon \ge \frac{\alpha}{64\log d}$. It is straightforward to verify that the inequality holds in either case.

For the inequality $\frac1{4d} \le \bar\epsilon \le \epsilon \le \frac18$, it suffices to check that $\bar\epsilon \le \epsilon$ as the other two inequalities are immediate. Since $\alpha \le -\log \rho_{\text{int}}$ by the definition of $\alpha$, it remains to check that $s \le 4\epsilon d$. Since $4\epsilon d \ge 1$ by the assumption on $\alpha$, this follows from the assumption on $\rho_{\text{int}}$.

For the inequality $2^{\fq+1} (\tfrac e{2\epsilon})^{4\epsilon d} (\rho_{\text{pat}}^{\text{bdry}})^{2d-4\epsilon d} \le \frac14 e^{-\alpha d}$, since $\frac1{4d} \le \epsilon \le \frac18$ and $\rho_{\text{pat}}^{\text{bdry}} \le e^{-\alpha}$, it suffices to show that $2^{\fq+1} (2ed)^{4\epsilon d} \le \frac14 e^{\alpha d/2}$. Since $2^{\fq+1} \le \frac14 e^{\alpha d/4}$, it suffices that $16\epsilon \log (2ed) \le \alpha$. Since $2ed \le d^4$, this follows from the definition of $\epsilon$.

For the inequality $|\phasemax| (\tfrac e{2\bar\epsilon})^{4\bar\epsilon d} (\rho_{\text{pat}}^{\text{bdry}})^{2d-4\bar\epsilon d} \le \frac14 e^{-\alpha d}$, since $\frac1{4d} \le \bar\epsilon \le \frac18$ and $|\phasemax|(\rho_{\text{pat}}^{\text{bdry}})^{3d/2} \le e^{-3\alpha d/2}$, it suffices to show that $4(2ed)^{4\bar\epsilon d} \le e^{\alpha d/2}$. Since $4 \le e^{\alpha d/4}$ and $\bar\epsilon \le \epsilon$, this follows as before.
\end{proof}

\begin{proof}[Proof of \cref{lem:main-cond-2}]
In the proof, we regard a function $\psi \colon [2d] \to \SS$ as specifying the values on the \emph{left side} of $K_{2d,2d}$ and we speak about possible values appearing on the \emph{right side} given such a~$\psi$. In homomorphism models, if a value $i \in \SS$ appears of the left side, then only values in $N(i)$ can appear on the right side. Thus, if $\psi$ appears on the left side, then the only possible values on the right side are those in $R(\psi([2d]))$. In non-homomorphism models, on the other hand, such hard constraints do not apply and more care is needed. This makes the computations for homomorphism models much simpler than in the general case. We therefore begin by explaining the simpler situation.

\medskip
\noindent\textbf{The homomorphism case.}
We begin with the simple observation that $Z(\Psi_J,\SS) \le (\lambda_J\lambda_{R(J)})^{2d}$ for any $R$-set $J$. More generally, if $\Psi \subset \Psi_J$ and $I \subset \SS$, then $Z(\Psi,I) = \omega(\Psi) \cdot (\lambda_{I \cap R(J)})^{2d}$, where $\omega(\Psi) = \sum_{\psi \in \Psi} \prod_{j \in [2d]} \lambda_{\psi(j)}$. This is the basis for all bounds.

Let $J$ be a side of a dominant pattern.
To see~\eqref{eq:cond-restricted-left}, let $\Psi \subset \Psi_J$, denote $J_j := \{ \psi(j) : \psi \in \Psi \}$, and note that
\begin{align*}
Z(\Psi,\SS)
 	\le \prod_{j \in [2d]} \lambda_{J_j} \cdot (\lambda_{R(J)})^{2d}
	&= (\lambda_J\lambda_{R(J)})^{2d} \cdot \prod_{j \in [2d]} \left(\tfrac{\lambda_{J_j}}{\lambda_J}\right) 	\\
	&\le \omega_{\text{dom}}^{2d} (\rho_{\text{pat}}^{\text{bdry}})^{k_\Psi} \le \omega_{\text{dom}}^{2d}  e^{-\alpha k_\Psi} .
\end{align*}
To see~\eqref{eq:cond-restricted-right}, let $I \subset \SS$ be such that $R(J) \not\subset R(R(I))$ and note that
\begin{align*}
Z(\Psi_{J,\epsilon,\bar\epsilon},I) \le Z(\Psi_J,I) &\le (\lambda_J\lambda_{I \cap R(J)})^{2d} \\&= \omega_{\text{dom}}^{2d} \left(\tfrac{\lambda_{I \cap R(J)}}{\lambda_{R(J)}}\right)^{2d} \le \omega_{\text{dom}}^{2d} (\rho_{\text{pat}}^{\text{bdry}})^{2d} \le \omega_{\text{dom}}^{2d}  e^{-2\alpha d} .
\end{align*}
To see~\eqref{eq:cond-unbalanced}, note first that
	\begin{equation}\label{eq:Z-bound-unbalanced}
	\begin{aligned}
	 \omega(\Psi_J \setminus \Psi_{J,\epsilon,\bar\epsilon}) = \omega(\Psi^1_{J,\epsilon})
	 	&\le \sum_{\substack{A \subsetneq J\\\text{side of dom pat}}} \sum_{k=1}^{\lceil 4\epsilon d \rceil - 1} \binom{2d}{k} \lambda_A^{2d-k} \lambda_{J \setminus A}^k \\
	 	&\le (\lambda_J)^{2d} \cdot 2^{\fq+1} (\tfrac e{2\epsilon})^{4\epsilon d} \max_{\substack{A \subsetneq J\\\text{$R$-set}}} \, (\tfrac {\lambda_A}{\lambda_J})^{2d-4\epsilon d} \\
	 	&\le (\lambda_J)^{2d} \cdot 2^{\fq+1} (\tfrac e{2\epsilon})^{4\epsilon d} (\rho_{\text{pat}}^{\text{bdry}})^{2d-4\epsilon d} \le (\lambda_J)^{2d} \cdot \tfrac14 e^{-\alpha d} ,
	\end{aligned}
	\end{equation}
	where we used \cref{lem:main-cond-1} in the last inequality (the $\frac14$ is not important here).
	Thus,
	\[ Z(\Psi_J \setminus \Psi_{J,\epsilon,\bar\epsilon},\SS) = \omega(\Psi_J \setminus \Psi_{J,\epsilon,\bar\epsilon}) \cdot (\lambda_{R(J)})^{2d} \le \omega_{\text{dom}}^{2d} \cdot e^{-\alpha d} .\]
Note that~\eqref{eq:cond-highly-energetic} is trivial since $Z(\Psi_J,\SS \setminus R(J))=0$.
Finally, \eqref{eq:cond-non-dominant} follows from
\[ \sum_{\substack{I \subset \SS\text{ side of maximal}\\\text{non-dominant pattern}}} Z(\Psi_I,\SS) = \sum_{\substack{I \subset \SS\text{ side of maximal}\\\text{non-dominant pattern}}} \lambda_I^{\langle 2d \rangle}\lambda_{R(I)}^{2d} = \omega_{\text{dom}}^{2d} \cdot (\rho_{\text{pat}}^{\text{bulk*}})^{2d} \le \omega_{\text{dom}}^{2d} \cdot e^{-2\alpha d} .\]

\medskip
\noindent\textbf{The non-homomorphism case.}
	Before diving into the proofs of the bounds as we did in the homomorphism case, we need to lay some groundwork.
	Fix an $R$-set $J$, a collection $\Psi \subset \Psi_J$ and a set $I \subset \SS$.
	Instead of working directly with functions $\psi \colon [2d] \to \SS$, it is convenient to only prescribe the number of times each value appears in the image of $\psi$.
	To this end, for $\xi \colon J \to \{0,1,2,\dots\}$, let $\Psi_\xi$ be the set of $\psi \in \Psi$ such that $|\psi^{-1}(j)|=\xi(j)$ for all $j \in J$.
	Say that $\xi$ is \emph{legal} if $\sum_{j \in J} \xi(j) = 2d$ and
	\[ \sum_{j \in J \setminus J'} \xi(j) > 0 \qquad\text{for any $J' \subset J$ such that $J' \not\simeq_R J$} .\]
	Note that $\Psi$ is the union of $\Psi_\xi$ over legal $\xi$.
	Thus, recalling~\eqref{eq:Z-Psi-I-def}, we may rewrite $Z(\Psi,I)$ as
	\begin{equation}\label{eq:Z-alternate-def}
	Z(\Psi,I) = \sum_{\xi\text{ legal}} |\Psi_\xi| \cdot Z^0_\xi \cdot (Z^1_\xi)^{2d} ,
	\end{equation}
	where
	\[ Z^0_\xi := \prod_{j \in J} (\lambda_j)^{\xi(j)} \qquad\text{and}\qquad Z^1_\xi := \sum_{i \in I} \lambda_i \prod_{j \in J} (\lambda_{i,j})^{\xi(j)} .\]
	Note that $Z^0_\xi$ is the contribution from all the values on the left side, while $Z^1_\xi$ is the contribution from a single value on the right side (taking into account its interaction with all the values on the left side).
	Next, we aim to divide the sum over legal $\xi$ into certain classes of $\xi$ on which we have control on $Z^0_\xi$ and $Z^1_\xi$.

To this end, for a legal $\xi$, define
	\[ S_\xi := R(S'_\xi), \qquad\text{where }S'_\xi := \Big\{ i \in \SS : \sum_{j \in J \setminus N(i)} \xi(j) < s \Big\} .\]
	Observe that $S'_\xi$ is the set of values that interact with maximum interaction weight with all but at most $s-1$ elements on the left side (i.e., in the image of any $\psi \in \Psi_\xi$), and, in particular, $R(J) \subset S'_\xi$. Thus, we may regard $S'_\xi$ as those values which have a ``good chance'' of appearing on the right side (at least in terms of their energetic cost). With this in mind, we may think of $S_\xi$ as the intersection of all $R$-subsets of $J$ which have a significant presence on the left side.
	Note that
	\begin{equation}\label{eq:Z-bound-S-supset}
	\{ j \in J : \xi(j) \ge s \} \subset S_\xi \subset J .
	\end{equation}
	Furthermore, we have that
	\begin{equation}\label{eq:Z-bound-when-z=0}
	S_\xi \neq J \quad\implies\quad \Psi_\xi \subset \Psi^2_{J,\bar\epsilon} .
	\end{equation}
	Indeed, $S_\xi \neq J$ implies that there exists $i \in S'_\xi \setminus R(J)$. In particular, $1 \le \sum_{j \in J \setminus N(i)} \xi(j) < s$ (since $i \notin R(J)$ implies that $J \not\subset N(i)$ and then $J'=J \cap N(i)$ is an $R$-set distinct from $J$), which implies that $\Psi_\xi \subset \Psi^2_{J,\bar\epsilon}$ since $s \le 4\bar\epsilon d$.

	For any $S \subset J$, let $\Xi_S$ be the set of all legal $\xi$ having $S_\xi = S$, so that
	\[ Z(\Psi,I) = \sum_{S \subset J} \Sigma_S ,\qquad\text{where}\quad \Sigma_S := \sum_{\xi \in \Xi_S} |\Psi_\xi| \cdot Z^0_\xi \cdot (Z^1_\xi)^{2d} .\]
	By definition, we have $\Sigma_S = 0$ whenever $S$ is not an $R$-set.
	Since $s \le \lceil2d/|\SS|\rceil$ by~\eqref{eq:s-ineq}, we have that $\Xi_\emptyset = \emptyset$ and $\Sigma_\emptyset = 0$ by~\eqref{eq:Z-bound-S-supset}.
	We proceed to bound $\Sigma_S$ for any non-empty $R$-set $S \subset J$. To this end, we first break the problem into two separate parts, concerning the values on the left and right sides. We do so by using the simple bound
	\begin{equation}\label{eq:Z-bounds-simple}
	 \Sigma_S \le \left( \sum_{\xi \in \Xi_S} |\Psi_\xi| \cdot Z^0_\xi \right) \cdot \left(\max_{\xi \in \Xi_S} Z^1_\xi \right)^{2d} .
	\end{equation}

Let us begin by considering the case $S=J$.
	In this case, for the term involving the values on the right side, every value outside of $R(J)$ comes at a large interaction cost, so that
	\[ \max_{\xi \in \Xi_J} Z^1_\xi \le \lambda_{I \cap R(J)} + \rho_{\text{int}}^s \lambda_{I \setminus R(J)} .\]
	Indeed, given $\xi \in \Xi_J$, any choice of $i \in R(J)$ contributes $\lambda_i$ to the sum in $Z^1_\xi$ as it yields a product over $j \in J$ which equals 1 (since $\lambda_{i,j}=1$ for all $j \in J$), and any choice of $i \notin R(J) = S'_\xi$ contributes at most $\rho_{\text{int}}^s \lambda_i$ as it yields a product which is at most $\rho_{\text{int}}^s$ (since $\lambda_{i,j}<1$ for all $j \in J \setminus N(i)$ and $\sum_{j \in J \setminus N(i)} \xi(j) \ge s$). When $R(J) \neq \emptyset$ (which holds whenever $J \neq \SS$), we may rewrite this last expression to obtain
	\[ \max_{\xi \in \Xi_J} Z^1_\xi \le \lambda_{R(J)} \left(
	1 - \tfrac{\lambda_{R(J) \setminus I}}{\lambda_{R(J)}} + \rho_{\text{int}}^s \tfrac{\lambda_{I \setminus R(J)}}{\lambda_{R(J)}} \right).\]
	For the term involving the values on the left side, we always have the trivial bound
	\[ \sum_{\xi \in \Xi_J} |\Psi_\xi| \cdot Z^0_\xi \le (\lambda_J)^{2d} . \]
	Plugging these bounds into~\eqref{eq:Z-bounds-simple}, we obtain that
	\begin{equation}\label{eq:Z-bounds-J}
	\Sigma_J \le (\lambda_J \lambda_{R(J)})^{2d} \left(
	1 - \tfrac{\lambda_{R(J) \setminus I}}{\lambda_{R(J)}} + \rho_{\text{int}}^s \tfrac{\lambda_{I \setminus R(J)}}{\lambda_{R(J)}} \right)^{2d} .
	\end{equation}
	When $J$ is a side of a dominant pattern, $\lambda_J \lambda_{R(J)} = \omega_{\text{dom}}$ and $\frac{\lambda_{I \setminus R(J)}}{\lambda_{R(J)}} \le \hat\rho_{\text{act}}$ by~\eqref{eq:rho-bulk-hat-and-rho-act-hat-def}, so that
\[ \Sigma_J \le \omega_{\text{dom}}^{2d} \left(
	1 - \tfrac{\lambda_{R(J) \setminus I}}{\lambda_{R(J)}} + \gamma \right)^{2d} .\]
Using this bound and some adaptations of it, we are already able to deduce~\eqref{eq:cond-restricted-left}, \eqref{eq:cond-restricted-right}, \eqref{eq:cond-highly-energetic} and half of~\eqref{eq:cond-unbalanced}, in a similar way as in the homomorphism case. For this part of the proof, we suppose that $J$ is a side of a dominant pattern.
 Note that, by~\eqref{eq:Z-bound-when-z=0}, we have that $Z(\Psi,I)=\Sigma_J$ whenever $\Psi$ is disjoint from $\Psi^2_{J,\bar\epsilon}$.

	\smallskip
	\noindent\textbf{Bound \eqref{eq:cond-restricted-left}:}
	Suppose that $\Psi \subset \Psi_{J,\epsilon,\bar\epsilon}$ and $I=\SS$. Since $Z(\Psi,I) = \Sigma_J \le \omega_{\text{dom}}^{2d} \cdot (1 + \gamma)^{2d}$, it remains only to obtain the additional factor $e^{-\alpha k_\Psi}$. Recall that $J_j := \{ \psi(j) : \psi \in \Psi \} \not\simeq_R J$ for $k_\Psi$ many indices $j \in [2d]$. Thus, the trivial bound on $\sum_{\xi \in \Xi_J} |\Psi_\xi| \cdot Z^0_\xi$ can be improved to
	\[ \sum_{\xi \in \Xi_J} |\Psi_\xi| \cdot Z^0_\xi \le \prod_{j \in [2d]} \lambda_{J_j} \le (\lambda_J)^{2d} \cdot (\rho_{\text{pat}}^{\text{bdry}})^{k_\Psi} \le (\lambda_J)^{2d} \cdot e^{-\alpha k_\Psi} ,\]
	which leads to
	\[ Z(\Psi,I) =\Sigma_J \le \omega_{\text{dom}}^{2d} \cdot (1 + \gamma)^{2d} \cdot e^{-\alpha k_\Psi} .\]

	\smallskip
	\noindent\textbf{Bound \eqref{eq:cond-restricted-right}:}
	Suppose that $\Psi = \Psi_{J,\epsilon,\bar\epsilon}$ and that $R(J) \not\subset R(R(I))$. Then
	\[ Z(\Psi,I) = \Sigma_J \le \omega_{\text{dom}}^{2d} \left(1- \tfrac{\lambda_{R(J) \setminus I}}{\lambda_{R(J)}} + \gamma \right)^{2d} \le \omega_{\text{dom}}^{2d} \cdot \big(\rho_{\text{pat}}^{\text{bdry}} + \gamma \big)^{2d} \le \omega_{\text{dom}}^{2d} \cdot e^{2\gamma d-2\alpha d} ,\]
	where in the last inequality we used that $1+\gamma \le e^{\gamma}$ and $\gamma \le \rho_{\text{int}}^{s/2} \le \sqrt{\rho_{\text{int}}}$ to deduce that $(\rho_{\text{pat}}^{\text{bdry}} + \gamma)e^{-\gamma} \le 1-(1-\rho_{\text{pat}}^{\text{bdry}})/(1+\gamma) \le 1-(1-\rho_{\text{pat}}^{\text{bdry}})(1-\sqrt{\rho_{\text{int}}}) \le e^{-\alpha}$.

	\smallskip
	\noindent\textbf{Half of bound \eqref{eq:cond-unbalanced}:}
	Suppose that $\Psi = \Psi^1_{J,\epsilon} \cup \Psi^2_{J,\bar\epsilon}$ and $I=\SS$.
	Since $\Psi$ is not disjoint from $\Psi^2_{J,\bar\epsilon}$, it no longer holds that $Z(\Psi,\SS)=\Sigma_J$. Here we only bound the contribution of $\Sigma_J$ to $Z(\Psi,\SS)$ and return later to the contributions of $\Sigma_S$ for $S \neq J$. We further separate between $\Psi=\Psi^1_{J,\epsilon}$ and $\Psi=\Psi^2_{J,\bar\epsilon}$, and denote by $\Sigma^1_J$ and $\Sigma^2_J$ the corresponding values, so that $\Sigma_J \le \Sigma^1_J + \Sigma^2_J$.
	For $\Psi=\Psi^1_{J,\epsilon}$, the same computation as in~\eqref{eq:Z-bound-unbalanced} yields that the trivial bound on $\sum_{\xi \in \Xi_J} |\Psi_\xi| \cdot Z^0_\xi$ may be improved in this case to
	\[ \sum_{\xi \in \Xi_J} |\Psi_\xi| \cdot Z^0_\xi \le (\lambda_J)^{2d} \cdot \tfrac14 e^{-\alpha d} ,\]
	which leads to the bound
	\[  \Sigma^1_J \le \omega_{\text{dom}}^{2d} \cdot (1 + \gamma)^{2d} \cdot \tfrac14 e^{-\alpha d} .\]
	For $\Psi=\Psi^2_{J,\bar\epsilon}$, an almost identical computation (where the sum in~\eqref{eq:Z-bound-unbalanced} is taken over all $R$-sets $A \subsetneq J$ and not just those that are sides of dominant patterns, so that $2^{\fq+1}$ becomes $|\phasemax|$, and we use the last inequality in \cref{lem:main-cond-1}) yields the same bound on $\Sigma^2_J$.

	\smallskip
	\noindent\textbf{Bound \eqref{eq:cond-highly-energetic}:}
	Suppose that $\Psi = \Psi_{J,\epsilon,\bar\epsilon}$ and $I=\SS \setminus R(J)$.
	Then $Z(\Psi,I) = \Sigma_J \le \omega_{\text{dom}}^{2d} \gamma^{2d}$.
	In fact, we claim that we can replace the term $\gamma=\hat\rho_{\text{act}} \rho_{\text{int}}^s$ here by $\hat\rho_{\text{act}} \rho_{\text{int}}^{4\bar\epsilon d}$. Indeed, since $\Psi$ is disjoint from $\Psi^2_{J,\bar\epsilon}$, every value outside of $R(J)$ comes at an even larger interaction cost than before as it interacts (with non-maximum interaction weight) with at least $4\bar\epsilon d$ spins on the left, so that by repeating the derivation of the above bound on $\max_{\xi \in \Xi_J} Z^1_\xi$, we get the improved bound $\max_{\xi \in \Xi_J} Z^1_\xi \le \lambda_{R(J)} \hat\rho_{\text{act}} \rho_{\text{int}}^{4\bar\epsilon d}$, which then gives that
	\[ Z(\Psi,I) = \Sigma_J \le \omega_{\text{dom}}^{2d} \cdot (\hat\rho_{\text{act}} \rho_{\text{int}}^{4\bar\epsilon d})^{2d} \le \omega_{\text{dom}}^{2d} \cdot \rho_{\text{int}}^{4\bar\epsilon d^2} \le \omega_{\text{dom}}^{2d} \cdot e^{-4\alpha\epsilon d^2} ,\]
	where in the last two inequalities we used that $\rho_{\text{int}}^{2\bar\epsilon d} \le \rho_{\text{int}}^{s/2} \le 1/\hat\rho_{\text{act}}$ by~\eqref{eq:s-ineq} and $\bar\epsilon \ge \frac{\alpha \epsilon}{-\log \rho_{\text{int}}}$.

\medbreak

It remains to establish~\eqref{eq:cond-non-dominant} and the other half of~\eqref{eq:cond-unbalanced}. For this, we now consider the case $S \neq J$. Note that by~\eqref{eq:Z-bound-S-supset}, this case can only occur when $s>1$.
	When $S \neq J$, certain values of $J$ do not appear often on the left side, potentially allowing some values outside of $R(J)$ to be placed on the right side at a low (but still positive) interaction cost. Values outside of $R(S)$ still incur a large interaction cost as before. Precisely, given $\xi \in \Xi_S$, any choice of $i \in R(S) \setminus R(J)$ contributes at most $\rho_{\text{int}} \lambda_i$ to $Z^1_\xi$ (since $\lambda_{i,j}<1$ for all $j \in J \setminus N(i)$ and $\xi(j) \ge 1$ for some $j \in J \setminus N(i)$) and any choice outside of $R(S)$ contributes at most $\rho_{\text{int}}^s\lambda_i$.
	Thus, for the term involving the values on the right side, we have that
	\begin{equation}\label{eq:Z-bounds-right-side}
	\begin{aligned}
	\max_{\xi \in \Xi_S} Z^1_\xi
	 &\le \lambda_{I \cap R(J)} + \rho_{\text{int}} \lambda_{I \cap R(S) \setminus R(J)} + \rho_{\text{int}}^s \lambda_{I \setminus R(S)} \\
	 &\le \lambda_{R(J)} +\rho_{\text{int}} \lambda_{R(S) \setminus R(J)} + \rho_{\text{int}}^s \lambda_{\SS \setminus R(S)} \\
	 &= \lambda_{R(S)} \left( 1 - (1-\rho_{\text{int}}) \tfrac{\lambda_{R(S) \setminus R(J)}}{\lambda_{R(S)}} + \rho_{\text{int}}^s \tfrac{\lambda_{\SS \setminus R(S)}}{\lambda_{R(S)}} \right) .
	\end{aligned}
	\end{equation}
	For the term involving the values on the left side, observe that by first choosing the subset of $[2d]$ on which $J \setminus S$ appears (by~\eqref{eq:Z-bound-S-supset}, each $j \in J \setminus S$ appears at most $s-1$ times), choosing a value from $J \setminus S$ for each position there, and then choosing a value from $S$ for each remaining position in $[2d]$, we obtain that
	\begin{equation}\label{eq:Z-bounds-left-side}
	 \sum_{\xi \in \Xi_S} |\Psi_\xi| \cdot Z^0_\xi
	\le \sum_{k=1}^{(s-1)|J \setminus S|} \binom{2d}{k} \lambda_{J \setminus S}^k \lambda_S^{2d-k} \le (\lambda_S)^{2d} \cdot \left(\tfrac{2d\lambda_\SS}{\lambda_S}\right)^{(s-1)|\SS|} ,
	\end{equation}
	where we used the inequality $\sum_{k=1}^m \binom nk \le n^m$.
	Plugging~\eqref{eq:Z-bounds-right-side} and~\eqref{eq:Z-bounds-left-side} into~\eqref{eq:Z-bounds-simple} yields that
	\[ \Sigma_S \le (\lambda_S \lambda_{R(S)})^{2d} \cdot
	\Big(1 - (1-\rho_{\text{int}})\tfrac{\lambda_{R(S)\setminus R(J)}}{\lambda_{R(S)}} + \rho_{\text{int}}^s \tfrac{\lambda_{\SS \setminus R(S)}}{\lambda_{R(S)}} \Big)^{2d} \left(\tfrac{2d\lambda_\SS}{\lambda_S}\right)^{(s-1)|\SS|} \qquad\text{for any }S \neq J.\]
	When $S$ is not a side of a dominant pattern, using the definition of $\hat\rho_{\text{pat}}^{\text{bulk}}$ from~\eqref{eq:rho-bulk-hat-and-rho-act-hat-def}, the definition of $\alpha$ from~\eqref{eq:alpha-def2} and that $\alpha d \ge \log 4$ by~\eqref{eq:parameter-inequalities-simple1}, we obtain that
	\begin{equation}\label{eq:Z-bounds-S-nondom}
	\Sigma_S \le \omega_{\text{dom}}^{2d} \cdot (\hat\rho_{\text{pat}}^{\text{bulk}})^{2d} \le \omega_{\text{dom}}^{2d} \cdot \tfrac{1}{|\phasemax|} e^{-2\alpha d} \le \omega_{\text{dom}}^{2d} \cdot \tfrac{1}{4|\phasemax|} e^{-\alpha d} .
	\end{equation}
	When $S$ is a side of dominant pattern, $\frac{\lambda_{R(J)}}{\lambda_{R(S)}} \le \rho_{\text{pat}}^{\text{bdry}}$ and $\tfrac{\lambda_{\SS \setminus R(S)}}{\lambda_{R(S)}},\tfrac{\lambda_\SS}{\lambda_S} \le \hat\rho_{\text{act}}$, so that
	\begin{equation}\label{eq:Z-bounds-S-dom}
	\begin{aligned}
	\Sigma_S &\le \omega_{\text{dom}}^{2d} \cdot \big(1 - (1-\rho_{\text{pat}}^{\text{bdry}})(1-\rho_{\text{int}}) + \gamma \big)^{2d} \cdot (2d \hat\rho_{\text{act}})^{(s-1)|\SS|} \\
	&\le \omega_{\text{dom}}^{2d} \cdot (1+\gamma)^{2d} \big(1 - (1-\rho_{\text{pat}}^{\text{bdry}})(1-\sqrt{\rho_{\text{int}}}) \big)^{2d} \cdot e^{\frac12 \alpha d} \le \omega_{\text{dom}}^{2d} \cdot 2^{-\fq-3} e^{2\gamma d - \alpha d} ,
	\end{aligned}
	\end{equation}
	where we used the definition of $\alpha$ from~\eqref{eq:alpha-def2} and that $\gamma \le \rho_{\text{int}}^{s/2} \le \sqrt{\rho_{\text{int}}}$ to deduce that $1-(1-\rho_{\text{pat}}^{\text{bdry}})(1-\rho_{\text{int}})/(1+\gamma) \le 1-(1-\rho_{\text{pat}}^{\text{bdry}})(1-\sqrt{\rho_{\text{int}}}) \le e^{-\alpha}$, and we used that $(2d \hat\rho_{\text{act}})^{(s-1)|\SS|} \le e^{-\alpha d/2}$ by~\eqref{eq:s-ineq} and $\frac12 \alpha d \ge (\fq+3)\log 2$ by~\eqref{eq:parameter-inequalities-simple1}.

	\smallskip
	\noindent\textbf{Bound \eqref{eq:cond-unbalanced}:}
	Suppose that $J$ is a side of a dominant pattern, $\Psi = \Psi^1_{J,\bar\epsilon} \cup \Psi^2_{J,\bar\epsilon}$ and $I=\SS$. We have already shown in this case that $\Sigma_J \le \omega_{\text{dom}}^{2d} \cdot (1 + \gamma)^{2d} \cdot \tfrac12 e^{-\alpha d}$. Thus, by~\eqref{eq:Z-bounds-S-nondom} and~\eqref{eq:Z-bounds-S-dom}, and since there are $|\phasemax|$ maximal patterns and at most $2^{\fq+1}$ dominant patterns,
	\[ Z(\Psi,\SS) = \Sigma_J + \sum_{S \neq J} \Sigma_S \le \omega_{\text{dom}}^{2d} \cdot e^{2\gamma d - \alpha d} .\]

	\smallskip
	\noindent\textbf{Bound \eqref{eq:cond-non-dominant}:}
	Suppose that $J$ is not a side of a dominant pattern, $\Psi=\Psi_J$ and $I=\SS$. When $R(J) \neq \emptyset$, using~\eqref{eq:Z-bounds-J} (and that $J$ is not a side of a dominant pattern), \eqref{eq:Z-bounds-S-nondom} and~\eqref{eq:Z-bounds-S-dom}, we get similarly to before that
	\[ Z(\Psi_J,\SS) = \sum_S \Sigma_S \le \omega_{\text{dom}}^{2d} \cdot e^{2\gamma d - \alpha d} .\]
	To obtain~\eqref{eq:cond-non-dominant}, we still need to sum over $J$ (including the case when $R(J) \neq \emptyset$).
	If we were to define $\alpha_2$ in~\eqref{eq:alpha-def2} by subtracting $\frac1d \log |\phasemax|$ instead of $\frac 2{3d} \log |\phasemax|$, then we would have an additional factor of $\frac{1}{|\phasemax|}$ in the above bound so that this would not be a problem. However, with the given definition of $\alpha_2$, we need to proceed differently. Let us add $J$ to the notation of $\Sigma_S$ by writing $\Sigma^J_S$. Then
	\[ \sum_{J\text{ nondom}} Z(\Psi_J,\SS) = \sum_{J\text{ nondom}, S \subset J} \Sigma^J_S ,\]
	where we write ``$A$ nondom'' as shorthand for ``$A$ is not a side of a dominant pattern''. We divide the latter sum into two according to whether $S$ is nondom, showing that each sum is at most $\omega_{\text{dom}}^{2d} \cdot \frac12 e^{2\gamma d - \alpha d}$. This will yield~\eqref{eq:cond-non-dominant}.

	We begin with the case when $S$ is dom. We claim that for any such $S$, the bound obtained in~\eqref{eq:Z-bounds-S-dom} holds also for the sum of over nondom $J$, so that $\sum_J \Sigma^J_S \le \omega_{\text{dom}}^{2d} \cdot 2^{-\fq-3} e^{2\gamma d - \alpha d}$. Indeed, $\max_J \max_{\xi \in \Xi^J_S} Z^1_\xi \le \lambda_{R(S)}(1-(1-\rho_{\text{int}})(1-\rho_{\text{pat}}^{\text{bdry}})+\gamma)$ by~\eqref{eq:Z-bounds-right-side}, and $\sum_J \sum_{\xi \in \Xi^J_S} |\Psi_\xi| Z^0_\xi \le (\lambda_S)^{2d} (2d\hat\rho_{\text{act}})^{(s-1)|\SS|}$ by the same computation as in~\eqref{eq:Z-bounds-left-side} with $J=\SS$, leading as in~\eqref{eq:Z-bounds-S-dom} to the claimed bound (where $\cup_J \Xi_S^J$ replaces $\Xi_S$ in~\eqref{eq:Z-bounds-simple}). Summing now over $S$ and using that there are at most $2^{\fq+1}$ dominant patterns, we obtain that
	\[ \sum_{S\text{ dom},J\text{ nondom}} \Sigma^J_S \le \omega_{\text{dom}}^{2d} \cdot \tfrac12 e^{2\gamma d - \alpha d} .\]

	We now consider the case when $S$ is nondom. Similarly to before, we claim that, for any $S$ nondom such that $R(S) \neq \emptyset$ (the latter can occur only if $S=\SS$), the bound obtained in~\eqref{eq:Z-bounds-S-nondom} holds also for the sum over $J$, so that  $\sum_J \Sigma^J_S \le \omega_{\text{dom}}^{2d} \cdot \tfrac{1}{4|\phasemax|} e^{-\alpha d}$.
	Indeed, $\max_J \max_{\xi \in \Xi^J_S} Z^1_\xi \le \lambda_{R(S)}(1+\rho_{\text{int}}^s \tfrac{\lambda_\SS}{\lambda_{R(S)}})$ by~\eqref{eq:Z-bounds-right-side}, and $\sum_J \sum_{\xi \in \Xi^J_S} |\Psi_\xi| Z^0_\xi \le (\lambda_S)^{2d} (2d\lambda_\SS/\lambda_S)^{(s-1)|\SS|}$ by the same computation as in~\eqref{eq:Z-bounds-left-side} with $J=\SS$ (the case of $S=J$, which corresponds to $k=0$ in the sum, can also be included in this computation), leading as in~\eqref{eq:Z-bounds-S-nondom} to the claimed bound. Summing now over $S$, we obtain that
	\[ \sum_{S\text{ nondom with }R(S)\neq \emptyset, J \supset S} \Sigma^J_S \le \omega_{\text{dom}}^{2d} \cdot \tfrac14 e^{-\alpha d} . \]
	It remains only to bound $\Sigma_\SS^\SS$ in the case when $R(\SS)=\emptyset$. In this case, we have $\max_{\xi \in \Xi^\SS_\SS} Z^1_\xi \le \rho_{\text{int}}^s \lambda_\SS$ and $\sum_{\xi \in \Xi^\SS_\SS} |\Psi_\xi| Z^0_\xi \le (\lambda_\SS)^{2d}$. Thus, using that $\gamma \le \rho_{\text{int}}^{s/2} \le \sqrt{\rho_{\text{int}}} \le e^{-\alpha}$ by~\eqref{eq:s-ineq},
	\[ \Sigma^\SS_\SS \le (\rho_{\text{int}}^s (\lambda_\SS)^2)^{2d} \le \omega_{\text{dom}}^{2d} \cdot (\hat\rho_{\text{act}} \rho_{\text{int}}^s)^{2d} = \omega_{\text{dom}}^{2d} \cdot \gamma^{2d} \le \omega_{\text{dom}}^{2d} \cdot e^{-2\alpha d} \le \omega_{\text{dom}}^{2d} \cdot \tfrac14 e^{-\alpha d}. \qedhere \]
\end{proof}

\section{Infinite-volume Gibbs states}
\label{sec:gibbs}
In this section, we prove \cref{thm:existence_Gibbs_states} and \cref{thm:characterization_of_Gibbs_states} (the quantitative versions of \cref{thm:existence_Gibbs_states-NQ} and \cref{thm:characterization_of_Gibbs_states-NQ}).
The former is about the existence of a limiting Gibbs state for each dominant pattern and the properties of this measure. The latter is about the characterization of all maximal-pressure Gibbs states. The first is proven in \cref{sec:P-pattern-Gibbs-state} (modulo the fact the measure has maximal pressure, which is deferred to \cref{sec:equilibrium-states}) and the second in \cref{sec:equilibrium-states}.
In \cref{sec:enum}, we establish the topological pressure formula of \cref{thm:enum}.
We assume throughout this section that $d \ge 2$ and that \cref{main-cond} holds, and we define $\tilde\alpha$ as in~\eqref{eq:alpha-tilde}.

Let us first provide a formal definition of a Gibbs state.
A probability measure $\mu$ on $\SS^{\Z^d}$ (with the natural product $\sigma$-algebra) is a \emph{Gibbs state} if it is supported on configurations $f$ satisfying $\lambda_{f(u),f(v)}>0$ for every edge $\{u,v\}$ of $\Z^d$, and a random function $f$ sampled from $\mu$ has the property that, for any finite $\Lambda \subset \Z^d$, given $f|_{\Lambda^c \cup \intB \Lambda}$, the conditional distribution of $f|_{\Lambda}$ is almost surely given by the weights in~\eqref{eq:config-weight} on the set of configurations in $\SS^\Lambda$ that agree with $f$ on $\intB \Lambda$.

For a distribution $\mu$ on $\SS^{\Z^d}$, we denote by $\mu|_U$ the marginal distribution of $\mu$ on $\SS^U$. Given two discrete distributions $\mu$ and $\nu$ on a common space, we denote the total-variation distance between $\mu$ and $\nu$ by $\distTV(\mu,\nu) := \max_{A} |\mu(A)-\nu(A)|$, where the maximum is over all events $A$. Recall that a domain is a finite, non-empty, connected and co-connected subset of $\Z^d$.

Several of the proofs in this section are very similar (some parts are even identical) to those given in~\cite{peledspinka2018colorings} for the proper coloring model. As the overlap is not very large, except for some prerequisite lemmas, we include the proofs here for the reader's convenience.

\subsection{Large violations}\label{sec:large-violations}
For the proofs of \cref{thm:existence_Gibbs_states} and \cref{thm:characterization_of_Gibbs_states}, we require extensions of \cref{thm:long-range-order} to larger violations of the boundary pattern rather than just single-site violations. Recall the definitions of $Z_P(f)$ and $Z_*(f)$ from~\eqref{eq:Z-def} and~\eqref{eq:Z_*-def}, and the definition of $Z_*^{+5}(f,V)$ from before \cref{lem:existence-of-breakup}.

\begin{prop}\label{prop:Z_*-bound}
	Let $\Lambda$ be a domain and let $V \subset \Z^d$ be finite. Then, for any $k \ge 1$,
	\begin{equation*}
	\Pr_{\Lambda,P_0}\big(|Z_*(f) \cap Z_*^{+5}(f,V)| \ge k\big) \le 2^{|V|} \cdot e^{- c\tilde\alpha k/d} .
	\end{equation*}
\end{prop}

\begin{proof}
	We omit $f$ from notation.
	Let $\Omega_{L,M,N}$ denote the event that there exists a breakup in $\breakups_{L,M,N}$ seen from~$V$. Let us show that $|Z_* \cap  Z_*^{+5}(V)| \ge k$ implies the occurrence of $\Omega_{L,M,N}$ for some $L,M,N \ge 0$ satisfying that $L/2 + M+N \ge k$.

	\cref{lem:existence-of-breakup} implies the existence of a breakup $X$ such that $X_*^{+5}=Z_*^{+5}(V)$.
	Note that this implies that $X_*= Z_* \cap Z_*^{+5}(V)$ so that $|X_*| \ge k$. Since every vertex in $\bigcup_P \intextB X_P$ is an endpoint of an edge in $\bigcup_P \partial X_P$, and since every edge has only two endpoints, we see that $X \in \breakups_{L,M,N}$ implies that $L/2+M+N \ge k$. Note also that \cref{lem:boundary-size-of-odd-set} implies that $\breakups_{L,M,N}=\emptyset$ when $L<d^2$.
	Therefore, by \cref{prop:prob-of-breakup-associated-to-V},
	\[ \Pr\big(|Z_* \cap Z_*^{+5}(V)| \ge k\big) \le 2^{|V|} \sum_{\substack{L \ge d^2,\,M,N \ge 0\\L/2+M+N \ge k}} \exp\left(- c\tilde\alpha \big( \tfrac{L}{d}+M+ \epsilon N \big) \right) .\]
	Using~\eqref{eq:alpha-cond} and $\epsilon \ge \frac{1}{4d}$, the desired inequality follows.
\end{proof}

Recall that, while the $P$-even vertices in $Z_P(f)$ are always in the $P$-pattern, the $P$-odd vertices there need not be. Let $\bar{Z}_P(f)$ denote the subset of $Z_P(f)$ that is in the $P$-pattern.
For a finite $V \subset \Z^d$, define $\cB_P(f,V)$ to be the union of the $(\Z^d)^{\otimes 2}$-connected components of $\bar{Z}_P(f)^c$ that intersect $V$.
For a set $U \subset \Z^d$, define $\diam^* U := 2m+\diam U_1 + \dots + \diam U_m$, where $\{U_i\}_{i=1}^m$ are the $(\Z^d)^{\otimes 2}$-connected components of~$U$ (here and below, $\diam$ and distances are with respect to the $\Z^d$ graph structure).

\begin{prop}\label{prop:B_P-bound}
	Let $\Lambda$ be a domain and let $V \subset \Z^d$ be finite. Then, for any $k \ge 1$,
	\begin{equation}\label{eq:bound-on-diameter-B_P}
	\Pr_{\Lambda,P}\big(\diam^* \cB_P(f,V) \ge k \big) \le 2^{|V|} \cdot e^{- c\tilde\alpha dk}.
	\end{equation}
\end{prop}

For the proof, we require the following lemma.
For $A \subset \Z^d$, denote
\begin{equation}\label{eq:def-partial+}
A_\iso := \{ v \in A : N(v) \cap A = \emptyset \}.
\end{equation}

\begin{lemma}[{\cite[Lemma~8.3]{peledspinka2018colorings}}]\label{lem:boundary-size-via-diameter2}
	Let $A \subset \Z^d$ be finite, odd and $(\Z^d)^{\otimes 2}$-connected. Then
	\[ |\partial A| + |\partial(A_\iso^+)| \ge \tfrac12(d-1)^2 (2+\diam A) .\]
\end{lemma}

\begin{proof}[Proof of \cref{prop:B_P-bound}]
	We denote $\cB:=\cB_P(f,V)$ and omit $f$ from notation.
	Let $\Omega_{L,M,N}$ denote the event that there exists a breakup in $\breakups_{L,M,N}$ seen from~$V$. Let us show that $\diam^* \cB \ge k$ implies the occurrence of $\Omega_{L,M,N}$ for some $L,M,N \ge 0$ satisfying that $L+2dM \ge cd^2k$.

	Note that $\intB \cB \subset \extB \bar{Z}_P$ so that, in particular, $\cB$ is an odd set.
	Note also that $\cB^{+2} \setminus \cB \subset \bar{Z}_P$ and that $\cB_\iso \subset Z_P \setminus \bar{Z}_P$.
	We claim that $\partial (\cB \setminus \cB_\iso) \subset \partial Z_P$ and $\cB_\iso^+ \subset Z_\hole$, so that, in particular, $\intextB \cB \subset Z_*$. To see the former, let $(u,v) \in \dpartial (\cB \setminus \cB_\iso)$ and $w \in N(u) \cap \cB$, so that $v \in Z_P$ and $w \notin Z_P$. It follows that $u \notin Z_P$ and hence that $\{u,v\} \in \partial Z_P$ as required. To see the latter, let $u \in \cB_\iso^+$ and $w \in u^+ \cap \cB_\iso$, so that $f(w) \notin P_\inner$ and $f(N(w)) \subset P_\bdry$. It follows that $w^+ \subset Z'_P$ so that $u \in Z_\hole$ as required.

	\cref{lem:existence-of-breakup} implies the existence of a breakup $X$ such that $X_*^{+5}=Z_*^{+5}(V)$.
	Note that this implies that $X_*= Z_* \cap Z_*^{+5}(V)$.
	Since $\intextB \cB \subset Z_*$ and since every $(\Z^d)^{\otimes 2}$-connected component of $\cB$ intersects $V$, it follows that $\intextB \cB \subset Z_*^{+5}(V)$ and hence that $\intextB \cB \subset X_*$. In particular, by~\eqref{eq:breakup-1} and~\eqref{eq:breakup-2}, $\partial (\cB \setminus \cB_\iso) \subset \partial X_P$ and $\cB_\iso^+ \subset X_\hole$.

	Let $L,M,N \ge 0$ be such that $X \in \breakups_{L,M,N}$.
	Applying \cref{lem:boundary-size-via-diameter2} to each $(\Z^d)^{\otimes 2}$-connected component of~$\cB$ yields that $|\partial \cB| +|\partial(\cB_\iso^+)| \ge c d^2 k$. Since $|\partial(\cB_\iso^+)| \le 2d|\cB_\iso^+|$, we conclude that $L+2dM \ge cd^2k$.
	Therefore, by \cref{prop:prob-of-breakup-associated-to-V},
	\[ \Pr\big(\diam^* \cB \ge k\big) \le 2^{|V|} \sum_{\substack{L,M,N \ge 0\\L+2dM\ge cd^2k}} \exp\left(- c\tilde\alpha \big( \tfrac{L}{d}+M+ \epsilon N \big) \right) .\]
	Using~\eqref{eq:alpha-cond}, the desired inequality follows.
\end{proof}

We note the following corollary of \cref{prop:B_P-bound}.
\begin{cor}\label{cor:prob-of-large-violation-of-P-pattern}
Let $\Lambda$ be a domain, $v \in \Lambda$ and $f \sim \Pr_{\Lambda,P}$.
Let $\cC'$ denote the connected component of vertices in the $P$-pattern containing $\intB \Lambda$, and let $\cC$ denote the $(\Z^d)^{\otimes 2}$-connected component of $\Lambda \setminus \cC'$ containing $v$.
Then, for any $k \ge 1$,
\[ \Pr\big(\diam \cC \ge k \big) \le e^{- c\tilde\alpha dk}. \]
\end{cor}
\begin{proof}
Note that for any $u \in \intB \cC$, we have $\cB(f,\{u\}) \supset \intB \cC$ and hence $\diam^* \cB(f,\{u\}) \ge \diam \cC$.
Fix a shortest path $(v_0,v_1,\dots,v_m)$ from $v$ to $\intB \Lambda$, and note that $\intB \cC$ intersects $\{v_0,\dots,v_{\min\{\diam \cC,m\}}\}$ whenever $\cC \neq \emptyset$.
Thus, $\Pr(\diam \cC = k) \le 2^{k+1} e^{-c\tilde\alpha dk}$ by \cref{prop:B_P-bound}, and the corollary follows by a union bound and~\eqref{eq:alpha-cond}.
\end{proof}

For the proof of \cref{thm:existence_Gibbs_states}, we also require a corollary of \cref{prop:B_P-bound} for violations of the boundary pattern in a pair of configurations.
Given two configurations $f$ and $f'$ of $\Z^d$, define $\cB_P(f,f',u)$ to be the $(\Z^d)^{\otimes 2}$-connected component of $u$ in $(\bar{Z}_P(f) \cap \bar{Z}_P(f'))^c$.

\begin{cor}\label{cor:prob-of-joint-breakup-core}
	Let $\Lambda$ and $\Lambda'$ be two domains and $f \sim \Pr_{\Lambda,P}$ and $f' \sim \Pr_{\Lambda',P}$ be independent. Then
	\[ \Pr\big(\diam \cB_P(f,f',u) \ge r \big) \le e^{- c\tilde\alpha dr} \qquad\text{for any }r \ge 1\text{ and }u \in \Z^d .\]
\end{cor}

For the proof of \cref{cor:prob-of-joint-breakup-core}, we require the following lemma.

\begin{lemma}[{\cite[Lemma~8.5]{peledspinka2018colorings}}]\label{lem:diam-witness}
	Let $U,V \subset \Z^d$ be finite and assume that $U \cup V$ is $(\Z^d)^{\otimes 2}$-connected. Then for any $u,v \in U \cup V$ there exists a path $p$ from $u$ to $v$ of length at most $\diam^* U_p + \diam^* V_p$, where $U_p$ and $V_p$ are the union of $(\Z^d)^{\otimes 2}$-connected components of $U$ and $V$ which intersect $p$.
\end{lemma}

\begin{proof}[Proof of \cref{cor:prob-of-joint-breakup-core}]
	Denote $\cB:=\cB_P(f,f',u)$ and suppose that $\diam \cB \ge r$.
	Let $v \in \cB$ be such that $\dist(u,v) \ge r/2$.
	Note that $\cB$ is contained in $\cB(f,\Z^d) \cup \cB(f',\Z^d)$.
	By \cref{lem:diam-witness} applied to $\cB \cap \cB(f,\Z^d)$ and $\cB \cap \cB(f',\Z^d)$, there exists a path $p$ from $u$ to $v$ of length $s \le \diam^* \cB(f,p) + \diam^* \cB(f',p)$.
	In particular, $T := \max \{\diam^* \cB(f,p),\diam^* \cB(f',p)\}$ is at least $s/2$.
	Thus, by a union bound on the choices of $p$ and $T$, \cref{prop:B_P-bound} and~\eqref{eq:alpha-cond},
	\[ \Pr\big(\diam \cB \ge r \big) \le \sum_{t=\lceil r/4 \rceil}^\infty 2 (2d)^{2t} 2^{2t+1} e^{- c\tilde\alpha dt} \le e^{- c\tilde\alpha dr} . \qedhere \]
\end{proof}

\subsection{The $P$-pattern Gibbs state}
\label{sec:P-pattern-Gibbs-state}

In this section, we fix a dominant pattern $P$ and prove that $\Pr_{\Lambda,P}$ converges as $\Lambda \uparrow \Z^d$ to an infinite-volume Gibbs state $\mu_P$ that satisfies a mixing property which, in particular, implies that $\mu_P$ is extremal. This is the content of the following two lemmas.

\begin{lemma}\label{lem:convergence}
	Let $\Lambda$ and $\Lambda'$ be two domains.
	Let $r \ge 1$ and let $U$ be a domain such that $U^{+r} \subset \Lambda \cap \Lambda'$.
	Then
	\[ \distTV\big(\Pr_{\Lambda,P}|_U, \Pr_{\Lambda',P}|_U\big) \le |U| \cdot e^{-c\tilde\alpha dr} .\]
\end{lemma}

\begin{lemma}\label{lem:almost-independence-of-colorings}
	Let $\Lambda$ be a domain, let $V \subset \Lambda$ be a domain, let $r \ge 1$ and let $U \subset \Lambda$ be such that $U^{+2r} \subset V$.
Then
\[ \distTV\big(\Pr_{\Lambda,P}|_{U \cup V^c}, \Pr_{\Lambda,P}|_U \times \Pr_{\Lambda,P}|_{V^c}\big) \le |U| \cdot e^{-c\tilde\alpha dr} .\]
\end{lemma}

\cref{lem:convergence} easily implies that the finite-volume $P$-pattern measures converge to an infinite-volume Gibbs state $\mu_P$.
Indeed, if $(\Lambda_n)$ is a sequence of domains increasing to $\Z^d$, then for any domain $U$, $\dist(U,\Lambda_n^c) \to \infty$ as $n \to \infty$, so that \cref{lem:convergence} implies that the sequence of measures $(\Pr_{\Lambda_n,P}|_U)_{n=1}^{\infty}$ is a Cauchy sequence with respect to the total-variation metric, and therefore, converges. This establishes the convergence of $\Pr_{\Lambda_n,P}$ as $n \to \infty$ towards an infinite-volume measure $\mu_P$ and it follows that this limit is a Gibbs state. Since this holds for any such sequence $(\Lambda_n)$, it follows that $\mu_P$ is invariant to all automorphisms preserving the two sublattices.
\cref{lem:almost-independence-of-colorings} then easily implies that $\mu_P$ satisfies the following mixing property: for any $0<\delta<1$, there exist constants $A,a>0$ such that
\[ \distTV\big(\mu_P|_{B_{\delta n} \cup (\Z^d \setminus B_n)}, \mu_P|_{B_{\delta n}} \times \mu_P|_{\Z^d \setminus B_n}\big) \le A e^{-an} \qquad\text{for all }n \ge 1 ,\]
where $B_m := [-m,m]^d \cap \Z^d$ (this property is termed \emph{quite weak Bernoulli with exponential rate} in~\cite{burton1995quite} in the context of translation-invariant measures). In particular, for any $k \ge 1$,
\[ \lim_{n \to \infty} \distTV\big(\mu_P|_{B_k \cup (\Z^d \setminus B_n)}, \mu_P|_{B_k} \times \mu_P|_{\Z^d \setminus B_n}\big) = 0 .\]
It is fairly standard to conclude from this that $\mu_P$ is tail trivial~(see~\cite[Proposition~7.9]{georgii2011gibbs}), which is equivalent to extremality within the set of all Gibbs states~(see~\cite[Theorem~7.7]{georgii2011gibbs}). Note that \cref{thm:long-range-order} implies that different~$P$ yield different measures $\mu_P$, and that \cref{cor:prob-of-large-violation-of-P-pattern} implies that $\mu_P$ is supported on configurations with an infinite connected component of vertices in the $P$-pattern, whose complement has only finite $(\Z^d)^{\otimes 2}$-connected components. Thus, \cref{thm:existence_Gibbs_states} will follow once we show that $\mu_P$ is of maximal pressure. We postpone this part to \cref{sec:equilibrium-states} (it is a consequence of \cref{prop:equilibrium-states-are-mixtures}).

The proofs of \cref{lem:convergence} and \cref{lem:almost-independence-of-colorings} make use of the following fact which exploits the domain Markov property of the model.
We say that a collection $\cS$ of proper subsets of $\Z^d$ is a \emph{boundary semi-lattice} if for any $S_1,S_2 \in \cS$ there exists $S \in \cS$ such that $S_1 \cup S_2 \subset S$ and $\partial S \subset \partial S_1 \cup \partial S_2$.
Two boundary semi-lattices which we require are $\cS(U,V) := \{ S \subsetneq \Z^d : U \subset S \subset V \}$ and $\cS(f,P) := \{ S \subsetneq \Z^d : \intextB S\text{ is in the $P$-pattern with respect to }f \}$.
The latter has the property that if $\cS$ is any boundary semi-lattice, then $\cS \cap \cS(f,P)$ is also a boundary semi-lattice.

\begin{lemma}\label{lem:marginal-distribution-given-agreement}
	Let $\Lambda,\Lambda' \subset \Z^d$ be finite and let $U \subset V \subset \Lambda \cap \Lambda'$ be non-empty.
	Let $f \sim \Pr_{\Lambda,P}$ and $f' \sim \Pr_{\Lambda',P}$ be independent.
	\begin{enumerate}[label=(\alph*)]
		\item \label{it:marginal-distribution-given-agreement-pair} $\distTV(\Pr_{\Lambda,P}|_U,\Pr_{\Lambda',P}|_U) \le \Pr(\cS(U,V) \cap \cS(f,P) \cap \cS(f',P)=\emptyset)$.
		\item \label{it:marginal-distribution-given-agreement-single} Assume that $U$ is connected, $V$ is co-connected and $\Pr(\cS(U,V) \cap \cS(f,P) \neq \emptyset)>0$. Then, conditioned on $\{\cS(U,V) \cap \cS(f,P) \neq \emptyset \}$, the distribution of $f|_U$ is a convex combination of the measures $\{ \Pr_{S,P}|_U \}_{S \in \cS^{\text{dom}}(U,V)}$, where $\cS^{\text{dom}}(U,V)$ is the collection of domains in $\cS(U,V)$.
	\end{enumerate}
\end{lemma}

\begin{proof}
	We shall prove both items together.
	To this end, let $f''$ be either $f$ or $f'$, and denote $\cS := \cS(U,V) \cap \cS(f,P) \cap \cS(f'',P)$.
	Since $\cS$ is a finite boundary semi-lattice, it has a unique maximal element $\sf S$ (if  $\cS = \emptyset$, we set $\sf S := \emptyset$).
	Let $S \neq \emptyset$ be such that $\Pr({\sf S}=S)>0$.
	Observe that the event $\{{\sf S}=S\}$ is determined by $f|_{(S^c)^+}$ and $f''|_{(S^c)^+}$.
	Therefore, by the domain Markov property, conditioned on $\{{\sf S}=S\}$, $f|_S$ and $f''|_S$ are distributed as $\Pr_{S,P}|_S$. In particular, conditioned on $\{\cS \neq \emptyset \}$, the distribution of both $f|_U$ and $f''|_U$ is $\sum_S \Pr({\sf S}=S \mid \cS \neq \emptyset) \Pr_{S,P}|_U$, from which the first item follows. Moreover, if $U$ is connected and $V$ is co-connected, then $\sf S$ is always a domain, since \cref{lem:co-connect-properties}\ref{it:co-connect-reduces-boundary} and \cref{lem:co-connect-properties}\ref{it:co-connect-kills-components} imply that the co-connected closure of $S$ (with respect to infinity) belongs to $\cS$ for any $S \in \cS$. Hence, the second item also follows.
\end{proof}

We are now ready to prove \cref{lem:convergence} and \cref{lem:almost-independence-of-colorings}.

\begin{proof}[Proof of \cref{lem:convergence}]
	Denote $S := U \cup \bigcup_{u \in \intextB U} \cB_P(f,f',u)^+$ and observe that, by definition, $\intextB S$ is in the $P$-pattern with respect to both $f$ and $f'$.
	Let $\cE$ be the event that $S$ intersects $(U^{+r})^c$, so that $S \subset U^{+r}$ on the complement of $\cE$.
	Then, by \cref{lem:marginal-distribution-given-agreement} and \cref{cor:prob-of-joint-breakup-core},
	\[ \distTV\big(\Pr_{\Lambda,P}|_U, \Pr_{\Lambda',P}|_U\big) \le \Pr(\mathcal E) \le \sum_{u \in U} \Pr\big(\diam \cB_P(f,f',u) \ge r\big) \le |U| \cdot e^{-c\tilde\alpha dr} . \qedhere \]
\end{proof}

\begin{proof}[Proof of \cref{lem:almost-independence-of-colorings}]
We begin with a simple observation. Let $X$ and $Y$ be discrete random variables and let $\mu_{X|Y}$ denote the conditional (random) distribution of $X$ given $Y$. Then
\[ \distTV(\mu_{(X,Y)}, \mu_X \times \mu_Y) = \E[\distTV(\mu_{X|Y},\mu_X)] ,\]
where we write $\mu_Z$ for the distribution of a random variable $Z$.
Indeed, the verification of this is straightforward using that $\distTV(\mu,\lambda)=\frac12 \sum_i |\mu(i)-\lambda(i)|$.

Let $\mu$ be the conditional (random) distribution of $f|_U$ given $f|_{V^c}$. Let $\cE'$ be the event that there exists a set $S$ such that $U^{+r} \subset S \subset V$ and such that $\intextB S$ is in the $P$-pattern.
By \cref{lem:marginal-distribution-given-agreement}, conditioned on $\cE'$, $\mu$ is a convex combination of measures $\Pr_{S,P}|_U$, where $S$ is a domain containing $U^{+r}$.
For any such $S$, by \cref{lem:convergence}, we have \[ \distTV(\Pr_{S,P}|_U,\Pr_{\Lambda,P}|_U) \le |U| \cdot e^{-c\tilde\alpha dr} .\]
Let $\cE$ be the event that $\cB_P(f,u)^+$ intersects $V^c$ for some $u \in U^{+r}$, and note that $\cE^c \subset \cE'$.
Hence,
\[ \E[\distTV(\mu,\Pr_{\Lambda,P}|_U)] \le |U| \cdot e^{-c\tilde\alpha dr} + \E[\mu(\cE)] .\]
By \cref{prop:B_P-bound},
\[ \E[\mu(\cE)] = \Pr(\cE) \le |U^{+r}| \cdot e^{-c\tilde\alpha dr} \le |U| \cdot (Cd)^r \cdot e^{-c\tilde\alpha dr} \le |U| \cdot e^{-c\tilde\alpha dr} .\]
Thus, $\E[\distTV(\mu,\Pr_{\Lambda,P}|_U)] \le |U| \cdot e^{-c\tilde\alpha dr}$, and the lemma follows from the above observation.
\end{proof}

\subsection{The maximal-pressure Gibbs states}
\label{sec:equilibrium-states}

The purpose of this section is to characterize all (periodic) maximal-pressure Gibbs states of the system.
Let us begin by defining the relevant notions.
Let $\mu$ be a probability measure on $\SS^{\Z^d}$.
Given a transformation $T \colon \Z^d \to \Z^d$, we say that $\mu$ is $T$-invariant if $\mu(T^{-1}A)=\mu(A)$ for any measurable event $A$. We say that $\mu$ is \emph{periodic} if it is $\Gamma$-invariant for a (full-dimensional) lattice $\Gamma$ of translations of $\Z^d$. Observe that every periodic measure $\mu$ is $(N \Z^d)$-invariant for some positive integer $N$.

To define the notion of maximal pressure, we first require some other definitions. Recall the weight $\omega_f$ of a configuration $f \in \SS^\Lambda$ from~\eqref{eq:config-weight}.
The \emph{partition function} with free boundary conditions in a domain $\Lambda$ and the \emph{topological pressure} (sometimes called the \emph{free energy}) of the system are given by\footnote{Note that every configuration $f \in \SS^{\Lambda_n}$ with $\omega_f>0$ may be extended to an admissible configuration of all of $\Z^d$, e.g., by iterated reflections.}
\[ Z^{\text{free}}_\Lambda := \sum_{f \in \SS^\Lambda} \omega_f \qquad\text{and}\qquad P_{\text{top}} := \lim_{n \to \infty} \frac{\log Z^{\text{free}}_{\Lambda_n}}{|\Lambda_n|} , \qquad \text{where }\Lambda_n := \{0,1,\dots,n\}^d .\]
The above limit exists by subadditivity (see~\cite[Lemma~15.11]{georgii2011gibbs}).
Let $\mu$ be a measure on $\SS^{\Z^d}$ which is periodic and supported on configurations $f$ satisfying $\lambda_{f(u),f(v)}>0$ for all $\{u,v\} \in E(\Z^d)$.
The \emph{measure-theoretic entropy} (also known as Kolmogorov-Sinai entropy) of $\mu$ is
\[ h(\mu) := \lim_{n \to \infty} \frac{\Ent(\mu|_{\Lambda_n})}{|\Lambda_n|} ,\]
which also exists by subadditivity (see~\cite[Theorem~15.12]{georgii2011gibbs}). Define a \emph{potential} by
\[ A_v(f) := \log \lambda_{f(v)} + \tfrac12 \sum_{u \sim v} \log \lambda_{f(u),f(v)} .\]
The \emph{measure-theoretic pressure} of $\mu$ is
\[ P(\mu) := h(\mu) + \frac{1}{N^d} \sum_{v \in [N]^d} \int A_v \, d\mu ,\]
where $N$ is any positive integer such that $\mu$ is $(N\Z^d)$-invariant.
The variational principle tells us that $P_{\text{top}}$ is the supremum of $P(\mu)$ over all such $\mu$ and that this supremum is achieved by some $\mu$ (see~\cite{misiurewicz1976short}). Such a $\mu$ is said to be of \emph{maximal pressure}.
For non-weighted homomorphism models (i.e., when the single-site activities are all 1 and the pair interactions are all 0 or 1), this notion reduces to that of a measure of maximal entropy.
A theorem of Lanford--Ruelle~\cite{lanford1969observables} tells us that every measure of maximal pressure is also a Gibbs state (so that there is some redundancy when speaking about a maximal-pressure Gibbs state).
We stress that a measure of maximal pressure is, by definition, always assumed to be periodic.

We wish to show that every maximal-pressure Gibbs state is a mixture of the $P$-pattern Gibbs states. In order to allow ourselves to appeal directly to \cref{prop:Z_*-bound} in the proof (instead of repeating similar arguments), we first show that certain configurations with periodic boundary conditions may be extended to $P$-pattern boundary conditions. This part of the argument may be simplified in some cases (e.g., when there exists $i$ such that $\lambda_{i,j}>0$ for all $j$), but is needed in general.

A function $f \colon U \to \SS$ is called a \emph{configuration on $U$}.
A configuration $f$ on $\{-n,\dots,n\}^{d-1}$ is \emph{symmetric} if $f(x_1,\dots,x_{d-1})=f(|x_1|,\dots,|x_{d-1}|)$ for all $x \in \{-n,\dots,n\}^{d-1}$.
A configuration $f$ on $U \subset \Z^{d-1}$ is \emph{$n$-periodic} if $f(x)$ depends only on $(x_1\text{ mod }n,\dots,x_{d-1}\text{ mod }n)$ for $x \in U$. A configuration on $\{-kn,\dots,kn\}^{d-1}$ is \emph{$n$-symmetric} if it is $2n$-periodic and its restriction to $\{-n,\dots,n\}^{d-1}$ is symmetric. A configuration on $\Lambda_{2kn} \subset \Z^d$ is $n$-symmetric if its restriction to any of the $2d$ faces of $\Lambda_{2kn}$ is $n$-symmetric (after an appropriate translation).
Finally, a configuration $f$ on $U \subset \Z^d$ is \emph{legal} if $\omega_f>0$ (i.e., if $\lambda_{f(u),f(v)}>0$ whenever $u,v \in U$ are adjacent), and it has $(a,b)$-boundary conditions if the even vertices in $\intB U$ take the value $a$ and the odd ones take $b$.

\begin{lemma}\label{lem:pattern-extends-to-ab-boundary-conditions}
Suppose that the graph $(\SS,\{ \{i,j\} : \lambda_{i,j}>0 \})$ is connected and let $a,b \in \SS$ be such that $\lambda_{a,b}>0$.
Let $f$ be an $n$-symmetric legal configuration on $\Lambda_{2kn}$. Then $f$ may be extended to a legal configuration on $(\Lambda_{2kn})^{+(dn+|\SS|)}$ having $(a,b)$-boundary conditions.
\end{lemma}

The above lemma follows easily from the following lemma.

\begin{lemma}
Any $n$-symmetric legal configuration $f$ on $\Lambda_{2kn}$ can be extended to a legal configuration on $(\Lambda_{2kn})^{+dn}$ having $(a,b)$-boundary conditions, where $a := f(0,\dots,0)$ and $b := f(1,0,\dots,0)$.
\end{lemma}
\begin{proof}
	The proof is identical to that of~\cite[Lemma~8.9]{peledspinka2018colorings} after replacing the complete graph $K_q$ there with the graph $(\SS,\{ \{i,j\} : \lambda_{i,j}>0 \})$ and replacing the condition $i \neq j$ with $i \sim j$ whenever $i,j \in \SS$.
\end{proof}

Recall the definition of $Z_*(f)$ from~\eqref{eq:Z_*-def}.

\begin{lemma}\label{lem:no-infinite-cluster-of-interface-vertices}
Suppose that $f$ is sampled from some (periodic) maximal-pressure Gibbs state. Then $Z_*(f)$ almost surely has no infinite $(\Z^d)^{\otimes 2}$-connected component.
\end{lemma}
\begin{proof}
Let $\mu$ be a maximal-pressure Gibbs state and let $f$ be sampled from $\mu$. Denote the lattice of $\mu$-preserving translations by $\Gamma$.
We call the elements of $Z_*$ interface vertices.
For a vertex~$u$, let $E_u$ be the event that $u$ belongs to an infinite $(\Z^d)^{\otimes 2}$-path of interface vertices. Since $\mu$ is $\Gamma$-periodic, $\mu(E_u)$ depends only on the $\Gamma$-equivalence class $[u]$ of $u$. Assume towards a contradiction that $\mu(E_u)>\delta$ for some $u$ and $\delta>0$. By ergodic decomposition, we may assume that $\mu$ is ergodic with respect to the $\Gamma$-action. Then by the ergodic theorem, the density of the set of vertices $v \in [u]$ for which $E_v$ occurs is $\mu(E_u)$ almost surely. In particular, $\mu(\cE_n) \to 1$ as $n \to \infty$, where $\cE_n$ is the event that at least a $\delta$-proportion of vertices $\Lambda_n$ are connected to $(\intB \Lambda_n)^{+4}$ by a $(\Z^d)^{\otimes 2}$-path of interface vertices in $\Lambda_n \setminus (\intB \Lambda_n)^{+2}$. Note that the event that a vertex $v$ is an interface vertex is measurable with respect to the values of $f$ on $v^{+3}$, and thus, $\cE_n$ is measurable with respect to the values of $f$ on $\Lambda_n$.

Define the partition function with $\tau$ boundary conditions on $B$ in domain $\Lambda$ by
\[ Z^{\tau,B}_\Lambda := \sum_{f \in \SS^\Lambda} \omega_f \cdot \1_{\{f=\tau\text{ on }B\}} .\]
Given a set $\cE \subset \SS^\Lambda$, we also denote by $Z^{\tau,B}_\Lambda(\cE)$ the above sum over elements in $\cE$.
Using that $\mu$ is a Gibbs state and \eqref{eq:entropy-chain-rule}-\eqref{eq:entropy-subadditivity} allows to write the pressure of $\mu$ as
\[ P(\mu) = \lim_{n \to \infty} \E[P(f,\Lambda_n)], \qquad\text{where }P(\tau,\Lambda) := \frac{\log Z^{\tau,\intB \Lambda}_\Lambda}{|\Lambda|} .\]
Since $\limsup_{\Lambda \to \Z^d} \max_\tau P(\tau,\Lambda) \le P_{\text{top}} = P(\mu)$, it follows that $P(f,\Lambda_n)$ converges to $P_{\text{top}}$ in probability as $n \to \infty$.
Similarly, if $Q_n(\tau) := Z^{\tau,\intB \Lambda}_{\Lambda_n}(\cE_n) / Z^{\tau,\intB \Lambda}_{\Lambda_n}$ denotes the probability of $\cE_n$ in the distribution corresponding to the model in domain $\Lambda_n$ with $\tau$ boundary conditions, then $\E[Q_n(f)] = \mu(\cE_n)$ so that $Q_n(f)$ converges to 1 in probability.
In particular, there exists a sequence of fixed (deterministic) configurations $\tau_n$ such that $P(\tau_n,\Lambda_n) \to P_{\text{top}}$ and $Q_n(\tau_n) \to 1$. We shall show that this is impossible.

The first step is to magnify the effect at a given scale $n$ by replicating it many times.
Namely, we take the model in domain $\Lambda_n$ with $\tau_n$ boundary conditions, and duplicate it to obtain a model in domain $\Lambda_{2kn}$, with each of the $(2k)^d$ shifted copies of the smaller box $\Lambda_n$ having the same boundary conditions (up to reflections). Indeed, by reflecting $\tau_n$ along the sides of the box $\Lambda_n$ some $2k-1$ number of times in each coordinate direction, we get boundary conditions $\tau_{n,k}$ defined on $B_{n,k} := n\{0,1,\dots,2k-1\}^d + \intB \Lambda_n$.
Let $\cE_{n,k}$ denote the event that at least a $\delta$-proportion of vertices in $\Lambda_{2nk}$ are connected to $B_{n,k}^{+4}$ by a $(\Z^d)^{\otimes 2}$-path of interface vertices in $\Lambda_{2nk} \setminus (\intB \Lambda_{2nk})^{+2}$. With a slight abuse of notation, we regard $\cE_{n,k}$ below as a collection of configurations on either $\Lambda_{2nk}$ or $U_{n,k} := \{-dn-|\SS|,\dots,2kn+dn+|\SS|\}^d$, according to the context.
Then
\[ \frac{\log Z^{\tau_{n,k},B_{n,k}}_{\Lambda_{2kn}}(\cE_{n,k})}{|\Lambda_{2kn}|} \ge (2k)^d \cdot \frac{\log Z^{\tau_n,\intB \Lambda_n}_{\Lambda_n}(\cE_n) - O(n^{d-1})}{{|\Lambda_{2kn}|}} \ge \frac{\log Z^{\tau_n,\intB \Lambda_n}_{\Lambda_n}(\cE_n)}{|\Lambda_n|} - o(1) \ge P_{\text{top}} - o(1).\]
By \cref{lem:pattern-extends-to-ab-boundary-conditions}, each legal configuration on $\Lambda_{2kn}$ having $\tau_{n,k}$ boundary conditions can be extended to a legal configuration on $U_{n,k}$ having $(a,b)$-boundary conditions for some $a \in A$, $b \in B$ and dominant pattern $P_{n,k}=(A,B)$. Thus,
\[ \log Z^{P_{n,k}, \intB U_{n,k}}_{U_{n,k}}(\cE_{n,k}) \ge \log Z^{\tau_{n,k},B_{n,k}}_{\Lambda_{2kn}}(\cE_{n,k}) - C_{d,\lambda} n^d k^{d-1} ,\]
where the presence of $\lambda$ in $C_{d,\lambda}$ indicates a dependence on $(\SS,(\lambda_i),(\lambda_{i,j}))$.
On the other hand,
\[ \log \Pr_{U_{n,k},P_{n,k}}(\cE_{n,k}) = \log Z^{P_{n,k}, \intB U_{n,k}}_{U_{n,k}}(\cE_{n,k}) - \log Z^{P_{n,k}, \intB U_{n,k}}_{U_{n,k}} \]
and
\[ \log Z^{P_{n,k}, \intB U_{n,k}}_{U_{n,k}} - \log Z^{\text{free}}_{\Lambda_{2kn}} \le \log Z^{\text{free}}_{U_{n,k} \setminus \Lambda_{2kn}} \le C_\lambda |U_{n,k} \setminus \Lambda_{2kn}| \le C_{d,\lambda} n^d k^{d-1} ,\]
so that
\[ P_{\text{top}} \le \frac{\log Z^{\text{free}}_{\Lambda_{2kn}}}{|\Lambda_{2kn}|} + \frac{C_{d,\lambda}}{k} + \frac{\log \Pr_{U_{n,k},P_{n,k}}(\cE_{n,k})}{|\Lambda_{2kn}|} + o(1) \qquad\text{as }n \to \infty .\]
Thus, we will arrive at a contradiction if
\[ \limsup_{k \to \infty} \limsup_{n \to \infty} \frac{\log \Pr_{U_{n,k},P_{n,k}}(\cE_{n,k})}{|\Lambda_{2kn}|} < 0 .\]
This follows from \cref{prop:Z_*-bound} as it implies that
\[ \Pr_{U_{n,k},P_{n,k}}(\cE_{n,k}) \le 2^{C_d k^d n^{d-1}} \cdot e^{-c_{d,\lambda}\delta (kn)^d} . \qedhere \]
\end{proof}

\begin{remark}
In the proof above, we implicitly used that the graph $H_{\text{pos}} := (\SS,\{ \{i,j\} : \lambda_{i,j}>0 \})$ is connected (when applying \cref{lem:pattern-extends-to-ab-boundary-conditions}). To handle the general case, note first that if $f$ is sampled from a maximal-pressure Gibbs state $\mu$, then its image $f(\Z^d)$ is almost surely contained in some connected component $T$ of $H_{\text{pos}}$. We may assume that $\mu$ is $\Gamma$-ergodic so that $T$ is almost surely constant. If $T$ contains a dominant pattern, then the proof of \cref{lem:no-infinite-cluster-of-interface-vertices} goes through unchanged. Otherwise, \cref{lem:weighted-shearer} implies that $P(\mu) \le \frac1{4d} \log Z(T^{[2d]},T)$, so that $P(\mu) \le \frac12 \log \omega_{\text{dom}} + \frac12 \gamma - \frac14 \alpha < \frac12 \log \omega_{\text{dom}}$ by~\eqref{eq:cond-non-dominant} and~\eqref{eq:alpha-cond}. Since $P_{\text{top}} \ge \tfrac12 \log \omega_{\text{dom}}$, we see that $\mu$ cannot have maximal pressure.
\end{remark}

\begin{prop}\label{prop:equilibrium-states-are-mixtures}
Every (periodic) maximal-pressure Gibbs state is a mixture of the $P$-pattern Gibbs states.
\end{prop}
\begin{proof}
Let $f$ be sampled from a Gibbs state $\mu$ under which $Z_*(f)$ almost surely has no infinite $(\Z^d)^{\otimes 2}$-connected components. In light of \cref{lem:no-infinite-cluster-of-interface-vertices}, it suffices to show that such a measure $\mu$ is a mixture of the $P$-pattern Gibbs states.

Let $U \subset \Z^d$ be finite and connected. Let us show that, almost surely, there exists a dominant pattern $P$ and a finite set $V$ containing $U$ such that $(\intB V)^+$ is in the $P$-pattern. Indeed, if we let $W$ denote the $(\Z^d)^{\otimes 2}$-connected component of $U \cup Z_*$ containing $U$, then $W$ is almost surely finite. Thus, if $V$ denotes the co-connected closure of $W^+$ with respect to infinity, then $V$ is finite, connected, co-connected and contains $U$. Since $\intextB V$ is connected by \cref{lem:int+ext-boundary-is-connected} and is contained in $\intextB W^+ = W^{+2} \setminus W$, which is disjoint from $Z_*$, it follows from the definition of $Z_*$ that $(\intextB V)^+$ is in the $P$-pattern for some~$P$.

Now consider the boxes $U_n := \{-n,\dots,n\}^d$ and let $P_n$ and $V_n$ be as above.
For a dominant pattern $P$, let $\cE_P$ be the event that $\{ n : P_n=P \}$ is infinite.
Since there are finitely many dominant patterns, $\bigcup_P \cE_P$ occurs almost surely. By a similar argument as in the proof of \cref{lem:marginal-distribution-given-agreement}, and using the fact that the finite-volume $P$-pattern measures converge, it follows that $\mu(\cdot \mid \cE_P)$ is precisely the $P$-pattern Gibbs state $\mu_P$. Thus, the events $\{\cE_P\}_P$ are disjoint and $\mu$ is the mixture $\sum_P \mu(\cE_P) \mu_P$.
\end{proof}

\subsection{Enumeration}\label{sec:enum}
In this section, we first make precise the intuition that, under $\Pr_{\Lambda,P}$, typically most vertices which are not in the $P$-pattern are singletons that are well-separated from each other. We then use this to obtain precise estimates on the partition function in finite volume and consequently of the topological pressure (defined in \cref{sec:equilibrium-states}), in particular, establishing \cref{thm:enum}.
For simplicity, we assume throughout this section that the spin system is fixed and consider asymptotics as $d \to \infty$ (with the usual $O,\Omega,\Theta,o,\omega$ notation). We also assume that $\lambda^{\text{int}}_{\text{max}}=1$ as a convenient normalization.

Fix a dominant pattern $P=(A,B)$.
Recall that $S_P(f)$ is the set of vertices in the $P$-pattern. Given $U \subset \Lambda$ and $\cE \subset \SS^\Lambda$, denote $Z^{P,U}_\Lambda(\cE) := \sum_{f \in \cE : U \subset S_P(f)} \omega_f$, and let $Z^P_\Lambda := Z^{P,\intB\Lambda}_\Lambda(\SS^\Lambda)$ be the normalization constant in the definition of $\Pr_{\Lambda,P}$.
Denote (as in \cref{sec:enum-intro})
\begin{equation}
\tilde\lambda_{A,B} := \sum_{i \in \SS \setminus A} \lambda_i \left(\sum_{b \in B} \lambda_b \lambda_{i,b} \right)^{2d} , \qquad \epsilon_{A,B} := \frac{\tilde\lambda_{A,B}}{\lambda_A \lambda_B^{2d}}, \qquad \delta_{A,B} := \frac{\epsilon_{A,B}}{1+\epsilon_{A,B}} .
\end{equation}
Recall that $\delta_{A,B}$ (respectively, $\delta_{B,A}$) is the probability that an even (respectively, odd) vertex violates the $P$-pattern given that all other vertices within distance two from it are in the $P$-pattern. Recall also that $\epsilon_{A,B}$ can be zero, but otherwise it is $e^{-c_0d(1+o(1))}$ for some constant $c_0>0$ (depending on the fixed spin system).
The following relates $\delta_{A,B}$ to the probability of violations when boundary conditions are arbitrarily far.

Denote $S:=S_P(f)$ and let $K_v$ denote the $(\Z^d)^{\otimes 2}$-connected component of $S^c$ which contains $v$.

\begin{thm}\label{thm:exact-violation-prob}
Let $\Lambda$ be a domain and let $v \in \Int(\Lambda)$ be even. Then, uniformly in $\Lambda$ and $v$,
\begin{align*}
 \Pr_{\Lambda,P}\big(|K_v|=1 \big) &= \delta_{A,B} \cdot (1 - e^{-\Theta(d)}) ,\\
 \Pr_{\Lambda,P}\big(|K_v| \ge 2 \big) &= \delta_{A,B} \cdot e^{-\Omega(d)} + e^{-\omega(d)} .
\end{align*}
When $v$ is an odd vertex, the same holds with $\delta_{B,A}$ replacing $\delta_{A,B}$.
\end{thm}

\begin{proof}
By symmetry it suffices to treat the case of even vertices.
Let us first prove the second estimate. We consider three overlapping regimes: small meaning that $2 \le |K_v|=o(d)$, moderate meaning that $|K_v|=\omega(\log^2d)$ and $\diam K_v=O(\log^2d)$, and large meaning that $\diam K_v = \omega(\log d)$. Note that this covers all possibilities.
For the large regime, we use \cref{cor:prob-of-large-violation-of-P-pattern} to obtain that $\Pr(\diam K_v \ge \omega(\log d)) \le e^{-\omega(d)}$.

Consider the moderate regime.
For any given $(\Z^d)^{\otimes 2}$-connected set $K$ containing $v$, we have
\begin{align*}
 \Pr(K_v=K)
  &= \Pr(K \subset S^c,~ K^{+2}\setminus K \subset S) \\
  &\le \frac{\Pr(K \subset S^c \mid K^{+2}\setminus K \subset S)}{\Pr(K \subset S \mid K^{+2}\setminus K \subset S)} = \frac{Z^{P,\extB K}_{K^+}(\{ K \subset S^c \})}{Z^{P,\extB K}_{K^+}(\{ K \subset S \})} .
  \end{align*}
  The denominator in the last expression equals $\lambda_A^{|\Even \cap K^+|}\lambda_B^{|\Odd\cap K^+|}$.
  For the numerator, by first considering the values at $K$ and then at $\extB K$, discarding the pair interactions between vertices of $K$ and itself in the first step, and only utilizing one (relevant) pair interaction between each vertex in $\extB K$ and its neighbors in the second step, we get the upper bound:
  \[ \lambda_{\SS \setminus A}^{|\Even \cap K|} \lambda_{\SS \setminus B}^{|\Odd \cap K|} \bar\lambda_{B,A}^{|\Even \cap \extB K|}\bar\lambda_{A,B}^{|\Odd \cap \extB K|} ,\]
  where $\bar\lambda_{A,B} := \max_{i \in \SS \setminus A} \sum_{b \in B} \lambda_b \lambda_{i,b}$ and $\bar\lambda_{B,A} := \max_{i \in \SS \setminus B} \sum_{a \in A} \lambda_a \lambda_{i,a}$. Thus,
\[ \Pr(K_v=K) \le \left( \tfrac{\lambda_{\SS\setminus A}}{\lambda_A}\right)^{|\Even \cap K|} \left(\tfrac{\lambda_{\SS\setminus B}}{\lambda_B}\right)^{|\Odd \cap K|} \left(\tfrac{\bar\lambda_{B,A}}{\lambda_A}\right)^{|\Even \cap \extB K|} \left(\tfrac{\bar\lambda_{A,B}}{\lambda_B}\right)^{|\Odd \cap \extB K|} .\]
Since every $i \in \SS \setminus B$ interacts with some $a \in A$ via a pair interaction $\lambda_{i,a}$ strictly less than $\lambda^{\text{int}}_{\text{max}}=1$, we see that $\bar\lambda_{B,A} < \lambda_A$. Similarly, $\bar\lambda_{A,B} < \lambda_B$. Thus, there exist constants $C_1,c_1>0$ such that
\[ \Pr(K_v=K) \le e^{C_1|K| - c_1|\extB K|} .\]
Since there are at most $d^{Cn}$ possible sets $K$ of size $n$ (i.e., that contain $v$ and are $(\Z^d)^{\otimes 2}$-connected), and since each such set $K$ satisfies $|\extB K| \ge \Omega(dn/\log^2d)$ when $\diam K = O(\log^2d)$ by \cref{lem:isoperimetry}, we deduce that
\[ \Pr(|K_v|=n,~\diam K_v=O(\log^2 d)) \le d^{Cn} e^{C_1n - c_1\Omega(dn/\log^2d)} \le e^{-\Omega(dn/\log^2d)} .\]
Summing over $n \ge \omega(\log^2d)$ now yields that $\Pr(|K_v|=\omega(\log^2d),~\diam K_v = O(\log^2d)) \le e^{-\omega(d)}$.

 Consider now the small regime. Observe first that $\delta_{A,B}$ must be non-zero in order for this regime to be non-trivial (in fact, in order for $K_v$ to be non-empty while not containing $N(v)$).
 We follow the same line of argument as in the moderate case, with a more careful estimate of $Z^{P,\extB K}_{K^+}(\{ K \subset S^c \})$ obtained by first considering the value at $v$, then at $N(v) \setminus K$, then at $K \setminus \{v\}$, and finally at $\extB K \setminus N(v)$, which leads to the upper bound:
  \[ \sum_{i \in \SS \setminus A} \lambda_i \left(\sum_{b \in B} \lambda_b \lambda_{i,b} \right)^{|N(v) \setminus K|} \lambda_{\SS \setminus A}^{|\Even \cap K|-1} \lambda_{\SS \setminus B}^{|\Odd \cap K|} \bar\lambda_{B,A}^{|\Even \cap \extB K|}\bar\lambda_{A,B}^{|\Odd \cap \extB K \setminus N(v)|} .\]
When $N(v) \setminus K \neq \emptyset$, this last expression is, in turn, at most
  \[ \tilde\lambda_{A,B} \cdot \underline\lambda_{A,B}^{-|K \cap N(v)|} \lambda_{\SS \setminus A}^{|\Even \cap K|-1} \lambda_{\SS \setminus B}^{|\Odd \cap K|} \bar\lambda_{B,A}^{|\Even \cap \extB K|}\bar\lambda_{A,B}^{|\Odd \cap \extB K \setminus N(v)|} ,\]
  where $\underline\lambda_{A,B} := \min \{ \sum_{b \in B} \lambda_b \lambda_{i,b} : i \in \SS \setminus A,~ \sum_{b \in B} \lambda_b \lambda_{i,b} > 0\}$ is positive since $\delta_{A,B}>0$. Thus,
\begin{align*}
  \Pr(K_v=K) \le \epsilon_{A,B} &\cdot \left(\tfrac{\lambda_B}{\underline\lambda_{A,B}}\right)^{|K \cap N(v)|} \left( \tfrac{\lambda_{\SS\setminus A}}{\lambda_A}\right)^{|\Even \cap K|-1} \left(\tfrac{\lambda_{\SS\setminus B}}{\lambda_B}\right)^{|\Odd \cap K|} \\&\cdot \left(\tfrac{\bar\lambda_{B,A}}{\lambda_A}\right)^{|\Even \cap \extB K|} \left(\tfrac{\bar\lambda_{A,B}}{\lambda_B}\right)^{|\Odd \cap \extB K \setminus N(v)|} .
\end{align*}
From here we obtain that
\[ \Pr(K_v=K) \le \epsilon_{A,B} \cdot e^{C_1(|K \cap N(v)|+|K|-1) - c_1|\extB K \setminus N(v)|} .\]
Since there are at most $d^{Cn}$ possible sets $K$ of size $n$, and since each such set $K$ satisfies $|\extB K| \ge 2dn - 2n^2$ by \cref{lem:isoperimetry}, we deduce that
\[ \Pr(|K_v|=n) \le \epsilon_{A,B} \cdot d^{Cn} e^{C_1(2n-1) - c_1(2dn-2n^2-2d)} \le \epsilon_{A,B} \cdot e^{-2c_1d(n-1)(1+o(1))} \]
for $2 \le n \le o(d)$. Summing over $n$ now yields that $\Pr(2 \le |K_v| \le o(d)) = \epsilon_{A,B} \cdot e^{-\Omega(d)}$.
Putting together the bounds for the small, moderate and large regimes yields that $\Pr(|K_v| \ge 2) = \epsilon_{A,B} \cdot e^{-\Omega(d)} + e^{-\omega(d)}$ as required, recalling that $\epsilon_{A,B}=\delta_{A,B}(1+o(1))$.

We now turn to estimating $\Pr(|K_v|=1)$. Suppose first only that $v \in \Int(\Lambda)$. We have
\[ \Pr(|K_v|=1)=\Pr(v \notin S, v^{++} \setminus \{v\} \subset S) = \Pr(v^{++} \setminus \{v\} \subset S) \cdot \Pr(v \notin S \mid v^{++} \setminus \{v\} \subset S) .\]
The second factor is precisely $\delta_{A,B}$. Thus, $\Pr(|K_v|=1) \le \delta_{A,B}$, which by what we have just shown implies that $\Pr(v \notin S) = \Pr(K_v \neq \emptyset) \le \delta_{A,B}\cdot (1+e^{-\Omega(d)}) + e^{-\omega(d)} \le e^{-\Omega(d)}$. Using that this also holds for odd vertices in $\Int(\Lambda)$, we see that the first factor above is at least $1-|v^{++}|e^{-\Omega(d)} = 1-e^{-\Omega(d)}$. Thus, $\Pr(|K_v|=1) \ge \delta_{A,B} \cdot (1-e^{-\Omega(d)})$.
\end{proof}

\begin{thm}\label{thm:finite-volume-enum}
Let $\Lambda$ be a domain with $|\Lambda \cap \Even|=|\Lambda \cap \Odd|$. Then, uniformly in $\Lambda$,
\[ \tfrac1{|\Lambda|} \log Z^P_\Lambda = \tfrac12 \log \omega_\text{dom} + \left(\epsilon_{A,B}\tfrac{|\Int(\Lambda) \cap \Even|}{|\Lambda|}+\epsilon_{B,A}\tfrac{|\Int(\Lambda) \cap \Odd|}{|\Lambda|}\right)(1 \pm e^{-\Omega(d)}) .\]
\end{thm}
\begin{proof}
We assume that $\epsilon_{A,B}$ and $\epsilon_{B,A}$ are both non-zero. The other cases require only cosmetic changes to the proof.

Sample $f$ from $\Pr_{\Lambda,P}$.
Denote $T := S^c$. Let $T_1$ be the union of the singleton $(\Z^d)^{\otimes 2}$-connected components of $T$, and let $T_2 := T \setminus T_1$. Let $S_1 := T_1^+$ and $S_2 := T_2^+$, and $S_0 := \Lambda \setminus (S_1 \cup S_2)$. Write $Z^P_\Lambda(T)$ for the partition function restricted to a particular realization of $T$ (which determines $T_1,T_2,S_0,S_1,S_2$), and observe that
\begin{equation}\label{eq:cond-partition-func}
\log Z^P_\Lambda = \Ent(T) + \E \log Z^P_\Lambda(T) .
\end{equation}
We have
\begin{align*}
 \log Z^P_\Lambda(T) &\le (\log \lambda_A)|S_0 \cap \Even| + (\log \lambda_B)|S_0 \cap \Odd| \\&\quad + (\log \tilde\lambda_{A,B})|T_1 \cap \Even| + (\log \tilde\lambda_{B,A})|T_1 \cap \Odd| \\&\quad + (\log \lambda_A)|\Even \cap S_2 \setminus T_2| + (\log \lambda_B)|\Odd \cap S_2 \setminus T_2| +(\log \lambda_\SS)|T_2| \\
 &= (\log \lambda_A) |\Even \cap \Lambda| + (\log \epsilon_{A,B}) |\Even \cap T_1| + (\log (\lambda_\SS /\lambda_A)) |\Even \cap T_2| \\&\quad + (\log \lambda_B) |\Odd \cap \Lambda| + (\log \epsilon_{B,A}) |\Odd \cap T_1| + (\log (\lambda_\SS /\lambda_B)) |\Odd \cap T_2| .
\end{align*}
Plugging this into~\eqref{eq:cond-partition-func}, noting that $(\log \lambda_A) |\Even \cap \Lambda| + (\log \lambda_B) |\Odd \cap \Lambda| = \tfrac12 (\log \omega_\text{dom}) |\Lambda|$, and using subadditivity of entropy for $\Ent(T)$, the upper bound of the theorem will follow once we show that
\[ H(p_v) + (\log \epsilon_{A,B}) p_{1,v} + (\log (\lambda_\SS /\lambda_A)) p_{2,v} \le \epsilon_{A,B} \cdot (1+e^{-\Omega(d)}) ,\]
for every even $v \in \Lambda$, where $p_{1,v} := \Pr(v \in T_1)$, $p_{2,v} := \Pr(v \in T_2)$ and $p_v := \Pr(v \in T)=p_{1,v}+p_{2,v}$, and a similar inequality for odd vertices. Here we use the notation $H(p):=-p\log p  - (1-p)\log(1-p)$.
Indeed, $p_v=0$ when $v \in \intB \Lambda$, and otherwise \cref{thm:exact-violation-prob} shows that $p_{1,v}=\epsilon_{A,B} \cdot (1-e^{-\Omega(d)})$ and $p_{2,v}=\epsilon_{A,B} \cdot e^{-\Omega(d)}$, so that the desired inequality follows using that
\begin{align*}
 H(p_v) &\le H(p_{1,v})+H(p_{2,v}) ,\\
 H(p_{1,v}) &\le p_{1,v} \log(1/p_{1,v}) + p_{1,v} \le p_{1,v}\log(1/\epsilon_{A,B})+\epsilon_{A,B} \cdot(1+e^{-\Omega(d)}),\\
 H(p_{2,v}) &\le \epsilon_{A,B} \cdot e^{-\Omega(d)} .
\end{align*}

Let us now turn to the lower bound on $Z^P_\Lambda$ for which we only take into account the contribution $\tilde Z^P_\Lambda$ from configurations having $T_2=\emptyset$.
Let $\tilde\Pr_{\Lambda,P}$ be the probability measure $\Pr_{\Lambda,P}$ conditioned on $T_2=\emptyset$. All random variables ($S,T,\dots$) are defined as before, but with respect to a configuration $f$ sampled from $\tilde\Pr_{\Lambda,P}$ (so that $T=T_1$).

Denote $\tilde p_v:=\Pr(v \in T)$. Note that it is still true that $\Pr(v \in T \mid v^{++} \setminus \{v\} \subset S) = \delta_{A,B}$ for even $v \in \Int(\Lambda)$ (and similarly with $\delta_{B,A}$ for odd $v$), whence it is straightforward to check that $\tilde p_v=\epsilon_{A,B}(1-e^{-\Omega(d)})$ for even $v \in \Int(\Lambda)$ and $\tilde p_v=\epsilon_{B,A}(1-e^{-\Omega(d)})$ for odd $v \in \Int(\Lambda)$.

Our starting point is $\log \tilde Z^P_\Lambda = \Ent(T) + \E \log Z^P_\Lambda(T)$. Since
\begin{align*}
 \log Z^P_\Lambda(T) = \tfrac12 (\log \omega_\text{dom}) |\Lambda| + (\log \epsilon_{A,B}) |\Even \cap T| + (\log \epsilon_{B,A}) |\Odd \cap T| ,
\end{align*}
then
\[ \log \tilde Z^P_\Lambda = \tfrac12 \log \omega_\text{dom}|\Lambda| + \Ent(\tilde T) + (\log \epsilon_{A,B}) \E |\Even \cap T| + (\log \epsilon_{B,A}) \E |\Odd \cap T| .\]
Using that $\Ent(T) \ge \sum_v \Ent(v \in T \mid T \setminus \{v\}) \ge \sum_v H(\tilde p_v)(1-e^{-\Omega(d)})$, it suffices to show that $H(\tilde p_v)(1-e^{-\Omega(d)}) + (\log \epsilon_{A,B})\tilde p_v \ge \epsilon_{A,B} \cdot (1-e^{-\Omega(d)})$ for even $v \in \Int(\Lambda)$, and a similar inequality for odd vertices. This follows using that $H(p)\ge p\log(1/p)+p(1-p)$.
\end{proof}

Recall that (free) topological pressure $P_{\text{top}}$ from \cref{sec:equilibrium-states}. We claim that
\begin{equation}\label{eq:partition-function-convergence-to-top-pres}
\lim_{n \to \infty} \frac{\log Z^P_{\Lambda_n}}{|\Lambda_n|} = P_{\text{top}} .
\end{equation}
Indeed, in one direction it is clear that $Z^P_{\Lambda_n} \le Z^{\text{free}}_{\Lambda_n}$. For the other direction, using \cref{lem:pattern-extends-to-ab-boundary-conditions} and reflections (as in the proof of \cref{lem:no-infinite-cluster-of-interface-vertices}), it follows that
\begin{equation}
Z^P_{\Lambda_{2kn + 2dn+2|\SS|}} \ge (Z^{\text{free}}_{\Lambda_n} / |\SS|^{|\intB \Lambda_n|})^{(2k)^d} .
\end{equation}
Taking logarithms and first the limit $n \to \infty$ and then $k \to \infty$ completes the proof of~\eqref{eq:partition-function-convergence-to-top-pres}. We note that~\eqref{eq:partition-function-convergence-to-top-pres} implies also that the (naturally defined) periodic topological pressure equals the free topological pressure $P_{\text{top}}$. It is also worth pointing out that the proof of~\eqref{eq:partition-function-convergence-to-top-pres} applies without assuming that all dominant patterns are equivalent and also for any any pattern $P$ (without assuming that it is dominant).

\cref{thm:enum} now follows from~\eqref{eq:partition-function-convergence-to-top-pres} and \cref{thm:finite-volume-enum}.

\section{Discussion and open questions}
\label{sec:discussion}

\subsection{Symmetry assumption}
Our results apply to models satisfying a certain symmetry assumption, namely, that all dominant patterns are equivalent. While we have seen that ``generic'' spin systems satisfy this symmetry condition (see \cref{sec:generic-systems}), as do many classical models of interest, some models do not. For models satisfying the symmetry condition, our non-quantitative results show that in sufficiently high dimensions each dominant pattern gives rise to an ordered Gibbs state and that any (periodic) maximal-pressure Gibbs state is a mixture of these.
What happens for models which do not satisfy the symmetry condition? It is plausible that in high dimensions only a certain subset of the dominant patterns are relevant in that only they give rise to ordered Gibbs states and that any other (periodic) maximal-pressure Gibbs state is a mixture of these. Is this the case? If so, how does one determine whether a given dominant pattern is relevant in this sense?

Let us give two examples.

\subsubsection{Homomorphisms to a path}\label{sec:hom-to-path}
Consider the spin system obtained in our framework when
\[ \SS = \{1,2,\dots,q\},\qquad \lambda_i = 1,\qquad\lambda_{i,j} = \1_{\{|i-j|=1 \}} .\]
This describes the model of homomorphisms to a path on $q$ vertices.
The model is degenerate when $q=2$ (there are only two possible configurations) and is also trivial when $q=3$ (the states at different vertices are independent given the boundary conditions).
The case $q=4$ is the hard-core model and was discussed in \cref{sec:hard-core model} (see also \cref{sec:hard-core-unequal-sublattice-activities-model}).
Suppose that $q \ge 5$. The maximal patterns are $(\{i\},\{i-1,i+1\})$ and its reversal for $2 \le i \le n-1$, and all are dominant. However, not all of the $2(q-2)$ dominant patterns are equivalent (e.g., $(\{2\},\{1,3\})$ and $(\{3\},\{2,4\})$).
One would expect that due to \emph{entropic repulsion}, the only dominant patterns which give rise to ordered Gibbs states (at least in high dimensions) are the ``central patterns'' -- the two corresponding to $i=\lceil \frac q2 \rceil$ when $q$ is odd, and the four corresponding to $i \in \{ \frac q2, \frac q2+1 \}$ when $q$ is even.

A similar situation occurs for homomorphisms to a path with loops, i.e., when
\[ \SS = \{1,2,\dots,q\},\qquad \lambda_i = 1,\qquad\lambda_{i,j} = \1_{\{|i-j| \le 1 \}} .\]
The looped $3$-path model is the Widom--Rowlinson model (at unit activity) discussed in \cref{sec:Widom-Rowlinson model} and the looped $4$-path model is the beach model (at unit activity) discussed in \cref{sec:beach model}.
We note that looped $q$-path model on $\Z^d$ is equivalent to the non-looped $(q+1)$-model on $\Z^d \times \{0,1\}$ via the mechanism described in \cref{sec:projection-systems}.

\subsubsection{Homomorphisms to a hypercube}\label{sec:hypercube}
Consider the spin system obtained in our framework when
\[ \SS = \{0,1\}^q,\qquad \lambda_i = 1,\qquad\lambda_{i,j} = \1_{\{\|i-j\|_1=1\}} .\]
This describes the model of homomorphisms to a $q$-dimensional hypercube.
The model is degenerate when $q=1$ (there are only two possible configurations) and is also trivial when $q=2$ (the states at different vertices are independent given the boundary conditions).
Suppose that $q \ge 3$.
There are two types of maximal patterns: those corresponding to a vertex and having the form $(\{v\},N(v)\}$ or its reversal for some $v \in \SS$, and those corresponding to a ``face'' and having the form $(\{v,v+e_i+e_j\},\{v+e_i,v+e_j\})$ for some $v \in \SS$ and distinct $e_i,e_j \in \SS$ such that $\|e_i\|=\|e_j\|=1$. There are $2^{q+1}$ patterns of the vertex type (two for each vertex) and $\binom q2 2^{q-1}$ patterns of the face type (two for each face).
When $q=3$, the dominant patterns are of the face type, and our results apply and show that in high dimensions each gives rise to an ordered Gibbs state.
When $q \ge 5$, the dominant patterns are of the vertex type, and our results apply and show that in high dimensions each gives rise to an ordered Gibbs state (since $\rho_{\text{pat}}^{\text{bulk}}=\frac 4q$, $\rho_{\text{pat}}^{\text{bdry}}=\frac 2q$, $|\phasemax|=2^{q+1}+\binom q2 2^{q-1}$ and $\fq=\log(1+2^q)$, condition~\eqref{eq:parameter-inequalities-simple-hom2} shows that this happens when $d \ge Cq^4$).
When $q=4$, both types of patterns are dominant so that our symmetry condition does not hold. Thus, the behavior of homomorphisms to the 4-dimensional hypercube remains an open problem.

Let us also briefly mention the related model of homomorphisms to a torus $\T_q^k$ with $k \ge 2$ and $q \ge 3$. The symmetry condition is satisfied when $k \ge 3$, but not when $k=2$. Our results thus do not directly apply for two-dimensional tori. However, we have seen in \cref{sec:product_systems} how to handle the case of $\T_3^2$, and the discussion in \cref{sec:covering systems} allows to then also handle the case $k=2$ and $q \ge 4$. The only remaining case is that of the 4-by-4 torus. In fact, $\T_4^2$ is graph-isomorphic to the 4-dimensional hypercube $\{0,1\}^4$ via the map that takes $x \in \{0,1\}^4$ to $(\psi(x_1+2x_2),\psi(x_3+2x_4)) \in \{0,1,2,3\}^2$, where $\psi$ is the permutation of $\{0,1,2,3\}$ that transposes 2 and 3. The two open problems are therefore the same.

\subsection{Dimension dependency}\label{sec:dim-dep}
Many spin systems of interest, such as those in \cref{sec:first applications} and \cref{sec:applications}, are parameterized by one or more variables (e.g., number of states, temperature and activity).
For such models, our non-quantitative results establish the existence of an ordered phase in sufficiently high dimensions.
In many cases, a disordered phase is also known to exist (e.g., via Dobrushin uniqueness), so that the model undergoes a phase transition between disordered and ordered phases (typically though a gap remains where the nature of the model is undetermined).
When only one parameter $p$ is allowed to vary and the other parameters are fixed, the sufficient condition for an ordered phase obtained from our quantitative results can be written in the form $p \ge \phi(d)$ for some function $\phi$.
In some cases, the optimal $\phi(d)$ is known to be between two powers of $d$, and determining the optimal power is an interesting problem. We now discuss this in several particular cases.

\subsubsection{Hard-core model}\label{sec:hardcore_model_open_questions}
Perhaps the most well-known and fundamental instance of this problem is to determine the largest $\alpha$ such that the hard-core model on $\Z^d$ at activity $\lambda>d^{-\alpha +o(1)}$ has multiple Gibbs states. As discussed in \cref{sec:hard-core-model}, Galvin--Kahn~\cite{galvin2004phase} were the first to establish that such an $\alpha$ exists, showing that $\alpha=1/4$ is possible (our results yield the same value), while the currently best-known value of $\alpha=1/3$ is due to Peled--Samotij~\cite{peled2014odd}. On the other hand, Dobrushin uniqueness tells us that $\alpha=1$ is the best we can hope for. It has been speculated~\cite{galvin2004phase} that $\alpha=1$ is indeed possible and this remains a big open problem.

\subsubsection{Antiferromagnetic Potts model}\label{sec:AF_Potts_open_questions}
The AF Potts model discussed in \cref{sec:AF Potts model} leads to several instances of this problem. The first instance is obtained when one varies $q$, the number of states, keeping the inverse temperature $\beta$ fixed (a canonical choice here is to take $\beta=\infty$, yielding the proper $q$-coloring model), in which case we are asking for the largest $\alpha$ such that the model has multiple Gibbs states (of maximal entropy if $\beta=\infty$) when $q \le d^{\alpha+o(1)}$. Our results show that $\alpha=1/12$ is possible (the more specialized analysis in~\cite{peledspinka2018colorings} yields $\alpha=1/10$ in the case $\beta=\infty$), and it is believed that $\alpha=1$ is possible, which would be optimal by Dobrushin uniqueness; see~\eqref{eq:Dobrushin uniqueness AF Potts} with~\eqref{eq:potts_param_ineq}.
The second instance is obtained when one fixes $q \ge 3$ and varies $\beta$, in which case we are asking for the largest $\alpha$ such that the model has multiple Gibbs states when $\beta \ge d^{-\alpha+o(1)}$. Our results show that $\alpha=1/4$ is possible, while it is believed that $\alpha=1$ is possible (which would again be optimal by Dobrushin uniqueness). An additional instance of the problem is obtained when considering the model with an external magnetic field $h$ applied to one state as discussed in \cref{sec:AF Potts with magnetic field}. Fixing $q$ and $\beta$, we now seek the largest $\alpha$ such that when $e^{\beta h} \ge d^{-\alpha+o(1)}$ the model has a Gibbs state under which the distribution of the sites in the first state is not translation invariant (e.g., the density of sites in the first state is different on the two sublattices). Our results imply that $\alpha=1/4$ is possible. There is no known upper bound for $\alpha$, but it is reasonable to expect that the best possible $\alpha$ is finite (perhaps also $1$).

\subsubsection{Beach model}\label{sec:dim-dep-beach}
As explained in \cref{sec:beach model}, the beach model undergoes an phase transition between disordered and ordered phases at a unique point $\lambda_c(d)$. We have seen in~\eqref{eq:beach_improved_lambda_c_bounds} that $|\lambda_c(d)-1| \le d^{-1/4+o(1)}$. One may seek the optimal power in this form, but we expect that the critical point is always greater than 1 and that it takes on the more precise form $\lambda_c(d) = 1 + d^{-\alpha+o(1)}$, perhaps with $\alpha=1$.

\subsubsection{Lipschitz height functions}

Our results show that $m$-Lipschitz height functions in high dimensions are localized (see~\cref{sec:clock-models}). However, the results apply only when $m$ is at most a small power of $d$~\eqref{eq:quantitative_clock_cond2}. What happens for larger $m$? By analogy with other height function models, it seems reasonable to predict that the model is localized for all $m\ge 1$ and $d\ge 3$, and has at least countably many extremal (translation-invariant) maximal-entropy Gibbs states (obtained in the thermodynamic limit from constant boundary conditions). Nonetheless, it may be that when $m$ is large in comparison to $d$, certain qualitative features of these Gibbs states differ from the features implied by our main results. For instance, when our results apply, one can conclude that the scaling limit of the height function (in one of the translation-invariant Gibbs states obtained from our results) is white noise. Is this still the case when $m$ is large compared with $d$? In the related context of the integer-valued Gaussian free field, G\"opfert and Mack~\cite{gopfert1982proof} discuss possible qualitative differences between high and low coupling constants which may be of similar nature.

\subsubsection{General spin systems}

In light of the above examples, it is a natural problem to try and improve the dependence on $d$ in our quantitative results. Limited improvements to the powers of $d$ in our quantitative conditions may be obtained by a more careful analysis of the places where the parameters $\rho_{\text{pat}}^{\text{bulk}}$ and $\rho_{\text{pat}}^{\text{bdry}}$ enter (something in this spirit was done in~\cite{peledspinka2018colorings} for the proper coloring model, leading to the power $1/10$ instead of $1/12$ as mentioned in \cref{sec:AF_Potts_open_questions}). However, to obtain further improvements it seems necessary to improve the dependence on $d$ in the approximations to breakups and odd cutsets discussed in \cref{sec:approx} (partial progress in this direction has been made in~\cite{peled2014odd} which led to the power $1/3$ instead of $1/4$ for the hard-core model as mentioned in \cref{sec:hardcore_model_open_questions}).

\subsection{Infinite spin space}
The general framework considered in this paper (introduced in \cref{sec:the model}) allows only for spin systems with a finite spin space $\SS$. While we have seen in \cref{sec:covering systems} that our results may sometimes be used to deduce results for spin systems with an infinite spin space (by first applying them to a spin system with a finite spin space, and then transferring the results to the desired spin system), it is reasonable to try to extend the results themselves to allow for spin systems with countably infinite spin spaces.
Our results and methods of proof should indeed permit some extensions of this type, though it is unclear what the extent of this would be.
For instance, when is a $m$-Lipschitz ($m \ge 1$) height function model with ``soft constraints'' localized? By the latter model we mean a spin system of the form
\[ \SS=\Z,\qquad \lambda_i=1,\qquad \lambda_{i,j} = f(|i-j|), \]
for a function $f$ satisfying that $f(r)=1$ when $0 \le r \le m$ and $0<f(r)<1$ otherwise (and with $f(r)$ decaying sufficiently fast so that the model is well defined).

\subsection{Gibbs states}
Two corollaries of our results are that when our assumptions are satisfied:
\begin{itemize}
 \item Every (periodic) maximal-pressure Gibbs state is invariant to parity-preserving translations.
 \item There are finitely many (translation-invariant) ergodic maximal-pressure Gibbs states.
\end{itemize}
 Is this always the case for the spin systems discussed in this paper (with no assumption on the parameters of the model or on the dimension)? We note that, even when our assumptions are satisfied, the first statement may fail for zero-pressure Gibbs states (e.g., periodic frozen proper $3$-colorings); could it also fail for positive but non-maximal pressure (periodic) Gibbs states?

\medskip\noindent\textbf{Data Availability Statement.}
We declare that the manuscript has no associated data and hence no data set has been used in the realization of this work.

\medskip\noindent\textbf{Conflict of interest.}
There are no financial or non-financial conflicts of interest.

\bibliographystyle{amsplain}
\bibliography{biblio}

\providecommand{\bysame}{\leavevmode\hbox to3em{\hrulefill}\thinspace}
\providecommand{\MR}{\relax\ifhmode\unskip\space\fi MR }
\providecommand{\MRhref}[2]{%
  \href{http://www.ams.org/mathscinet-getitem?mr=#1}{#2}
}
\providecommand{\href}[2]{#2}
\begin{thebibliography}{10}

\bibitem{Aiz94}
Michael Aizenman, \emph{On the slow decay of {${\rm O}(2)$} correlations in the
  absence of topological excitations: remark on the {P}atrascioiu-{S}eiler
  model}, J. Statist. Phys. \textbf{77} (1994), no.~1-2, 351--359. \MR{1300539}

\bibitem{alon2019mixing}
Noga Alon, Raimundo Brice{\~n}o, Nishant Chandgotia, Alexander Magazinov, and
  Yinon Spinka, \emph{Mixing properties of colourings of the{${\bf Z}^d$}
  lattice}, Combinatorics, Probability and Computing \textbf{30} (2021), no.~3,
  360--373.

\bibitem{balister2007counting}
PN~Balister and B~Bollob{\'a}s, \emph{Counting regions with bounded surface
  area}, Communications in mathematical physics \textbf{273} (2007), no.~2,
  305--315.

\bibitem{banavar1980ordering}
Jayanth~R. Banavar, Gary~S. Grest, and David Jasnow, \emph{Ordering and phase
  transitions in antiferromagnetic {P}otts models}, Physical Review Letters
  \textbf{45} (1980), no.~17, 1424--1428.

\bibitem{van1994disagreement}
Jacob van~den Berg and Christian Maes, \emph{Disagreement percolation in the
  study of {M}arkov fields}, The Annals of Probability \textbf{22} (1994),
  no.~2, 749--763.

\bibitem{van1994percolation}
Jacob van~den Berg and Jeffrey~E Steif, \emph{Percolation and the hard-core
  lattice gas model}, Stochastic Processes and their Applications \textbf{49}
  (1994), no.~2, 179--197.

\bibitem{berker1980ground}
AN~Berker and Leo~P Kadanoff, \emph{Ground-state entropy and algebraic order at
  low temperatures}, Journal of Physics A: Mathematical and General \textbf{13}
  (1980), no.~7, L259.

\bibitem{Bol06}
B{\'e}la Bollob{\'a}s, \emph{The art of mathematics: Coffee time in {M}emphis},
  Cambridge University Press, 2006.

\bibitem{brightwell1999nonmonotonic}
Graham~R Brightwell, Olle H{\"a}ggstr{\"o}m, and Peter Winkler,
  \emph{Nonmonotonic behavior in hard-core and {W}idom--{R}owlinson models},
  Journal of statistical physics \textbf{94} (1999), no.~3-4, 415--435.

\bibitem{burton1994non}
Robert Burton and Jeffrey~E Steif, \emph{Non-uniqueness of measures of maximal
  entropy for subshifts of finite type}, Ergodic Theory and Dynamical Systems
  \textbf{14} (1994), no.~02, 213--235.

\bibitem{burton1995new}
\bysame, \emph{New results on measures of maximal entropy}, Israel Journal of
  Mathematics \textbf{89} (1995), no.~1-3, 275--300.

\bibitem{burton1995quite}
Robert~M Burton and Jeffrey~E Steif, \emph{Quite weak {B}ernoulli with
  exponential rate and percolation for random fields}, Stochastic processes and
  their applications \textbf{58} (1995), no.~1, 35--55.

\bibitem{carstens2012percolation}
Sebastian~Maurice Carstens, \emph{Percolation analysis of the two-dimensional
  {W}idom--{R}owlinson lattice model}, Ph.D. thesis, lmu, 2012.

\bibitem{chandgotia2017four}
Nishant Chandgotia, \emph{Four-cycle free graphs, height functions, the pivot
  property and entropy minimality}, Ergodic Theory and Dynamical Systems
  \textbf{37} (2017), no.~4, 1102--1132.

\bibitem{chandgotia2018delocalization}
Nishant Chandgotia, Ron Peled, Scott Sheffield, and Martin Tassy,
  \emph{Delocalization of uniform graph homomorphisms from $\mathbb{Z}^2$ to
  $\mathbb{Z}$}, Communications in Mathematical Physics \textbf{387} (2021),
  no.~2, 621--647.

\bibitem{chayes1995aggregation}
L~Chayes, R~Kotecky, and SB~Shlosman, \emph{Aggregation and intermediate phases
  in dilute spin systems}, Communications in mathematical physics \textbf{171}
  (1995), no.~1, 203--232.

\bibitem{chung1986some}
Fan~RK Chung, Ronald~L Graham, Peter Frankl, and James~B Shearer, \emph{Some
  intersection theorems for ordered sets and graphs}, Journal of Combinatorial
  Theory, Series A \textbf{43} (1986), no.~1, 23--37.

\bibitem{cohen2017rarity}
Omri Cohen-Alloro and Ron Peled, \emph{Rarity of extremal edges in random
  surfaces and other theoretical applications of cluster algorithms}, Annals of
  Applied Probability \textbf{30} (2020), no.~5, 2439--2464.

\bibitem{dobrushin1985phase}
RL~Dobrushin, J~Kolafa, and SB~Shlosman, \emph{Phase diagram of the
  two-dimensional {I}sing antiferromagnet (computer-assisted proof)},
  Communications in Mathematical Physics \textbf{102} (1985), no.~1, 89--103.

\bibitem{Dobrushin1968TheDe}
Roland~L. Dobrushin, \emph{The description of a random field by means of
  conditional probabilities and conditions of its regularity}, Theor. Probab.
  Appl. \textbf{13} (1968), 197--224.

\bibitem{dobrushin1968problem}
\bysame, \emph{The problem of uniqueness of a {G}ibbsian random field and the
  problem of phase transitions}, Functional analysis and its applications
  \textbf{2} (1968), no.~4, 302--312.

\bibitem{duminil2017lectures}
Hugo Duminil-Copin, \emph{Lectures on the {I}sing and {P}otts models on the
  hypercubic lattice}, PIMS-CRM Summer School in Probability, Springer, 2017,
  pp.~35--161.

\bibitem{duminil2019logarithmic}
Hugo Duminil-Copin, Matan Harel, Benoit Laslier, Aran Raoufi, and Gourab Ray,
  \emph{Logarithmic variance for the height function of square-ice},
  Communications in Mathematical Physics \textbf{396} (2022), no.~2, 867--902.

\bibitem{engbers2012h2}
John Engbers and David Galvin, \emph{H-coloring tori}, Journal of Combinatorial
  Theory, Series B \textbf{102} (2012), no.~5, 1110--1133.

\bibitem{engbers2012h1}
\bysame, \emph{H-colouring bipartite graphs}, Journal of Combinatorial Theory,
  Series B \textbf{102} (2012), no.~3, 726--742.

\bibitem{feldheim2013rigidity}
Ohad~N. Feldheim and Ron Peled, \emph{Rigidity of 3-colorings of the discrete
  torus}, Annales de l'Institut Henri Poincar{\'e}, Probabilit{\'e}s et
  Statistiques, vol.~54, Institut Henri Poincar{\'e}, 2018, pp.~952--994.

\bibitem{feldheim2016growth}
Ohad~N. Feldheim and Yinon Spinka, \emph{The growth constant of odd cutsets in
  high dimensions}, Combinatorics, Probability and Computing (2017), 1--20.

\bibitem{feldheim2015long}
\bysame, \emph{Long-range order in the 3-state antiferromagnetic {P}otts model
  in high dimensions}, Journal of the European Mathematical Society \textbf{21}
  (2019), no.~5, 1509--1570.

\bibitem{friedli2017statistical}
Sacha Friedli and Yvan Velenik, \emph{Statistical mechanics of lattice systems:
  a concrete mathematical introduction}, Cambridge University Press, 2017.

\bibitem{frohlich1980phase}
J{\"u}rg Fr{\"o}hlich, Robert~B Israel, Elliott~H Lieb, and Barry Simon,
  \emph{Phase transitions and reflection positivity. {II}. {L}attice systems
  with short-range and {C}oulomb interactions}, Journal of Statistical Physics
  \textbf{22} (1980), no.~3.

\bibitem{Galvin2003hammingcube}
David Galvin, \emph{On homomorphisms from the {H}amming cube to {${\bf Z}$}},
  Israel J. Math. \textbf{138} (2003), 189--213.

\bibitem{galvin2006bounding}
\bysame, \emph{Bounding the partition function of spin-systems}, JOURNAL OF
  COMBINATORICS \textbf{13} (2006), no.~3, R72.

\bibitem{galvin2004phase}
David Galvin and Jeff Kahn, \emph{On phase transition in the hard-core model on
  {${\bf Z}^d$}}, Combinatorics, Probability and Computing \textbf{13} (2004),
  no.~02, 137--164.

\bibitem{galvin2012phase}
David Galvin, Jeff Kahn, Dana Randall, and Gregory Sorkin, \emph{Phase
  coexistence and torpid mixing in the 3-coloring model on {${\bf Z}^d$}}, SIAM
  Journal on Discrete Mathematics \textbf{29} (2015), no.~3, 1223--1244.

\bibitem{galvin2011multistate}
David Galvin, Fabio Martinelli, Kavita Ramanan, and Prasad Tetali, \emph{The
  multistate hard core model on a regular tree}, SIAM Journal on Discrete
  Mathematics \textbf{25} (2011), no.~2, 894--915.

\bibitem{galvin2007torpid}
David Galvin and Dana Randall, \emph{Torpid mixing of local {M}arkov chains on
  3-colorings of the discrete torus}, Proceedings of the eighteenth annual
  ACM-SIAM symposium on Discrete algorithms, Society for Industrial and Applied
  Mathematics, 2007, pp.~376--384.

\bibitem{galvin2004weighted}
David Galvin and Prasad Tetali, \emph{On weighted graph homomorphisms}, DIMACS
  Series in Discrete Mathematics and Theoretical Computer Science \textbf{63}
  (2004), 97--104.

\bibitem{georgii2011gibbs}
Hans-Otto Georgii, \emph{Gibbs measures and phase transitions}, vol.~9, Walter
  de Gruyter, 2011.

\bibitem{GeoHagMae01}
Hans-Otto Georgii, Olle H{{\"a}}ggstr{{\"o}}m, and Christian Maes, \emph{The
  random geometry of equilibrium phases}, Phase transitions and critical
  phenomena, {V}ol. 18, Phase Transit. Crit. Phenom., vol.~18, Academic Press,
  San Diego, CA, 2001, pp.~1--142. \MR{2014387 (2004h:82022)}

\bibitem{gopfert1982proof}
Markus G{\"o}pfert and Gerhard Mack, \emph{Proof of confinement of static
  quarks in 3-dimensional {U}(1) lattice gauge theory for all values of the
  coupling constant}, Communications in Mathematical Physics \textbf{82}
  (1982), no.~4, 545--606.

\bibitem{grimmett2004random}
Geoffrey Grimmett, \emph{The random-cluster model}, Probability on discrete
  structures, Springer, 2004, pp.~73--123.

\bibitem{haggstrom1996phase}
Olle H{\"a}ggstr{\"o}m, \emph{On phase transitions for subshifts of finite
  type}, Israel Journal of Mathematics \textbf{94} (1996), no.~1, 319--352.

\bibitem{haggstrom1997ergodicity}
\bysame, \emph{Ergodicity of the hard-core model on {${\bf Z}^2$} with
  parity-dependent activities}, Arkiv f{\"o}r Matematik \textbf{35} (1997),
  no.~1, 171--184.

\bibitem{haggstrom1998random}
\bysame, \emph{Random-cluster analysis of a class of binary lattice gases},
  Journal of statistical physics \textbf{91} (1998), no.~1-2, 47--74.

\bibitem{haggstrom2002monotonicity}
\bysame, \emph{A monotonicity result for hard-core and {W}idom--{R}owlinson
  models on certain $ d $-dimensional lattices}, Electronic Communications in
  Probability \textbf{7} (2002), 67--78.

\bibitem{hallberg2004gibbs}
Per Hallberg, \emph{Gibbs measures and phase transitions in {P}otts and beach
  models}, Ph.D. thesis, Matematik, 2004.

\bibitem{higuchi1982coexistence}
Yasunari Higuchi, \emph{Coexistence of the infinite (*) clusters:—{A} remark
  on the square lattice site percolation}, Zeitschrift f{\"u}r
  Wahrscheinlichkeitstheorie und Verwandte Gebiete \textbf{61} (1982), no.~1,
  75--81.

\bibitem{higuchi1983applications}
\bysame, \emph{Applications of a stochastic inequality to two-dimensional
  {I}sing and {W}idom-{R}owlinson models}, Probability Theory and Mathematical
  Statistics, Springer, 1983, pp.~230--237.

\bibitem{higuchi2004some}
Yasunari Higuchi and Masato Takei, \emph{Some results on the phase structure of
  the two-dimensional {W}idom--{R}owlinson model}, Osaka Journal of Mathematics
  \textbf{41} (2004), no.~2, 237--255.

\bibitem{imrich1984explicit}
Wilfried Imrich, \emph{Explicit construction of regular graphs without small
  cycles}, Combinatorica \textbf{4} (1984), no.~1, 53--59.

\bibitem{jenssen2020homomorphisms}
Matthew Jenssen and Peter Keevash, \emph{Homomorphisms from the torus},
  Advances in Mathematics \textbf{430} (2023), 109212.

\bibitem{kahn2001entropy}
Jeff Kahn, \emph{An entropy approach to the hard-core model on bipartite
  graphs}, Combinatorics, Probability and Computing \textbf{10} (2001), no.~03,
  219--237.

\bibitem{Kahn2001hypercube}
\bysame, \emph{Range of cube-indexed random walk}, Israel J. Math. \textbf{124}
  (2001), 189--201.

\bibitem{kahn1999generalized}
Jeff Kahn and Alexander Lawrenz, \emph{Generalized rank functions and an
  entropy argument}, Journal of Combinatorial Theory, Series A \textbf{87}
  (1999), no.~2, 398--403.

\bibitem{kahn2020number}
Jeff Kahn and Jinyoung Park, \emph{The number of 4-colorings of the {H}amming
  cube}, Israel Journal of Mathematics (2020), 1--21.

\bibitem{kelly1991loss}
Frank~P Kelly, \emph{Loss networks}, The Annals of Applied Probability (1991),
  319--378.

\bibitem{korshunov1981number}
Aleksej~D. Korshunov, \emph{On the number of monotone {B}oolean functions},
  Problemy Kibernetiki \textbf{38} (1981), 5--108.

\bibitem{Korshunov1983Th}
Aleksej~D. Korshunov and Alexander.~A. Sapozhenko, \emph{The number of binary
  codes with distance 2}, Problemy Kibernet (Russian) \textbf{40} (1983),
  no.~1, 111--130.

\bibitem{kotecky1985long}
Roman Koteck{\`y}, \emph{Long-range order for antiferromagnetic {P}otts
  models}, Physical Review B \textbf{31} (1985), no.~5, 3088.

\bibitem{kotecky2014entropy}
Roman Koteck{\`y}, Alan~D. Sokal, and Jan~M. Swart, \emph{Entropy-driven phase
  transition in low-temperature antiferromagnetic {P}otts models}, Comm. in
  Math. Phys. \textbf{330} (2014), no.~3, 1339--1394.

\bibitem{lanford1969observables}
OE~Lanford and David Ruelle, \emph{Observables at infinity and states with
  short range correlations in statistical mechanics}, Communications in
  Mathematical Physics \textbf{13} (1969), no.~3, 194--215.

\bibitem{lebowitz1998improved}
J.~L. Lebowitz and A.~E. Mazel, \emph{Improved {P}eierls argument for
  high-dimensional {I}sing models}, J. Statist. Phys. \textbf{90} (1998),
  no.~3-4, 1051--1059. \MR{1616958}

\bibitem{lebowitz1971phase}
JL~Lebowitz and G~Gallavotti, \emph{Phase transitions in binary lattice gases},
  Journal of Mathematical Physics \textbf{12} (1971), no.~7, 1129--1133.

\bibitem{lovasz1975ratio}
L{\'a}szl{\'o} Lov{\'a}sz, \emph{On the ratio of optimal integral and
  fractional covers}, Discrete mathematics \textbf{13} (1975), no.~4, 383--390.

\bibitem{mazel1991random}
AE~Mazel and Yu~M Suhov, \emph{Random surfaces with two-sided constraints: an
  application of the theory of dominant ground states}, Journal of statistical
  physics \textbf{64} (1991), no.~1, 111--134.

\bibitem{mceliece2002theory}
Robert McEliece, \emph{The theory of information and coding}, vol.~3, Cambridge
  University Press, 2002.

\bibitem{meyerovitch2014independence}
Tom Meyerovitch and Ronnie Pavlov, \emph{On independence and entropy for
  high-dimensional isotropic subshifts}, Proc. London Math. Soc. (2014),
  pdu029.

\bibitem{misiurewicz1976short}
Michael Misiurewicz, \emph{A short proof of the variational principle for a
  {${\bf Z}_+^N$}-action on a compact space}, Ast{\'e}risque \textbf{40}
  (1975), 147--157.

\bibitem{peled2010high}
Ron Peled, \emph{High-dimensional {L}ipschitz functions are typically flat},
  The Annals of Probability \textbf{45} (2017), no.~3, 1351--1447.

\bibitem{peled2014odd}
Ron Peled and Wojciech Samotij, \emph{Odd cutsets and the hard-core model on
  {${\bf Z}^d$}}, Annales de l'Institut Henri Poincar{\'e}, Probabilit{\'e}s et
  Statistiques, vol.~50, Institut Henri Poincar{\'e}, 2014, pp.~975--998.

\bibitem{peled2017condition}
Ron Peled and Yinon Spinka, \emph{A condition for long-range order in discrete
  spin systems with application to the antiferromagnetic potts model},
  arXiv:1712.03699 (2017).

\bibitem{peled2019lectures}
\bysame, \emph{Lectures on the spin and loop $o(n)$ models}, Sojourns in
  Probability Theory and Statistical Physics-I, Springer, 2019, pp.~246--320.

\bibitem{peled2020three}
\bysame, \emph{Three lectures on random proper colorings of {${\bf Z}^d$}},
  arXiv:2001.11566 (2020).

\bibitem{peledspinka2018colorings}
\bysame, \emph{Rigidity of proper colorings of {${\bf Z}^d$}}, Inventiones
  mathematicae \textbf{232} (2023), no.~1, 79--162.

\bibitem{pinson1998slow}
Hare~T Pinson, \emph{A slow decay of a connectivity function in a broad class
  of sos models}, Nuclear Physics B \textbf{525} (1998), no.~3, 664--670.

\bibitem{pirogov1975phase}
Sergey~Anatol'evich Pirogov and Ya~G Sinai, \emph{Phase diagrams of classical
  lattice systems}, Theoretical and Mathematical Physics \textbf{25} (1975),
  no.~3, 1185--1192.

\bibitem{pirogov1976phase}
\bysame, \emph{Phase diagrams of classical lattice systems continuation},
  Theoretical and Mathematical Physics \textbf{26} (1976), no.~1, 39--49.

\bibitem{racz1980phase}
Zolt\`an R\`acz, \emph{Phase boundary of {I}sing antiferromagnets near $h=h_c$
  and $t=0$: {R}esults from hard-core lattice gas calculations}, Physical
  Review B \textbf{21} (1980), no.~9, 4012.

\bibitem{ray2020proper}
Gourab Ray and Yinon Spinka, \emph{Proper-colorings of {${\bf Z}^2$} are
  {B}ernoulli}, Ergodic Theory and Dynamical Systems \textbf{43} (2023), no.~6,
  2002--2027.

\bibitem{runnels1974phase}
LK~Runnels and JL~Lebowitz, \emph{Phase transitions of a multicomponent
  {W}idom-{R}owlinson model}, Journal of Mathematical Physics \textbf{15}
  (1974), no.~10, 1712--1717.

\bibitem{runnels1976analyticity}
\bysame, \emph{Analyticity of a hard-core multicomponent lattice gas}, Journal
  of Statistical Physics \textbf{14} (1976), no.~6, 525--533.

\bibitem{salas1997absence}
Jes{\'u}s Salas and Alan~D Sokal, \emph{Absence of phase transition for
  antiferromagnetic {P}otts models via the {D}obrushin uniqueness theorem},
  Journal of Statistical Physics \textbf{86} (1997), no.~3-4, 551--579.

\bibitem{sapozhenko1989number}
Aleksandr~Antonovich Sapozhenko, \emph{The number of antichains in ranked
  partially ordered sets}, Diskretnaya Matematika \textbf{1} (1989), no.~1,
  74--93.

\bibitem{sapozhenko1987onthen}
Alexander.~A. Sapozhenko, \emph{On the number of connected subsets with given
  cardinality of the boundary in bipartite graphs}, Metody Diskretnogo Analiza
  (Russian) \textbf{45} (1987).

\bibitem{sapozhenko1991number}
Alexander~A. Sapozhenko, \emph{On the number of antichains in multilevelled
  ranked posets}, Discrete Math. Appl. \textbf{1} (1991), no.~2, 149--170.

\bibitem{timar2013boundary}
{\'A}d{\'a}m Tim{\'a}r, \emph{Boundary-connectivity via graph theory},
  Proceedings of the American Mathematical Society \textbf{141} (2013), no.~2,
  475--480.

\bibitem{van1993uniqueness}
Jacob van~den Berg, \emph{A uniqueness condition for {G}ibbs measures, with
  application to the 2-dimensional {I}sing antiferromagnet}, Communications in
  mathematical physics \textbf{152} (1993), no.~1, 161--166.

\bibitem{yamagata1995absence}
Atsushi Yamagata, \emph{Absence of re-entrant phase transition of the
  antiferromagnetic {I}sing model on the simple cubic lattice: {M}onte {C}arlo
  study of the hard-sphere lattice gas}, Physica A: Statistical Mechanics and
  its Applications \textbf{215} (1995), no.~4, 511--517.

\end{thebibliography}

\end{document}